\documentclass{article}
\usepackage[utf8]{inputenc}
\usepackage{xcolor}
\usepackage{amsfonts, amsmath,  amssymb, mathrsfs, amscd, amsthm}
\usepackage{verbatim}
\usepackage{graphicx}
\usepackage{ulem} 
\usepackage{mathtools}
\usepackage{upgreek}
\usepackage{enumitem}
\usepackage{tikz-cd}
\usepackage[all]{xy}
\usepackage[nonumberlist,sort=def]{glossaries}
\usepackage{nicefrac}
\usepackage{faktor} 
\usepackage{MnSymbol}
\usepackage{transparent}
\usepackage{makebox}
\usepackage{stackrel}

\graphicspath{{Diagrams/}{Figs2/}{Figs/}{Figs3/}}

\newcommand{\C}{{\mathbb C}}
\newcommand{\Z}{{\mathbb Z}}
\newcommand{\R}{{\mathbb R}}
\newcommand{\GG}{{\mathcal G}}
\newcommand{\HH}{{\mathcal H}}

\newcommand{\mM}{{\mathsf M}}

\newcommand{\CC}{{\mathcal{ C}}}
\newcommand{\MM}{\underline{M}}  

\newcommand{\Homf}{\mathrm{Hom}}

\newcommand{\Hom}[1][M]{\mathrm{Homeo}_{#1}}
\newcommand{\Homempty}{\mathrm{Homeo}_M(\emptyset,\emptyset)}
\newcommand{\Homn}{\mathrm{Homeo}_M(N,N')}

\newcommand{\sh}[1]{\mathfrak{#1}}

\newcommand{\shmor}[4]{\mathfrak{#1}
\colon #3 \stackrel{\lcurvearrowright}{} #4}
\newcommand{\shmot}[4]{#1
\colon #3 \stackrel{\lcurvearrowright}{} #4}

\newcommand{\mcg}[1][M]{\mathrm{MCG}_{#1}}
\newcommand{\mcgfix}[2]{\mathrm{MCG}_{#1}^{#2}}

\newcommand{\mc}[1]{\mathfrak{#1}}

\newcommand{\Motion}{\mathrm{Mot}}
\newcommand{\Mot}[1][M]{\Motion_{#1}}

\newcommand{\Q} {\mathbb{Q}}
\newcommand{\N} {\mathbb{N}}

\newcommand{\Topo}{{\mathbf{Top}}}
\newcommand{\TOPO}{{\mathbf{TOP}}}
\newcommand{\Vect}{{\mathbf{Vect}}}

\newtheorem{theorem}{Theorem}[section]
\newtheorem{proposition}[theorem]{Proposition}
\newtheorem{lemma}[theorem]{Lemma}
\newtheorem{corollary}[theorem]{Corollary}
\newtheorem{conjecture}[theorem]{Conjecture}
\newcommand{\prop}[1]{\begin{proposition} #1 \end{proposition}}
\newcommand{\lemm}[1]{\begin{lemma} #1 \end{lemma}}

\theoremstyle{definition}
\newtheorem{defin}[theorem]{Definition}
\newtheorem{example}[theorem]{Example}
\newcommand{\defn}[1]{\begin{defin} #1 \end{defin}}
\newcommand{\exa}[1]{\begin{example} #1 \end{example}}
\newtheorem{remark}[theorem]{Remark}
\newcommand{\rem}[1]{\begin{remark} #1 \end{remark}}
\newtheorem{para}[theorem]{}

\newcommand{\ali}[1]{\begin{align*} #1 \end{align*}}
\newcommand{\beq}{\begin{equation}}
\newcommand{\eq}{\end{equation}} 

\newcommand{\Power}{{\mathcal P}}
 
\newcommand{\stationary}[1][N]{$#1$-stationary}

\newcommand{\tauco}[2]{\tau^{co}_{#1#2}} 

\newcommand{\II}{{\mathbb I}} 
\newcommand{\Id}{{\mathrm{Id}}}  
\newcommand{\id}{{\mathrm{id}}}  

\newcommand{\Path}{{\mathfrak P}}
\newcommand{\too}{\rcurvearrowright}    

\newcommand{\Mtc}{\mathbf{Mt}_M}
\newcommand{\Mtcmag}{\mathrm{Mt}_M}
\newcommand{\Mtcmagmov}{\mathrm{Mt}_M^{mov}}
\newcommand{\Mtcmaghom}{\mathrm{Mt}_M^{hom}}
\newcommand{\Mtmagdot}{\Mtcmag^{\dot{\phantom{x}}}}
\newcommand{\Mtmagstar}{\Mtcmag^{*}}
\newcommand{\Motc}[3]{\mathrm{Mt}_{#1}({#2},{#3})}
\newcommand{\MotAc}[3]{\mathrm{Mt}^A_{#1}({#2},{#3})}
\newcommand{\mot}[4]{#1
\colon #3\too #4}

 \newcommand{\simp}{\stackrel{p}{\sim}}   
\newcommand{\classp}[1]{[#1]_{\!\mbox{\tiny p}}}

\newcommand{\rel}[4]{(#1^{#3}_{#2}\!#4)} 

\newcommand{\simrp}{\stackrel{\scriptscriptstyle{rp}}{\sim}}  
\newcommand{\classrp}[1]{[#1]_{\!\mbox{\tiny rp}}} 
 
\newcommand{\simi}{\stackrel{i}{\sim}}  

\newcommand{\simiA}{\stackrel{i_A}{\sim}}   

\newcommand{\classi}[1]{[#1]_{\!\mbox{\tiny i}}}

\newcommand{\classiA}[1]{
[#1]_
{\!\mbox{\tiny{i}}_{\tiny{A}}
}
}

\newcommand{\simm}{\stackrel{m}{\sim}}   
\newcommand{\classm}[1]{[#1]_{\!\mbox{\tiny m}}}

\newcommand{\simfk}{\stackrel{fk}{\sim}} 

\newcommand{\classfk}[1]{[#1]_{\!\mbox{\tiny fk}}}

\newcommand{\simww}{\stackrel{w}{\sim}}   
\newcommand{\classww}[1]{[#1]_{\!\mbox{\tiny w}}}

\newcommand{\premo}[1]{\mathrm{Flow}_{#1}}
\newcommand{\premot}{flow}

\newcommand{\Premot}{Flow}
\newcommand{\premots}{\premo{M}}

\newcommand{\premotsmov}[1][M]{\premo{#1}^{mov}}
\newcommand{\premotshom}{\premo{M}^{hom}}

\newcommand{\axiomM}{manifold}

\newcommand{\setstat}[1][M]{\overline{\mathrm{SetStat}}_{#1}}

\newcommand{\Set}{{\mathsf{Set}}} 

\newcommand{\coball}[4]{B_{#1#2}(#3,#4)}

\newcommand{\F}{\mathsf{F}} 
\newcommand{\FA}{\mathsf{F}} 

\newcommand{\acts}{\triangleright}

\newcommand{\agrpd}[3]{#3/\!/_#2 \, #1}

\newcommand{\W}{\mathbf{W}}

\newcommand{\intt}[1]{\mathrm{int}(#1)}

\newcommand{\fakemot}[3]{#1\colon #2 \not \too #3}
\newcommand{\FMotmag}[1]{\mathrm{FMt}_{#1}}
\newcommand{\FMot}[1]{\mathrm{FMot}_{#1}}

\newcommand{\Braid}{{\mathbf{B}}}

\newcommand{\fakemotion}{fake motion} 
\newcommand{\dahm}[3]{\Topo_{\Braid{},#2}[#1,#3]}
\newcommand{\dahmA}[3]{\Topo_{\Braid{},#2}^A[#1,#3]}
\newcommand{\dahmfix}[4]{\Topo_{\Braid{},#2}^{#4}[#1,#3]}

\newcommand{\T}{\mathbf{T}}

\newcommand{\braid}{\mathbf{B}_n}

\newcommand{\Artinb}{Artin braid}

\newcommand{\ppm}[1]{#1} 

\DeclareMathVersion{normal2}

\usepackage{hyperref}

\makeglossaries

\newglossaryentry{powerset}{
	name = $\Power M$,
	description = {The power set of a set $M$}}
\newglossaryentry{top}
{
	name = $\Topo$, 
	description={The category of topological spaces and continuous maps}
}

\newglossaryentry{set}
{
	name = $\Set$, 
	description={The category of sets and functions between sets}
}

\newglossaryentry{p equivalence}
{
	name = $\simp$, 
	description={Indicates paths related by homotopy relative to end-points, see Definition \ref{de:pe}}
}
\newglossaryentry{p class}
{
	name = $\classp{\gamma}$, 
	description={Equivalence class of a paths up to path-equivalence, see Definition \ref{de:pe}}
}

\newglossaryentry{co top}{
	name = $\tauco{X}{Y}$, 
	description={The compact-open topology on the set $\Topo(X,Y)$, see Definition~\ref{de:compact-open}}
}

\newglossaryentry{bijections in Top}{
	name = {$\Topo^h(M,M)$}, 
	description={The submonoid of $\Topo(M,M)$ containing homeomorphisms}
}
\newglossaryentry{space of bijections in Top}{
	name = {$\TOPO^h(M,M)$}, 
	description={The set $\Topo^h(M,M)$ equipped with subspace topology from $\tauco{M}{M}$}
}
\newglossaryentry{Hom}{
	name={$\Hom^A$},
	description= {Groupoid with objects $\Power M$ and morphisms homoeomorphisms,
		 see Def.~\ref{Def:homeoMA}}}
\newglossaryentry{self-homeomorphisms}{
	name = {$\shmor{f}{A}{N}{N'}$}, 
	description={Notation for morphisms in $\Hom^A$}
}
\newglossaryentry{Premots}
{
	name = $\premots$, 
	description={Set of all \premot{}s in $M$, $f \in   \Topo^{}(\II,\Topo^h(M,M))$ with  $f_0 = \id_M$, see Def.~\ref{de:premots}}
}

\newglossaryentry{Identitypremot}
{
	name = $\Id_M$, 
	description={\Premot{} in $M$ which is the path $f_t=\id_M$ for all $t$}
}
\newglossaryentry{motion}{
	name = {$\mot{f}{}{N}{N'}$}, 
	description={A motion from $N$ to $N'$ in the specified manifold, see Definition~\ref{de:mot1}}
}
\newglossaryentry{worldline}{
	name={$\W(\mot{f}{}{N}{N'})$},
	description ={Worldline of a motion $\mot{f}{}{N}{N'}$, see Definition~\ref{def:worldline}}
}
\newglossaryentry{set of motions N to N'}{
	name={$\Motc{M}{N}{N'}$},
	description={The set of all motions from $N$ to $N'$ in $M$}}
\newglossaryentry{set of motions}{
	name={$\Mtc$},
	description={The set of all motions in $M$}}

\newglossaryentry{m equivalence}
{
	name = $\simm$, 
	description={Indicates motions related by motion-equivalence, see Proposition~\ref{pr:me1}}
}
\newglossaryentry{m class}
{
	name = $\classm{\mot{f}{}{N}{N'}}$, 
	description={Class of a motion $\mot{f}{}{N}{N'}$ up to motion-equivalence, see Proposition~\ref{pr:me1}}
}

\newglossaryentry{motion gpd}
{
	name = $\Mot^A$, 
	description={Motion groupoid of a manifold $M$, fixing $A\subset M$, see Theorems \ref{th:mg} and \ref{th:mg2A}}
}

\newglossaryentry{rp equivalence}
{
	name = $\simrp$, 
	description={Indicates motions related by relative path-equivalence, see Definition \ref{de:rp equiv}}
}
\newglossaryentry{rp class}
{
	name = $\classrp{\mot{f}{}{N}{N'}}$, 
	description={Class of a motion $\mot{f}{}{N}{N'}$ up to relative path-equivalence, see Lemma \ref{le:rpe}}
}

\newglossaryentry{wequivalence}
{
	name = $\simww$, 
	description={Equivalence relation on motions in terms of worldlines, see Definition~\ref{de:wequivalence} }
}
\newglossaryentry{classwequivalence}
{
	name = $\classww{\mot{f}{}{N}{N'}}$, 
	description={Class of motions related by $\simww$, see Definition~\ref{de:wequivalence}}
}

\newglossaryentry{fakemotions}
{
	name = $\dahm{N}{M}{N'}$, 
	description={Set of fake motions from $N$ to $N'$ in a manifold $M$, see Definition~\ref{de:fakemotion}}
}

\newglossaryentry{strongisotopy}
{
	name = $\simfk$, 
	description={Indicates fake motions related by strong isotopy, see Definition~\ref{de:strongiso}}
}

\newglossaryentry{clstrongisotopy1}
{
	name = $\classfk{f\colon N\protect \not \too N'}$,
	description={Class of \fakemotion{}s related by strong isotopy, see Lemma~\ref{lem:Fake-braids}}
}

\newglossaryentry{fakemotiongroupoid}{
	name = $\FMot{M}$,
	description = {Groupoid of strong isotopy classes of fake motions, see Lemma~\ref{lem:Fake-braids}}
}

\newglossaryentry{i equivalence}
{
	name = $\simi$, 
	description={Indicates self-homeomorphisms related by isotopy, see Definition \ref{de:isotopy}}
}
\newglossaryentry{i class}
{
	name = $\classi{\shmor{f}{}{N}{N'}}$, 
	description={Class of self-homeomorphisms related by isotopy, see Lemma \ref{le:isotopy}}
}
\newglossaryentry{mcg}
{
	name = $\mcg^A$, 
	description={Mapping class groupoid of a manifold $M$, fixing $A\subset M$,
		 see Thms.~\ref{th:mcg}, \ref{th:mcga}}
}
\newglossaryentry{interval}
{
	name = $\II$, 
	description={The space $[0,1]\subset \R$ with the subset topology}
}
\newglossaryentry{disk}
{
	name = $D^n$, 
	description={The $n$-disk $\{x\in \R^n\vert \,|x|\leq 1 \}\subset \R^n$ with the subset topology}
}
\newglossaryentry{circle}
{
	name = $S^n$, 
	description={The circle $\{x\in \R^{n+1}\vert \,|x|= 1 \}\subset \R^{n+1}$ with the subset topology}
}

\usepackage{authblk}

\title{Motion groupoids and mapping class groupoids}

\author{{Fiona Torzewska}\footnote{f.m.torzewska@leeds.ac.uk}, 
{Jo\~ao Faria Martins\footnote{j.fariamartins@leeds.ac.uk} },
{Paul Purdon Martin}\footnote{{p.p.martin@leeds.ac.uk}}\\
{\small School of Mathematics, University of Leeds, UK}
}

\date{\today}
\begin{document}

\maketitle

\begin{abstract}
 
Here $\MM$ denotes a pair $(M,A)$ of a manifold and a subset (e.g. $A=\partial M$ or $A=\emptyset$).
We construct for each 
$\MM$ its 
{\it motion groupoid}  $\Mot[\MM]$, 
whose object set
is the power set $ \Power M$ of $M$,
and whose morphisms are certain equivalence classes of continuous flows
of the `ambient space' $M$,
that fix $A$, acting on $\Power M$.
These groupoids generalise
the classical definition of a motion group associated to a manifold $M$ and a 
{submanifold} $N$, which can be recovered by considering the automorphisms 
in $\Mot[\MM]$
of $N\in \Power M$. 

We also construct the
{\it mapping class groupoid} $\mcg[\MM]$ associated to a pair $\MM$ with the same object class, whose morphisms are now equivalence classes of homeomorphisms of $M$, that fix $A$. 
We 
recover the classical definition of the mapping class group of a pair by taking automorphisms at the appropriate object.

For each pair $\MM$ we explicitly construct a
functor
$\F\colon \Mot[\MM] \to \mcg[\MM]$, which is the identity on objects,
and prove that this is full and faithful, and hence an isomorphism, if 
$\pi_0$ and $\pi_1$ of the appropriate space of self-homeomorphisms of $M$ are trivial.
In particular, we have an isomorphism in
the physically important case $\MM=([0,1]^n, \partial [0,1]^n)$, 
for any $n\in \N$.

We show that the congruence relation used in the construction $\Mot[\MM]$ can be formulated entirely in terms of a level preserving isotopy relation on the trajectories of objects under flows --
worldlines (e.g. monotonic `tangles').

We examine several explicit examples of $\Mot[\MM]$ and $\mcg[\MM]$ 
demonstrating the
utility of the constructions. 
\end{abstract}

\medskip 


\medskip

\noindent \textbf{Acknowledgements:} {JFM and PM were partially funded by the Leverhulme trust {research project} grant ``RPG-2018-029: 
Emergent Physics From Lattice Models of Higher Gauge Theory''.
FT was funded by a University of Leeds PhD Scholarship and is now funded by EPSRC.
FT thanks Carol Whitton; and PM thanks Paula Martin  for useful conversations. 
{We all thank Celeste Damiani and Martin Palmer-Anghel for comments,  
Arnaud Mortier for useful email discussions, 
and Hadeel Albeladi, Basmah Alsubhi 
and Manar Qadi
for useful discussions.}

An initial version of this material was presented by FT in a series of LMS funded lectures available \href{https://media.ed.ac.uk/playlist/dedicated/51612401/1_dfxocmez/1_bm66fhc8}{\underline{here}},
FT also thanks Simona Paoli for the invitation to give these lectures. 
}

\newpage 
\tableofcontents

\newpage
\section{Introduction}\label{sec:Intro}

The paper is about constructing algebraic structures that capture
(in a broad sense) topological aspects of particle motion.
Although we eventually construct groupoids, these are neither a convenient starting point, nor close in any absolute sense to the
underlying physics. 
(The convenience of groupoids as an endpoint lies in the relatively
well-developed state of their own representation theory; and the fact that eventually one must give a scheme for predicting results of
physical measurement -- i.e. pass to linear operators and their 
spectra, and so, in the algebraic language, pass to 
$\Vect{}$, where the images of our constructions are indeed categories.)
More natural are structures that, like categories, have partial compositions, but where (for quite different reasons) 
both the unit and the associativity conditions are relaxed.
Here our starting points are `engines' that take a manifold $M$ as input --- the choice of ambient space;
and produce, 
as output
for each $M$, 
a {\it `magmoid'} (a triple of objects, morphisms and a partial composition with no conditions).
In each case, the object class
{is}
$\Set(M,\{0,1\}) \;$ 
 (as a point of reference, note that 
if object $f\in\Set(M,\{0,1\} )$ has finite support then 
it is a `collection of point particles', and 
endomorphisms are called braids).
{ 
The main remaining step is to construct 
congruences on these magmoids that wash out irrelevant detail
from particle motion,
while retaining useful features... 
and yielding quotient groupoids.}

\medskip
    
To facilitate 
our main
Introduction, we begin 
with a prelude
summarizing some
key constructions
of  
engines and congruences.
We 
postpone proofs and 
explanations,
but we are explicit. This will allow us to be technically specific
even
in our Introduction. 

\medskip

\begin{figure}
\newcommand{\bit}{2.95} 
\hspace{-.1in} 
    \includegraphics[width=\bit cm]{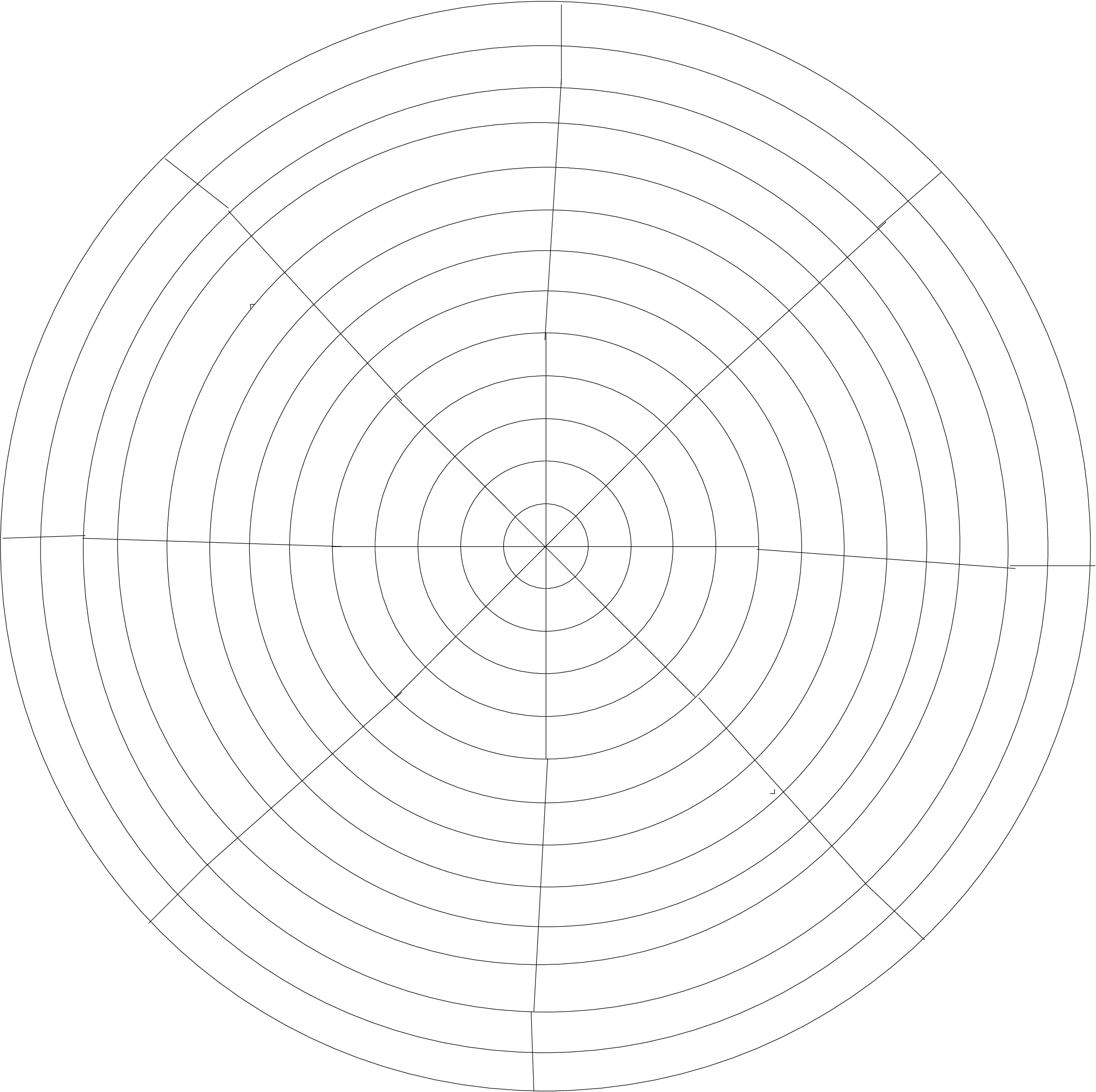}
    \includegraphics[width=\bit cm]{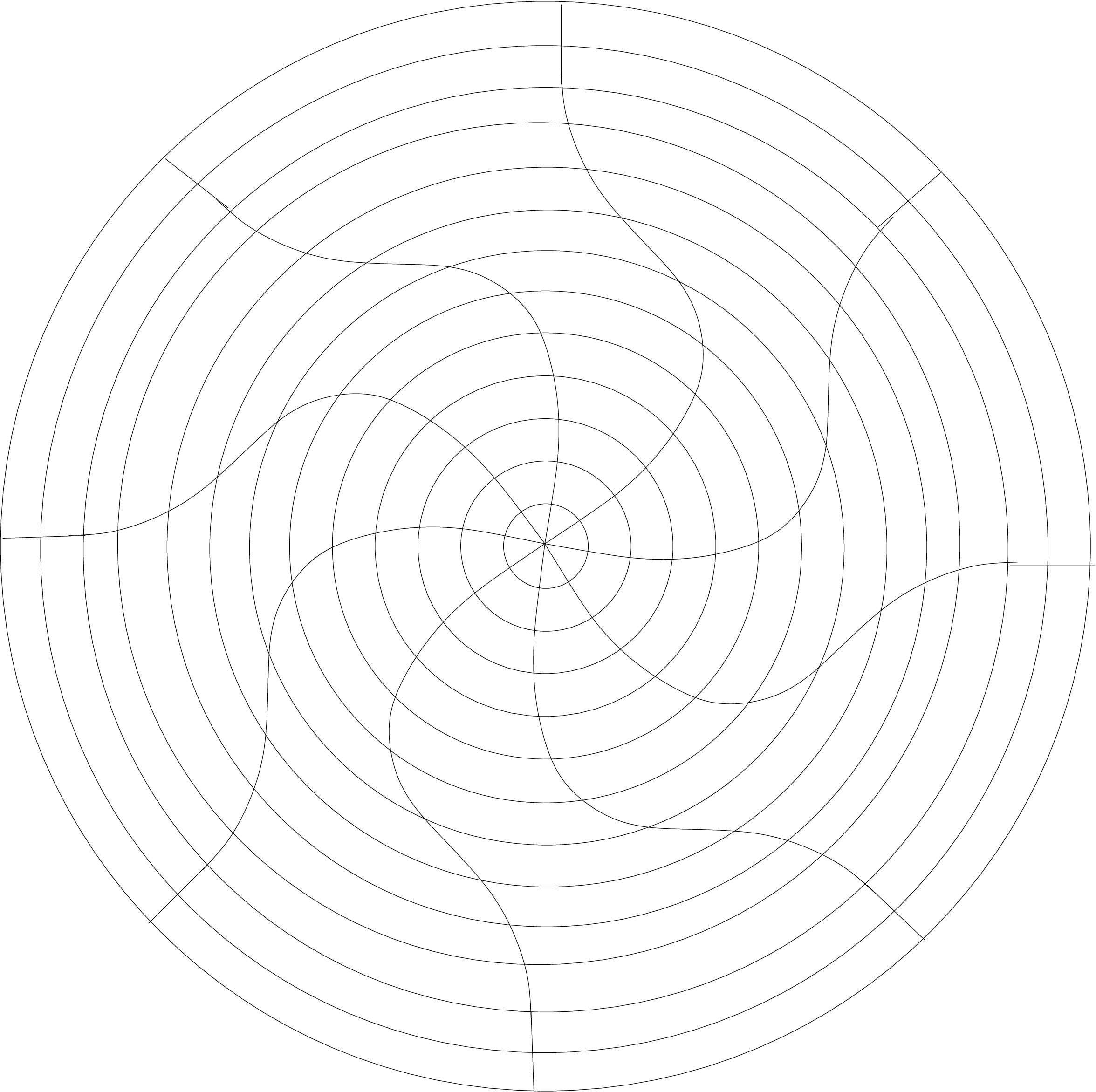}
    \includegraphics[width=\bit cm]{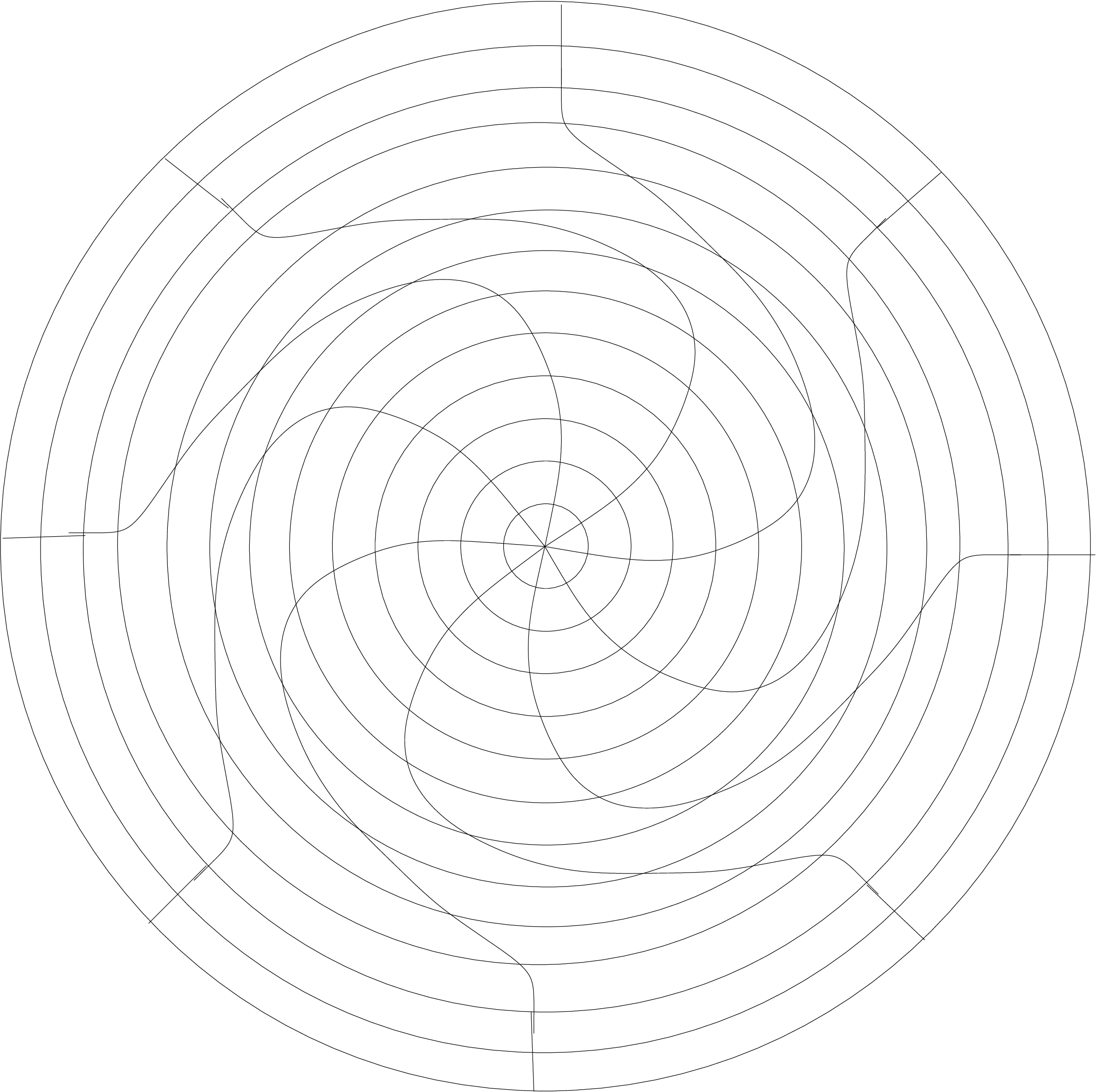}
    \includegraphics[width=\bit cm]{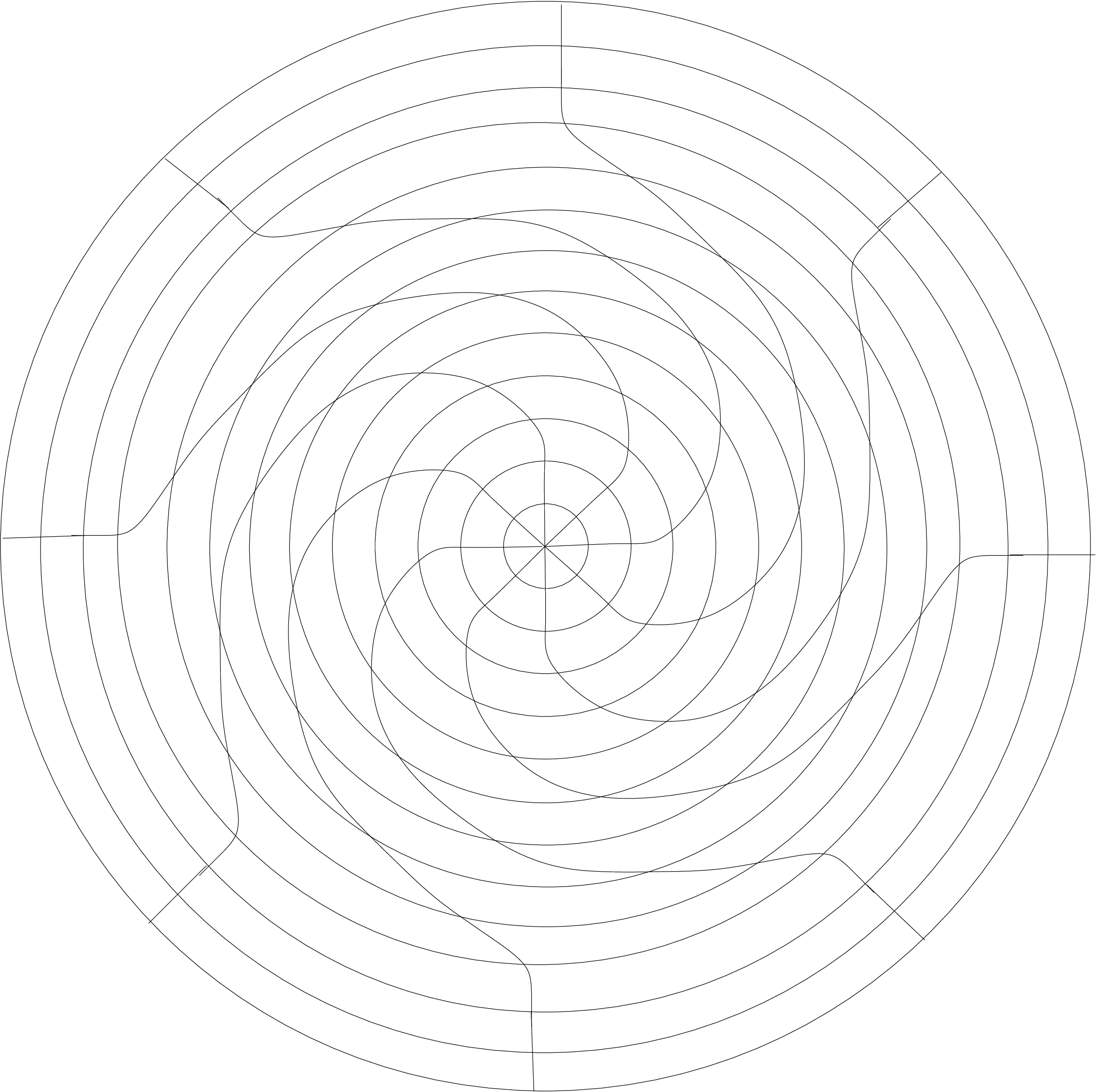}
    \includegraphics[width=\bit cm]{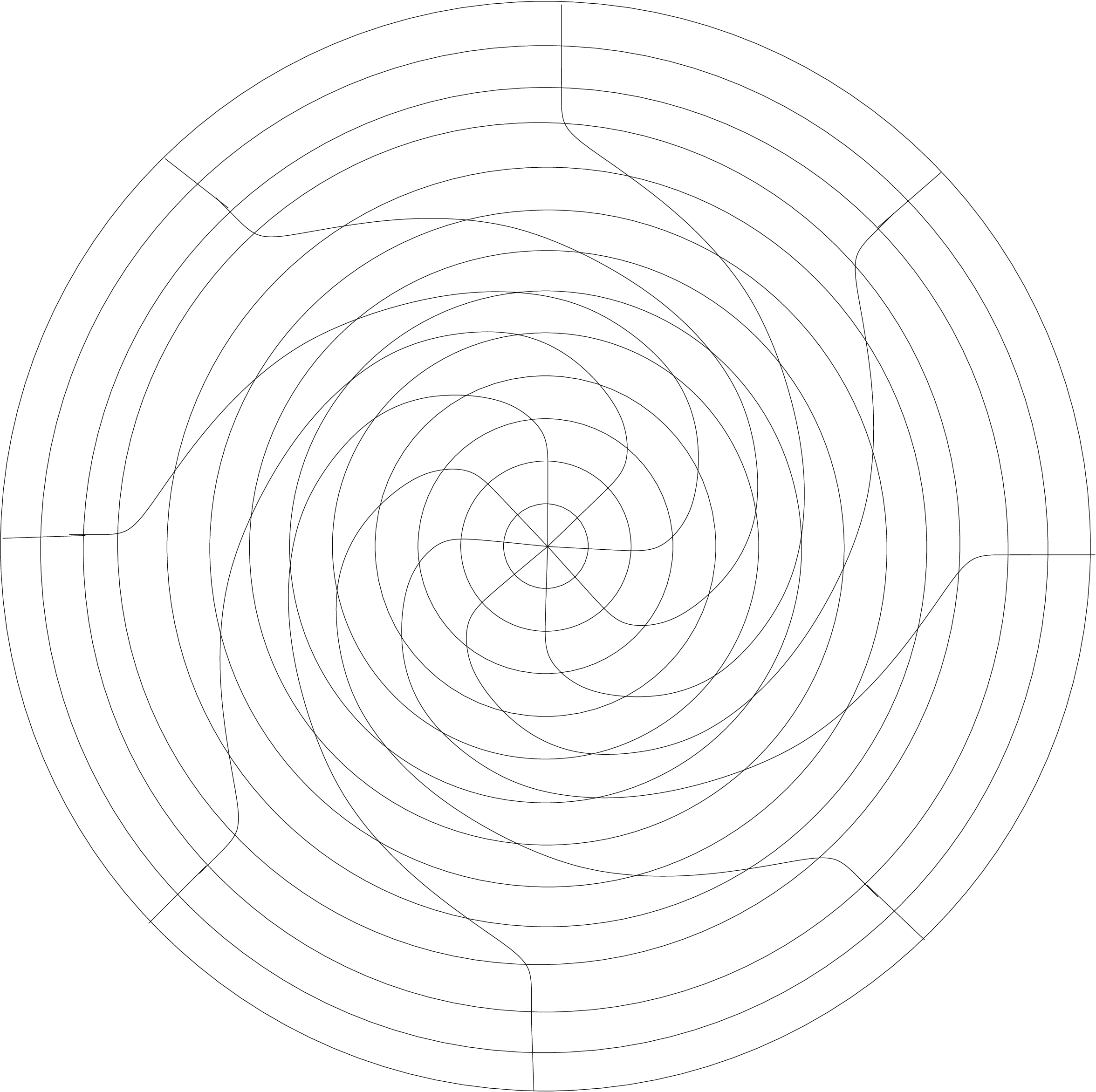}
    \caption{
    Schematics representing a sequence of self-homeomorphisms of
    manifold $D^2$.
}
    \label{fig:my_label01}
\end{figure}

\newcommand{\phomotopy}[3]{\Topo(\II^2, #1)\!(\! #2 (\!\!\! ) #3 ) }
\newcommand{\phisotopy}[3]{\Topo(#1\times\II, #1)\!(\! #2 (\!\!\! ) #3 ) }
\newcommand{\phisotop}[5]{\Topo(#1\times\II, #1)\!(\! {\tiny #2}^{#4}_{#3}\!\!  {\tiny #5} ) }
\newcommand{\pathphisotop}[5]{\Topo(\II, \TOPO^h(#1,#1))\!(\! {\tiny #2}^{#4}_{#3}\!\!  {\tiny #5} ) }

\newcommand{\simw}{\stackrel{\omega}{\sim}} 

\newcommand{\mug}{\boldsymbol{\mu}}  
\newcommand{\alphaa}{\alpha}         

Let $\Topo$ denote the category of topological spaces and continuous maps, and $\Topo^h$ the subcategory of homeomorphisms. 
Then let  $\TOPO^h(X,Y)$ denote the space with underlying set $\Topo^h(X,Y)$ and the compact open topology.
For example, Figure~\ref{fig:my_label01} indicates a sequence of self-homeomorphisms of the disk such that successive terms are 
close in $\TOPO^h(D^2,D^2)$.
Let $\II =[0,1]\subset\R$ with the usual topology.
Fixing, for now, a 
manifold-subset pair $\MM=(M,\emptyset)$,
so writing simply
$M$ for $\MM$,

$$
\premots \; = \; \{ f \in   \Topo^{}(\II,\TOPO^h(M,M)) \; | \; f_0 = \id_M   \} 
$$
-- a set of gradual deformations of $M$
over some standard unit of time. 
Thus
Figure~\ref{fig:my_label01} indicates sequential points on a path in $\premo{D^2}$.
$\;$ 
For	$(f,g)\in \premots\times \premots$, 
note that $g*f$ given by
	\begin{align*}
		(g*f)_t = \begin{cases}
			f_{2t} & 0\leq t\leq 1/2, \\
			g_{2(t-1/2)}\circ f_1 & 1/2\leq t \leq 1.
		\end{cases}
	\end{align*}
is in $\premots$; 
as is  $\bar{f}$  given by 
$ 
		\bar{f}_t 
		\;=\;  f_{(1-t)}\circ f_1^{-1} .
$ 
\hspace{.1cm}
Thus $(\premots,*)$ is a magma.

A magma action of a magma $P$ on a set $S$ 
is a map $\alphaa: P\times S \rightarrow S$ with $q(ps)=(qp)s$ (denoting the composition in $P$ of a pair $(p,q)$ as $qp$, and the image $\alpha(p,s)$ as $ps$).
Given 
such
a magma action $\alpha$, 
the 
{\em action magmoid} 
$\mug_\alphaa$ 
is a triple consisting of objects $S$, morphisms which are triples 
$(p,s,p s) \in P \times S \times S$, 
and a partial composition  
$((p,s,ps),(q,ps,qps))\mapsto(qp,s,qps)$.

\newcommand{\actss}{\gtrdot} 
\newcommand{\actsss}{\blacktriangleright} 
\newcommand{\Homs}{Hom_{\Set}}
\newcommand{\SSS}{\Sigma}  

Fix a manifold $M$. Let $\SSS$ be a space.
A $\SSS$-field configuration is a function in $\Topo(M,\SSS)$. 
There is a magma action 
$\actsss$ of  $(\premots,*)$ on $M$ 
given by 
$\actsss(f, m) = f_1(m)$; 
and 
	on 
$\Topo(M,\SSS)$ 
given 
	by 
	$\chi \mapsto \chi\circ f_1^{-1}$. 
	In particular 
	on the power set
	$\Power M \cong 
	\Topo(M,\Z_2)$
	($\SSS = \Z_2$ with the indiscrete topology),
the action is
denoted
$\acts$, so
$\acts(f, N) = f_1(N)$. 
The 
action magmoid 
$\mug_{\acts}$ 
is denoted
    $\Mtcmag^* \; =\; ( \Power M , \Mtcmag^*(-,-)  , * )$.
{A morphism  $(f,N,N')$ here is a gradual deformation of  $M$ that
carries  
the initial object subset to the final object subset: $f_1(N)=N'$, hence called a {\em motion}.}
    The motion magmoid $\Mtcmag^*$ is large, but amenable to various natural quotients, 
    as we shall see.

Let $X$ be a space. 
Paths 
$\gamma,\gamma'\in \Topo(\II,X)$ 
are {\it path homotopic}
if $ \phomotopy{X}{\gamma}{\gamma'} \neq \emptyset$ where
$$
\phomotopy{X}{\gamma}{\gamma'} \; \;
= \; \{ H\in \Topo(\II^2,X) \;|\; H(-,0)=\gamma, \; H(-,1)=\gamma', \; H(0,-)=\gamma_0, \; H(1,-)=\gamma_1  
\}.
$$
Letting $X=\TOPO^h(M,M)$,
path homotopy gives a congruence on $(\premots,*)$
(with group quotient).
But alternatively with
 $f,g$ flows and $N,N' {\subset M}$ subsets, 
let 
\begin{multline*}
   \Topo(\II^2,\TOPO^h(M,M))\rel{f}{N}{N'}{g} 
 \;=\; 
 \{ H \in \Topo(\II^2,\TOPO^h(M,M)) \; | \; \forall t\colon H(t,0)=f_t,
 H(t,1) = g_t, \\ \forall s\colon H(0,s)=\id_M, H(1,s)(N)=N'=f_1(N) \}.
\end{multline*}

\begin{theorem} \label{th:main000}
	For $M$ a manifold, there is a congruence on 
$\Mtcmag^*$ given by  the relation $(f,N,N')\simrp (g,N,N')$ if
$\Topo(\II^2,\TOPO^h(M,M))\rel{f}{N}{N'}{g} \neq\emptyset$.
	 The quotient is a groupoid -- the {\em motion groupoid} 
	$$
	\Mot\;\; 
	=\;\;\; \Mtcmag^*/\simrp \;\;\;\; 
	=\;\; \;
	(\Power M,\; \Mtcmag^*(N,N')/\simrp,*,
	\classrp{\Id_M}, \; 
	\classrp{f} \mapsto \classrp{\bar{f}} )  .
	$$		
\end{theorem}

\newcommand{\gammaf}{\sh{f}}
\newcommand{\gammaff}{\sh{g}}

{ 

Now let $\Hom$ denote the action groupoid
of the 
group action of $\Topo^h(M,M)$ on $\Power M$ given by $\sh{f}\acts N = \sh{f}(N)$. 
Let $N,N' \in \Power M$ and $\gammaf,\gammaff \in \Hom(N,N')$. 
Then define 
$$
    \pathphisotop{M}{\gammaf}{N}{N'}{\gammaff}
= \{ H \in \Topo(\II, \TOPO^h(M, M) \; | \; 
H(0)=\gammaf, H(1)=\gammaff,
\forall t\colon H(t)(N)=N'
\}. 
$$
The relation
$\gammaf\simi\gammaff \in \Hom(N,N')$
if   
$ 
\pathphisotop{M}{\gammaf}{N}{N'}{\gammaff}
\neq \emptyset$ 
gives a congruence on $\Hom$. 
The
{\em mapping class groupoid} $\mcg$ is the quotient $\Hom/\simi$.
}
\hspace{.05in}
A
functor 
\[
\F\colon \Mot[M] \to \mcg[M] 
\]
is given by `forgetting' $\classrp{f}\mapsto \classi{f_1}$. 

Finally in this pre-Introduction, we introduce one more 
`engine'  creating a power-set magmoid from each manifold;
and,
for this construction,
one more type of congruence.
For $M$ a manifold and $N$ a subset let 
$\dahm{N}{M}{N'}$ 
be the subset of elements
$f \in \Topo(N\times\II,M)$ 
such that $ f(-,t)$ is an embedding for each $t$,
 $f(n,0)=n$ for all $n \in N$, and
$f(N,1)=N'$.  
\hspace{.05in} 
One indeed obtains 
another power-set magmoid for $M$ upon noting that
the formula
\begin{align*}\label{eq:introcomposition_fakemotion}
g\smallsquare f(n,t):= 
\begin{cases} f(n,2t), & 0\leq t \leq 1/2 \\
g(f(n,1),2t-1), &  1/2\leq t\leq 1\, 
\end{cases}
\end{align*}
gives a composition 
$\dahm{N}{M}{N'}\times\dahm{N'}{M}{N''}
\rightarrow \dahm{N}{M}{N''}$. 
\hspace{.1cm}
Define 
\[
\Topo( (N\times\II)\times\II, M)[N^{f}_{f'} N']
=
\left\{H\in \Topo ( (N\times\II)\times\II, M) 
\; \middle\vert \; \parbox{6.75cm}{$\forall t\colon H(n,t,0)=f(n,t), \;
H(n,t,1)=f'(n,t) \; ,  \\ \;
\text{for each } s\in \II, \; H(-,-,s) \in  \dahm{N}{M}{N'}$}
\right\}.
\]
Then   
$f \simfk f'$ 
if $\Topo ( (N\times\II)\times\II, M)[N^{f}_{f'} N']\neq \emptyset$. This defines a congruence,
thus
yielding the 
{\em \fakemotion\ groupoid}    $\; \FMot{M}$.
\hspace{.1in} 
A functor 
\[
\T:\Mot \rightarrow \FMot{M}\] 
is given by sending $\classm{(f,N,N')}$ to the class of the map 
$(n,t)\mapsto f_t(n)$, where $n\in N$.

{The engines $\Mot[-]$, $\mcg[-]$ and $\FMot{-}$
all specialise
when considering  with  
$M=\R^2$ and $N=N'$ finite,
to constructions 
isomorphic to the braid group.
However in general they differ 
from each other
in a number of beautiful,
intriguing and useful ways, as we shall see.}

\medskip

{\bf{Main Introduction. } } 
We will show that
embedding the various groups $\Mot(N,N)$ 
(motion groups, generalising \cite{dahm, goldsmith})
into the groupoid $\Mot$
yields several significant benefits.
By way of introduction, we now discuss
these (in no particular order).

\medskip

\begin{figure}
\begin{minipage}{.325\textwidth}
\newcommand{\bit}{3.8} 
    \centering
    \includegraphics[width=\bit cm]{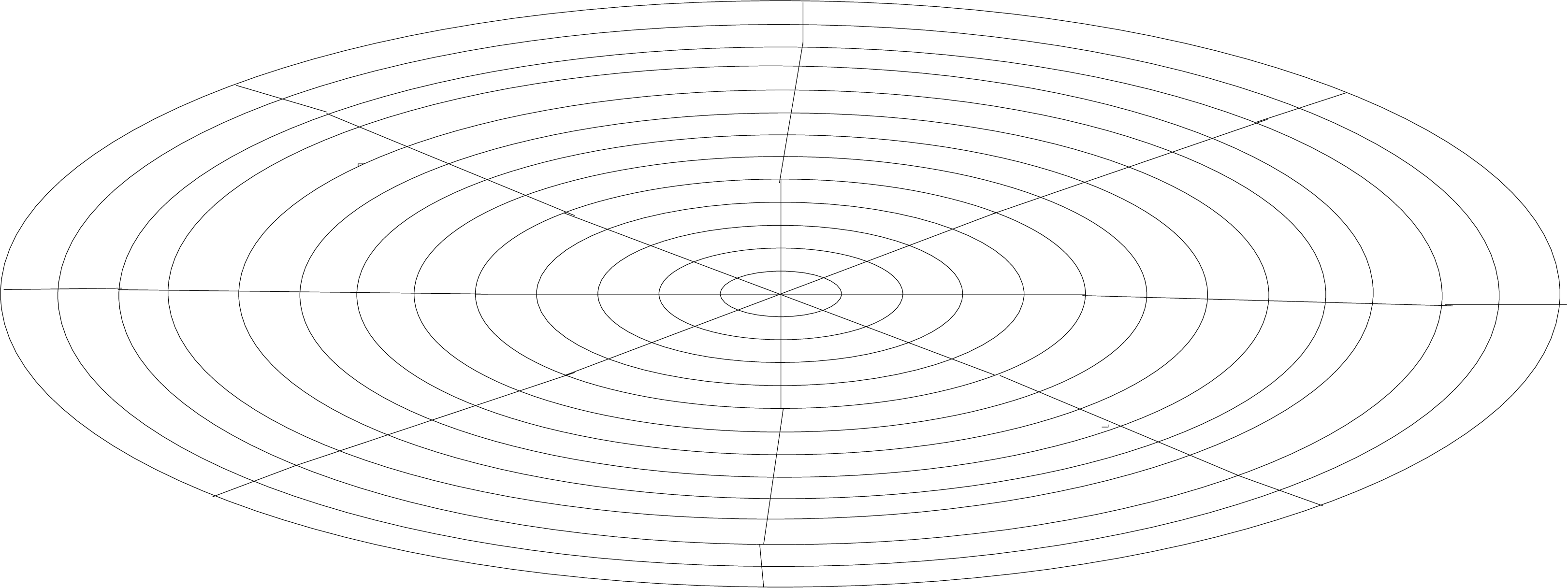}
    \\
      \includegraphics[width=\bit cm]{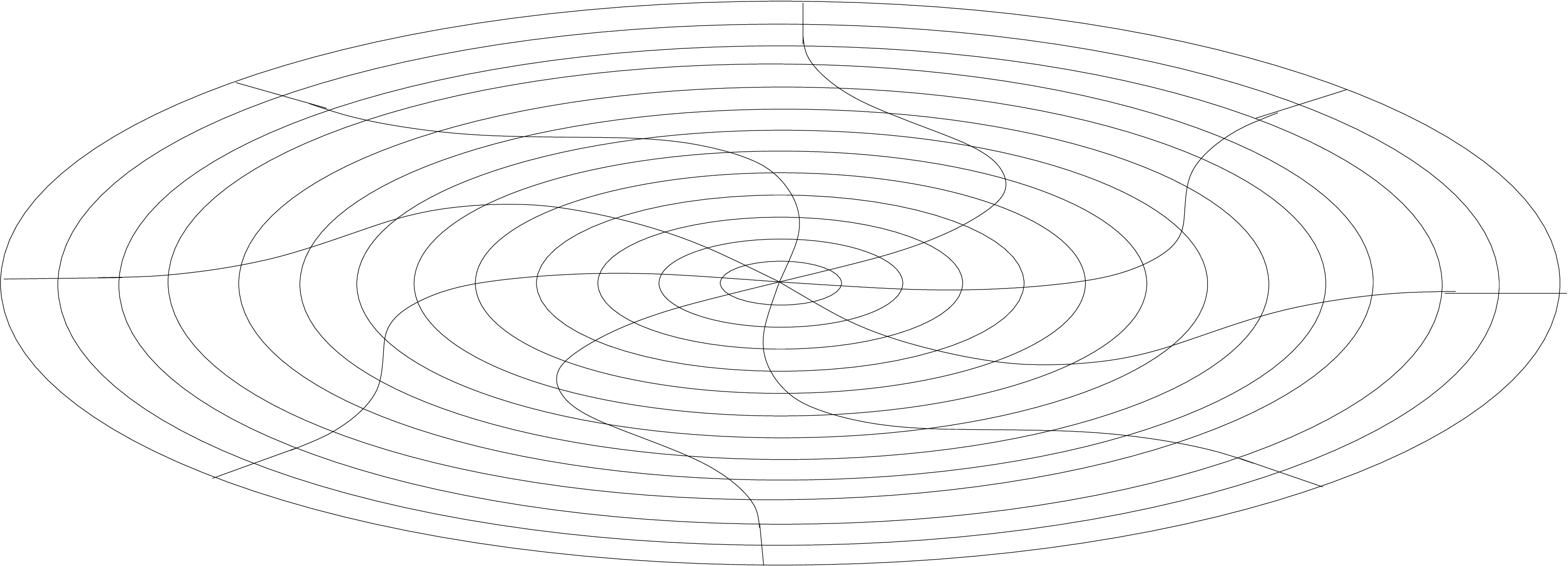}
      \\
      \includegraphics[width=\bit cm]{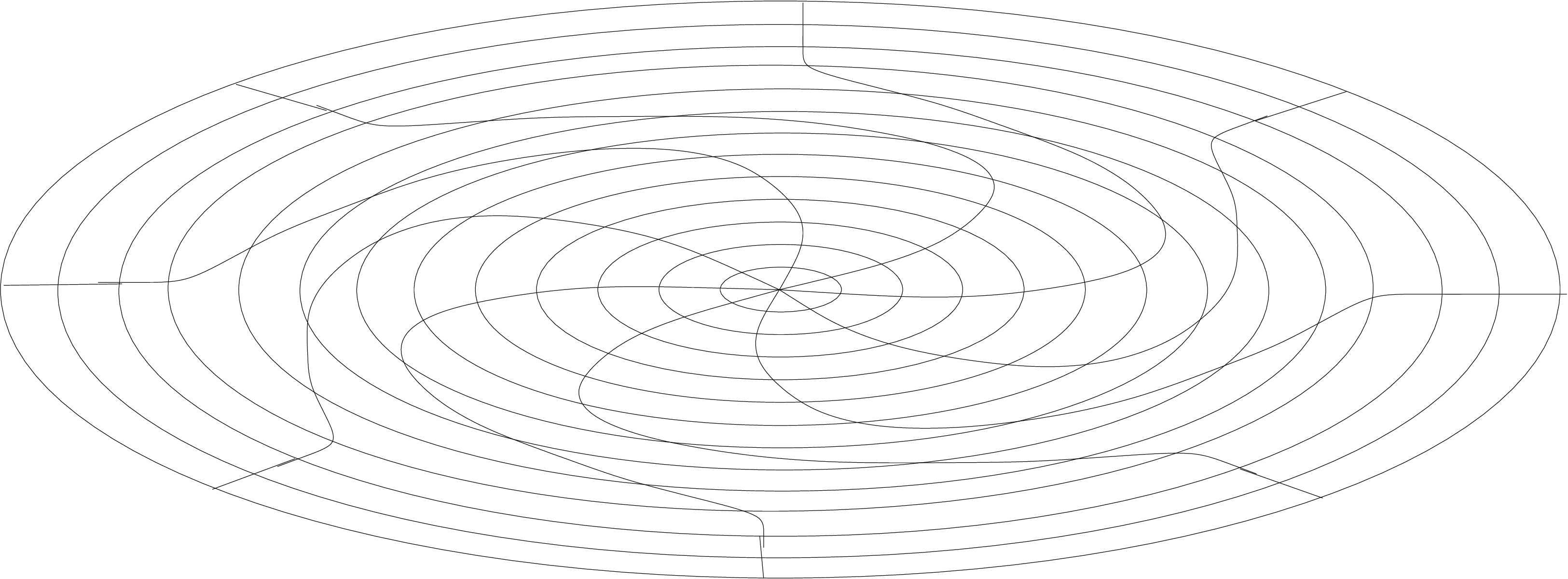}
      \\
      \includegraphics[width=\bit cm]{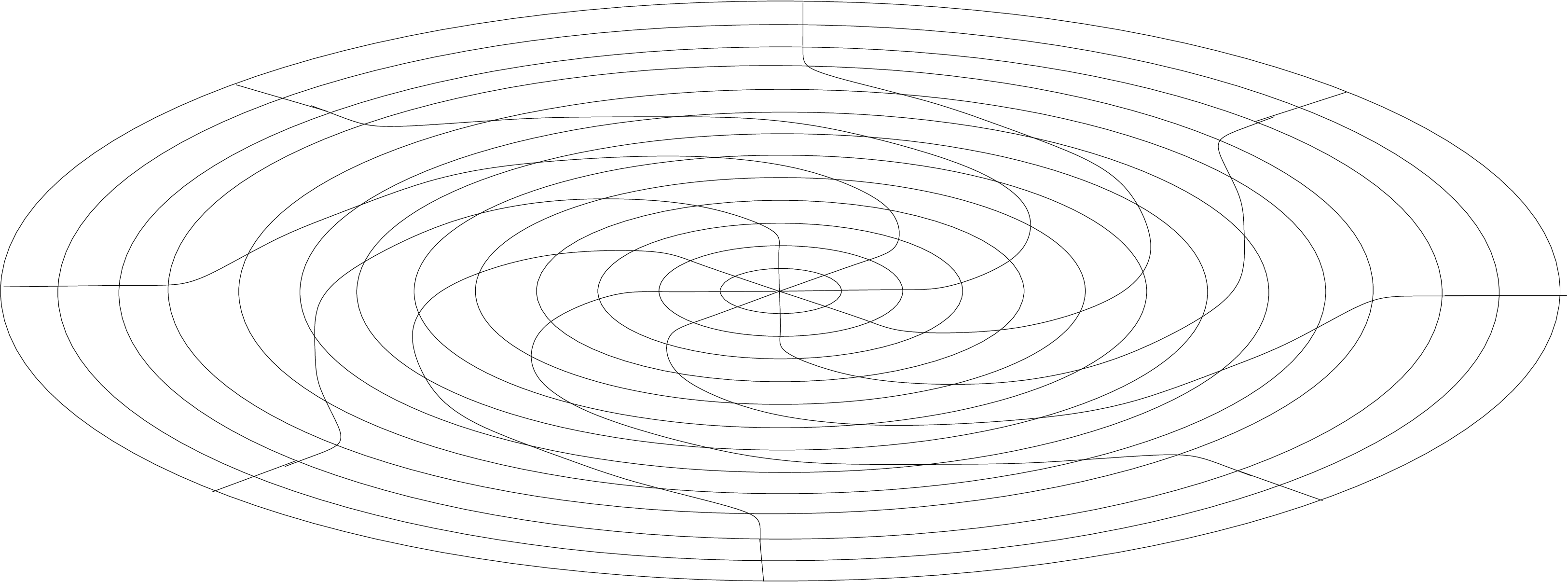}
      \\
      \includegraphics[width=\bit cm]{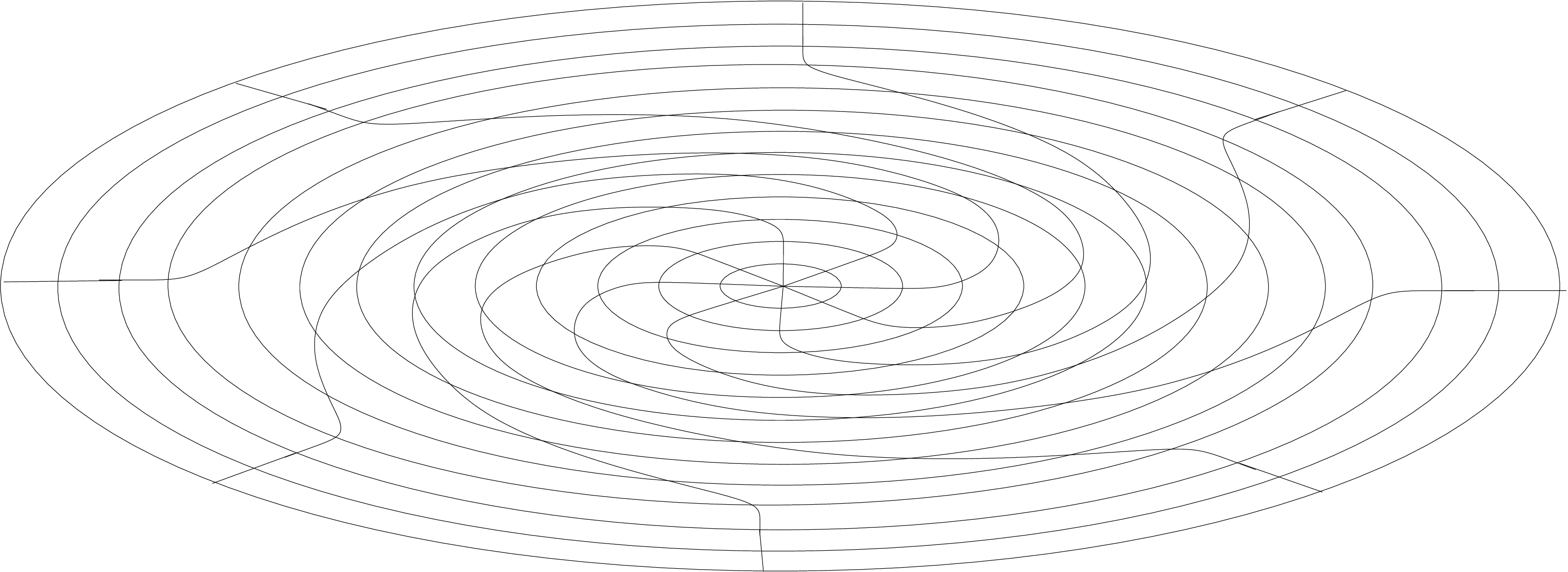}
      \\
      \includegraphics[width=\bit cm]{Figs/dehn0006xs.eps}
\end{minipage}
\begin{minipage}{.325\textwidth}
\newcommand{\bit}{4.065} 
  \centering
    \includegraphics[width=\bit cm]{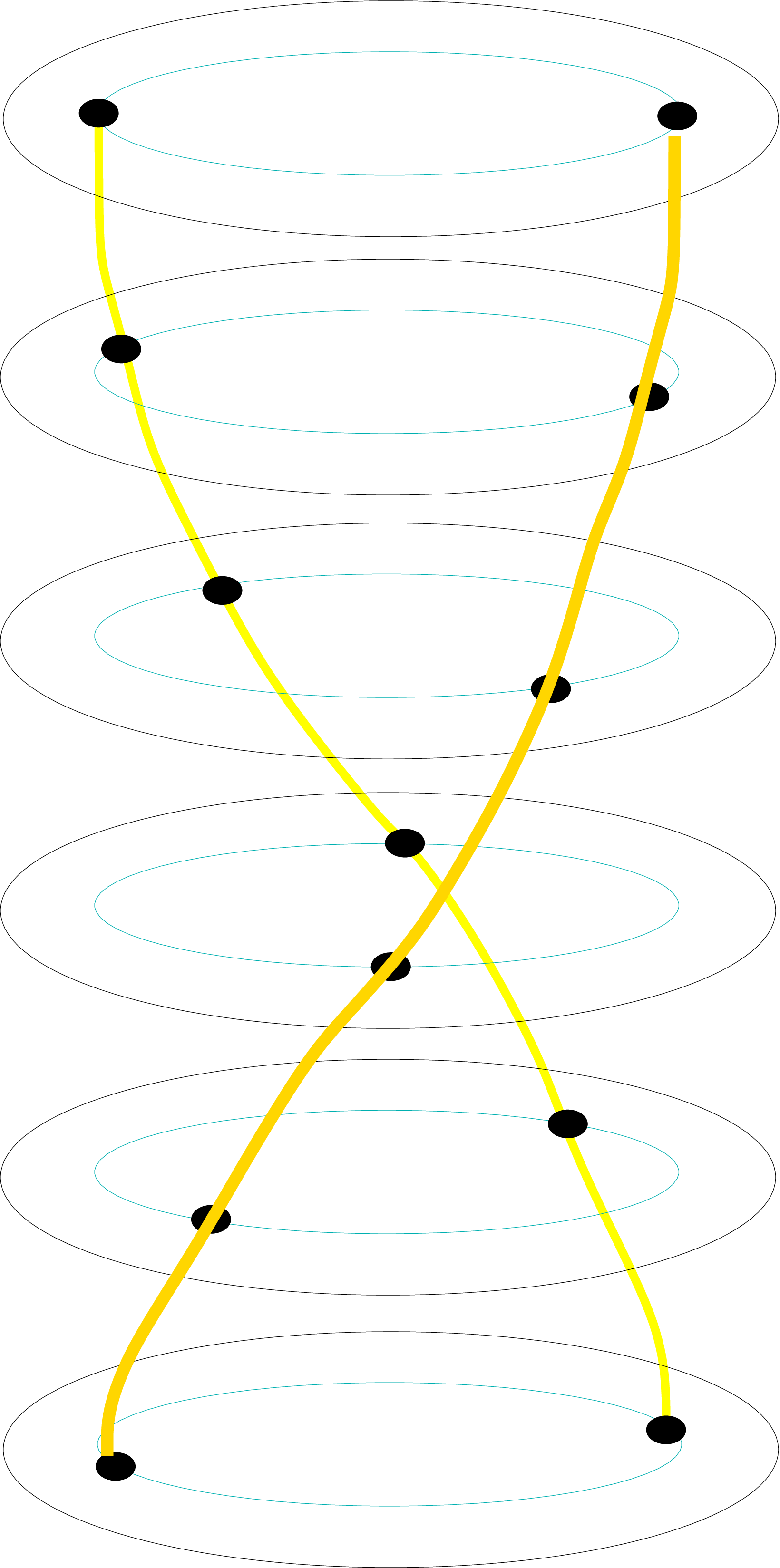}
\end{minipage}
\begin{minipage}{.325\textwidth}
\newcommand{\bit}{4.065} 
  \centering
    \includegraphics[width=\bit cm]{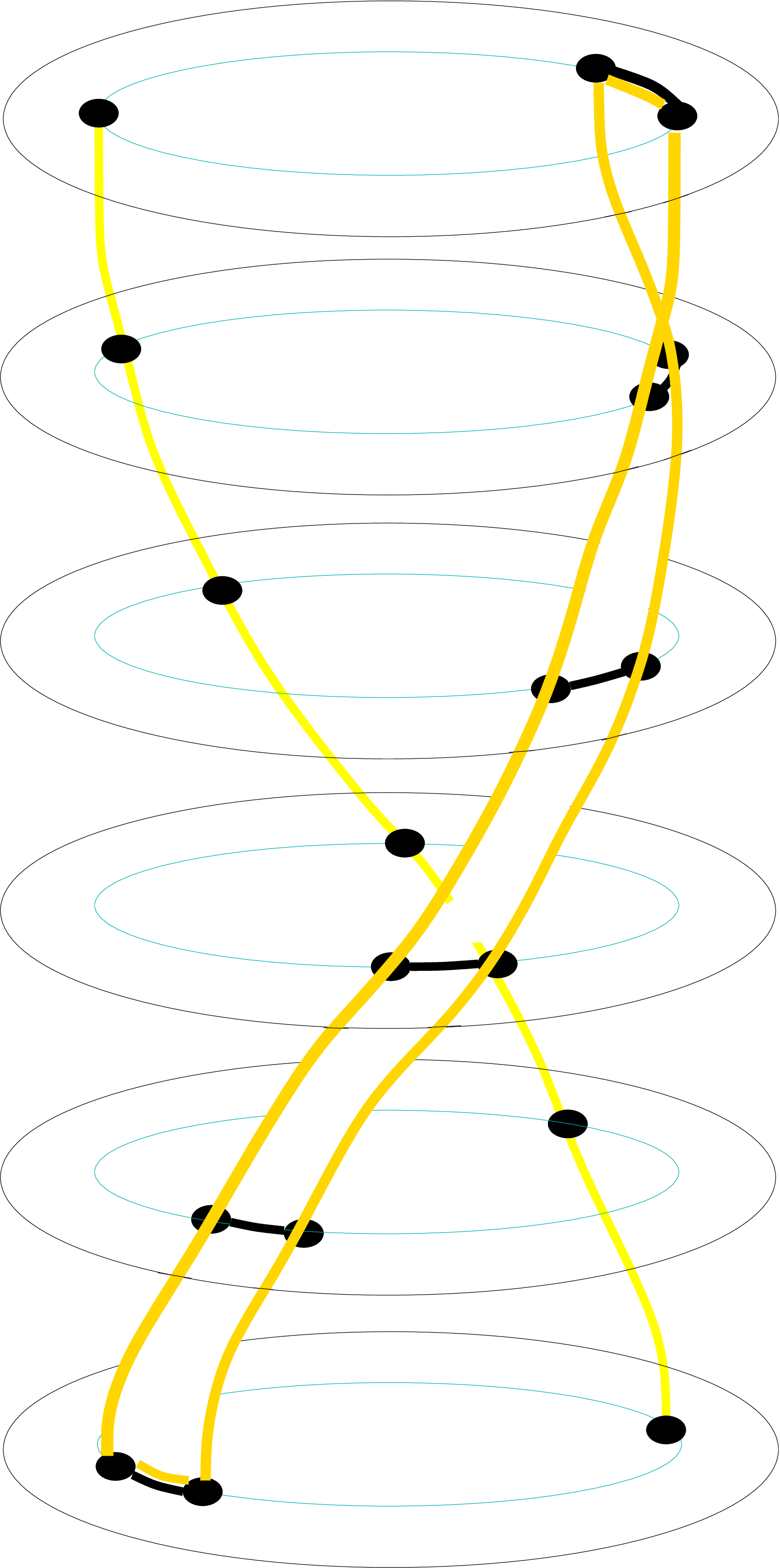}
\end{minipage}
    \caption{Left: A flow $f$ in $M=D^2\subset \R^2$ 
    (from Fig.\ref{fig:my_label01})
    as a level-preserving homeomorphism in $M\times \II$.
    Middle:
    Induced motion of 
    two 
    point particles as induced by the flow $f$, drawn as paths in $M\times\II$.
  Right: Induced motion of 
    a line and a
    point particle as induced by the flow $f$.
      }
    \label{fig:my_label002}
\end{figure}

Firstly, the groupoid formalism allows us to study the set of objects in a unified way.
The reader should think of finite sets passing to braids; compact manifolds passing to generalised loop-braids;
hybrids passing to some beautiful new 
algebraic structures; and non-compact manifolds and various other generalisations passing to structures depending heavily on details of the
embedding. 
This facility has many aspects -- 
leading to monoidal structures, and vast generalisations thereof, as well as to other higher categorical formalisms;
and, as we will consider here, different but often equivalent notions of congruence. 
Together these lead in turn to the question of non-isomorphims 
such as particle fusion,
and hence to generalised 
tangle categories (cf. \cite{BaezDolan, Picken}).
In this direction we go 
little 
further than framing the construction.
We do, however, make significant progress on the connection between {\it isomorphisms} in generalised tangle categories, and motion groupoids by 
realising motion equivalence in terms of a level preserving isotopy relation on `worldlines' of motions.

Another benefit lies in 
object-changing morphisms themselves.
Equipped with these, 
motions can be factored into compositions of simpler motions
in many more ways. 
This is 
advantageous, for example,  
for giving 
generating sets 
of morphisms in motion subgroupoids. 
It is an extension of something that can already be seen
in the classical setting of braid groups, in the sense that 
the Artin presentation of the braid group has 
fewer 
generators than the usual presentation of the pure braid group, 
where each individual point 
must
return to 
its starting position.

Moreover, 
depending on the object $N$, in the group setting 
it is typically true that there exists an open neighbourhood of the identity homeomorphism 
such that all motions realised by a flow staying in this neighbourhood are identified with the identity: for instance for finite $N$ all non-trivial flows connecting the same object are at positive ``distance" from the identity. 
Thus Lie-theory/Lie-algebra-like linearisation strategies are hard to {access}.
In contrast in the groupoid setting, even a flow very close to the identity
can change the object,
and then 
of course is not identified with the identity motion.
It follows that 
morphisms of the groupoid can be obtained as a composition of
arbitrarily “small” motions. A range of powerful techniques thus become 
accessible. 
This will be explored for example
when proving the relation between motion groupoids and 
Artin braid groups in \S\ref{sec:braids}.

\medskip

 We will 
 see here that motion groupoids 
are 
 useful for physical modelling.
 A bicycle inner-tube with a puncture is useless
--- except as an example of a topological space, but a real 
rubber tube is by no means a uniform continuum. It is useful to 
be able to model the {\it emergence} of a realistic puncture.
A toy version of this is to compare the punctured Hopf-link and
the punctured unlink as objects in $M=\R^3$.
They are connected in $\FMot{M}$ but not 
in $\Mot{}$
--- see Examples~\ref{ex:punctured_hopf} and \ref{ex:fakeflow_circletoint}.
(It is interesting to note that the question of connection is
 the {\it same} in the category as in the underlying magmoid.)

This leads us naturally to the question of the relationship between 
the different engines at fixed $M$.
Localised to the `braid' setting, 
motion groups and mapping classes have been related, at least implicitly, in the particular cases of the braid group and loop braid group \cite{BB,damiani,dahm}.
The functor $\F\colon \Mot \to \mcg$ generalises this.
It yields immediate concrete applications:
providing a route to using results about mapping class groups to inform about about motion (sub)groupoids.

The `braid group' has several  
realisations, each with different flavours 
-- see for example \cite{BB, rolfsen,Birman_book}.
Each of the classical topological realisations  
 consists of some `concrete' elements, together with a composition and an equivalence \cite{BB,birman,goldsmith, dahm}.
These constructions are not generally pairwise equivalent when considering the concrete elements. 
Bridges between the concrete elements in the different topological realisations 
can be illustrated using the left and middle images in Fig.\ref{fig:my_label002}.
The middle picture
may represent a motion, where the whole space is moving in the way prescribed by Figure~\ref{fig:my_label01} and the induced movement of the points is marked.
Further the picture may represent an element in $\Topo_{\Braid}(N\times \II,M)$, where $N$ is two points and the image at each $t\in\II$ drawn ascending up the page, this is Artin's formulation of braids. 
 Alternatively the picture may represent a monotonic embedding of two unit intervals in the cylinder, which, passing to the image, gives a concrete element of the tangle category (see e.g. \cite{Kassel}); 
 or a path in the configuration space of two points in the disk. Then a concrete element of the mapping class groupoid is the endpoint of the motion.
Notice here that motions contain the most information: 
maps into  
the other settings `forget' information.

Even from this simplest possible perspective, then, 
we see that we should begin by constructing the engine $\Mot[-]$,
since it is the one  
which keeps track of the most information. 
We then investigate in which cases it is possible to forget information whilst retaining the same algebraic structure.
We will see that in general the `forgetting' functors we construct are neither surjective, nor injective.

Generalisations of  
all
the aforementioned constructions of braids to unknotted, unlinked, loops in $3$ dimensions exist. 
There they lead again 
to isomorphic groups \cite{damiani,dahm,brendlehatcher}.

\medskip

Yet another motivation for the study of motion groupoids is that we expect rich representation theory, as has already been found to be the case with motion groups. 
Firstly, in certain cases motion groups have the representation theoretically useful property that they can be finitely presented: a finite presentation of Artin's braid group was given in \cite{artin}, and 
Dahm gave a classification of the group of motions of  a configuration of $n$ unknotted, unlinked circles in $\R^3$ in terms of an isomorphism to the automorphism group of the free group with $n$ generators \cite[Sec.III]{dahm}, 
which 
is 
known to have finite presentation \cite[p.131,Th.3.2]{magnus}.
A plethora of approaches to the study of the representation theory of the latter group are present in the literature, see
\cite{baez, kadar, qiu, BFM} for example.
A motion group of a collection of Hopf links and trivial links in $\R^3$ is suggested to have a finite presentation  in \cite{damianiseiichi}. Similarly a presentation of a motion group of a `necklace' in $\R^3$ is given in \cite{bellingeri}, and its representation theory studied in \cite{BKMR}.

\medskip

Having 
discussed the motivation for
our construction we now introduce the content of the paper.

\medskip

Firstly, for both the motion groupoid and the mapping class groupoid, we have a version which fixes a distinguished subset $A\subset M$ pointwise, where $A=\partial M$ for example. These are denoted $\Mot^A$ and $\mcg^A$.
There are a few reasons for this considering this level of generality. For one thing it leads to interesting algebra;
and for another it facilitates passage to various natural generalisations
that are monoidal categories and/or canopolis or cubical algebras
(whose main discussion we will postpone).
This facility will also be necessary to find isomorphisms from the motion groupoid to the mapping class groupoid in the case $M=D^n$.

\medskip

In
this work we will give several equivalent realisations of both motions and of the motion groupoid; each of the different realisations will have useful applications.
In \S\ref{sec:motions} we begin by defining motions and giving two different compositions of motions. 
Define a 
map
$\W:\premots\times\Power M \rightarrow \Power (M\times\II)$  
by 
$\W(f,N) =   \bigcup_{t \in [0,1]} f_t(N) \times \{t\}$. This is {what we call} the {\em worldline} of the corresponding motion.
The $*$ composition given above represents the physical picture; thinking of motions as a gradual deformation of space over time, the composition of two motions is the first followed by the other.
This is made precise by Lemma~\ref{lem:concatW}, which says
that the worldline of the $*$ composition of two motions can be written in terms of a composition of the worldlines.
The $*$ composition is necessary to construct the functor $
\T:\Mot \rightarrow \FMot{M}$
in \S\ref{sec:artinbraids}; as well as for functors into generalised tangle categories.
The second is a pointwise composition $(g\cdot f)_t=g_t\cdot f_t$, which is introduced as a computational convenience as it simplifies a number of proofs. The compositions lead to distinct magmoids, although in the motion groupoid both compositions descend to the same well defined composition on equivalence classes.

Still in \S\ref{sec:motions}, 
we give our first construction of $\Mot$ (Theorem~\ref{th:mg}); the morphisms are equivalence classes of motions with the relation that $f\simm g\in \Mtcmag^*(N,N')$ if $\bar{f}*g$ is path homotopic to  an $N$\textit{-stationary motion} -- these are the motions which leave $N$ fixed setwise for all $t\in \II$.
We prove that this leads to a groupoid by performing
a two stage quotient, first identifying path equivalent motions, and then identifying each class that contains a stationary motion with the corresponding identity.
We start with this particular equivalence relation as this is a direct generalisation of the relation used in Dahm's motion groups \cite{dahm, goldsmith}, the construction we set out to generalise. 
There are also practical advantages, certain results are more straightforward to see with this relation; for example, Proposition~\ref{prop:dahm} that motions with the same worldline are equivalent.

In \S~\ref{ss:examples}
we will present some key examples to demonstrate the richness of our construction. 
The unified treatment of objects discussed above leads to questions about skeletons {(equivalent categories without extraneous isomorphisms \cite[IV.2]{MacLane})}, and about objects whose automorphism groups are related by canonical isomorphisms outside of the motion groupoid.
Note that the existence of a homeomorphism between subspaces, or indeed a homeomorphism of the ambient space sending one subspace to the other is not enough to ensure that the underlying sets are connected by a morphism in the motion groupoid. 
In Section~\ref{sec:Iexamples} we show that in $\Mot[\II]$, the existence of morphisms in $\Mot[\II](N,N')$ depends only on the topology of $N$ and $N'$ if $N$ and $N'$ are compact subspaces in the interior of $\II$, but on details of the embeddings when they are non-compact manifolds. We also have Example~\ref{ex:R} which shows that even for $1$-dimensional manifolds, automorphism groups can be non-trivial.
In \S~\ref{sec:outerisos} we give examples relating automorphism groups $\Mot(N,N)$ and $\Mot(N',N')$ when $N$ and $N'$ are not connected in the motion groupoid. 

In \S\ref{sec:schematics}
we give two
alternative realisations of motions from $N$ to $N'$ in $M$,
as certain elements of $\Topo(M\times \II, M)$ and of $\Topo^h(M\times \II,M\times \II)$. The former connects to a direct generalisation of Artin braids \cite{artin25,artin}. We develop this further in \S\ref{sec:braids}.
The latter suggests the existence of a map into a generalised tangle category, where concrete elements are embeddings $N\times \II\to M\times \II$, with conditions.
These realisations also lead to some useful schematic representations of motions.

\medskip

In Section~\ref{sec:equivalence}, we discuss alternative ways to formulate the relation $\simm$. 
Firstly, in Theorem~\ref{th:mg2}, we prove that $\simm $ defines the same relation on motions
as the relation $\simrp $ 
given above. The relation  {$\simrp$ is} the same relation used to construct relative homotopy sets, thus facilitating the use of the long exact sequence of homotopy groups 
to
investigate properties of
the functor $\F\colon \Mot^A \to \mcg^A$ in \S\ref{sec:mg_to_mcg}.

As noted above, a flow $f$, together with $N\subset M$, naturally gives rise to a subset $\W(f,N)  \subset M\times \II$, the {worldline} of the corresponding motion. If $N$ is a finite subset in the interior of $D^2$, and $f$ is a flow of $D^2$, then $\W(f,N)$ is a tangle in $D^2 \times \II $, and one such example is as shown on the middle image in Fig.~\ref{fig:my_label002}. If $N\subset \mathrm{int}(D^2)$ is the union of a point and a line, and $f$ the flow of $D^2$ depicted in the left-hand-side of Fig.~\ref{fig:my_label002}, then $\W(f,N)$ would look like
the right-hand-side of Fig.~\ref{fig:my_label002}. For the case of the motion of an unlink in $D^3$, the worldline is a disjoint union of embedded tubes in $D^3\times \II$. It should be observed that motions and worldlines are 
very 
different:
the former are paths in the homeomorphism group of $M$, the latter are subsets of $M\times \II$.

 In Section~\S\ref{sec:laminated} we prove that it is, in fact, possible to understand motion equivalence in terms of 
 level preserving isotopies between worldlines (Theorem~\ref{th:me_bflpai}). 
  This is a significant step towards proving the well-definedness, and potentially injectivity, of functors from motion groupoids into generalised tangle categories \cite{BaezDolan}, and from our understanding this is the only such result in the literature. It also has applications to defining representations of motion groups, as explored in \cite{Torzewska}.

 Precisely, in Theorem~\ref{th:me_bflpai}, we prove that two motions 
 and 
 are motion equivalent if, and only if, their worldlines 
 and 
 are \textit{level preserving ambient isotopic} subsets of $M\times \II$, through an isotopy fixing $M\times \{0\}$ and $M\times \{1\}$.
 The proof relies on the above result that motions are equivalent if and only if they are relative path-equivalent, along with Proposition~\ref{prop:dahm} that motions with the same worldline are equivalent.
 
 Still in Section~\ref{sec:laminated},
we also 
define a further relation on 
$\premots \times \Power M $
(and hence on motions)
by $(f,N)\sim (f',N')$
if $\W(f,N)=\W(f',N')$ or if 
$f,f'$ are path homotopic.
Let $\simww$ be the equivalence closure of this relation. Using Theorem~~\ref{th:me_bflpai}, we prove that $\simw$ defines the same relation on motions as $\simm$
(this is Theorem~\ref{th:motequiv_worldlineequiv})

In Section~\ref{sec:braids} we construct, for a manifold $M$, the category $\FMot{M}$. The morphisms in $\FMot{M}$ are {\it `\fakemotion{}s'}. We use this nomenclature since these are maps $N\times \II\to M$ thus appear to specify a movement of $N$ in $M$, although there does not always exist a motion $M\times \II\to M$, which restricts to a given \fakemotion{}.
We prove that there is a functor $\T\colon \FMot{M}^A\to \Mot^A$ (Theorem~\ref{thm:MottoFmot}).
In Section~\ref{sec:artinbraids} we restrict to the case that $N$ is a finite set of points in the interior of a manifold. In this case we get that $\T$ restricts to an isomorphism on automorphism groups. This makes the connection with Artin's presentation of the braid group \cite{artin}, thus proving that there exists subgroupoids of $\Mot$ which have finite presentation.
More generally $\T$ is an isomorphism of groupoids in the case of finite points (Theorem~\ref{th:groupoid_artinbraid}).
For $K$ a finite subset in the interior of a disk, the results in this section also rigorously prove isomorphisms $\Mot[\R^2](K,K)\cong \Mot[D^2]^{\partial D^2}(K,K)\cong \Mot[D^2\setminus \partial D^2](K,K)$, explaining why these various settings are often used interchangeably in the literature.

\medskip 

In Theorems~\ref{th:mcg} and \ref{th:mcga} we have the mapping class groupoid and its $A$-fixing version. These generalise the mapping class group of a submanifold $N$ in a manifold $M$ \cite{birman} cf. for example \cite{farb,ivanov}.
For $C\subset D^3$ a subset consisting of unknotted, unlinked circles in the $3$-disk, then $\mcg[D^3]^{\partial D^3}(C,C)$ is the extended loop braid group as in \cite{damiani}.
We also have some examples demonstrating the utility of results on homeomorphism spaces available in the literature.

\medskip

In \S\ref{sec:mg_to_mcg} we relate the motion groupoid and mapping class groupoid formalisms, giving the functor $\F\colon \Mot^A\to \mcg^A$ in Theorem~\ref{le:functor_Mot_to_MCG} and proving in Theorem~\ref{th:mot_to_mcg} that this is an isomorphism if the space of $A$-fixing homeomorphisms of $M$ is path connected and has trivial fundamental group.
We also give examples to demonstrate the utility of the functor.
In \S\ref{sec:D^m}, we show that $\Mot[D^m]^{\partial D^m}\cong \mcg[D^m]^{\partial D^m}$ for all integer $m$. The case $m=2$ gives an isomorphism between two realisations of the braid group. We also show, in \S\ref{sec:D2_movebdy}, that if we remove the condition that $\partial D^m$ is fixed, we no longer have an isomorphism.
In \S\ref{ex:S1} and \S\ref{ex:2sphere} we consider the (non-isomorphism) functor $\F$ in the case $M=S^1$ and $M=S^2$.
In each of these examples we rely heavily on known results on mapping class groups and spaces of homeomorphisms \cite{hamstrom}.

\medskip

To complete the introduction, we make a couple of comments on our technical choices.

In the current paper we do all constructions in the topological category, following \cite{dahm,goldsmith}. 
We expect a construction in the smooth category to be possible, and this is certainly and interesting direction of future study.
Various smooth motion group constructions exist in the literature, see 
 \cite{baez,qiu,Wattenberg}.

For convenience we have formulated everything for ambient space $M$ a manifold.
 However, the constructions in this paper will also lead to a theory of motion groupoids, fake motion groupoids and mapping class groupoids  of spaces if the word ``manifold''   is substituted by ``locally compact, locally connected, Hausdorff topological space'' 
 e.g. finite graphs, or locally finite CW-complexes, which are Hausdorff, locally connected and locally compact \cite[Corollary 1.3.3, Proposition 1.5.10]{FP}.
 Note, however, that our results relating these groupoids in the manifold case may no longer hold in this generality.

\subsection{Paper Overview} \label{ss:overview}
In Section~\ref{sec:prelim}
we begin 
by 
introducing notation and recalling technology that will be useful in our construction.
In Section~\ref{ss:compact_open} we recall the compact-open topology on spaces of continuous maps. 
We then give conditions under which the usual product-hom adjunction in the category of sets lifts to an adjunction in the category of topological spaces.
In Section~\ref{sec:magmoids} we introduce magmoids and congruences on magmoids as a tool for constructing groupoids.
In Section~\ref{sec:path} we obtain the fundamental groupoid of a 
topological space $X$
(Proposition~\ref{pr:fundamentalgroupoid}) from the \textit{path magmoid}. 
In  Section~\ref{ss:selfhomeos} we
give the construction of a groupoid 
of self-homeomorphisms $\Hom$ corresponding to a manifold $M$, with object class the power set $\Power M$ (Definition~\ref{Def:homeoMA}).

In Section~\ref{sec:motions} the first main theorem is Theorem~\ref{th:mg}, 
the construction of the motion groupoid 
$\Mot$ 
of a  manifold $M$.
Picking a single set in $\Power(M)$ and looking at the group of automorphisms we get back the motion group.
We also have Theorem~\ref{le:motion_MxI} which says that motions $N$ to $N'$ in $M$ are in correspondence with a certain subset of $\TOPO^h(M\times \II, M\times \II)$.
In Section~\ref{ss:examples} we have 
examples.

In Section~\ref{sec:rel groupoid} we have  Theorem~\ref{th:mg2} which says that relative path-equivalence defines the same equivalence relation on motions as motion equivalence. We then have  Theorems~\ref{th:lpai_setwise},\ref{th:me_bflpai} and \ref{th:motequiv_worldlineequiv}, which each prove various equivalence relations on worldlines of motions define the same relation as motion equivalence.

In Section~\ref{sec:braids} we have
the construction of a groupoid of \fakemotion{}s (Lemma~\ref{lem:Fake-braids}), and the construction of the functor from $\T\colon \FMot{M}^A\to \Mot^A$ (Theorem~\ref{thm:MottoFmot}). There is also Theorem~
\ref{th:geometric_braids} which says that there is an isomorphism $\Mot^{\partial M}(K,K)\cong \FMot{M}^{\partial M}(K,K)$ where $K$ is a  finite set in the interior of the manifold, and Theorem~\ref{th:groupoid_artinbraid}  
which says this extends to a groupoid isomorphism for the subgroupoids of all such $K$.
Finally we have Theorem~\ref{th:geometricbraid_nofix}
which says that 
the same isomorphisms hold without fixing the boundary in  $\Mot$.

In Section~\ref{sec:mcg} we construct the mapping class groupoid of a manifold $M$ (Theorem~\ref{th:mcg}).  Theorem~\ref{th:mcga} is a subset-fixing version.

In Section~\ref{sec:mg_to_mcg} we construct the functor $\F\colon \Mot^A\to \mcg^A$ (Theorem~\ref{th:mot_to_mcg}).
In Theorem~\ref{th:mot_to_mcg} we
prove that $\F$ is an isomorphism if $\pi_0(\TOPO^h(M,M),\id_M)$ and $\pi_1(\TOPO^h(M,M),\id_M)$ are trivial.
In Section~\ref{ss:longexactseq} we give some examples demonstrating the utility of the functor.

\newpage

\newpage
\vspace{-1cm}
\glsaddall
\setlength{\glsdescwidth}{0.8\textwidth}
\printglossary[style=super]

\newpage

\section{Preliminaries} \label{sec:prelim}

In this section we recall concepts that we will need for our construction, and fix notation. First we have the compact-open topology on continuous maps, which leads to a product-hom adjunction in the category of topological spaces in Section~\ref{sec:tensor_hom}.
We then have magmoids, and groupoids defined in terms of magmoids in Section~\ref{sec:magmoids}.
Our first example here is the path magmoid of a topological space (Definition~\ref{de:pathspace}).
In Section~\ref{sec:congruence} we have congruences on magmoids, and on groupoids, which lead to quotient structures.
Finally we give a construction of the fundamental groupoid of a topological space; as a quotient of the path magmoid.

\subsection{The compact-open topology}\label{ss:compact_open}

Denote by $\Topo$ the category with topological spaces (spaces) as objects, continuous maps as morphisms and the usual composition.
Thus we denote the set of morphisms from a space $X$ to a space $Y$ by $\Topo(X,Y)$. 
We assume familiarity with $\Topo$ (see e.g. Chapter 1 of \cite{dieck}),
thus with product spaces. 
We denote the categorical product of spaces $X$ and $Y$ by $X\times Y$.
We similarly assume familiarity with $\Set$, the category of sets, functions and function composition.

\medskip

To formalise `flows' in a manifold, we will require a topology on the sets $\Topo(X,Y)$. 
For this we use the compact-open topology.
Here we give the definition and give some results demonstrating its intuitive naturality (Prop.\ref{pr:comsup1}).

\defn{
For $X$ a set, $\Power X$ denotes the power set.
}

\defn{ 
Given a set $X$, and a
subset $Y$ of $\Power X$ with {$\cup_{A \in Y} A=X$}, 
we write $\overline{Y}$ for the 
topology closure of $Y$. 
Hence the open sets in the topological space $(X,\overline{Y})$ are arbitrary unions of finite intersections of elements in $Y$.
We say that $Y$ is a {\em subbasis} of
$(X,\tau)$ if $\overline{Y} = \tau$.
(NB: in general, $\tau= \overline{Y}$ does not uniquely determine $Y$.)
}

\defn{ 
A {\em neighbourhood basis} of $(X,\tau)$ at $x\in X$ is a subset $B\subseteq \tau$, whose members are called basic neighbourhoods of $x$, such that every neighbourhood\footnote{ Our convention is that a neighbourhood of $x$ is a subset of $X$ containing an open set containing $x$.}  of $x$ contains an element of $B$.
}
\defn{ \label{de:compact-open}
Let $(X,\tau_X)$ and $(Y,\tau_Y)$ be
topological spaces, then the 
{\em compact-open} topology 
$\tauco{X}{Y}$
on $\Topo(X,Y)$ has subbasis 
\[
b_{XY}=\{\coball{X}{Y}{K}{U}\vert K\subseteq X \text{ is compact}, U\in\tau_Y \}
\]
where
		\[
	    \coball{X}{Y}{K}{U}=\left\{f\colon X \to Y \, \vert \, f(K)\subseteq U\right\}.
		\]
That is $\tauco{X}{Y} = \overline{b_{XY}}$.
\footnote{
There are two conventions for compact-open topology. The one written here (which is the classical one) and the one where we additionally impose that each $K$ in the $\coball{X}{Y}{K}{U}$ be Hausdorff. For example, \cite[Chapter 5]{May} takes the latter convention. This creates an a priori smaller set of open sets in the function space. However there will be no ambiguity issues in this paper as we will only be working with Hausdorff topological spaces.}
}

\prop{Let $Y$ be a space, and $X$ be the space with a single point. Then the $\tauco{X}{Y}$ is the same in the obvious sense as the topology on $Y$.}
\begin{proof}
Maps from $X$ to $Y$ can be labelled by their image in $Y$. 
The only compact set $K$ is the single point set $X$.
The set $\coball{X}{Y}{K}{U}$ is the set of maps labelled by the elements of $U$.
\end{proof}

\prop{\label{coprod} Let $Y$ be a space, and $X$ be the space consisting of $n$ points with discrete topology. Then $\tauco{X}{Y}$ is the same in the obvious sense as the topology on $Y^n=Y\times \ldots\times Y$.}
\begin{proof}
Maps from $X$ to $Y$ are in bijection with tuples $(y_1,\dots,y_n)\in Y^n$ where $y_i$ is the image of $x_i\in X$.
All subsets $K$ of $X$ are compact and we have, using this bijection, 
\[
\coball{X}{Y}{K}{U}\cong\{(y_1,\dots,y_n)\;\vert \; y_i\in U \text{ if } x_i\in K\}.
\]
Hence elements of the subbasis of $\tauco{X}{Y}$ are open sets in $Y^n$.

Basis elements in the topology on $Y^n$ are obtained from the compact open topology as follows.
Let $U^n$ be a basis open set in the topology on $Y^n$. Then $U^n$ is of the form $U_1\times\ldots \times U_n$, with $U_1,\dots, U_n$ open in $Y$.
Now
\[
\coball{X}{Y}{\{x_i\}}{U_i}= \{(y_1,\dots,y_n)\;\vert \; y_i\in U_i\},
\]
and $\cap_{i}\coball{X}{Y}{\{x_i\}}{U_i}=U^n$.
\end{proof}

In the case $Y$ is a metric space we have an interpretation of the compact-open topology in terms of a metric on $\Topo(X,Y)$. 
\prop{ \label{pr:comsup1}
{\rm (A.13 in \cite{hatcher}) }
Let $X$ be a compact space and $Y$ a metric space with metric $d$.
Then 
\\
(i) the function
\[ d'(f,g) \; := \; \sup_{x \in X} d(f(x),g(x))  \]
is a metric on $\Topo(X,Y)$; and 
\\
(ii)
the compact open topology on $\Topo(X,Y)$ is the same as the one defined by the metric
$d'$.
}
\begin{proof}
See A.13 in \cite{hatcher}.
\end{proof}

\subsection{The space \texorpdfstring{$\TOPO(X,Y)$}{} and the product-hom adjunction}\label{sec:tensor_hom}

In addition to its intuitive naturality,
the compact-open topology allows us to find a partial lift of the classical product-hom adjunction in $\Set$ to an adjunction in $\Topo$ (Theorem \ref{th:tensorhom}).
We will make extensive use of this adjunction throughout this paper.
First, in Theorem~\ref{le:motion_MxI} to prove that motions in $M$ 
to have an additional interpretation as homeomorphisms from $M\times \II$ to $M\times \II$.
Then later in Section~\ref{sec:laminated}, to understand motion equivalence in terms of a relation on worldlines, and again in Lemma~\ref{le:MCG_pi0}.

This adjunction holds subject to some conditions which are not too restrictive for us.
In particular, the compact-open topology allows us to define a right-adjoint to the functor $-\times Y\colon \Topo \to \Topo$ (see Lemma~\ref{le:product_functor}), when $Y$ is a locally compact Hausdorff space. (The case $Y=[0,1]$ was one of the examples given in the original reference on adjoint functors \cite[pp 294]{Kan}.)

\medskip

In this section we also give some results regarding the continuity of the composition in $\Topo$ with respect to the compact open topology.

\medskip 

We will use capital $\TOPO(X,Y)$ to indicate the morphism set 
$\Topo(X,Y)$ considered as a space with the compact-open topology, so
\[
\TOPO(X,Y)=(\Topo(X,Y),\tauco{X}{Y}).
\]

\lemm{\label{le:product_functor}
Fix a topological space $Y$. There exists a functor  
$-\times Y\colon\Topo \to \Topo$ constructed
as follows.
A space $X$ is sent to the product space $X\times Y$. 
A continuous map $f\colon X\to X'$ is sent to the map 
$f\times \id_Y \colon X\times Y\to X'\times Y$,
$(x,y)\mapsto (f(x),y)$. 
We will refer to this as the {\em product functor}.
}
\begin{proof}
We first show that, for a map $f\colon X\to X'$, $f\times \id_Y$ is a continuous map $X \times Y$ to $X'\times Y$.
Let $U'\times V$ be a basis open set in $X'\times Y$. Then the preimage under $f\times \id_Y$ is $f^{-1}(U')\times V$ which is open since $f$ is continuous. 
It is clear that the product functor preserves the identity and respects the composition.
\end{proof}

\lemm{\label{le:homfunctor}
Fix a topological space $Y$. 
There exists a functor ${\TOPO( Y,-)\colon}\Topo \to \Topo$ constructed as follows.
A space $Z$ is sent to the space $ \TOPO(Y,Z)$.
A continuous map $f\colon Z\to Z'$ is sent to
$f\circ -\colon \TOPO(Y,Z)\to \TOPO(Y,Z')$,
$g\mapsto f\circ g$. 
We will refer to this as the {\em hom functor}.
}

\begin{proof}
We first show that $f\circ -$ is a continuous map. 
Open sets in the subbasis of $\tauco{Y}{Z'}$ are of the form $\coball{Y}{Z'}{K}{U}$ for some $K\subseteq Y$ a compact set and $U\subseteq Z'$ an open set.
The set $f^{-1}(U)$ is open in $Z$ since $f$ is a continuous map.
Hence $\coball{Y}{Z}{K}{f^{-1}(U)}$ is an open set in $\tauco{Y}{Z}$.
We show that the inverse image of $\coball{Y}{Z'}{K}{U}$ under  $f\circ -$ is precisely $\coball{Y}{Z}{K}{f^{-1}(U)}$.
For any $g\in \coball{Y}{Z}{K}{f^{-1}(U)}$ we have $f\circ g\in \coball{Y}{Z'}{K}{U}$.
Conversely suppose $h\in \coball{Y}{Z'}{K}{U}$ can be written in the form $f\circ g'$ for some $g'\in \TOPO(Y,Z)$, then $g'\in \coball{Y}{Z}{K}{f^{-1}(U)}$. 

It is straightforward to see that $\TOPO(Y,-)$ preserves the identity and respects composition.
\end{proof}

For a category $\CC$, we will use $\Homf\colon \CC^{op}\times \CC\to \Set$ to denote the usual bifunctor, see \cite[pgs.~34,38]{MacLane}. (Note that by fixing first argument and topologising the image of objects, we get back the hom functor of Lemma~\ref{le:homfunctor}.)

\medskip

The following Lemma gives conditions under which the usual hom-tensor correspondence from $\Set$
is well-defined in $\Topo$.

\lemm{\label{th:tensorhom}
Let $Y$ be a locally compact Hausdorff topological space.
The product functor 
$-\times Y $
is left adjoint to the hom functor $\TOPO(Y,-)$.
In particular,
for objects $X,Y,Z\in\Topo$,
 this gives a set map
\ali{
\Phi\colon\Topo(X, \TOPO(Y,Z)) &\to \Topo(X\times Y, Z) \\
f &\mapsto ((x,y)\mapsto f(x)(y))
}
that is a bijection{, natural in the variables $X$ and $Z$}.
\footnote{There is in fact an adjustment of the compact open topology which, with an adjustment to the product, gives an adjunction without the need to restrict $Y$. See \cite[Sec.5.9]{brownt+g} for more information. 
}
}

\begin{proof}
    	That we have a bijection of sets is proved in Proposition A.14 of \cite{hatcher}. 
	It remains to prove that this bijection is natural.
	Suppose we have continuous maps $\alpha\colon X'\to X$ and 
	$\beta\colon Z\to Z'$, then we must show we have a commuting diagram of the form
	\[
	\begin{tikzcd}
		\Topo(X,\TOPO(Y,Z)) \ar[r,"{\Phi}"]\ar[d,"{\Homf(\alpha,-\circ \beta)}"']& \Topo(X\times Y,Z) \ar[d,"{\Homf(\alpha\times \id_Y,\beta)}"]\\
		\Topo(X',\TOPO(Y,Z'))\ar[r,"{\Phi}"'] & \Topo(X'\times Y,Z').
	\end{tikzcd}
	\]
	Looking first at the left hand vertical arrow, a map $f\colon X\to \TOPO(Y,Z)$ is sent to the map $X'\to \TOPO(Y,Z')$, $x'\mapsto \beta \circ f(\alpha(x'))$,
	and then to $(x',y)\mapsto (\beta \circ f(\alpha(x')))(y)$ in $\Topo(X'\times Y,Z')$.
	Going first along the top, a map $f$ is sent to the map $ X\times Y \to Z$, $(x,y)\mapsto f(x)(y)$ and then to the map $X'\times Y\to Z'$ defined by $(x',y) \mapsto(\beta \circ f(\alpha(x')))(y)$. 
\end{proof}

For any space $X$ then $\Topo(X,X)$ is a monoid,
with identity the identity map. 
The subset of maps which are set bijections is a submonoid.
Let $\Topo^h$ be the subcategory of $\Topo$
with the same objects as $\Topo$ and morphisms which are homeomorphisms.
(Note that the indicated subset is in fact closed under composition.)
Then $\Topo^h(X,X)$ is the group of homeomorphisms $\sh{f}\colon X \to X$.
Denote by $\TOPO^h(X,X)$ the subspace of $\TOPO(X,X)$ with underlying set $\Topo^h(X,X)$.

In Section~\ref{sec:motions} we will be interested in formalising how certain paths in $\TOPO^h(M,M)$, where $M$ is a manifold, induce  `motions' of subsets in $M$.
We will introduce a `pointwise' composition and inverse of such motions which requires that $\TOPO^h(M,M)$ is a topological group.

\begin{theorem}[{\cite[Thm.4]{arens}}]\label{le:top_group}
If $X$ is a locally connected, locally compact Hausdorff space then $\TOPO^h(X,X)$, is a topological group.
(This means the composition $(\sh{f},\sh{g})\mapsto \sh{g} \circ \sh{f} $ and the map $\sh{f}\mapsto \sh{f}^{-1}$ are both continuous.)
\end{theorem}

\begin{proof}
See Section~\ref{sec:TopGroupProof}.
\end{proof}

Notice that if a space $X$ satisfies the conditions of Theorem~\ref{le:top_group}, then $X$ also satisfies the conditions of Lemma~\ref{th:tensorhom}.

\medskip 

In general, for fixed topological spaces $X$, $Y$ and $Z$, the composition map $\TOPO(X,Y) \times \TOPO(Y,Z) \to \TOPO(X,Z)$ is continuous in each variable, despite the fact that it is not always continuous as a function of two variables (\cite[page 259, 2.1 and 2.2.]{dugundji}). We give the proof here, and will use this weaker result (in comparison to the previous lemma) where possible to emphasise where 
continuity of the composition map is really necessary for a given construction.

\begin{lemma}\label{le:comp_continuous_eachvariable}
Let $Y$ be a space. \\
(I) For any $\sh{g}\in \TOPO(X,Y)$, the map $-\circ\sh{g}\colon\TOPO(Y,Y)\to \TOPO(X,Y)$, $\sh{f}\mapsto \sh{f}\circ\sh{g}$ is continuous, and\\
(II) for any $\sh{g}\in \TOPO(Y,Z)$, the map
$\sh{g}\circ-\colon\TOPO(Y,Y)\to \TOPO(Y,Z)$, $\sh{f}\mapsto \sh{g}\circ\sh{f}$ is continuous.
\end{lemma}
\begin{proof}
(II)  For a subbasis open set $\coball{X}{Y}{K}{U}$ (with notation as in Definition~\ref{de:compact-open}), with $K\subseteq X$ compact and $U\subseteq Y$ open, 
	we have $\sh{f\circ g} \in \coball{X}{Y}{K}{U}\iff \sh{f}(\sh{g}(K))\subseteq U \iff \sh{f}\in \coball{Y}{Y}{\sh{g}(K)}{U}$, where the latter subset of the function space is open.\\
(I) For a subasis open set $\coball{Y}{Z}{K}{U}$ 
with $K\subseteq Y$ compact and $U\subseteq Y$ open, we have $\sh{g\circ f}\in \coball{Y}{Z}{K}{U}\iff \sh{g}(\sh{f}(K))\subseteq U\iff \sh{f}\in \coball{Y}{Y}{K}{\sh{g}^{-1}(U)}$, where the latter subset of the function space is open.
\end{proof}

\subsection{Groupoids and magmoids}
\label{sec:magmoids}

In this work constructions of groupoids are a recurrent theme.
Such constructions will often start from something 
`concrete' with a 
non-associative composition. 
Equivalence classes of these concrete things eventually become the morphisms of the constructed groupoid.
So it will be useful to have a general machinery for studying such constructions.
For example we can think of the underlying idea of a category as  objects, morphisms between objects, and a composition which is not necessarily associative, or unital 
- a categorified magma, or  {\it `magmoid'}.
We can then study congruences on these magmoids, some of which will lead to groupoids.

\defn{ \label{de:magmoid}
	A {\em magmoid} ${\mM}$ 
	is a triple 
	$$
	{\mM} \; = \;  (Ob(\mM),\mM(-,-),\Delta_{\mM})
	$$ 
	consisting of
	\begin{itemize}
		\item[(I)] a collection $Ob(\mM)$ of \textit{objects},
		\item[(II)] for each pair $X,Y\in Ob(\mM)$ a set $\mM(X,Y)$ of \textit{morphisms from $X$ to $Y$} (we use $f\colon X\to Y$ to indicate that $f$ is a morphism from $X$ to $Y$), and 
		\item[(III)] for each triple $X,Y,Z \in Ob(\mM)$ a {\it composition}
		\[
		\Delta_\mM\colon \mM(X,Y)\times \mM(Y,Z)\to \mM(X,Z).
		\]
	\end{itemize}
}
\noindent Magmoids can be compared with cubical sets with composition, as defined in \cite{Brown_Higgins_cubes}. Quotients there lead to cubical $n$-groupoids.

\medskip
For our example below we will need some notation.

\defn{ \label{de:pathspace}
	Let $X$ be a topological space.
	An element of $\Topo(\II,X)$ is called a {\em path} in $X$ and
	$ \TOPO(\II,X)$ is
	called the {\em path space of $X$}.
	\\
	Notation: Let $\gamma \in \Topo(\II,X)$.
	We use $\gamma_t$ for $\gamma(t)$. We  say $\gamma$ is a path from $x$ to $x'$
	when $\gamma_0 = x$ and $\gamma_1=x'$.
	For $x,x'\in X$, let 
	\[\Path X(x,x')=\{\gamma \colon \II \to X \;\vert\; \gamma\in \Topo(\II,X), \, \gamma_0=x, \, \gamma_1=x'\}.
	\]
	}

\newcommand{\PsiOps}{\Gamma_{\!\frac{1}{2}}} 

\prop{ \label{pr:path_comp}
	Let $X$ be a topological space.
	For any $x,x',x''\in X$, there exists a composition 
	\ali{
		\PsiOps \colon \Path X(x,x')\times \Path X(x',x'') &\to \Path X(x,x'')\\
		(\gamma, \gamma')&\mapsto \gamma'\gamma
	}
	(note the null composition symbol here) with 
	\begin{align} 
		\label{eq:pathcomp} 
		(\gamma'\gamma)_t= \begin{cases}
			\gamma_{2t} & 0\leq t\leq 1/2, \\
			\gamma'_{2(t-1/2)} & 1/2\leq t \leq 1.
		\end{cases}
	\end{align}	
	(Note the convention to choose distinguished point $t=1/2$, we could have chosen any $a \in (0,1)$.)
}
\begin{proof}
	We have $\gamma_1=\gamma'_0$ so Equation \eqref{eq:pathcomp} defines a continuous map.
	Notice $(\gamma'\gamma)_0=\gamma_0=x$ and $(\gamma'\gamma)_1=\gamma'_1=x''$.
	Hence $\gamma'\gamma\in \Path X(x,x'')$.
\end{proof}
\rem{We find the above convention for ordering path composition to be convenient as we will later map paths to functions.}

\defn{\label{de:path_magmoid}
Let $X$ be a topological space. 
From Prop.\ref{pr:path_comp} we may 
define the
{\em path magmoid} 
$$
\Path X  \;=\; (X,\Path X(-,-), \PsiOps ) .
$$
}
Note that the magmoid $\Path X$  is neither associative nor unital.

\defn{Let $\mM$ and $\mM'$ be magmoids. A {\em magmoid morphism} 
	$F\colon \mM\to \mM'$
	is a map sending each object $X\in Ob(\mM)$ to an object $F(X)\in Ob(\mM')$ and each morphism $f\colon X\to Y$ in $\mM$ to a morphism $F(f)\colon F(X)\to F(Y)$ in $\mM'$ such that for any  morphisms $f,g \in \mM$ 
	\[
	F(\Delta_{\mM}( f,g))= \Delta_{\mM'}( F(f),F(g))
	\]
	wherever $\Delta_\mM(f,g)$ is defined.
}

 Magmoid representation theory can be framed in terms of magmoid morphisms in a category of magmoids.
\begin{proposition}\label{pr:magmor_comp}
    For any set $S$ of magmoids, there exists a small category with objects $S$ and morphisms all magmoid morphisms between the elements of $S$.\\
	The \ppm{(partial)} composition of magmoid morphisms sends a pair of magmoid morphisms $F\colon \mM\to \mM'$ and $F'\colon \mM'\to \mM''$ to $F'\circ F\colon \mM\to \mM''$ with
	\[
	F'\circ F(f\colon X\to Y )=
	F'(F(f))\colon F'(F(X))\to F'(F(Y)).
	\]
\end{proposition}
\begin{proof}
	It is straightforward to check that $F'\circ F$ is well defined and is a magmoid morphism.
	Associativity and identities follow from the properties of morphisms in $\Set$.
\end{proof}

\subsection{Algebraic congruence and magmoid congruences}\label{sec:congruence}

\defn{A {\em congruence} $C$ on a magmoid $\mM$ consists of, for each pair $X,Y\in Ob(\mM)$ an equivalence relation $R_{X,Y}$ on $\mM(X,Y)$, such that 
	$f'\in[f]$, $g'\in[g]$ implies $\Delta_\mM(f',g')\in [\Delta_\mM(f,g)]$ where defined.
}

\exa{\label{ex:congruence_paths}
Let $X$ be a topological space.
For each $x,x'\in X$, define a relation on $\Path X(x,x')$ by $\gamma\sim \gamma'$ if there exists $f\in\Topo^h_{\partial \II}(\II,\II)$ such that $\gamma_t=\gamma'_{f(t)}$ for all $t\in \II$. (Intuitively, $\gamma'$ is a reparametrisation of $\gamma$.) Then it can be shown that the family of relations $(\Path X(x,x'),\sim)$ is a congruence on $\Path X$.}

\defn{Let $\mM=(Ob(\mM),\mM(-,-),\Delta_{\mM})$ be a magmoid and $C$ a congruence on $\mM$. The {\em quotient magmoid} of $\mM$ by $C$ is $\mM/C=(Ob(\mM),\mM(X,Y)/R_{X,Y},\Delta_{\mM/C})$ where for each triple $X,Y,Z\in Ob(\mM/C)$: 
		\ali{
			\Delta_{\mM/C}\colon\mM/C(X,Y)\times \mM/C(Y,Z)&\to \mM/C(X,Z)\\
			([f],[g])&\mapsto [\Delta_{\mM}(f,g)].	
		}
		(That the composition is well defined follows directly from the definition of a congruence.)
}

In practice we will often use the notation for the composition in $\mM$ to denote also the composition in the quotient.

\medskip
It follows directly from the definition of congruence that we have the following.

\lemm{
	Let $\mM$ be a magmoid and $C$ a congruence on $\mM$.
	There is an induced magmoid morphism  $\tilde{C}
	\colon  \mM \to \mM/C$, called the {\em quotient morphism}, which is the identity on objects and which sends morphisms to their equivalence class under $C$. \qedhere
}

\begin{defin}
	Let $\mM$ be a magmoid, $C$ a congruence.
	If there exists a magmoid $\mM'$ and full, non-identity magmoid morphisms $G\colon \mM \to \mM'$ and $H\colon \mM' \to \mM/C$ such that $\tilde{C}=H\circ G$, we say that the congruence $C$ has a {\em factor}.
\end{defin}

The classical definition of a groupoid (see e.g. \cite[Ch.~6]{brownt+g}) can be given as a magmoid plus extra structure. We give it here to fix notation.
We will see shortly that congruences have some useful characterisations when considered on magmoids which are also groupoids.

\defn{
	A {\em groupoid} ${\GG}$ is a tuple $\GG=(Ob(\GG),\GG(-,-),*_\GG,1_-,(-)\mapsto (-)^{-1})$ consisting of a magmoid $(Ob(\GG),\GG(-,-),*_\GG)$ such that 
	$Ob(\GG)$ is a set, and 
	\begin{itemize}
		\item[(IV)] for each $X\in Ob(X)$ a morphism
		$1_X\in \GG(X,X)$		
		called the \textit{identity};
		\item[(V)] for each pair $(X,Y)\in Ob(\GG)\times Ob(\GG)$ a function
		\ali{
		(-)^{-1}\colon \GG(X,Y)&\to \GG(Y,X)\\
		f&\mapsto f^{-1}
	}
		called the \textit{inverse assigning}, or just inverse, function;
	\end{itemize}
	such that the following axioms are satisfied.
	\begin{itemize}
		\item[$(\GG1)$] \textbf{Identity law:} for any morphism $f\colon X \to Y$, we have $1_Y*_\GG f = f = f*_\GG 1_X$.
		\item[$(\GG2)$] \textbf{Associativity:} for any triple of morphisms $ f\colon X\to Y$, $g\colon Y \to Z$ and $h\colon Z\to W$ we have $h*_\GG(g*_\GG f)=(h*_\GG g)*_\GG f$.
		\item [$(\GG3)$] \textbf{Inverse:}
		for any morphism $f\colon X\to Y$, we have $f^{-1}*_\GG f=1_X$ and $f*_\GG f^{-1}=1_Y$.
	\end{itemize}
We will sometimes replace $-$ notation with generic symbols where convenient. Let $\GG$ be a groupoid. By abuse of notation we will refer also to the underlying magmoid as $\GG$. 
}

\rem{
Note that the identities and inverses of a groupoid $\GG$ are uniquely determined from the underlying magmoid of $\GG$.}

\rem{Magmoid morphisms between groupoids automatically send identities to identities and inverses to inverses. So a functor of groupoids is simply a magmoid morphism between  underlying magmoids. This is not true for categories, where preservation of identities does not follow automatically from the preservation of composition. }

\prop{\label{pr:gpd_cong}
	Suppose $\GG=(Ob(\GG),\GG(-,-),*_G,1_,(-)^{-1})$ is a groupoid. Then for any congruence $C$ on $\GG$, there is a {\em quotient groupoid} $\GG/C=(Ob(\GG),\GG/C(-,-),*_{\GG/C},[1_{-}],[f]\mapsto [f^{-1}])$.}
\begin{proof}
    	$(\GG1)$ For all $[f]\colon X\to Y$ we have 
	\[
	[f]*_{\GG/C}[1_X]= [f*_{\GG} 1_X]=[f]=[1_Y*_{\GG} f ]=[1_Y]*_{\GG/C}[f].
	\]
	$(\GG2)$ Let $[f],[g],[h]$ be composable morphisms in $\GG/C$. Then
	\[
	[h]*_{\GG/C}([g]*_{\GG/C}[f])=h*_{\GG} g*_{\GG} f=([h]*_{\GG/C}[g])*_{\GG/C}[f].
	\]
	$(\GG3)$
	Any $[f]\in \GG/C(X,Y)$ has inverse $[f^{-1}]$ since 
	\ali{
		[f^{-1}]*_{\GG/C} [f]=[f^{-1}*_\GG f]=[1_X], \; \text{ and} &&
		[f]*_{\GG/C}[f^{-1}]=[f*_{\GG}f^{-1}]=[1_Y].&\qedhere
	}
\end{proof}
 
As discussed in Section~\ref{ss:stat_mot} below,
we are interested in starting from a magmoid, which describes a physical system,
and applying congruences until we arrive at a finitely generated category
(hopefully without pushing the interesting physics into the kernel).
Often we will find it convenient to do this by passing through a factor. 
When this factor is a groupoid $\GG$ we can construct a congruence on $\GG$ from a subgroupoid which is \textit{normal} and thus obtain a quotient \ppm{groupoid},
mirroring quotienting groups by normal subgroups.
We make this explicit here.

Everything in 
the remainder of
this section can be found in Section 1.4.3 of \cite{browngrpds}.
 
 \medskip
 
A subgroupoid $\HH$ of groupoid $\GG$ is said to be {\em wide} if $Ob(\HH)=Ob(\GG)$.
 
\defn{Let $\GG$ be a groupoid and $\HH$ a wide subgroupoid.  $Ob(\HH)=Ob(\GG)$. 
Then $\HH$ is said to be {\em normal} if
  for any morphism $g\colon X\to Y$ in $\GG$ and any $h\colon Y\to Y$ in $\HH$ we have $g^{-1}*_\GG h*_\GG g\colon X\to X$ is in $\HH$.\\
  We say $\HH$ is {\em totally disconnected} if for any $X,Y\in Ob(\HH)$ with $X\neq Y$ we have $\HH(X,Y)=\emptyset$.
 }

\lemm{\label{le:nsubgrpd_cong}
(See e.g. \cite[8.3.1]{brownt+g}.)
	Let $\GG$ be a groupoid and $\HH$ a normal, totally disconnected subgroupoid. For each $X,Y\in Ob(\GG)$ and $g,g'\in \GG(X,Y)$ the relation $g\sim g'$ if $g'^{-1}*_{\GG}g\in \HH$, is an equivalence relation on $\GG(X,Y)$.
	Moreover all such relations together form a congruence on $\GG$.\\
	We will denote this congruence also by $\HH$, the meaning will be clear from context. \qed}

\rem{Note that this is the weakest congruence such that all morphisms of the form $h\colon X\to X$ in $\HH$ become equivalent to the appropriate identity.}

\subsection{Interval 
\texorpdfstring{$\II=[0,1]$}{I=[0,1]}, 
space \texorpdfstring{$\TOPO(\II,X)$}{} 
and path-homotopy}
\label{sec:path} 
In this section we spend some time focusing on the space $\TOPO(\II,X)$ of paths in $X$.
We obtain the fundamental groupoid (Proposition \ref{pr:fundamentalgroupoid}) by quotienting the path magmoid by a congruence (Definition~\ref{de:path_magmoid}).
Some careful constructions of the fundamental groupoid can be found in the literature, for example in \cite{dieck} and \cite{brownt+g}, although our magmoid approach is non-standard and we will use (more radical versions of) similar ideas repeatedly in later sections so we think this `warm up' is worthwhile.

This also allows us to give a first example of the utility of the product-hom adjunction, Lemma~\ref{th:tensorhom};
paths in the fundamental groupoid are equivalent if and only if there is a path between them in the space of paths (Lemma~\ref{le:tensorhom_path}).
Throughout the rest of this paper we will use path-equivalence alongside several other equivalence relations so we also introduce some careful notation here.

\defn{ \label{de:pe}
	Let $X$ be a topological space.
	Define a relation on $\Path X (x,x')$ as follows.
	Suppose we have paths $\gamma,\gamma'\in\Path X(x,x')$, then $\gamma\simp \gamma'$ 
	if there exists a continuous map
	$H\colon \II \times \II \to X$
	such that 
	\begin{itemize}
		\item for all $t\in \II$, $H(t,0)=\gamma(t)$,
		\item for all $t\in \II$, $H(t,1)=\gamma'(t)$, and
		\item for all $s\in \II$,
		$H(0,s)=x$ and $H(1,s)=x'$.
	\end{itemize}
	Notation:
	We call such an $H$ a {\em path-homotopy} from 
	$\gamma$ to $\gamma'$.
}

\prop{\label{pr:pe}
	Let $X$ be a topological space. 
	For each pair $x,x'\in X$, $\simp$ is an equivalence relation on $\Path X(x,x')$.\\
	Notation: If $\gamma\simp \gamma'$ we say $\gamma$ and $\gamma'$ are {\em path-equivalent }.
 	We use $\classp{\gamma}$ for the path-equivalence class of $\gamma$.}
\begin{proof}
	We show that $\simp$ is reflexive, symmetric and transitive. Let $\gamma\in \Path X(x,x')$, $\gamma'\in \Path X(x,x')$ and $\gamma''\in \Path X(x,x')$ be paths with $\gamma\simp \gamma'$ and $\gamma'\simp \gamma''$.
	
	The relation is reflexive since the function $H(t,s)=\gamma_t$ is a path-homotopy from $\gamma$ to $\gamma$.
	
	By assumption, there exists a path-homotopy, say $H_{\gamma,\gamma'}$, from $\gamma$ to $\gamma'$. 
	The function $H_{\gamma',\gamma}(t,s)=H_{\gamma,\gamma'}(t,1-s)$ 
	is a 
	path-homotopy from $\gamma'$ to $\gamma$, hence the relation is symmetric.
	
	There also exists a path-homotopy, say $H_{\gamma',\gamma''}$, from $\gamma'$ to $\gamma''$.
	The function 
	\[
	H_{\gamma,\gamma''}(t,s)=\begin{cases}
		H_{\gamma,\gamma'}(t,2s) & 0\leq s\leq \frac{1}{2} \\
		H_{\gamma',\gamma''}(t,2(s-\frac{1}{2})) & \frac{1}{2} \leq s \leq 1.
	\end{cases}
	\]
	is a path-homotopy from $\gamma$ to $\gamma''$, so $\simp$ is transitive.
\end{proof} 

\begin{lemma}\label{le:tensorhom_path}
	Let $X$ be a topological space. 
	Let $\gamma,\gamma'\in \Path X(x,x')$ be paths.
	Then $\gamma\simp \gamma'$ if and only if there is a path
	$\tilde{H}\colon \II\to \Topo(\II,X)$ such that 
	$\tilde{H}(0)=\gamma$, $\tilde{H}(1)=\gamma'$ and for all $t\in \II$, $\tilde{H}(t)\in \Path X(x,x')$.
\end{lemma}
\begin{proof}
	We have that $\II$ is a locally compact Hausdorff topological space so Lemma~\ref{th:tensorhom} gives that there is a bijection between continuous maps $\II\times \II \to X$ and continuous maps $\II\to \TOPO(\II,X)$.
	We obtain the appropriate conditions by looking at the image of a path homotopy under this bijection.
\end{proof}

\lemm{\label{le:path_comp_well_defined}
	Let $X$ be a topological space.
	The equivalence relations $(\Path X(x,x'),\simp)$ for each pair $x,x'\in X$ are a congruence on $\Path X$.
}
\begin{proof}
	Suppose $\gamma,\gamma'\in \Path X(x,x')$ are path-equivalent and so there exists a path homotopy, say $H_{\gamma,\gamma'}$ from $\gamma$ to $\gamma'$.
	And suppose $\delta,\delta'\in \Path X(x',x'')$ are path-equivalent and so there exists a path homotopy, say $H_{\delta,\delta'}$ from $\delta$ to $\delta'$.
	Notice $H_{\gamma,\gamma'}(1,s) =
	H_{\delta,\delta'}(0,s)=x'$ and so the function
	\[
	H(t,s)=\begin{cases}
		H_{\gamma,\gamma'}(2t,s) & 0\leq t\leq \frac{1}{2} \\
		H_{\delta,\delta'}(2(t-\frac{1}{2}),s) & \frac{1}{2} \leq t\leq 1
	\end{cases}
	\]
	is a homotopy from $\delta\gamma$ to $\delta'\gamma'$.
\end{proof}

\prop{ \label{pr:fundamentalgroupoid}
	Let $X$ be a topological space. 
	There exists a groupoid
	 $$
	 \pi(X) \;=\; \Path X / \simp  \;=\; (X,\Path X(-,-)/\simp, \PsiOps ,\classp{e_x},\classp{\gamma^{rev}}) 
	 $$
	 with $\Path X$ as in Definition~\ref{de:path_magmoid}.
	 Here the identity morphism $\classp{e_x}$ at each object $x$ is the path-equivalence class of the constant path $(e_x)_t=x$ for all $t\in \II$.
	The inverse of a morphism $\classp{\gamma}$ from $x$ to $x'$ is the path-equivalence class of $\gamma^{rev}_t=\gamma_{1-t}$.\\
	We have obtained the {\em fundamental groupoid} of $X$.
	}
\begin{proof}
	($\GG1$)\; Suppose $\gamma\in \Path X(x,x')$,
	a suitable choice of path homotopy from $e_x\gamma$ to $\gamma$ is:
	\[
	H_{id}(t,s) = \begin{cases}
		\gamma_{\frac{t}{\frac{s}{2}+\frac{1}{2}}} & 0\leq t \leq \frac{s}{2}+\frac{1}{2}\\
		x & \frac{s}{2}+\frac{1}{2} \leq t\leq 1.
	\end{cases}
	\]
	A choice for $\gamma e_{x'}\simp \gamma$ can is given by splitting the segments at $\frac{1}{2}-\frac{s}{2}$, and using the subscript $\frac{t-\frac{1}{2}+\frac{s}{2}}{\frac{1}{2}+\frac{s}{2}}$.\\
	($\GG2$)\; The following function is a path homotopy $\gamma''(\gamma'\gamma)$ to $(\gamma''\gamma')\gamma$:
	\[
	H_{ass}(t,s)= \begin{cases}
		\gamma_{\frac{t}{\frac{s}{4}+\frac{1}{4}}} & 0\leq t\leq \frac{s}{4}+\frac{1}{4} \\
		\gamma'_{4(t-\frac{s}{4}-\frac{1}{4})}& \frac{s}{4}+\frac{1}{4}\leq t \leq \frac{s}{4}+\frac{1}{2} \\
		\gamma''_\frac{{t-\frac{s}{4}-\frac{1}{2}}}{\frac{1}{2}-\frac{s}{4}}& \frac{s}{4}+\frac{1}{2}\leq t \leq 1.
	\end{cases}
	\]
	($\GG3$)\;
	The following function is a homotopy $\gamma^{rev}\gamma$ to $e_x$:
	\[
	H_{in}(t,s)=\begin{cases}
		\gamma_{2t} & 0 \leq t\leq \frac{1}{2}-\frac{s}{2} \\
		\gamma_{1-s}& \frac{1}{2}-\frac{s}{2}\leq t \leq \frac{1}{2}+\frac{s}{2} \\
		\gamma_{1-2(t-\frac{1}{2})}& \frac{1}{2}+\frac{s}{2}\leq t \leq 1.
	\end{cases}
	\]
	The same segments, with first and last term independent of $s$, and middle term $\gamma_s$, defines a path homotopy  $\gamma\gamma^{rev}\simp e_x$.
\end{proof}

\rem{Let $X$ be a topological space and $x\in X$ be a point, we have that 
	$ \pi(X)(x,x)
	$
	is the fundamental group based at $x \in X$.}

\subsection{Action groupoid \texorpdfstring{$\Hom$}{Homeo(M)} of the action of self-homeomorphisms on subsets}\label{ss:selfhomeos}
In this paper, manifold means a Hausdorff topological manifold, which in particular is locally compact and locally connected.

    From here we will work always with $M$ a \axiomM{} possibly with boundary. 
    Then we have that $\TOPO^h(M,M)$ is a topological group (Theorem \ref{le:top_group}) and we can use the product-hom adjunction (Lemma \ref{th:tensorhom}).
    
    \medskip

\defn{ \label{def:actiongroupoid}
Let $P$ be a magma. 
A magma action of $P$ on a set $S$ 
is a map $\alphaa: P\times S \rightarrow S$ with $q(ps)=(qp)s$ (denoting the composition in $P$ of a pair $(p,q)$ as $qp$, and the image $\alpha(p,s)$ as $ps$).
Given  a magma action $\alpha$, we can construct an {\em action magmoid} 
$\mug_\alphaa$, which is a triple consisting of objects $S$, morphisms which are triples 
$(p,s,p s) \in P \times S \times S$, 
and a partial composition  
$((p,s,ps),(q,ps,qps))\mapsto(qp,s,qps)$.

If $P$ additionally has the structure of a group with identity $e\in P$, and $\alphaa$ additionally satisfies $es=s$, then $\mug_\alphaa$ is a groupoid with $(p,s,ps)^{-1}=(p^{-1},ps,s)$, called the {\em action groupoid}.}

For an action $\alphaa \colon P\times S\to S$, we will find it useful to keep track of both $P$ and $S$ in our notation for the action groupoid, so we denote $\mug_\alphaa$ as
$
\agrpd{P}{\alphaa}{S}.
$

Observe that since every magma (group) action on a set $\alphaa \colon P\times S\to S$ induces an action on the power set $\Power S$, 
for every action there is a corresponding action magmoid (groupoid) of the action on $\Power S$, which we denote $
\agrpd{P}{\alphaa}{\Power S}.
$

\medskip

Let $M$ be a \axiomM{} and $A\subset M$ a subset.
Recall the definition of $\Topo^{h}(M,M)$ from \S\ref{sec:tensor_hom}.
Let $A\subset X$ a subset and let $\Topo^h_A(M,M)$ denote the subset of $\Topo^h(M,M)$ of homeomorphisms which fix $A$ pointwise.

Here we organise the elements of $\Topo^{h}_A(M,M)$ into a groupoid $\Hom^A$ with objects $\Power M$, constructed from an action on $M$.
In general this category 
is too large to be an interesting object of study itself but it is a natural first step in the construction that follows.

\medskip

\defin{\label{Def:homeoMA}
The group $\Topo^h_A(M,M)$, of homeomorphisms fixing $A$ pointwise, acts on the set $M$ as $\sh{f}\acts m=\sh{f}(m)$. We denote the action groupoid of the induced action on $\Power M$ as
\[
\Hom^A= \agrpd{\Topo^h_A(M,M)}{\acts}{\Power M}.
\]}

We will denote triples $(\sh{f},N,\sh{f}(N))\in \Hom^A(N,N')$ as 
$\shmor{f}{A}{N}{N'}$. 
In this notation, 
the identity at each object $N$ is
$\shmot{\id_M}{A}{N}{N}$
 where $\id_M$ denotes the identity homeomorphism, and, given a morphism $\shmor{f}{A}{N}{N'}$, the inverse is the morphism $\shmor{f^{-1}}{A}{N'}{N}$.

We will use just $\Hom$ to denote $\Hom^{\emptyset}$, so morphism sets are of the form $\Homn$.

	    \lemm{\label{le:homcompl}
	    Let $M$ be a \axiomM{} and $A\subset M$ a subset. For any subsets $N,N'\subset M$ we have 
	    \[
	    \Hom^A(N,N')\cong \Hom^A(M\setminus N, M\setminus N').
	    \]}
	    \begin{proof}
Since any $\shmor{f}{A}{N}{N'}$ is a bijection, 
$\sh{f}(N)=N'$ if and only if $\sh{f}(M\setminus N)=M\setminus N'$.
\end{proof}

Abusing notation, we will also use $\Hom^A(N,N')$ to denote the set obtained by projecting to the first element of the triple. 
Then we have $\Topo^h(M,M)=\Hom(\emptyset,\emptyset)=\Hom(M,M)$ and every $\Hom^{A}(N,N')\subseteq \Topo^h(M,M)$. Notice each self-homeomorphism $\mc{f}$ of $M$ will belong to many such $\Hom^A(N,N')$.

	    		\lemm{\label{le:Hom(N,N')_togroup}
	    		Let $M$ be a \axiomM{} and $A\subset M$ a subset. With the induced topology, each $\Hom^A(N,N)$ becomes a topological subgroup of $\TOPO^h(M,M)$.}
		\begin{proof}
		Note that any subgroup of a topological group is itself a topological group with the induced topology.
		
		We check that $\Hom^A(N,N)\subseteq \TOPO^h(M,M)$ is a subgroup.
		Suppose we 
		are given self-homeomorphisms 
		$\shmor{f}{A}{N}{N}$ and $\shmor{g}{A}{N}{N}$, then $\sh{f}\circ \sh{g} (N)= \sh{f}(N)=N $
		and for all $a \in A$: $\sh{f}\circ \sh{g} (a)= \sh{f}(a)=a.$ So $\shmor{f\circ g}{A}{N}{N}$ is in $\Hom^A(N,N)$.
		Similarly $\shmor{f^{-1}}{A}{N}{N}$ is in $\Hom^A(N,N)$.
		\end{proof}

 \rem{There are various ways in which we could equip the subsets of $M$ with extra structure.
		For example we could let $N$ and $N'$ be submanifolds of $M$ equipped with an orientation and then consider homeomorphisms which preserve these orientations.
		}

\section{Motion groupoid \texorpdfstring{$\Mot^A$}{Mot(M)}}
\label{sec:motions}
In this section we construct the motion groupoid associated to a pair of a manifold $M$ and a pointwise fixed subset $A\subset M$ (e.g. $A=\partial M$ or $A=\emptyset$).
The core topological ideas used in this section are present in \cite{goldsmith}, and first appeared in \cite{dahm} (see also \cite{goldsmiththesis}), which construct a group of classes of motions which return a subset $N$, in the interior of a manifold $M$, to its initial position.

We proceed by first defining {\it \premot{}s} in a manifold $M$, and giving two choices of composition, $*$ and $\cdot$.
At this point there are no `objects', one choice of composition gives a magma, the other a group.
We obtain motions by considering an action of \premot{}s on $ M$, and thus on $\Power M$.
The two compositions on \premot{}s lead to two action magmoids, each of which has object set $\Power M$ and \textit{motions} as morphisms, and one of which is also a groupoid.

The $*$ composition is the intuitive composition where one motion is carried out, followed by another, similar to path composition. In addition it is only with this composition that we are able to interpret motion composition in terms of a composition of the \textit{wordlines} of the motion (Lemma~\ref{lem:concatW}). This leads to it being a more useful setting to work with for examples.
The $\cdot$ composition is introduced as it will be more convenient for many proofs, Lemma~\ref{le:stat_mots} for example.

For physical/engineering purposes, it is often necessary to have something finitary, thus we add a congruence to our magmoids.
We will find it most straightforward to construct this congruence in two stages, first quotienting by a congruence using path homotopy, under which these magmoids become the same groupoid.
This groupoid has, in general, uncountable morphism sets, and thus we add a further equivalence.
By
quotienting by a normal subgroupoid of classes containing a \textit{set-stationary motion},
we obtain the motion groupoid $\Mot$ (Theorem~\ref{th:mg}).
The object set is the power set $\Power M$ and the morphisms are equivalence classes of motions.

To make the notation more manageable we only give the full details of proofs when working in $\Hom$. In Section \ref{ss:Afixmot} we also construct a version using $\Hom^A$, i.e fixing a distinguished choice of subset $A\subset M$.
This leads to the motion groupoid $\Mot^A$.

In Section~\ref{ss:examples}
we have some examples which frame some of the questions that our construction allows us to ask. For example we can think about skeletons of our motion groupoids, or equivalently which subsets of a manifold $M$ are connected in the motion groupoid. 
Or we could instead look for subsets which are not connected by a morphism in the motion groupoid, but which do have isomorphic automorphism groups.

\subsection{\Premot{}s:  elements in \texorpdfstring{$\Topo^{}(\II,\TOPO^h(M,M))$}{Topo(I,Topoh(M,M)}
}

Here we define \premot{}s and introduce two compositions.

\defn{\label{de:premots}
Fix a \axiomM{} $M$.
A {\em \premot{} in $M$} is a path in $\TOPO^h(M,M)$ starting at $\id_M$. 
We define notation for the set of all \premot{}s in $M$,
$$
\premots \; = \; \{ f \in   \Topo^{}(\II,\TOPO^h(M,M)) \; | \; f_0 = \id_M   \} .
$$}

\exa{ \label{ex:Id_premot}
For any \axiomM{} $M$ 
the path $f_t = \id_M$ for all $t$, is a \premot.
We will denote this \premot{} $\Id_M$.}

\exa{\label{S1_pre-mot}
For $M=S^1$ (the unit circle) we may parameterise by 
$\theta\in\R/2\pi$ in the usual way.
Consider the functions $\tau_\phi : S^1 \rightarrow S^1$ 
($\phi \in \R$) given by $\theta \mapsto \theta+\phi$,
and note that these are homeomorphisms.
Then consider the path $f_t = \tau_{t \pi}$ (`half-twist'). 
This is a \premot{}.}

\begin{lemma} 
\label{le:raisin}
Let $S,R$ be  manifolds  and 
$\psi\colon S\to R$
a homeomorphism. 
Then there exists a bijection
 $
\premo{S} \to \premo{R}
$
denoted $g \mapsto g^\psi$
where 
\[ 
(g^\psi )_t(x)= \psi \circ g_t \circ \psi^{-1}(x).\]
\end{lemma}
\begin{proof}
That $g^\psi\colon \II \to \TOPO^h(R,R)$ is a continuous map follows from Lemma~\ref{le:comp_continuous_eachvariable}. It is clear that $(g^\psi)_0=\id_R$. 
The inverse $\premo{R} \to \premo{S}$ is given by $f\mapsto f^{{(\psi^{-1})}}$.
\end{proof}

\lemm{\label{le:pw_inv}
Let $M$ be a \axiomM{}. For any \premot{} $f$ in $M$, then $(f^{-1})_t=f_t^{-1}$ is a \premot{}.
}
\begin{proof}
By Theorem \ref{le:top_group} we have that $\Topo^h(M,M)$ is a topological group, so we have that the map $g \in \Topo^h(M,M) \mapsto g^{-1} \in \Topo^h(M,M)$ is continuous. 
It follows that the composition $t \mapsto f_t\mapsto f_t^{-1}$ is continuous.
Notice also that $(f^{-1})_0=\id_M^{-1}=\id_M$.
\end{proof}

\medskip
\noindent
\textbf{Composition of \premot{}s}\\
The usual non-associative `stack+shrink' composition of paths in 
$\Topo(\II,X)$ (see \eqref{eq:pathcomp} \ppm{on p.\ref{eq:pathcomp}})  is a partial composition, precisely $gf$ is  a path if the end of the path $f$ is the start of the path $g$.
Now suppose $X=\TOPO(Y,Y)$ for some space $Y$ 
and $f,g\in \Topo(\II,\TOPO(Y,Y))$. We can use the function composition in $\TOPO(Y,Y)$ to construct paths $g_0\circ f_t$ and $g_t\circ f_1$ which share an endpoint, and thus we can use the usual path composition on these modified paths.
\prop{ \label{pr:prestar_comp}
	Let $Y$ be a space.
	There exists a composition
	\ali{
		*\colon \Topo(\II,\TOPO(Y,Y))\times \Topo(\II,\TOPO(Y,Y))&\to \Topo(\II,\TOPO(Y,Y))\\
		(f,g)&\mapsto g*f
	}
	where
	\begin{align}
		(g*f)_t = \begin{cases}
			g_0\circ f_{2t} & 0\leq t\leq 1/2, \\
			g_{2(t-1/2)}\circ f_1 & 1/2\leq t \leq 1.
		\end{cases}
		\label{def:precomp} 
	\end{align}
}

\begin{proof}
     It follows from Lemma~\ref{le:comp_continuous_eachvariable} that $g*f$ is continuous on each segment.
     We also have that 
     the
     functions agree at $t=1/2$, hence Equation \eqref{def:precomp} defines 
an element
     in $\Topo(\II,\TOPO(Y,Y))$.
\end{proof}

\prop{ \label{de:premot_comp}
	Let $M$ be a \axiomM{}. There exists a composition
	\ali{
		*\colon \premots\times \premots&\to \premots\\
		(f,g)&\mapsto g*f
	}
	where
	\begin{align}
		\label{def:comp1} 
		(g*f)_t = \begin{cases}
			f_{2t} & 0\leq t\leq 1/2, \\
			g_{2(t-1/2)}\circ f_1 & 1/2\leq t \leq 1.
		\end{cases}
	\end{align}
	We denote the magma $(\premots,*)$.
}
\begin{proof}
    This is the restriction of the $*$ function of Proposition~\ref{pr:prestar_comp} to $\premots$
    so we need only to check that $g*f\in \premots$. 
    We have $(g*f)_0=f_0=\id_M$ and for all $t\in \II$, $(g*f)_t$ is a homeomorphism as it is the composition of two homeomorphisms.
\end{proof}
Note that the previous composition of flows
does not require $M$ a manifold.
Given a manifold $M$, we can also define another `pointwise' composition of paths in $\premots$ which relies on the fact that $\TOPO^h(M,M)$ is a topological group. 

\lemm{\label{le:dot_premot_comp}
	Let $M$ be a manifold.
	(I) There is an associative composition
	\ali{
		\cdot\;\colon \premots\times \premots&\to \premots\\
		(f,g)&\mapsto
		g\cdot f
	}
	where $(g\cdot f)_t=g_t\circ f_t$.\\
	(II) There is a group $(\premots, \cdot)$, with identity $\Id_M$ and inverse map $f\mapsto f^{-1}$ with $f^{-1}$ as defined in Lemma~\ref{le:pw_inv}.
}

\begin{proof}
	(I) We first check that $g\cdot f$ is a path. This can be seen by rewriting as
	\begin{align*}
	\II&\to \TOPO^h(M,M)\times\TOPO^h(M,M) \to \TOPO^h(M,M) \\
	t &\mapsto \makebox*{$\TOPO^h(M,M)\times\TOPO^h(M,M)$}{$(f_t,g_t)$} \mapsto g_t\circ f_t.
	\end{align*}
	The map into the product is continuous because it is continuous on each projection, and the second map is continuous because $\TOPO^h(M,M)$ is a topological group, by Theorem \ref{le:top_group}.
	Notice also that $(g\cdot f)_0=g_0\circ f_0=\id_M$, so we have that $g\cdot f\in \premots$.
	Associativity of the composition follows from the associativity of function composition in $\Set$.\\
	(II) For all $t\in \II$, $(\Id_M\cdot f)_t=\id_M\circ f_t=f_t=f_t\circ \id_M=(f\cdot \Id_M)_t$ and $(f\cdot f^{-1})_t=f_t\circ f_t^{-1}= \id_M=f_t^{-1}\circ f_t =(f^{-1}\cdot f)_t$.
\end{proof}

The following lemma says that, up to path-equivalence, both compositions are the same.

\begin{lemma}\label{le:pms_star_equiv_dot}
	Let $M$ be a manifold and $f,g\in \premots$. Then $g*f \; \simp\; g\cdot f$.
\end{lemma}
Before the proof, let us fix some conventions. 
\Premot{}s are paths $f\colon \II \to \TOPO^h(M,M)$ and then homotopies of paths are maps $H\colon \II\times \II \to \TOPO^h(M,M) $.
We will always think of the first copy of $\II$ in a homotopy as the one parameterising the \premot{}, and will continue to use the parameter $t$.
For the second copy of $\II$, which parameterises the homotopy, we will use $s$. 

\begin{proof}
	The following function is a suitable path homotopy to prove the path-equivalence
	\begin{align}\label{eq:star pequiv to dot}
		H(t,s) = \begin{cases}
			g_{ts}\circ f_{2t(1-s)+ts} & 0\leq t\leq \frac{1}{2},\\
			g_{2(t-1/2)(1-s)+ts}\circ f_{(1-s)+ts} & \frac{1}{2}\leq t\leq 1.
		\end{cases}
	\end{align}
	Notice $H(t,0)=(g*f)_t$, $H(t,1)=(g\cdot f)_t$ and for all $s\in \II$ we have $H(0,s)=g_0\circ f_0 =\id_M$ and $H(1,s)= g_1\circ f_1$.
	\ppm{Note that continuity of each segment uses that $\TOPO^h(M,M)$ is a topological group by Theorem~\ref{le:top_group}.}
\end{proof} 

\rem{\label{rem:starprime}
There are other choices of compositions of \premot{}s which assign paths $g$ and $f$ to a path which is path-homotopic to $g*f$ and $g\cdot f$. For example
\ali{
(g*'f)_t = \begin{cases}
			g_{2t} & 0\leq t\leq 1/2, \\
			g_1\circ f_{2(t-1/2)} & 1/2\leq t \leq 1.
		\end{cases}
}
}

We can also generate from any \premot{} $f$, a \premot{} $\bar{f}$ which reverses the path. Intuitively
$\bar{f}$ is obtained from $f$ by first changing the direction of travel along the path, and then precomposing at each $t$ with $f_1^{-1}$ to force the reversed path to start at the identity.
\prop{\label{pr:rev}
	Let $M$ be a \axiomM{}.
	There exists a set map 
	\ali{
		\bar{\phantom{f}}\colon \premots&\to \premots{} \\
		f &\mapsto \bar{f}
	}
	with 
	\begin{align}\label{eq:rev}
		\bar{f}_t 
		\;=\;  f_{(1-t)}\circ f_1^{-1}.
\end{align}}
\begin{proof}
    By Lemma~\ref{le:comp_continuous_eachvariable}, the composition with $f_1^{-1}$ is continuous and so $\bar{f}$ is continuous.
	Also notice $\bar{f}_0=f_{1}\circ f_1^{-1}=\id_M$, and $\bar{f}_t$ is a composition of homeomorphisms, thus a homeomorphism.
\end{proof}

\rem{\label{rem:starbar_pequive_identity}
The operation $f \mapsto \bar{f}$ is an involution, namely $\bar{\bar{f}}=f$. Notice also that for a \premot{} $f$, $\bar{f}*f=f^{rev}f$, with path composition as in \eqref{eq:pathcomp} and $f^{rev}$ as in Proposition~\ref{pr:fundamentalgroupoid}. Thus we have already shown in the proof of Proposition~\ref{pr:fundamentalgroupoid} that $\bar{f}*f\simp \Id_M$. }

\subsection{Motions: the action of \premot{}s on subsets} \label{ss:motions}  

For a manifold $M$, there is a set map from $\premots \times M$ to $M$ defined by $(f,m)\mapsto f_1(m)$.
This lifts to a magma
action of $(\premots, *)$ on $M$, and a group action 
of $(\premots, \cdot)$ on $M$.
Thus we can form the action magmoid and action groupoid respectively of these actions on $\Power M$ (see Definition~\ref{def:actiongroupoid} and the following text).

A \textit{motion} is a morphism in either of these magmoids, whose morphisms are, by construction, the same.

\defn{\label{de:mot1}
Fix a \axiomM{} $M$.
A {\em motion in $M$} is a triple $(f,N,f_1(N))$ consisting of a \premot{} $f\in\premots{}$, a subset $N\subseteq M$ and the image of $N$ at the endpoint of $f$, namely $f_1(N)$. (Note $f_1(N)=N'$ if and only if $f_1\in \Homn$.)

\noindent { Notation:}
We will denote such a triple by $\mot{f}{}{N}{N'}$ where $f_1(N)=N'$, and say it is a motion from $N$ to $N'$. 
For subsets $N,N'\subseteq M$ we define
\[
\Motc{M}{N}{N'}=\{(f,N,f_1(N)) \text{ a motion in $M$}\,\vert\,  f_1(N)=N'\}.
\]
A motion is uniquely determined by a pair of a \premot{} $f$ and a subset $N\subseteq M$. This implies
\[
\Mtc=\bigcup_{N,N'\in \Power M}\Motc{M}{N}{N'} 
   \;\;\; \cong\; \premots \times \Power M
   ,
\]
where the union is over all pairs $N,N'\subseteq M$.}
As with $\Hom$, where convenient
we will also use $\Motc{M}{N}{N'}$ to denote the set obtained by projecting to the first element of the triple. 
Then each $f\in \premots$ will belong to many $\Motc{M}{N}{N'}$.

\medskip

The bar operation generates from any motion from $N$ to $N'$, a motion from $N'$ to $N$. 
\prop{\label{pr:bar_mot}
	Let $M$ be a \axiomM{}. For any 
	subsets $N,N'\subseteq M$ there is a set map 
	\ali{
		\bar{\phantom{f}}\colon \Motc{M}{N}{N'} &\to
		\Motc{M}{N'}{N}\\
		\mot{f}{}{N}{N'}&\mapsto \mot{\bar{f}}{}{N'}{N}
	}
	where $\bar{f}_t 
		\;=\;  f_{(1-t)}\circ f_1^{-1}$, as in Equation \eqref{eq:rev}.}
\begin{proof}
    Proposition \ref{pr:rev} gives that $\bar{f}$ is a \premot.
    Note that we have $\bar{f}_1(N')=N$,
	hence $\mot{\bar{f}}{}{N'}{N} \;\in \Motc{M}{N'}{N}$.
\end{proof}
	
\exa{ 
	For a manifold $M$, a subset $N\subseteq M$ and the \premot{} $\Id_M$ as in Example \ref{ex:Id_premot},
	$\mot{\Id_M}{}{N}{N}$ is a motion. 
	We will call this the `trivial motion' from $N$ to $N$.
	Note that the \premot{} $\Id_M$ becomes a motion from $N$ to $N$ for any $N$, but not a motion
	from $N$ to $N'$ unless $N=N'$.}

\exa{ \label{ex:half-twist motion}
The half-twist of $S^1$ (see Example \ref{S1_pre-mot}) becomes a motion in $S^1$ from $N$ to $\tau_\pi(N)$ for any $N\subseteq S^1$.}

\rem{\label{rem:motion_implies_homeo}
Suppose $\mot{f}{}{N}{N'}$ is a motion in $M$, then $N$ and $N'$ are homeomorphic.
}

\prop{\label{pr:mot_comp}
	Let M be a \axiomM{}. 
     There is a magma action $*\colon (\premots,*)\times M\to M$, $(f,m)\mapsto f_1(m)$. Hence we can construct the action magmoid (Def.~\ref{def:actiongroupoid}) of the corresponding action on $\Power M$
	\[
    \Mtcmag^* \; =\; \agrpd{\premots}{*}{\Power M} \;=\;
     ( \Power M , \Mtcmag(-,-)  , * ).
    \]}
    \begin{proof}
For all $m\in M$ and $f,g\in \premots$, $(g*f)_1(m)= g_1\circ f_1 (m) =g_1(f_1(m))$.
\end{proof}
   In our notation, the composition in $\Mtcmag^*$ is given by
    \ali{
		*\colon \Motc{M}{N}{N'}\times \Motc{M}{N'}{N''}&\to \Motc{M}{N}{N''}\\
		(\mot{f}{}{N}{N'},\mot{g}{}{N'}{N''})&\mapsto
		(\mot{g}{}{N'}{N''})*(\mot{f}{}{N}{N'}),
	}
	where 
	$(\mot{g}{}{N'}{N''})*(\mot{f}{}{N}{N'})=\mot{g*f}{}{N}{N''}$ with $g*f$ as defined in Equation \eqref{def:comp1}.

\prop{\label{pr:grpd_dot}\label{pr:dot_comp}
	Let $M$ be a manifold. There is a group action 
	$\cdot\colon (\premots,\cdot)\times M\to M $, $(f,m)\mapsto f_1(m)$.
	Hence we can construct the action groupoid (Def.~\ref{def:actiongroupoid}) of the corresponding action on $\Power M$
	\[
	\Mtcmag^{\;\cdot} \; = \;
	\agrpd{\premots}{\cdot}{\Power M}
	\; =\;
	 (\Power M , Mt_M(-,-) , \cdot,\Id_M,(f^{-1})_t=(f_t)^{-1}).
	\]
}
\begin{proof}
We have that for all $m\in M$ and $f,g\in \premots$, $(g\cdot f)_1(m)=g_1\circ f_1(m)=g_1(f_1(m))$ and $(\Id_M)_1(m)=\id_M(m)=m$.
\end{proof}

In our notation the composition in $\Mtcmag^{\;\cdot}$ is given by
	\ali{
		\cdot \;\colon \Motc{M}{N}{N'}\times \Motc{M}{N'}{N''}&\to \Motc{M}{N}{N''}\\
		(\mot{f}{}{N}{N'},\mot{g}{}{N'}{N''})&\mapsto
		(\mot{g}{}{N'}{N''})\cdot(\mot{f}{}{N}{N'})
	} 
where 
	$(\mot{g}{}{N'}{N''})\cdot(\mot{f}{}{N}{N'})=\mot{g\cdot f}{}{N}{N''}$ and $(g\cdot f)_t=g_t\circ f_t$.

Note that in the last three entries of the pentuple $\Mtcmag$ we give only information about what happens to the group element in each morphism, as it is clear what should happen to the elements of $\Power M$. We do this to keep notation readable and it will be common in our constructions.

\subsection{Schematic for \texorpdfstring{$\TOPO^h(M,M)$}{TOPOh(M,M)}}

In Fig.\ref{fig:my_lalala} we represent
the space $\TOPO^h(M,M)$ 
and elements of
$\Topo(\II,\TOPO^h(M,M))$ 
schematically 
(i.e. on the page, which is to say, the plane)
for arbitrary $M$.
Figure \ref{fig:my_lalala} further gives, schematically, two examples of 
motions in  $M$. 
Here  $\TOPO^h(M,M)$ is represented as 
(a couple of disconnected) 
regions of the plane, so we have that the various $\Hom(N,N')$s are 
possibly intersecting subregions.
The blue path $(a)$ represents a motion from $N$ to $N$. 
Notice this is a path starting and ending in the same shaded region of $\Hom(N,N)$. 
This is possible since $\Hom(N,N)$ must contain the identity.
(Although $\Hom(N,N)$ may also have path
connected components which do not contain the identity, as pictured.)
        The red path (b) is a motion from $N$ to $N'$ where $N\neq N'$. 
        
        Note a \premot{} corresponds to precisely one path in $\TOPO^h(M,M)$, although many motions can have the same underlying \premot{}, thus to make such a  diagram
        convey a motion 
        it is necessary to explicitly state the subsets in addition to the schematic representation of the path.

\begin{figure}
    \centering
\def\svgwidth{0.8\columnwidth}
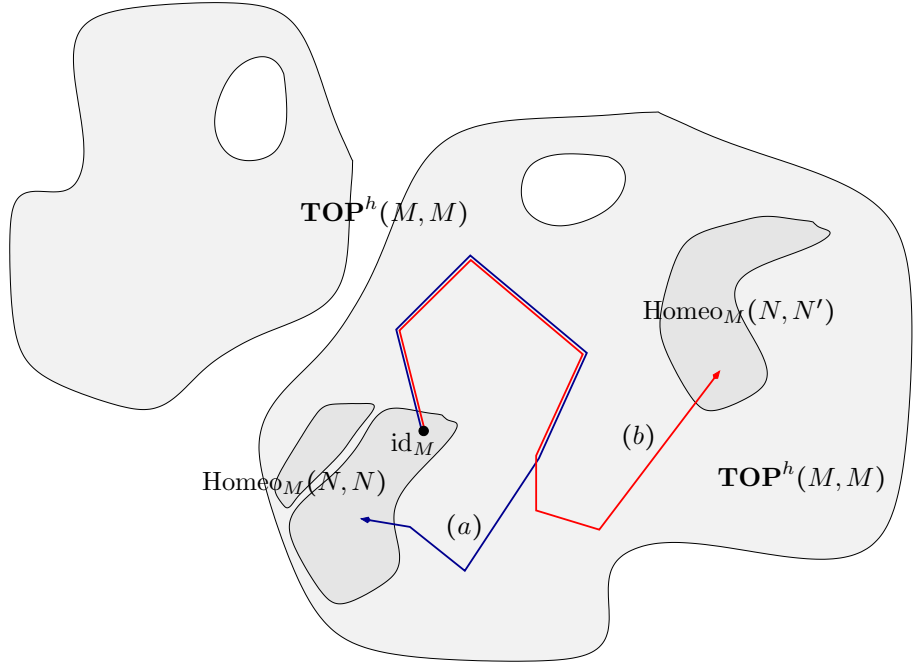
    \caption{
    A schematic representation of 
    $\Topo^h(M,M)$, 
    for a fixed but arbitrary $M$,
    as a not-necessarily connected, not-necessarily simply-connected subset of $\R^2$.
    In practice we are only interested in the connected component of the point $\id_M$. 
 The blue line (a) is then a motion from $N$ to $N$ and the red line $(b)$ a motion from $N$ to $N'$.
     }
    \label{fig:my_lalala}
\end{figure}

\subsection{Path homotopy congruence on motion magmoids}

\renewcommand{\too}{\rcurvearrowright}

Here we 
show that path-equivalence is a congruence on $\Mtcmag^*$, and that the corresponding quotient magmoid is a groupoid.
 We then show the same equivalence is a congurence on $\Mtcmag^{\;\cdot}$ and that the quotient magmoid is precisely the groupoid obtained from $\Mtcmag^*$.

\medskip

Notice that, since path homotopies fix the endpoints, for any motion $\mot{f}{}{N}{N'}$ and any  path-equivalence $f\simp f'$, $\mot{f'}{}{N}{N'}$ is a motion.

\lemm{\label{le:simp_cong}
  Let $M$ be a \axiomM. \\
  (I) For each pair $N,N'\subseteq M$ of subsets 
  the relation
  $$
  (\mot{f}{}{N}{N'})\sim(\mot{f'}{}{N}{N'}) \;\mbox{ if }\;
  f\simp f'
  $$ 
  is an equivalence relation on $\Mtcmag(N,N')$
  (see Definition~\ref{de:pe}
  for the definition of $\simp$
  ).\\
 (II) The equivalence relations $(\Mtcmag(N,N'),\sim \,)$ for each pair $N,N'\subseteq M$ are a congruence on $\Mtcmag^{*}$.
 \\
{Notation}: By abuse of notation we will also use $\simp$ to denote this relation.
We  
use $\classp{\mot{f	}{}{N}{N'}}$ 
or $\classp{f}$
for the path-equivalence class of $\mot{f	}{}{N}{N'}$.
}   

\begin{proof}
	(I) For any pair $N,N'$, we have that $\Mtcmag(N,N')\subset \Topo(\II,\Topo(M,M))$, thus the proof that path-homotopy is an equivalence relation on $\Topo(\II,\Topo(M,M))$ (Proposition~\ref{pr:pe}) is sufficient.
	\\
	(II) Suppose we have pairs of equivalent motions $(\mot{f}{}{N}{N'})\simp(\mot{f'}{}{N}{N'})$
	and $(\mot{g}{}{N'}{N''})\simp(\mot{g'}{}{N'}{N''})$. Then there exists a path homotopy, say $H_f$ from $f$ to $f'$ and a path homotopy, say $H_g$ from $g$ to $g'$.
	Notice that, since path homotopies fix the endpoints, for all $s\in \II$ we have $H_f(1,s)=f_1$.
	Thus the map
	\[
	H(t,s)=\begin{cases}
		H_f(2t,s) & 0\leq t\leq 1/2\\
		H_g(2(t-1/2),s)\circ f_1 & 1/2 \leq t\leq 1
	\end{cases}
	\]
	is a path homotopy $g*f$ to $g'*f'$.
\end{proof}

\lemm{\label{le:grpd_motstar}
    Let $M$ be a \axiomM{}.
The pentuple
	$$
	\Mtcmag^*/\simp \;\;\; = \;
	(\Power M ,\; \Mtcmag(N,N')/\simp,\; *,\; \classp{\Id_M},\; \classp{f}\mapsto \classp{\bar{f}})
	$$
	is a groupoid.
	}
\begin{proof}
First notice that, by Lemma~\ref{le:simp_cong}, $\simp$ is a congruence, hence there is a quotient magmoid $(\Power M ,\; \Mtcmag(N,N')/\simp,\; *)$.

	We have proved in Lemma~\ref{le:pms_star_equiv_dot} that $g*f\simp g\cdot f$, and 
	by Proposition~\ref{pr:mot_comp} 
	$\cdot$ is associative and unital with unit $\Id_M$. This is sufficient to prove $(\GG1)$ and $(\GG2)$.\\ 
	($\GG 3$)\;  Since we are considering a different inverse to the inverse in the group $(\premots,\cdot)$, we prove this directly.
	Note that for any morphism $\classp{\mot{f}{}{N}{N'}}$, $\mot{\bar{f}}{}{N'}{N}$ is well defined by Proposition~\ref{pr:bar_mot}. For any morphism $\classp{\mot{f}{}{N}{N'}}$, the following function
	\begin{align}\label{eq:mot_refl}
	H_{inv}(t,s)=
	\begin{cases}
	f_{2t(1-s)} & 0\leq t \leq \frac{1}{2}, \\
	f_{(1-2(t-1/2))(1-s)} & \frac{1}{2}\leq t \leq 1
	\end{cases}
	\end{align}
	is a homotopy from $\bar{f}*f$ to $\Id_M$.
	Observe that for each fixed $s$, the path $H_{\bar{f}*f}(t,s)$ starts at the identity, follows $f$ until $f_{(1-s)}$, 
	and then follows $f_{(1-t)}$ back to $\id_M$.
\end{proof}
\rem{Note that $\Mtcmag^*/\simp$ is the action groupoid 
$\agrpd{\left((\premots,*)/\simp\right)}{\acts}{\Power M}$,
where $\classp{f}\acts N=f_1(N)$. The proof of Lemma~\ref{le:grpd_motstar} is essentially a proof that $(\premots,*)/\simp$ is a group. The downside of this approach is that motions are obscured, and since our motivation is to model particle trajectories which  do not a priori include a choice of equivalence relation, we favour our approach.}

\lemm{\label{le:Mtcmag}
Let $M$ be a \axiomM{}. 
$(I)$
The relations $(\Mtcmag(N,N'),\simp)$ for each $N,N'\subseteq M$ are a congruence on $\Mtcmag^\cdot$, thus we have a groupoid: 
\[
\Mtcmag^{\;\cdot}/\simp\;\;=\;\;(\Power M, \Mtcmag(N,N')
/\simp 
,\cdot,\classp{\Id_M},\classp{f}\mapsto \classp{f^{-1}})
\]
$(II)$ We have that $\Mtcmag^\cdot =\Mtcmag^{*}$.\\
Notation: We will now denote this groupoid by just $\Mtcmag/\simp$.
}
\begin{proof}
$(I)$  
By Lemma~\ref{le:pms_star_equiv_dot} $f\cdot g\simp f*g$ for all \premot{}s, hence that $\simp$ is a congruence follows from Lemma~\ref{le:simp_cong}.
By Proposition~\ref{pr:dot_comp} $\Mtcmag^{\; \cdot}$ is a groupoid, thus, by Proposition~\ref{pr:gpd_cong}, the quotient is also a groupoid. \\
(II) By construction the two categories have the same objects and morphisms.
		By Lemma~\ref{le:pms_star_equiv_dot} the composition is the same up to path-equivalence. Thus the underlying magmoids are the same. By uniqueness of inverses and identities, they are the same groupoid.
\end{proof}

\rem{\label{rem:barinv_dotinv}
Note in particular that Lemma~\ref{le:Mtcmag} implies 
$\bar{f}\simp f^{-1}$.
}
The previous lemma allows us to work interchangeably with either choice of composition or inverse according to which simplifies each proof, this will be used throughout the paper.

\defn{\label{de:disk}
	We define the topological space $D^2$, called the 2-disk, as $\left\{x\in \R^2 \;\vert \; \; \vert x\vert\leq 1\right\}\subset \R^2$ with the subset topology.
}

\ppm{The groupoid $\Mtcmag/\simp $ typically has uncountable sets of morphisms. Morally this is because path-homotopy completely fixes endpoints, and thus there is still too much information being kept track of. To be more precise,} let $M$ be a manifold and $N,N'\subseteq M$ be  subsets.
Given two motions $\mot{f,f'}{}{N}{N'}$ such that $f_1\neq f_1'$, then their path-homotopy classes (which we recall are relative to end-points) are different, so $[
\mot{f}{}{N}{N'}]_p \neq [\mot{f'}{}{N}{N'}]_p$. In general there uncountably many choices of endpoints of homeomorphisms of $M$ sending $N$ to $N'$.

In particular, let $M=D^2$ and $N\subset \mathrm{int}(D^2)$ be a finite set in the interior of $D^2$.
Fix an $x\in \mathrm{int}(D^2)\setminus N$. For any $y\in \mathrm{int}(D^2)\setminus N$ there is a motion $\mot{{f^y}}{}{N}{N}$ with $f^{y}_1(x)=y$  and $f^y_t(N)=N$ for all $t\in \II$. (We are using the homogeneity of connected smooth manifolds discussed in \cite[\S 4]{Milnor_top}, together with the fact that if $M$ is connected, of dimension $\ge 2$, then $M$ minus a finite set of points is still connected.)
Note that $\classp{\mot{{f^y}}{}{N}{N}}\neq \classp{\mot{{f^{y'}}}{}{N}{N}}$ if $y\neq y'$, as $f^y_1(x)=y$ whereas $f^{y'}_1(x)=y'$.
There are uncountably many such $f^y$, hence the set $\Mtcmag/\simp(N,N)$ is uncountable.

Both the braid groups and the loop braid groups have presentations with a finite number of generators, thus are not uncountable \cite{artin,damiani}.
In the next section we impose a further quotient that will identify the motions $f^y$ and $f^{y'}$ .

\subsection{The motion groupoid \texorpdfstring{$\Mot$}{Mot}: congruence induced by  set-stationary motions} \label{ss:stat_mot}

Motivated by engineering/physical considerations,
we aim to
construct models 
where physical configurations are 
arranged into a countable/finitely-generated set of classes 
i.e. combinatorially. 
See
\cite{Fradkin, Bais1, Leinaas,BFM}
for examples of systems where the interesting physics is modelled by such structures.
Accordingly, 
by imposing path equivalence,
we have washed out some distinctions that do not affect the induced movement of 
our object subsets
(without this our sets are certainly  larger than combinatorial). 
However, for general subsets, so far, we are still only allowing motions to be equivalent if their underlying paths share the same end point, and these sets can still be very large.
Dahm's idea of `motion groups' partially addresses this problem \cite{dahm}.
Intuitively, the quotient used to construct a motion group can be thought of as de-emphasising the motion of the ambient space.
(This point of view is made precise by Proposition~\ref{prop:dahm} and Section~\ref{sec:laminated}.)
Here we prove there is a lift of Dahm's idea
to the groupoid setting.

\medskip
We start by defining \textit{\stationary{}} motions, 
motions $N\too N$ which leave $N$ fixed setwise. 
We then show that there is a normal subgroupoid in $\Mtcmag/\simp$ whose morphisms are those classes containing an \stationary{} motion for some $N\in M$, hence this normal subgroupoid induces a congruence.
This leads to the motion groupoid $\Mot$ in Theorem~\ref{th:mg}.

\newcommand{\topi}{\Topo(\II, \Topo^h(M,M))}  
\newcommand{\sstat}{\topi} 

\defn{
Let $M$ be a manifold, and $N\subseteq M$ a subset.
A motion 
	$\mot{f}{}{N}{N}$ in $M$ is said to be {\em \stationary} if $f_t\in \Hom(N,N)$ ($f_t(N)=N$) for all $t\in\II$.
	Define
\[
\mathrm{SetStat}_M(N,N)=\left\{ \mot{f}{}{N}{N} \; \vert \;  f_t(N)=N
\text{ for all } t \in \II \right\}.
\]
}

\begin{example}\label{ex:discrete_premot}
	Let $M=D^2$, the $2$-disk and let $N\subset M$ be a finite set of points. 
	Then a motion $\mot{f}{}{N}{N}$ in $M$ is \stationary{} if, and only if, $f_t(x)=x$, for all $x\in N$ and $t\in \II$. More generally this holds if $N$ is a totally disconnected subspace of $M$, e.g. $\mathbb{Q}$ in $\mathbb{R}$. 
\end{example}

\begin{example}\label{ex:bdyfix_D^2}
	Let $M= D^2$.
	Consider $D^2$ as a subset of $\C$.
	 Let $f$ in $\TOPO^h(D^2,D^2)$ be constructed as follows.
    Consider a continuous function $g\colon [0,1] \to \R^-_0,$ such that $g(0)=0$ and $g(1)=0$, and with a single local minimum in $x\in(0,1)$ with $f(x)=-\pi$. Let  $f_t$ be the map $z\mapsto z \exp(i g(|z|) t)$.
    Then sequential points on $f\in \premo{D^2}$ are represented by  Figure \ref{fig:my_label01}.
    
	Now pick a subset $N\subset D^2$ which is any circle centred on the centre of the disk, i.e. the set of all points a fixed distance from the centre using the metric induced from the complex plane.
	The motion $\mot{f}{}{N}{N}$ is \stationary{}.
	\end{example}

\begin{lemma}\label{le:stat_mots}
	Let $M$ be a \axiomM{} and $N,N'\subseteq M$ subsets. Let $\setstat(N,N)$ be the subset of
	$\Mtcmag/\simp (N,N)$ of those classes that intersect $\mathrm{SetStat}_M(N,N)$.
	Let $\setstat (N,N') = \emptyset$ if $N \neq N'$.
There is a totally disconnected, normal, wide subgroupoid of $\Mtcmag/\simp$,
\[
\setstat \; =\; (\Power M, \;\;
\setstat(N,N') , \;
* , \; \classp{\Id_M}, \; \classp{f}\mapsto \classp{\bar{f}} ).
\]
Note that 
\[
\setstat(N,N') = \{ \classp{\mot{f}{}{N}{N'}} \; \;
     |\;\;   
     \exists \text{ \stationary{} } \mot{f'}{}{N}{N'} \in \classp{\mot{f}{}{N}{N'}}
     \}.
\]
\end{lemma}

\newcommand{\setstatN}{\setstat(N,N)} 

\begin{proof}
	
	First we will show that the tuple $\setstat$ is a subgroupoid.
	
	For each $N\subseteq M$, the identity $\classp{\mot{\Id_M}{}{N}{N}}$ is in $\setstatN$ as for all $t\in \II$,  $(\Id_M)_t(N)=\id_M(N)=N$.
	
	For the existence of inverses, and for closure of composition, observe that there is nothing to show unless $N=N'$. 
	For each $\classp{\mot{x}{}{N}{N}}\in \setstatN$, we may assume without loss of generality that $\mot{x}{}{N}{N}$ is a \stationary{} motion. Then the inverse $\classp{\mot{\bar{x}}{}{N}{N}}$ is in 
	$\setstatN$,
	since for all $t\in \II$, $\bar{x}_t(N)=x_{1-t}\circ x_1^{-1}(N)=x_{1-t}(N)=N$.
	
	Let $\classp{\mot{x}{}{N}{N}}$ and $\classp{\mot{x'}{}{N}{N}}$ be in $\setstatN$, we may again assume without loss of generality, that $\mot{x}{}{N}{N}$ and $\mot{x'}{}{N}{N}$ are \stationary{}.
	For all $t\in[0,1/2]$ we have that $(x'*x)_t(N)=x_t(N)=N$ and for $t\in [1/2,1] $ that $(x'*x)_t(N)=x'_t\circ x_1(N)=x'_t(N)=N$.
	Thus composition closes, and so $\setstat$ is a groupoid.
	
	Observe now that $\setstat$ is totally disconnected and wide by construction.

	Finally,	we have that $\setstat$ is normal, since for any morphism $\classp{\mot{f}{}{N}{N'}}\in \Mtcmag/\simp$ and for any $\classp{\mot{x}{}{N'}{N'}}$ in 
	$\setstat(N',N')$,
	with $\mot{x}{}{N'}{N'}$ $N'$-stationary{}, the following function
	\[
	H(t,s)=f_{t(1-s)+s}^{-1}  \circ x_t \circ f_{t(1-s)+s}
	\]
	is a path homotopy from $f^{-1}\cdot x \cdot f$ to the path
	$t \mapsto f^{-1}_1\circ x_t \circ f_1 \in \TOPO^h(M,M)$, which is an \stationary{} motion. 
\end{proof}

\newcommand{\simS}{\stackrel{\mathrm{Stat}}{\sim}}
 Hence,
by Lemma~\ref{le:nsubgrpd_cong}, we can form a quotient groupoid 
$(\Mtcmag/\simp)/\setstat$.
Precisely, morphisms $\classp{\mot{f}{}{N}{N'}}$ and  $\classp{\mot{g}{}{N}{N'}}$ in $\Mtcmag/\simp$ are related in the quotient, denoted 
\[\classp{\mot{f}{}{N}{N'}} \simS
\classp{\mot{g}{}{N}{N'}},\]
if $\classp{\mot{\overline{g}}{}{N'}{N}} * \classp{\mot{f}{}{N}{N'}}=\classp{\mot{\overline{g}*f}{}{N}{N}}$ is in $\setstat(N,N)$. (Or equivalently if $\classp{\mot{{g^{-1}}\cdot f}{}{N}{N}}$ is in $\setstat(N,N)$.)

We give our first definition of the motion groupoid as a groupoid which is canonically isomorphic to $(\Mtcmag/\simp)/\setstat$.

\prop{\label{pr:pe_implies_me}
	Let $M$ be a \axiomM{}. 
	Let $(\mot{f}{}{N}{N'})\simp (\mot{g}{}{N}{N'})$ be path-equivalent motions in $M$.
	Then
	$\mot{\bar{g} * f}{}{N}{N}$ is  path-equivalent to an \stationary{} motion. 
}

\begin{proof}
	We have $\classp{\mot{f}{}{N}{N'}}=\classp{\mot{g}{}{N}{N'}}$, hence, as $\Mtcmag/\simp$ is a groupoid (Lemma~\ref{le:grpd_motstar}), thus with unique inverses, $\classp{\mot{f}{}{N}{N'}}^{-1}=\classp{\mot{\bar{g}}{}{N}{N'}}$. This implies there is a path-homotopy $H$ from $\bar{g}*f$ to $\Id_M$, which is an \stationary{} motion.
\end{proof}

\prop{\label{pr:me1}
	For a \axiomM{} $M$ and $N,N' \subseteq M$, denote by
	$\simm$ \ppm{(`$m$' for motion)}
	the relation
	\[
	\mot{f}{}{N}{N'} \simm \mot{g}{}{N}{N'} \;\;\mbox{ if }\;\; 
	\classp{\overline{g} * f} \in\setstat(N,N)
	\]
	on $\Mtcmag(N,N')$.
	(I) This is an equivalence relation.\\
	(II) There exists a canonical bijection $\Mtcmag(N,N')/\simm \, \cong (\Mtcmag/\simp)/\setstat$ sending $\classm{\mot{f}{}{N}{N'}}$ to the $\simS$ equivalence class of $\classp{\mot{f}{}{N}{N'}}$.\\
	Notation: We call this {\em motion-equivalence} and denote by $\classm{\mot{f}{}{N}{N'}}$ the motion-equivalence class of $\mot{f}{}{N}{N'}$.
	}

\begin{proof}
First note the following general fact. Consider a set $X$, with an equivalence relation $\sim_1$, and  furthermore another equivalence relation $\sim_2$ on $X/\sim_1$.
Then $x\sim y$, if $[x]_{\sim_1} \sim_2 [y]_{\sim_1} $, is an equivalence relation on $X$. And moreover we have a bijection $X/\sim \to {(X/\sim_1)}/\sim_2$ 
such that $[x]_\sim \mapsto {[[x]_{\sim_1}]}_{\sim_2}$.

 Notice that another way to write $\simm$ is as
\[\mot{f}{}{N}{N'} \simm \mot{g}{}{N}{N'} \;\;\mbox{ if }\;\; \classp{\mot{f}{}{N}{N'}} \simS  \classp{\mot{g}{}{N}{N'}}.
\] 
Further $\simp$ is an equivalence relation by Lemma~\ref{le:simp_cong}, and $\simS$ is an equivalence relation, since it is the relation obtained from Lemma~\ref{le:nsubgrpd_cong} applied to $\setstat$, which is normal by Lemma~\ref{le:stat_mots}.
Thus both results follow from the first paragraph.
\end{proof}

In the following theorem we continue our convention of only giving the first element of the triple corresponding to a motion in the tuple defining a groupoid, to keep the notation readable. \\
Thus $\classm{f}$ denotes, for some choice of $N$, and $N'=f_1(N)$, $\classm{\mot{f}{}{N}{N'}}$.
We note however that the relation $\simm$ depends on the subsets.

\medskip

By Proposition~\ref{pr:me1} we have:

\begin{theorem}\label{th:mg}
	Let $M$ be a \axiomM{}.
	There is a groupoid 
	\[
	\Mot\;=\; 
	(\Power M,\; \Mtcmag(N,N')/\simm,*,\classm{\Id_M}, \; \classm{f}\mapsto \classm{\bar{f}}) \] 
	where
	
	\begin{itemize}
		\item[(I)] objects are subsets of $M$;
		\item[(II)] morphisms between subsets $N,N'$ are motion-equivalence classes $\classm{\mot{f}{}{N}{N'}}$ of motions, explicitly
		$$
		\mot{f}{}{N}{N'}\simm \mot{g}{}{N}{N'} 
		$$
		if 
		$\bar{g}*f \colon N \too N $ is path equivalent to an \stationary{} motion;
		\item[(III)] composition of morphisms
		is given by 
		\[
		\classm{\mot{g}{}{N'}{N''}}*\classm{\mot{f}{}{N}{N'}} = \classm{\mot{g*f}{}{N}{N''}}
		\]
		where
		\begin{align}
		(g*f)_t = \begin{cases}
		f_{2t} & 0\leq t\leq 1/2, \\
		g_{2(t-1/2)}\circ f_1 & 1/2\leq t \leq 1;
		\end{cases}
		\end{align}
		\item[(IV)] the identity at each object $N$ is the motion-equivalence class of $\mot{\Id_M}{}{N}{N}$, $(\Id_M)_t(m)=m$ for all $m\in M$;
		\item[(V)] the inverse for each morphism $\classm{\mot{f}{}{N}{N'}}$ is the motion-equivalence class of $\mot{\bar{f}}{}{N'}{N}$
		where $\bar{f}_t=f_{(1-t)}\circ f_1^{-1}$.
	\end{itemize}
Moreover, we have a canonical isomorphism of groupoids:
 \[
 \pushQED{\qed}
 (\Mtcmag/\simp)/\setstat \cong  \Mot.\qedhere
 \popQED\] 
\end{theorem}

\begin{remark}
We have from Lemma~\ref{le:Mtcmag} that $\Mtmagdot/\simp$ and $\Mtmagstar/\simp$ are the same groupoid. Thus in the statement of the theorem we could have chosen to define the composition in terms of $\cdot$, or the inverse in terms of the pointwise inverse $(f^{-1})_t=f_t^{-1}$, or both.
It follows that $\mot{f}{}{N}{N'}\simm \mot{g}{}{N}{N'} $ if and only if $\mot{g^{-1}\cdot f}{}{N}{N}$ is path-equivalent to an \stationary{} motion.
\end{remark}

\begin{remark}\label{rem:DahmMG}
\ppm{The automorphism groups $\Mot(N,N)$, where $N\subseteq M$ is a submanifold, are isomorphic to the motion groups constructed in \cite{dahm,goldsmith}, at least in the case when $M$ is compact. We do not know the precise relation between our motion groups and those of \textit{loc cit} when  $M$ is not compact, since the motion groups constructed by Dahm and Goldsmith only consider for homeomorphism of compact support, and furthermore all path-homotopies are required to only trace homeomorphisms of compact support.
We settle that the two settings give the same result for the case of
finite subsets contained in the interior of a manifold in
Remark~\ref{rem:braidDahm}.}
\end{remark}

In the rest of this Section, now that we have our core construction, we collect some straightforward but useful properties of $\Mot$.

\medskip

For certain choices of $M$ and $N$, there exist subgroupoids of $\Mot$ which have finite presentation.
We discuss this in Section~\ref{sec:artinbraids}.

\begin{example}\label{ex:M_premot}
	Let $M$ be a \axiomM{}, then $\Mot(M,M)$ is trivial. 
	This is because for any $f\in \premots$, $\mot{f}{}{M}{M}$ is a motion, and it is 
	$M$-stationary.
	Similarly $\Mot(\emptyset,\emptyset)\cong\Mot(M,M)$ is trivial.
\end{example}

\begin{defin}\label{def:worldline}
The {\em worldline} of a motion $\mot{f}{}{N}{N'}$ in a manifold $M$ is
\[\W\left(f\colon N \too N'\right):=  \bigcup_{t \in [0,1]} f_t(N) \times \{t\} \subseteq M\times \II.
\]
\end{defin}

It follows directly from the definitions that an \stationary{} motion is precisely a motion whose worldline is equal to $N\times \II$.

The group case of the following result is effectively proved in \cite[pg.6]{dahm}.

\begin{proposition}\label{prop:dahm}
Let $f,g\colon N \too N'$ be motions with the same worldline, so we have
\[\W(f\colon N \too N')=\W(g\colon N \too N').\]
Then $f\colon N \too N'$ and $g\colon N \too N'$ are motion equivalent.
\end{proposition}

\begin{proof} 
By using Lemma~\ref{le:Mtcmag}, there is a motion $\mot{{g^{-1}}}{}{N'}{N}$ and 
\begin{align*}
\mot{\bar{g}*f}{}{N}{N}&\mot{\simp g^{-1} \cdot f}{}{N}{N}.
\end{align*}
The equality $\W(f\colon N \too N')=\W(g\colon N \too N')$ means:
\[ \bigcup_{t \in [0,1]} f_t(N) \times \{t\}=\bigcup_{t \in [0,1]}g_t(N) \times \{t\}.\]
In particular for each $t \in [0,1]$, we have $f_t(N)=g_t(N).$ 
This implies that for all $t\in \II$, $(g^{-1}\cdot f )_t(N)=(g^{-1}_t ( f_t(N))= (g^{-1}_t(g_t(N))=g^{-1}_t\circ g_t(N)=N $. Thus $g^{-1} \cdot f$ is $N$-stationary, and hence
$\mot{\bar{g}* f}{}{N}{N}$ is path-equivalent to an \stationary{} motion.
\end{proof}

In Section~\ref{sec:laminated} we will prove a generalisation of 
Prop.\ref{prop:dahm}, 
fully describing motion equivalence in terms of a notion of isotopy of worldlines of motions.

The following lemma says that the $*$ composition of motions descends to a composition of worldlines, the same is not true for the $\cdot$ or the $*'$ compositions of motions (Proposition~\ref{pr:grpd_dot} and Remark~\ref{rem:starprime}). 

\begin{lemma}\label{lem:concatW}
Consider motions $f\colon N \too N'$ and $g\colon N' \too N''$, then 
\[\W\left(g*f\colon N \too N''\right)\hspace{-0.12pt}=\hspace{-0.12pt}
\big\{(m,t/2) \hspace{-1pt}\mid\hspace{-1pt} (m,t) \in \W(f\colon N \too N')\big \} \cup \big\{(m,t/2+1/2) \hspace{-1pt}\mid \hspace{-1pt}(m,t) \in \W(g\colon N' \too N'')\big \},
\]
where $*$ composition is as in Equation~\eqref{def:comp1}.
\end{lemma}
\begin{proof}This follows from the definition of $*$ in  Proposition \ref{de:premot_comp}.
\end{proof}

\lemm{\label{Relating_Mot}
Let $M$ and $M'$ be \axiomM{}s such that there exists a homeomorphism $\psi\colon M\to M'$. Then there is an isomorphism of categories
	\[
	\Psi\colon \Mot[M]\to \Mot[M']
	\]
	defined as follows.
	On objects $N\subseteq M$, $\Psi(N)=\psi(N)$.
	For a motion $\mot{f}{}{N}{N'}$ in $M$, let
	$(\psi\circ f \circ\psi^{-1})_t=\psi \circ f_t \circ \psi^{-1}$.
	Then $\Psi$ sends the equivalence class $\classm{\mot{f}{}{N}{N'}}$ to the equivalence class
	$\classm{\psi\circ f \circ\psi^{-1}\colon \psi(N)\to \psi(N')}				$.}
\begin{proof}
    We have from Lemma~\ref{le:raisin} that $\psi\circ f \circ \psi^{-1}$ is in $\premo{M'}$.
	Notice also that $(\psi\circ f \circ \psi^{-1})_1(\psi(N))=\psi\circ f_1 \circ \psi^{-1}(\psi(N))=\psi\circ f_1(N)=\psi (N')$.
	
	We check $\Psi$ is well defined. Suppose $\mot{f}{}{N}{N'}$ and $\mot{f'}{}{N}{N'}$ are equivalent motions in $M$, so there is a path homotopy $\bar{f'}*f$ to a path, say $x$, such that $\mot{x}{}{N}{N}$ is an \stationary{} motion, let us call this $H$.
	It is straightforward to check that 
	the function 
	$(\psi\circ H\circ \psi^{-1})(t,s)=\psi\circ H(t,s)\circ \psi^{-1}$ is a homotopy making $\overline{\Psi(f')}*\Psi(f)$ path-equivalent to $\mot{\psi\circ x\circ {\psi^{-1}}}{}{\psi(N)}{\psi(N)}$ which is a $\psi(N)$-stationary{} motion.
	
	We check $\Psi$ preserves composition.
	Suppose $\mot{f}{}{N}{N'}$ and $\mot{g}{}{N'}{N''}$ are motions, then $\Psi(g\cdot f )_t=\psi\circ g_t\circ f_t \circ \psi^{-1}=\psi\circ g_t\circ\psi^{-1}\circ \psi \circ f_t \circ \psi^{-1}=(\psi(g)\cdot\psi(f))_t$.
	
	The inverse functor to $\Psi$ is defined in the same way using the  homeomorphism $\psi^{-1}\colon M'\to M$.
\end{proof}

\exa{
\label{ex:R2D2} There exists a homeomorphism from $D^2\setminus \partial D^2$ to $\R^2$, thus, 
by Lemma~\ref{Relating_Mot}, 
there is an isomorphism $\Mot[D^2\setminus \partial D^2]\cong \Mot[\R^2]$.
}

\begin{corollary}\label{co:Relating_Mot}
Let $M$ be a \axiomM{} and $N,N'\subseteq M$ subsets such that there exists a homeomorphism $\sh{f}\colon M\to M$ with $\sh{f}(N)=N'$. Then there is a group isomorphism 
\[
\Mot[M](N,N)\xrightarrow{\cong}\Mot[M](N',N').
\]
\end{corollary}
\begin{proof}
Letting $\psi=\sh{f}$ in Lemma~\ref{Relating_Mot}
gives the isomorphism.
\end{proof}

\begin{example}\label{ex:finiteset} It is folklore, and discussed for instance in \cite{n-trans}, with roots in \cite[\S 4]{Milnor_top}, that, if $M$ has dimension $\ge 2$, and is connected, then given any positive integer $n$, and a pair of $n$-tuples, $(x_1,\dots, x_n)$ and $(y_1,\dots, y_n)$, of points in the interior of $M$, there exists a homeomorphism $f\colon M \to M$ such that $f(x_i)=y_i$, for $i=1,\dots, n$. This is not true in general if $\dim(M)=1$, but, given that $M$ is homeomorphic\footnote{Here we are assuming that $M$ is not only Hausdorff, but also second countable.} to either $S^1, [0,1)$, $[0,1]$ or $\R$, we can still prove that, if $K$ and $K'$ are finite subsets of the same cardinality in the interior of $M$, then there exists a homeomorphism of $M$ sending $K$ to $K'$. 
From this it follows that, whenever $M$ is connected and $K$ is a finite subset of $M\setminus \partial M$ the isomorphism type of $\Mot[M](K,K)$ depends only on the cardinality of $K$.
This is proved in Lemma~\ref{le:motion_exists_between_finitesets}.
\end{example}

\begin{lemma}\label{le:Mot_auto}
Let $M$ be a \axiomM{}.
There is an involutive automorphism
\[
\Omega\colon \Mot \to \Mot
\]
which sends an object $N\subseteq M$ to its complement $M\setminus N$ and which sends a morphism $\classm{\mot{f}{}{N}{N'}}$ to  $\classm{\mot{f}{}{M\setminus N}{M\setminus N'}}$.
\end{lemma}
\begin{proof}
First notice that by Lemma \ref{le:homcompl}, $f_1\in \Hom(M\setminus N,M\setminus N')$.
We also need to check this functor is well defined. 
Suppose $\mot{f}{}{N}{N'}$ and $\mot{f'}{}{N}{N'}$ are motion-equivalent, so there is a path homotopy $f'*f$ to a stationary motion. 
So then $\mot{f}{}{M\setminus N}{M\setminus N'}\simm \mot{f'}{}{M\setminus N}{M\setminus N'}$ using the same homotopy.
It is immediate from the definitions that composition is preserved.
\end{proof}

\subsection{
Pointwise
\texorpdfstring{$A$}{A}-fixing motion groupoid, \texorpdfstring{$\Mot^A$}{MotA}}\label{ss:Afixmot}
So far we have avoided working with  homeomorphisms which fix a distinguished subset to avoid overloading the notation and thus make the exposition clearer.
Everything we have done so far in this section could have been done by working instead with paths in $\TOPO^h_A(M,M)$ for $M$ a manifold and $A\subset M$ a subset.
Recall $f\in \TOPO^h_{A}(M,M)$ is a self-homeomorphism with $f(a)=a$ for all $a \in A$. We have the following adjusted definitions.

\defn{
	Fix a \axiomM{} $M$ and a subset $A\subset M$.
	An {\em $A$-fixing  \premot{}} in $M$ is a path in $\TOPO^h_{A}(M,M)$ starting at $\id_M$.
	We define notation for the set of all A-fixing \premot{}s in $M$,
	$$
	\premots^A \; = \; \{ f \in   \Topo^{}(\II,\Topo^h_A(M,M)) \; | \; f_0 = id_M   \} .
	$$}

\defn{
	Let $M$ be a \axiomM{} and $A\subset M$ a subset.
	An { \em $A$-fixing motion} in $M$ is a triple $(f,N,f_1(N))$ consisting of an A-fixing \premot{} $f\in\premots^A{}$, a subset $N\subset M$ and the image of $N$ at the endpoint of $f$, $f_1(N)$.\\
	\noindent { Notation:}
	We will denote such a triple by $\mot{f}{A}{N}{N'}$ where $f_1(N)=N'$, and say it is an $A$-fixing motion from $N$ to $N'$.
	For subsets $N,N'\subset M$ we define
	\[
	\MotAc{M}{N}{N'}=\{(f,N,f_1(N)) \text{ a motion in $M$}\,\vert\,  f_1(N)=N'\}.
	\]
}
In practice we will often be interested in the case $A=\partial M$.

\exa{All motions of $\II$ are $\partial \II$-fixing motions.}

\exa{The half-twist motions of $S^1$, described in Example \ref{ex:half-twist motion} are not $A$-fixing motions for any non-empty subset $A\subset S^1$.}

\exa{Let $f$ be the flow in $D^2$ described in Example~\ref{ex:bdyfix_D^2} (and represented in Figure~\ref{fig:my_label01}), then $f\colon N\to N'$ a $\partial D^2$-fixing motion for any $N,N'\subseteq D^2$.}

We have the following version of the motion groupoid where morphisms are classes of $A$-fixing motions up to an equivalence of paths in
$\TOPO_A^h(M,M)$.
\begin{theorem} \label{th:mg2A} 
Let $M$ be a \axiomM{} and $A\subset M$ a subset. 
There is a groupoid
\[
\Mot^A \; =\; (\Power M,\; \Mtcmag^A(N,N')/\simm,*,\classm{\Id_M}, \; \classm{f}\mapsto \classm{\bar{f}})
\]
where $A$-fixing motions $\mot{f}{A}{N}{N'}$ and $\mot{g}{A}{N}{N'}$ are equivalent if $\bar{g}*f$ is path-equivalent to an $A$-fixing \stationary{} motion as paths in $\TOPO_A^h(M,M)$.
\end{theorem}
\begin{proof}
 Notice that if $\mot{f}{A}{N}{N'},\mot{g}{A}{N'}{N''}$ are $A$-fixing motions
then $\bar{f}$, $f^{-1}$, $g*f$ and $g\cdot f$ are all $A$-fixing motions.
All motions constructed in homotopies required for the proof of Theorem \ref{th:mg} and associated lemmas are $A$-fixing if the input paths are $A$-fixing.
Thus all proofs work in exactly the same way for $A$-fixing motions.
\end{proof}

\lemm{\label{Relating_MotA}
Let $M$ and $M'$ be \axiomM{}s such that there exists a homeomorphism $\psi\colon M\to M'$. Then there is a isomorphism of categories
				\[
				\Psi\colon \Mot[M]^A\to \Mot[M']^{\psi(A)}
				\]
				defined as in Lemma \ref{Relating_Mot}.
}
\begin{proof}
We can use the same proof as in Lemma \ref{Relating_Mot} together with the fact that for any motion $\mot{f}{A}{N}{N}$, the path $(\psi\circ f\circ\psi^{-1})_t=\psi\circ f_t\circ\psi^{-1}$ fixes $A$ pointwise.
\end{proof}
\begin{corollary}
	Let $M$ be a \axiomM{} and $A \subset M$  subset. Let $N,N'\subset M$ be subsets such that there exists a homeomorphism $\sh{f}\colon M\to M$ with $\sh{f}(N)=N'$ and $\sh{f}(a)=a$ for all $a\in A$. Then there is a group isomorphism
	\[
	\Mot[M]^A(N,N)\xrightarrow{\cong}\Mot[M]^A(N',N').
	\]
\end{corollary}
\begin{proof}
	As for Corollary~\ref{co:Relating_Mot}. 
\end{proof}

\begin{lemma}
	Let $M$ be a \axiomM{}.
	There is an involutive automorphism
	\[
	\Omega\colon \Mot^A \to \Mot^A
	\]
	which sends and object $N\subset M$ to it's complement $M\setminus N$ and which sends a morphism $\classm{\mot{f}{A}{N}{N'}}$ to  $\classm{\mot{f}{A}{M\setminus N}{M\setminus N'}}$.
\end{lemma}
\begin{proof}
	This is the same as for Lemma~\ref{le:Mot_auto}.
\end{proof}

\subsection{Examples}\label{ss:examples}

Here we will consider some examples 
which serve to illustrate
some
key aspects of the
richness of  the construction.

By Lemma~\ref{Relating_MotA} we have that if $M$ and $M'$ are homeomorphic manifolds, $\Mot$ and $\Mot[M']$ are isomorphic groupoids.
Thus it is enough to consider one $M$ for each homeomorphism class.

An interesting problem in each case is to give a characterisation of a skeleton.
This is far from straightforward, even if we restrict to objects that are themselves
manifolds. Note that subsets $N,N'\subseteq M$ being homeomorphic submanifolds is not a sufficient condition to ensure an isomorphism connecting them in the motion groupoid.
For example, let $M=\II^2$ and let $N\subset \mathrm{int}(\II^2)$
	be the circle of points a distance $1/4$ from the point $(1/2,1/2)$.
	Let $L$ be the point $(1/2,1/2)$, and $L'$ the point $(3/4,3/4)$.
	Then $N\cup L$ and $N\cup L'$ are homeomorphic but $\Mot[\ppm{\II^2}](N\cup L,N\cup L')=\emptyset$.
Sections~\ref{sec:Iexamples} and \ref{sec:Rexamples} below discuss isomorphisms between objects in $\Mot[\II]$ and $\Mot[\R]$.

We can think of the skeleton question as looking for `inner' isomorphisms, objects which are connected by an isomorphism in the motion groupoid.
This perspective frames a comparison  
with `outer' isomorphisms.  
By this we mean
the phenomenon illustrated by the following question:
for a manifold $M$, which subsets $N,N'\subseteq M$ have 
a constructible group isomorphism $\chi\colon \Mot(N,N) \rightarrow$ 
$\Mot(N',N')$,
but with $\Mot(N,N')$ empty? 
Such isomorphisms are potentially useful tools in the construction of specific motion (sub)groupoids. 
\ppm{We give examples} in Section~\ref{sec:outerisos}.

Observe that even in a skeleton most objects are undefinable so it is a good exercise to restrict to a full subgroupoid of particular interest.
Given a subset $Q$ of the object class $\Power M$ of $\Mot[M]^{A}$, 
we write $\Mot[M]^{A}|_{Q}$ for the corresponding full subgroupoid.
For example, let $M$ be $S^3$, the $3$-sphere, and $Q$ the set of subsets which are homeomorphic, using the subset topology, to a disjoint union of copies of $S^1$. Then
finding the
connected components of 
$\Mot[S^3]|_{Q}$
is equivalent to classifying unoriented links up to isotopy.
Note, however, that motion groupoids  do not only bookkeep links, up to isotopy, but also the theory of possible ambient isotopies  between links, up to motion equivalence.
The latter has rich features, as indicated for instance by the case of an unlink in $D^3$, whose automorphism group in $\Mot[D^3]^{\partial D^3}$ gives the extended loop braid group, \cite{damiani}; see Definition~\ref{de:loopbraid} and Proposition~\ref{pr:D^m} below. Given the fact that  the components of the  space of embeddings of $S^1$ inside $S^3$ typically have non-trivial fundamental group \cite{Budney}, we expect that   $\Mot[S^3](K,K)$, where $K$ is a knot, will usually be non-trivial.

\subsubsection{On \texorpdfstring{$\Mot[\II]$}{MotI}}\label{sec:Iexamples}

\begin{proposition}\label{pr:compactsubsetI}
Suppose $N\subset \II\setminus \{0,1\}$ is a compact subset with a finite number of connected components, so $N$ is a union of points and closed intervals.
We can assign a word 
in $\{ a,b \}$
to $N$ as follows: each point in $N$ is represented by an $a$ and each interval by a $b$, ordered in the obvious way using the natural ordering on $\II$.
Let $N'\subset \II\setminus \{0,1\}$ be another compact subset with a finite number of connected components.
 Then $|\Mot[\II](N,N')|=1$ if the word assigned to $N$ and $N'$ is the same, otherwise $\Mot[\II](N,N')=\emptyset$.
\end{proposition}
\begin{proof}(Sketch)
It is possible to construct even a piecewise linear motion from $N$ to $N'$ if the word assigned to $N$ and $N'$ is the same.
The result that all motions $N$ to $N'$ are equivalent follows from the intermediate value theorem.
That $\Mot[\II](N,N')=\emptyset$ if the words assigned to $N$ and $N'$ are different follows from the fact that any orientation preserving homeomorphism of $\II$ is order preserving.
\end{proof}

Note homeomorphisms send boundary points to boundary points and interior points to interior points, so any continuous path of homeomorphisms $\II\to \II$ starting at the identity fixes the boundary points.
Thus we can generalise the previous proposition to: for $A,B\subset \II$ compact subsets with a finite number of connected components, we have $|\Mot[\II](A,B)|=1$ if and only if $A\cap \{0,1\}=B\cap\{0,1\}$ and $A$ and $B$ give the same word. Otherwise $\Mot[\II](A,B)=\emptyset$.

\medskip

If we consider non-compact subsets we must also pay attention to the embeddings.
\begin{example}
Suppose $N=(1/4,1/2)\cup (1/2,3/4)$ and $N'=(1/4,3/8)\cup (5/8,3/4)$,
then we have $\Mot[\II](N,N')=\emptyset$.
\end{example}

The automorphism group $\Mot[\II](N,N)$ for $N\subset \II$ with a finite number of connected components is always trivial. This changes dramatically if more complicated subsets of $\II$ are considered.
	We will prove in Proposition~\ref{rem:mcgI} that for $M=\II$ and $N=\II \cap \mathbb{Q}$, $\Mot[\II]^{\partial \II}(N,N)$ is uncountably infinite.

\subsubsection{On \texorpdfstring{$\Mot[\R]$}{MotR}} \label{sec:Rexamples}

\begin{proposition}
 $\Mot[\R](\Q,\Z)=\emptyset$.
 \end{proposition}
\begin{proof}
This can be seen by observing that there is no homeomorphism $\theta\colon \R \to \R$ sending $\Q$ to $\Z$, since homeomorphisms of $\R$ must map dense subsets  to dense subsets. 
\end{proof}

In contrast $\Mot[\R](\Q,\Q)$ is uncountable. A proof follows the same ideas used in the proof of Theorem~\ref{th:mcgI} together with Theorem~\ref{th:mot_to_mcg} below.

\medskip 

\noindent{Question:}
Let $N\neq N'$ be countable dense subsets of $\R$.
Then does this imply the existence of a motion
$\mot{f}{}{N}{N'}$ in $\R$?

\begin{example}\label{ex:R}
	Let $M=\R$. 
	Then there is a group isomorphism $\phi\colon (\Z,+) \xrightarrow{\cong}\Mot[\R](\Z,\Z) $ such that, for  $n \in (\Z,+)$, $\phi(n)$ is the motion-equivalence class of the motion $\mot{f}{}{\Z}{\Z}$ such that $f_t (x)= x+tn$.
\end{example}

\subsubsection{Relating automorphism groups in \texorpdfstring{$\Mot[M]$}{MotM}}\label{sec:outerisos}

To construct motion (sub)groupoids,
it 
is useful to be able to obtain the automorphism group of an object in terms of the automorphism group of another object. If objects are connected in the motion groupoid then this is straightforward. Otherwise we may still be able to construct a canonical `outer' isomorphism between automorphism groups, or we may be able to construct a group homomorphism.
The following examples investigate this in various cases.

\exa{\label{ex:outer_iso2}
	In general, there will not exist a morphism in $\Mot(N,M\setminus N)$. 
However, by Lemma~\ref{le:homcompl} we have a group isomorphism $\Mot(N,N)\cong\Mot(M\setminus N,M\setminus N)$.

(Although we can construct specific cases 
	for which $\Mot(N,M\setminus N)\neq\emptyset$.
	For example let $M=S^1$, 
	and $\mot{\tau_{t\pi}}{}{N}{\tau_{\pi}(N)}$ the half-twist motion as in Example~\ref{ex:half-twist motion} --
	letting $N=[0,\pi)\subset S^1$, we have that 
	the motion-equivalence class is in $\Mot[S^1]([0,\pi),[\pi,0))$.)

	}

\begin{proposition}\label{pr:closuremap}
Let $M$ be a manifold, and let $\mathrm{cl}(- )$ denote the closure of a subset.
(I) \ppm{Let $S,T\subseteq M$ be subsets.} There is a well defined set map: 
 \[\Gamma^M_{S,T}\colon \Mot[M](S,T) \to \Mot[M]\big(\mathrm{cl}(S),\mathrm{cl}(T)\big),\] which sends a motion equivalence class $\classm{\mot{f}{}{S}{T}}$ to the motion equivalence class $\classm{\mot{f}{}{\mathrm{cl}(S)}{\mathrm{cl}(T)}}$. \\
(II) Moreover, there is a functor $\mathbf{\Gamma}^M\colon \Mot[M] \to \Mot[M]$ which is defined by $N \mapsto \mathrm{cl}(N)$ on objects, and on morphisms restricts to the map $\Gamma^M_{N,N'}\colon \Mot[M](N,N') \to \Mot[M]\big(\mathrm{cl}(N),\mathrm{cl}(N')\big).$\\
It follows directly from (II) that when $S=T$, the map $\Gamma^M_{S,S}$ is a group homomorphism. In this case, we denote the map by $\Gamma^M_{S}$.
\end{proposition}
\begin{proof}
(I) Let $h\colon M\to M$ be a homeomorphism such that $T=h(S)$, then $h\big(\mathrm{cl}(S)\big)=\mathrm{cl}(T)$.
It follows that if $f$ is a \premot{} in $M$ such that $\mot{f}{}{S}{T}$ is a motion, then $\mot{f}{}{\mathrm{cl}(S)}{\mathrm{cl}(T)}$ is a motion. It also follows that if $S=T$, and $\mot{f}{}{S}{S}$ is
$S$-stationary, then $\mot{f}{}{\mathrm{cl}(S)}{\mathrm{cl}(S)}$ is also $\mathrm{cl}(S)$-stationary. \\
(II) That composition is preserved follows directly from the definition.
\end{proof}

This map is, in general, neither injective, nor surjective.
Example~\ref{ex:interval} gives a case for which  $\Gamma^M_{S,T}$ is not surjective, and Example~\ref{ex:circle_D2} gives a case for which  $\Gamma^M_{S,T}$ is not injective.

\exa{\label{ex:interval} 
Let $M=D^2$ (here $D^2 =\{x\in \C\vert \,|x|\leq 1 \}\subset \C$).
Let $N=[-a,a]$ be a closed interval in the real axis with $0<a<1$, and let $N'=(-a,a]$.

There is a path in $\Topo(\II,\TOPO^h(D^2,D^2))$, which we label $\tau_{\pi}$, such that $\tau_{\pi t}$ is a $\pi t$ rotation of $D^2$.
Now $\tau_{\pi}$ gives a motion from $N$ to $N$, but not from $N'$ to $N'$.
Any motion $\mot{f}{}{N'}{N'}$ must satisfy $f_1(a)=a$.

An \stationary{} motion $\mot{s}{}{N}{N}$ must satisfy, for all $t\in \II$, $s_t(a)=a$ and $s_t(-a)=-a$, as there is no path in $\Hom(N,N)$ starting at the identity and ending in a homeomorphism sending $a$ to $-a$.
Suppose $(\mot{f}{}{N}{N})\simm (\mot{\tau_\pi}{}{N}{N})$, then $\bar{f}*\tau_\pi\simp s$, where $\mot{s}{}{N}{N}$ is some stationary motion. 
So we have $(\bar{f}*\tau_\pi)_1(a)=a$. We know $\tau_{\pi1}(a)=-a$, so this implies $\bar{f}_1(-a)=a$, and hence $f_1(a)=-a$. 
So all $f\in \classm{\tau_\pi}$ satisfy $f_1(a)=-a$. Hence $\classm{\tau_\pi}$ has no preimage in $\Mot(N',N')$ under $\Gamma^{D^2}_{N'}$.
	
	Notice that we do have a group isomorphism $\Gamma^{D^2}_{(-a,a)}\colon \Mot[D^2]((-a,a),(-a,a))\xrightarrow{\sim} \Mot[D^2](N,N)$ (using the notation of Proposition~\ref{pr:closuremap}).
	This can be shown by constructing an inverse, which is possible since a motion of  a closed interval induces a motion of the open interval obtained by removing the boundary. 
}

\exa{
\label{ex:circle_D2}
	As in the previous example, let $M=D^2$ seen as a subset of $\C$. Let $N=\{x\in \C\vert\; |x|=1/2 \}$ be a circle centred on the centre of the disk, and $N'=N\setminus \{-1/2,1/2\}$, so $N'=cl(N)$. 
	Let $\tau_{\pi}$ be the gradual rotation in $\TOPO^h(D^2,D^2)$, as constructed in Example~\ref{ex:interval}. 
	Then we have that $\mot{\tau_{\pi}}{}{N}{N}$ and $\mot{\tau_{\pi}}{}{N'}{N'}$
	are motions.
	
	Notice that $\tau_{\pi t}(N)=N$ for all $t\in \II$, thus $\tau_{\pi t}$ is a stationary motion, and $(\mot{\tau_{\pi}}{}{N}{N}) \simm (\mot{\Id_M}{}{N}{N})$.
	
	For $N'$, $\tau_{\pi t}(N')\neq N'$ unless $t\in \{0,1\}$, thus we do not obtain a motion-equivalence between $\mot{\tau_{\pi}}{}{N'}{N'}$ and $\mot{\Id_M}{}{N'}{N'}$ in the same way. In fact $\mot{\Id_M*\tau_{\pi}}{}{N'}{N'}$ is not path equivalent to any stationary motion. This is because $(\Id_M*\tau_{\pi})_1(x)=y$, and there is no path in $\Hom(N',N')$ starting at the identity and ending at a homeomorphism sending $x$ to $y$.
 Thus $\Id_M*\tau_\pi$ is not path equivalent to a stationary motion.

	Using the notation of Proposition~\ref{pr:closuremap}, there is a map from $\Gamma^{D^2}_{N'}\colon\Mtcmag(N',N') \to \Mtcmag(N,N)$ sending a motion to the motion with the same underlying \premot{}. By the previous paragraphs, this is not injective.
 } 

The following example shows that, in certain cases, we can keep track of the change to the automorphism group in a motion groupoid when a single point is removed from the subset being considered.

\exa{
Let $M=S^1$ and $N=S^1\setminus \{x\}$ where $x\in S^1$ is any point. The group $\Mot[S^1](S^1,S^1)$ is trivial. The group $\Mot[S^1](S^1\setminus \{x\}, S^1\setminus \{x\})\cong\Z$, this  follows from the fact that $\Mot[S^1](\{x\},\{x\})\cong\Z$ (Example~\ref{ex:S1}) and Lemma~\ref{le:Mot_auto}.
}

An interesting question is to determine the kernel of the map  $\Gamma^{D^2}_{{S^1\setminus\{x\}}}\colon\Mot[D^2](S^1\setminus\{x\},S^1\setminus\{x\})\to \Mot[D^2](S^1,S^1)$, where $D^2$ is the $2$-disk, $S^1$ is a circle centred on the centre of the disk and $x$ is any point in $S^1$.

\exa{\label{ex:punctured_hopf}
Let $N\subset \II^3$ a subset which is a Hopf link in the interior of $\II^3$. Let $N'=N\setminus \{x\}$ where $x\in N$ is any point.
		Then 
		$\Mot[\II^3](N,N')=\emptyset$. There is a homomorphism 
		$\Gamma^{\II^3}_{N'}\colon \Mot[\II^3](N',N')\to \Mot[\II^3](N,N)$. Again it would be interesting to characterise the kernel of this map.
	
	 Let $K\subset \II^3$ be a subset with $2$ unknotted unlinked connected components homeomorphic to $S^1$, and let $K'=K\setminus \{y\}$ where $y\in K$ is any point.
	Then $\Mot(N,K)=\emptyset$, this can be seen by noticing that the fundamental group of the complement of $N$ and $K$ are different.
	It follows that $\Mot(N',K')=\emptyset$ since, if such a morphism were to exist, it would have image in $\Mot(N,K)$ under $\Gamma_{N',K'}^{\II^3}$, contradicting the previous sentence.
	}

\exa{\label{ex:dehn}
Let $M$ be the torus $T^2=S^1\times S^1$, and let $N=S^1\times \{1\}.$
Let $N'$ be the image of $N$ under a Dehn twist about $\{1\}\times S^1$.  Then the curves $N$ and $N'$ are not isotopic so there is no path $f$ in $\TOPO^h(T^2,T^2)$, starting in $\id_{T^2}$ and with $f_1(N)=N'$.
However $\Mot[T^2](N,N)\cong\Mot[T^2](N',N')$.
This is just a case of Corollary~\ref{co:Relating_Mot}.}

For extended examples and first steps in representation theory see \cite{TorzewskaMartinMartinsIV}.

\subsection{Motions as maps from \texorpdfstring{$M\times \II$}{M times I}, schematics and movie
representations}\label{sec:schematics}

In this section we give two 
further equivalent ways to define
motions in a manifold $M$, 
in terms of certain maps from $M\times \II$.
Equivalence Theorem~\ref{le:motion_MxI} 
is significant because it
indicates that we can 
connect to the embedded cobordism/generalised tangle picture of manifolds embedded in $M\times \II$,
as trailed in \S\ref{sec:Intro}. 
(Note that the equivalences are still different so this does not immediately imply a functor between the two settings, this will be investigated further in Section~\ref{sec:laminated}.)
The various definitions of motions lead
us to some useful schematic representations,
so we also discuss these 
below.

\medskip

We now give an interpretation of \premot{}s, and hence motions in a manifold $M$ as a subset of $\Topo(M\times \II, M)$.

\defn{\label{de:premotsmov}
Fix a \axiomM{} $M$. Let $\premotsmov\subset\Topo(M\times \II,M)$ denote the subset of elements $g\in \Topo(M\times \II,M)$
such that:
\begin{itemize}
    \item[(I)] for all $t\in \II$, $m\mapsto g(m,t)$ is a homeomorphism $M \to M$, and
    \item[(II)] for all $m\in M$, $g(m,0)=m$.
\end{itemize}}

Let $M$ be a \axiomM{}.
Letting $X=\II$ and $Y=Z=M$ in Lemma~\ref{th:tensorhom}, and composing the bijection with a map flipping $\II$ and $M$ in the product, gives a bijection which, by abuse of notation we label also as $\Phi$:
\begin{align*}
    \Phi\colon \Topo(\II, \TOPO(M,M))&\to \Topo(M\times \II,M) \\
    f &\mapsto ((m,t)\mapsto f_t(m)).
\end{align*}

\lemm{\label{le:premotsmov}
Let $M$ be a \axiomM{}. The restriction of the map $\Phi$ given before the lemma,
yields a bijection 
\[
\Phi\colon \premots \xrightarrow{\sim} \premotsmov.
\]
}
\begin{proof}
	We have that $\Phi$ is a bijection so we just need to check that $\Phi(\premots)\subseteq\premotsmov$ and that $\Phi^{-1}(\premotsmov)\subseteq\premots$ where $\Phi^{-1}$ sends a map $g\colon M\times \II\to M$ to the map $t\mapsto(m\mapsto g(m,t))$.
	
	Let $f\in \premots$ be a \premot. 
	Then $\Phi(f)|_{M\times \{t\}}=f_t$ which is a homeomorphism and $\Phi(f)(m,0)=f_0(m)=m$.
	Hence $\Phi(f)\in \premotsmov$.
	
	Let $g\in \premotsmov$. 
	Then $m\mapsto g(m,t)$ is a homeomorphism for all $t\in \II$ and $\Phi^{-1}(g)(0)=(m\mapsto g(m,0))=id_M$. Hence $\Phi^{-1}(g)\in \premots$.
\end{proof}

Let $M$ be a \axiomM{} and $g\in\Topo(M\times \II, M)$ be in  $\premotsmov$.
Our first schematics are based on the `movie presentations' of
\cite{Carter}. 
A movie presentation of $g$ consists of a number of pictures where each picture corresponds to a chosen value of $t$, ordered by $t\in \II$.
We may also
add `grid line' subsets in $M$  
--
these help to show the homeomorphism at $t$ of $M$. An example is given by
Figure~\ref{fig:my_label01}. 
Here the grid lines are a polar grid
at $t=0$.
By Lemma~\ref{le:premotsmov} $g$ corresponds to a motion and thus movie presentations can be considered as representations of motions.
Movie presentations are used in \cite{Carter} for schematics representing the images of isotopies, and elements of $\premotsmov$ are precisely isotopies.

\defn{Let $M$ be a \axiomM{} and $N,N'\subseteq M$.
Let $\Mtcmagmov(N,N')\subseteq \premotsmov$ denote the subset of elements $g\in \premotsmov$
such that $g(N\times \{1\})=N'$.
}
\lemm{\label{le:movie_mot}
Let $M$ be a \axiomM{} and $N,N'\subset M$.
The restriction of the map $\Phi$ from Lemma~\ref{le:premotsmov}
 yields a bijection
\[
\Phi\colon\Mtcmag(N,N')\mapsto \Mtcmagmov(N,N').
\]
}
\begin{proof}
	Notice first that each $\Mtcmagmov(N,N')$ is a subset of $\premotsmov$.
	We have from Lemma~\ref{le:premotsmov} that $\Phi$ gives a bijection $\premots\cong \premotsmov$ so 
	we only need to check that $\Phi(\Mtcmag(N,N
	))\subseteq\Mtcmagmov(N,N')$ and $\Phi^{-1}(\Mtcmagmov(N,N
	))\subseteq\Mtcmag(N,N')$.
	
	Suppose $\mot{f}{}{N}{N'}$ is a motion,
	then $\Phi(f)(N\times \{1\})=f_1(N)=N'$.
	Suppose $f'\in \Mtcmagmov(N,N')$, then $\Phi^{-1}(f')_1(N)=f'(N\times \{1\})=N'$.
\end{proof}

Using Lemma~\ref{le:movie_mot}, we may add subsets to movie schematics.
This can be seen in Figure~\ref{fig:my_label002}: the middle schematic shows a subset consisting of two points in the disk, and the right most a point and a line; the images of each subset at various ordered $t\in \II$ are shown in each disk with $t$ progressing up the page.
Note that in this case we have not included grid lines. 
Roughly the homeomorphisms shown in Figure~\ref{fig:my_label01} moves the subsets as shown in the first five disks.

\medskip

Next we give our second interpretation of motions in a manifold $M$ as a subset of $\Topo^h(M\times \II, M\times \II)$.

\defn{\label{de:premotshom}
Fix a \axiomM{} $M$. 
Let 
 $ \premotshom \subset \Topo^h(M\times \II,M\times \II)$
denote the subset of homeomorphisms $g\in\Topo^h(M\times \II,M\times \II)$ 
such that
	\begin{itemize}
		    \item[(I)] $g(m,0)=(m,0)$ for all $m\in M$, and
		    \item[(II)] $g(M\times \{t\})=M\times \{t\}$ for all $t\in \II$. 	\end{itemize}}
\rem{
To prove the following we need both  that $\TOPO^h(M,M)$ is a topological group and the product-hom adjunction of Lemma~\ref{th:tensorhom}.
\footnote{Note we use the fact that $M$ is a manifold, so that $\TOPO^h(M,M)$ is a topological group. An alternative proof of this result that holds if $M$ is compact (and not necessarily a manifold) follows from the fact that any continuous bijection between compact Hausdorff spaces is a homeomorphism.}}
\lemm{\label{le:homeo_MxItoMxI} 
Let $M$ be a \axiomM{}.
There is a bijection
\ali{\Theta\colon\premots
\;&\to\; \premotshom, \\
f\;&\mapsto\; ((m,t)\mapsto (f_t(m),t)).}
}
\begin{proof}
	We first check the $\Theta$ is well defined.
	Let $f\in \premots$.
	Then $\Theta(f)$ is continuous since the projection onto the first coordinate of the map $(m,t)\mapsto (f_t(m),t)$ is $\Phi(f)$ with $\Phi$ as in Lemma~\ref{le:premotsmov}, and the projection on the second coordinate is clearly continuous.
	
	We also have $\Theta(f)(m,0)=(f_0(m),0)=(m,0)$ and $\Theta(f)(M\times \{t\})=f_t(M)\times \{t\}=M\times \{t\}$. 
	
	To complete the proof of well definedness, it remains only to check that $\Theta(f)$ is a homeomorphism. The map $(m,t)\mapsto (f_t(m),t)$ has inverse $(m,t) \mapsto (f_t^{-1}(m),t)$. 
	Let us see that the inverse is continuous.
	We have that $f$ is a \premot{} and so Lemma~\ref{le:pw_inv} gives that $f^{-1}$ is a \premot{}, specifically it is a continuous map $\II\to\TOPO(M,M)$. Hence $(m,t)\mapsto (f_t^{-1}(m),t)$, which is the image of $f^{-1}$ under $\Theta$, is continuous.
	Thus $\Theta$ is a well defined homeomorphism.
	
	We now show that $\Theta$ is a bijection, by constructing an inverse.
	Consider the following map.
	\ali{
		\Theta^{-1}\colon\premotshom &\to \premots\\
		g&\mapsto (t\mapsto (m\mapsto p_0\circ g(m,t))
	}
	It is straightforward to check that for any $f\in \premots$ we have $\Theta^{-1}\circ\Theta(f
	)=f$ and that for any $g\in \premotshom$ we have $\Theta\circ \Theta^{-1}(g)=g$.
	It remains to check that $\Theta^{-1}$ is well defined.
	Let $g\in \premotshom$.
	The map $\Theta^{-1}(g)$ is continuous as it is equal to $\Phi^{-1}(p_0\circ g)$, with $\Phi$ as in Lemma~\ref{le:premotsmov}.
	
	We have $(\Theta^{-1}(g))_0(m)=p_0\circ g(m,0)=m$ so $\Theta^{-1}(g)_0=\id_M$.
	For all $t\in\II$ the restriction $g|_{M\times \{t\}}$ is also a homeomorphism onto its image which, by (II) in the definition of $\premotshom$, is $M\times \{t\}$.
	The projection $p_0\colon M\times \{t\}\to M$ is an isomorphism. Hence for all $t\in \II$, $\Theta^{-1}(g)_t= p_0\circ g|_{M\times \{t\}}$ is in $\TOPO^h(M,M)$.
\end{proof}

\begin{figure}
    \centering
        \includegraphics[width=1.965in]{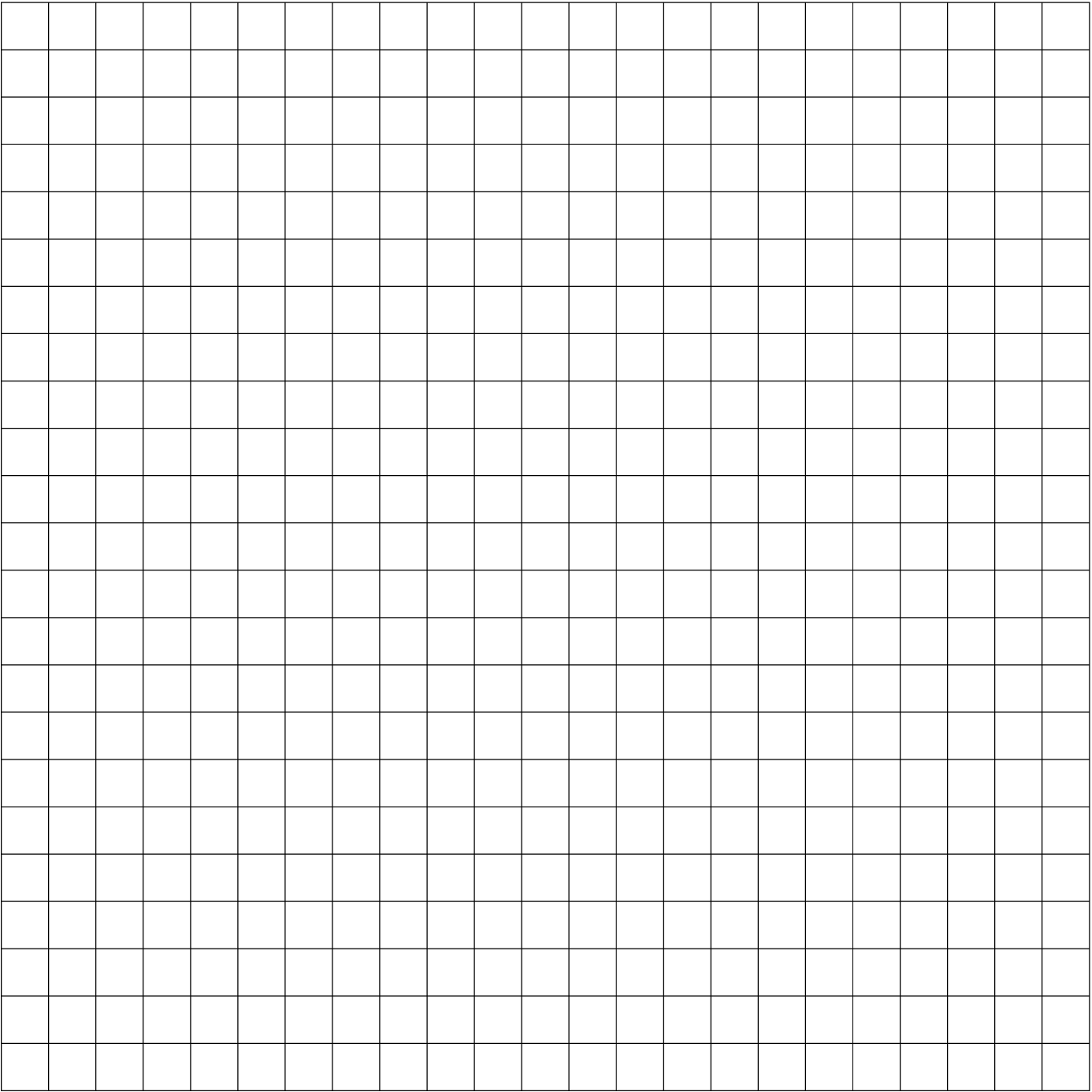}
\hspace{.061in} { } 
\raisebox{.851in}{$\mapsto$}
    \hspace{.1in} 
    \includegraphics[width=1.965in]{Figs/grid1a.eps}
\caption{Flare schematic for the homeomorphism $\Id_{\II\times\II}$ with $F\subset \II\times \II$ the marked grid.
This is also the image under $\Theta$ of the constant path in $\TOPO^h(\II,\II)$ starting at $\id_{\II}$. 
}
\label{fig:flare-id}
\end{figure}

\begin{figure}
    \centering
        \includegraphics[width=1.965in]{Figs/grid1a.eps}
\hspace{.061in} { } 
\raisebox{.851in}{$\mapsto$}
    \hspace{.1in} 
    \includegraphics[width=1.965in]{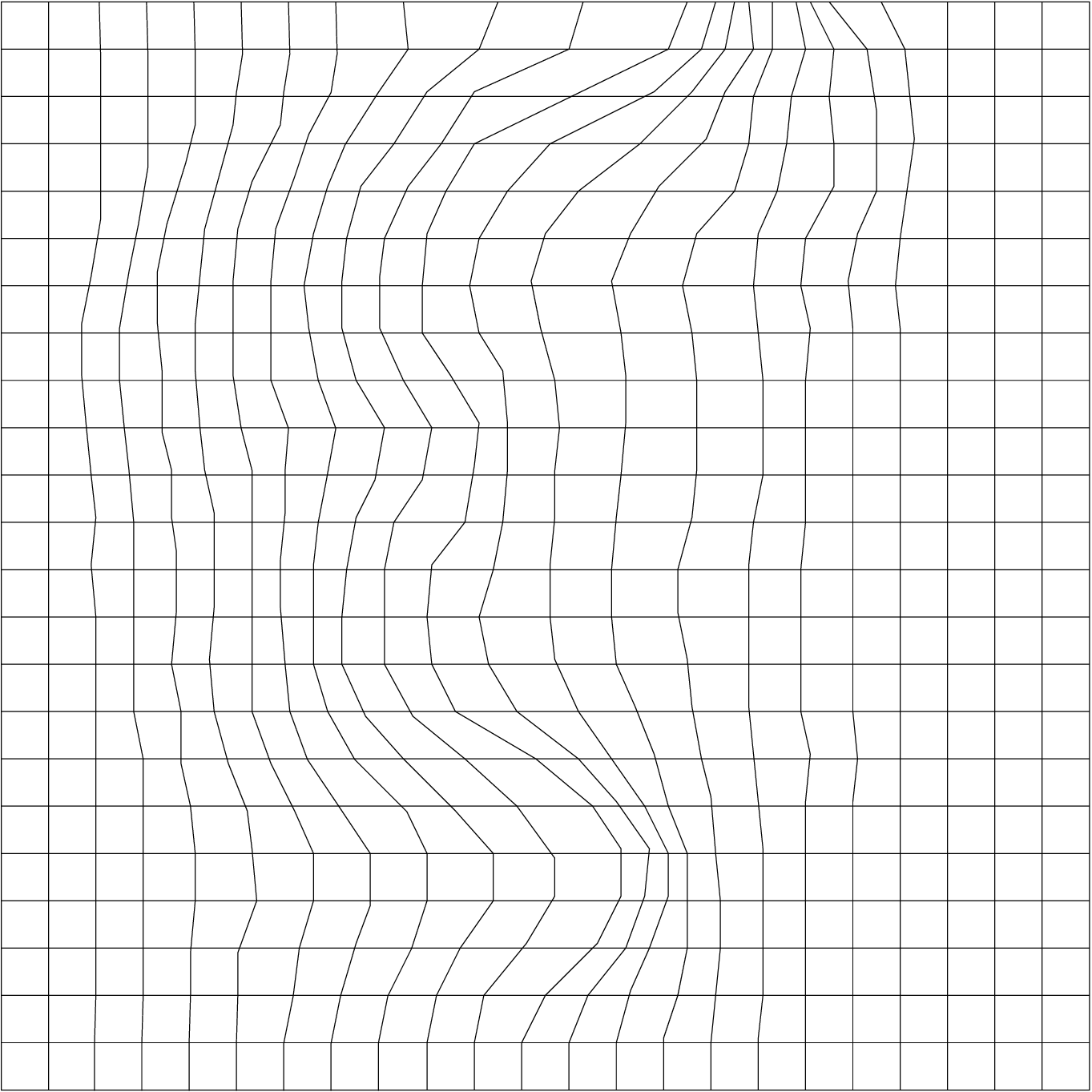}
\caption{Flare line schematic for a non-identity self-homeomorphism of $\II\times\II$ in $\premotshom$. 
This homeomorphism restricts to the identity on the south, east and west
but not the north part of the boundary.
It is the image 
under $\Theta$ 
of a path in $\TOPO^h(\II,\II)$ starting at $\id_{\II}$ 
(mapped to the southern edge)
but not ending at $\id_{\II}$.
}
\label{fig:flare-noid}
\end{figure}

We now introduce `flare schematic' representations for individual \premot{}s (at least for $M$ of low dimension). 
These are to be understood as follows.
For a \axiomM{} $M$, a flare schematic represents 
a homeomorphism $g\colon M\times \II \to M\times \II $ in $\premotshom$, which by Lemma~\ref{le:homeo_MxItoMxI} represents a \premot{}.
To construct a flare schematic for $g$ we proceed as follows.
We first choose a subset $F\subset M\times \II$ which is a $1$-dimensional complex, i.e. the union of a set of (perhaps neatly) embedded real intervals 
-- flare lines. The idea is to choose $F$ such that 
the image of $F$ under $g$ reveals information about $g$.
We note that this is only a useful exercise for sufficiently well behaved $g$, and in fact we consider only cases for which we have
a one-size-fits-all $F$ that is 
some regular array of lines.
Finally we make
a picture (i.e. a black-in-white-out representation of a projection onto the plane) of 
the 
subset $F$; together with a picture of the image $g(F)$. 

Our first examples are Figures~\ref{fig:flare-id} and \ref{fig:flare-noid} where $M=\II$.
The ambient space $\II$ is oriented horizontally left to right, and $t$ in the second copy of $\II$ progresses up the page. We fix a choice of $F\subset \II$ which is shown on the left hand side of both figures.
 The right hand figures represent the image 
 of $F$ under some homeomorphism in $\premotshom$.
 In Figure~\ref{fig:flare-id} this is the identity morphism $\id_{\II\times \II}$ and in Figure~\ref{fig:flare-noid} we have a non-identity homeomorphism.

Figures~\ref{fig:1-2-1},
\ref{fig:1-1-1}, \ref{fig:S1.1} and \ref{fig:S1-7} show
self-homeomorphisms of $M\times \II$ in $\premotshom$ where $M=S^1$.
Again we 
choose a subset $F\subset S^1\times\II$ which consists of the product of eight marked points in $S^1$ with $\II$, together with `horizontal' (in the sense of being orthogonal to $\II$) lines.
We have drawn $-\times \II$ radially, thus marked points become radial lines, and
horizontal lines 
will here
become concentric circles.
We note that in this case we vary the number of horizontal lines in $F$ depending on each self-homeomorphism.
These are by construction fixed setwise and thus are really only a guide to the eye.

We turn now to the paths themselves.
The paths in $\TOPO(S^1,S^1)$ represented by Figures~\ref{fig:1-1-1} and \ref{fig:1-2-1} both end at a different self-homeomorphism to $\id_{S^1}
$.
The paths represented by Figures~\ref{fig:S1.1} and \ref{fig:S1-7} instead both end at $\id_{S^1}$.

The 
path in Figure~\ref{fig:S1.1} is path homotopic to the constant path.
The 
path in Figure~\ref{fig:S1-7} is not.

\newcommand{\sdisk}{4}  

\newcommand{\Sonecaption}{Illustration of a path of self-homeomorphisms of the circle $M = S^1$,
realised as a homeomorphism $M\times\II \rightarrow M\times\II$. 
    The circle is drawn together with eight marked points upon it, 
    to reveal the space `moving' under the path of self-homeomorphisms.
    In this case the $-\times\II$ is drawn 
    radially, outside-to-inside, rather than
    bottom-to-top on the page (so the drawing scale changes
    with radial distance; while the angular coordinate does not).
    }

\begin{figure} 
    \centering
    \includegraphics[width=\sdisk cm]{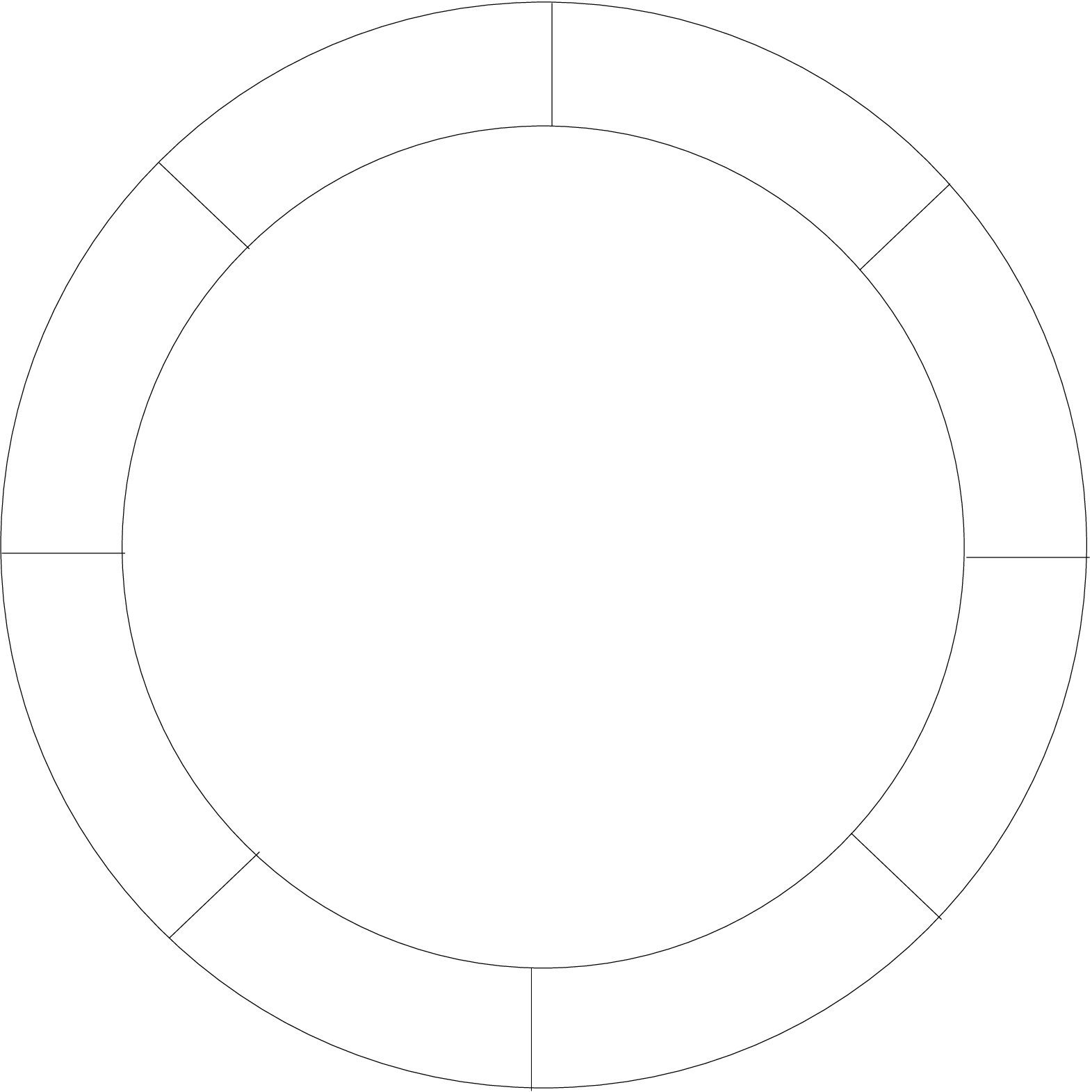}
    \raisebox{1.9cm}{$\;\;\mapsto\;\;$}
    \includegraphics[width=\sdisk cm]{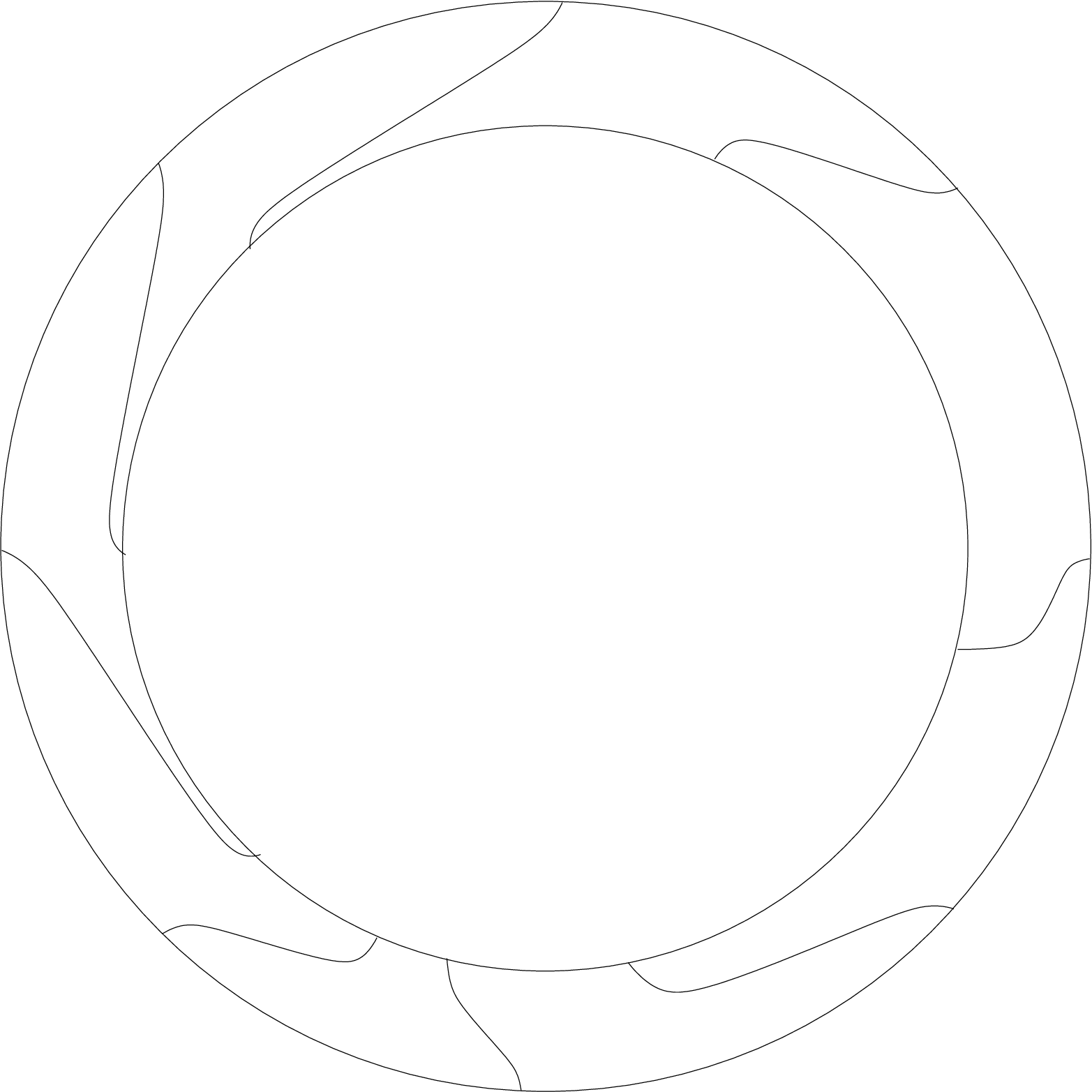}
\caption{\Sonecaption
The path in $\Topo^h(S^1, S^1)$ illustrated here does not end at the same homeomorphism in which it begins.}
\label{fig:1-2-1}
\end{figure}

\begin{figure} 
    \centering
    \includegraphics[width=\sdisk cm]{Figs3/S1-1-1.eps}
    \raisebox{1.9cm}{$\;\;\mapsto\;\;$}
    \includegraphics[width=\sdisk cm]{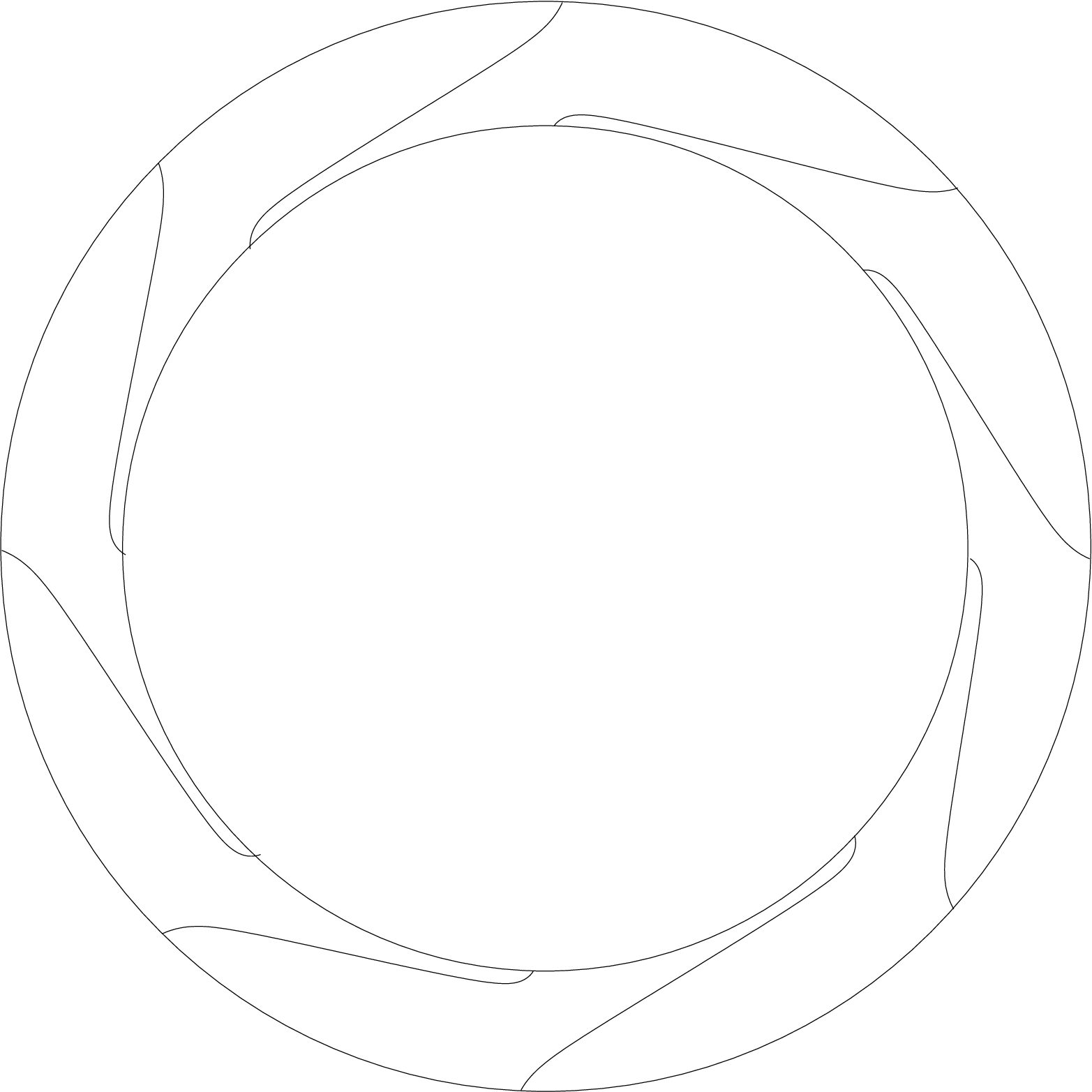}
\caption{
The path in $\Topo^h(S^1, S^1)$ illustrated here does not end at the same homeomorphism in which it begins.}
\label{fig:1-1-1}
\end{figure}

\begin{figure}
  \centering
    \includegraphics[width=6.08cm]{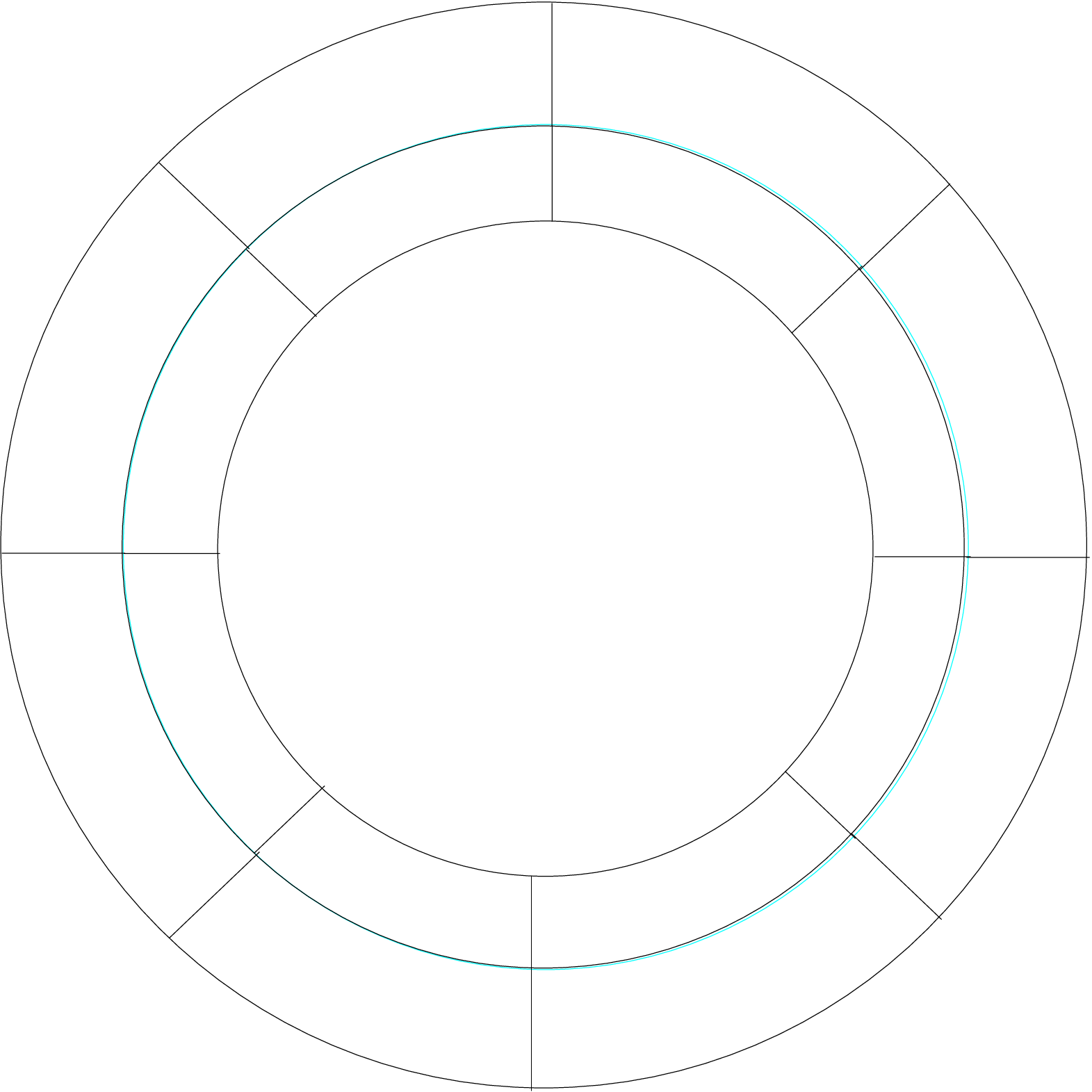}
    \raisebox{2.9cm}{$\;\;\mapsto\;\;$}
    \includegraphics[width=6.08cm]{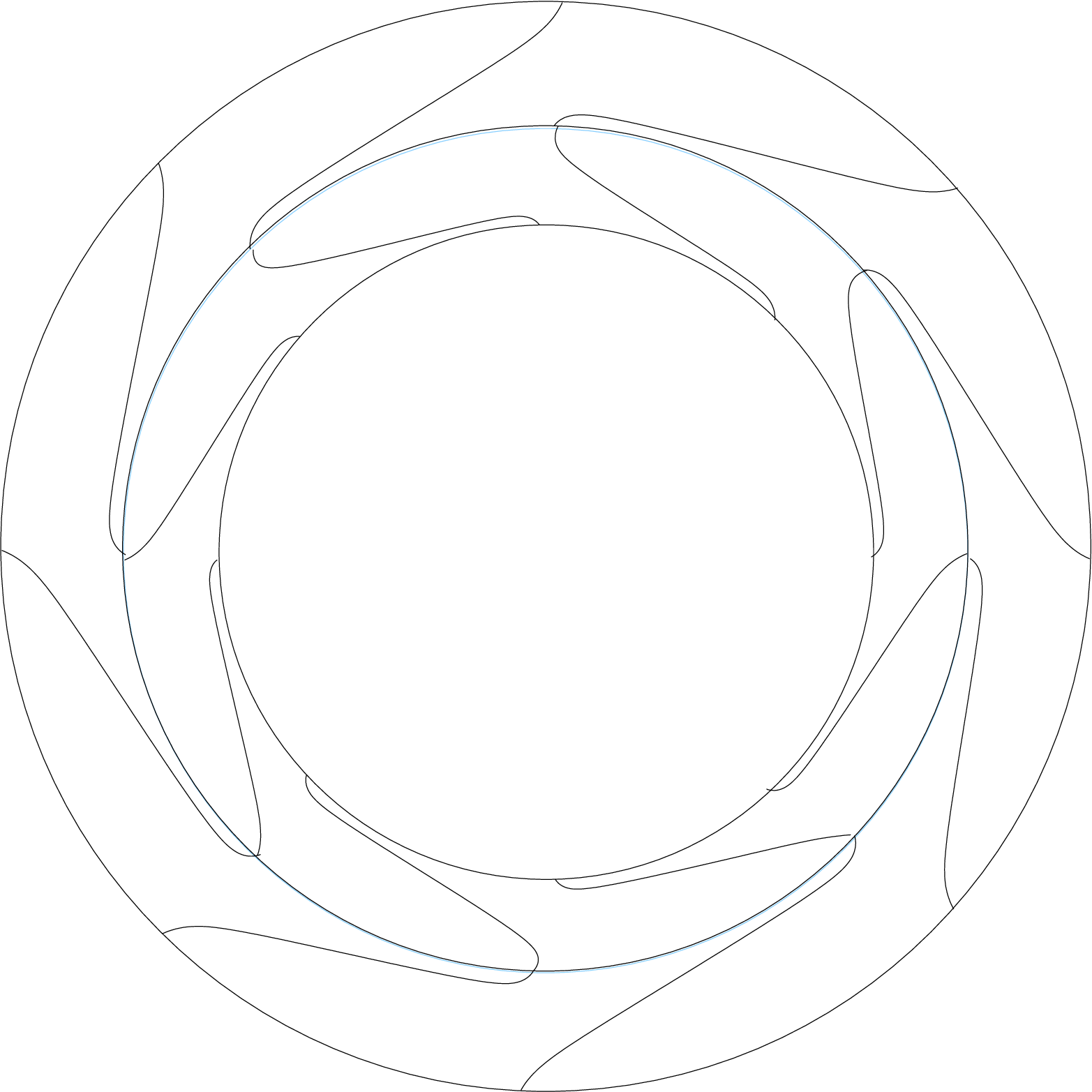}
\caption{Illustration of a path of self-homeomorphisms of the circle $M = S^1$,
realised as a homeomorphism $M\times\II \rightarrow M\times\II$. 
 }
\label{fig:S1.1}
\end{figure}

\begin{figure} 
\hspace{-.5cm}
    \includegraphics[width=7.8cm]{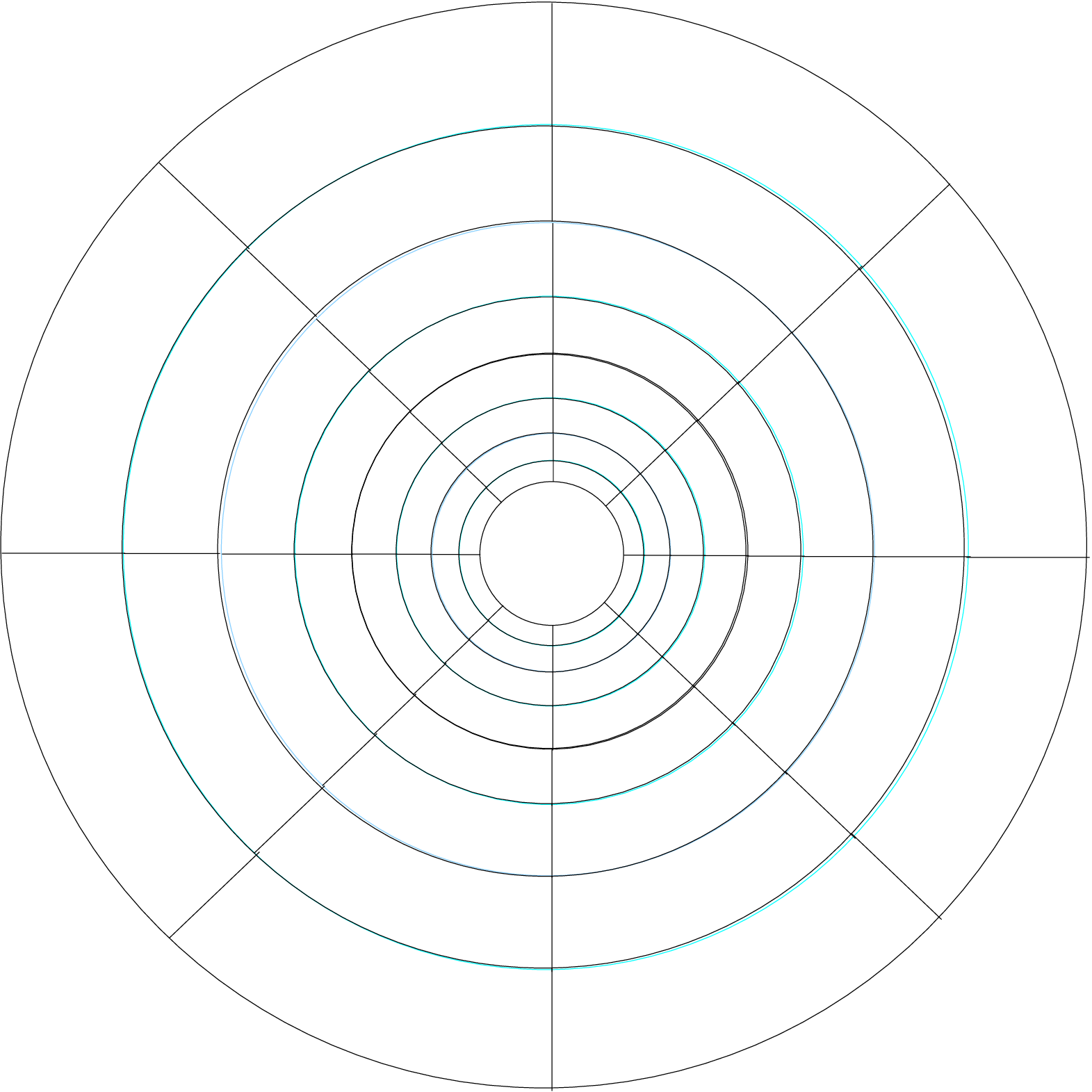}
    \raisebox{3.7cm}{$\;\;\mapsto\;\;$}
    \includegraphics[width=7.8cm]{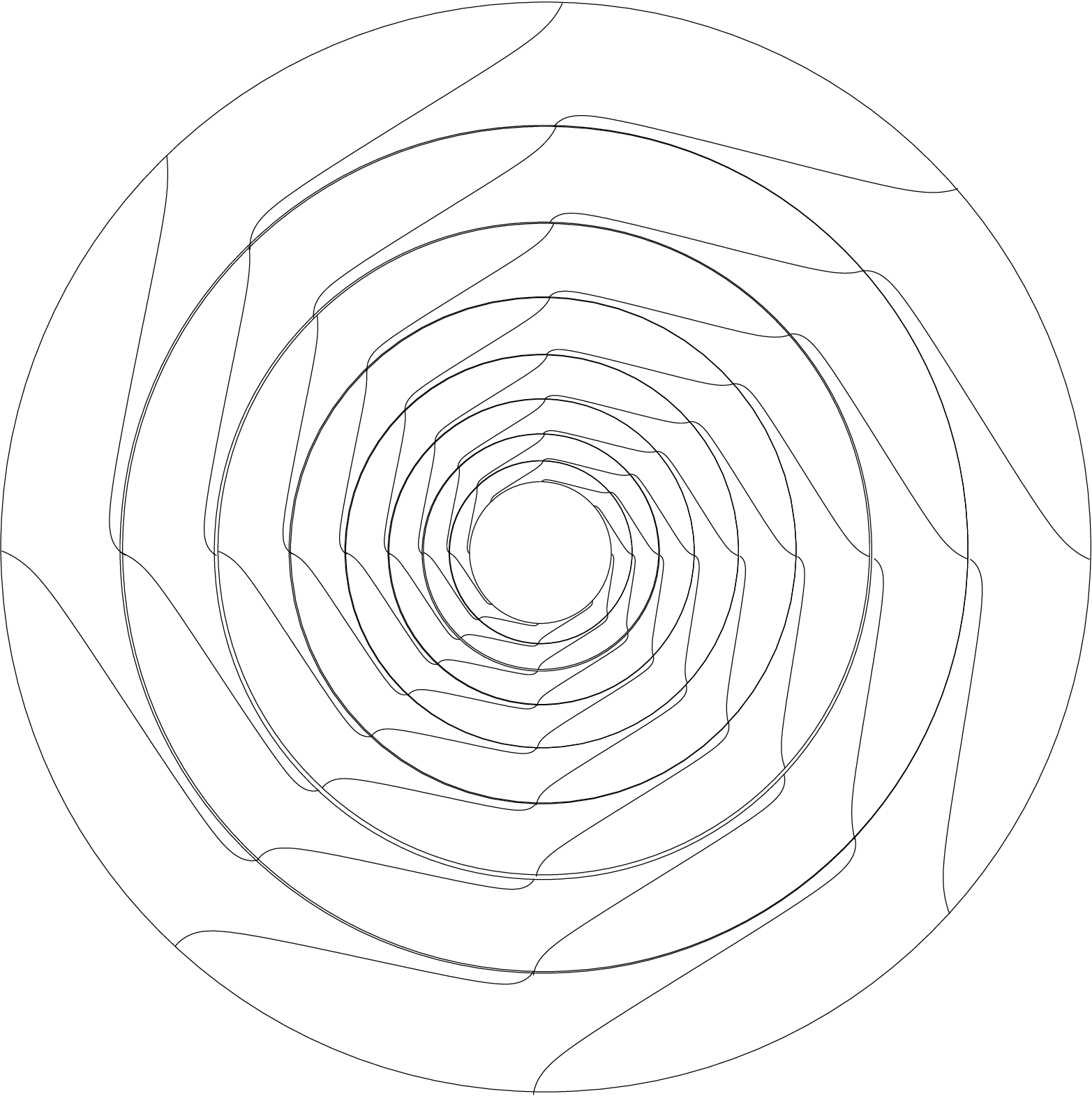}
    \caption{
    Illustration of a path of self-homeomorphisms of the circle $M = S^1$. 
    Comparing with Fig.\ref{fig:S1.1},
    both paths can be taken to start at $\Id_\II$, and both finish at the same point.
    }
    \label{fig:S1-7}
\end{figure}

\defn{\label{de:Mthom}
Let $M$ be a \axiomM{} and $N,N'\subseteq M$.
Let $ \Mtcmaghom(N,N') \subseteq \premotshom$
denote the subset of 
homeomorphisms
$g\in\premotshom$ 
such that $g(N\times \{1\})=N'\times \{1\}$.	
}

\begin{theorem}
\label{le:motion_MxI}
 Let $M$ be a \axiomM{} and $N,N'\subseteq M$.
The restriction of $\Theta$ as in Lemma~\ref{le:homeo_MxItoMxI} yields a bijection
\ali{\Theta\colon\Mtcmag(N,N')
\;&\to\; \Mtcmaghom(N,N')
.}
\end{theorem}
\begin{proof}
	Notice each $\Mtcmaghom(N,N')$ is a subset of $\premotshom$.
	Lemma~\ref{le:homeo_MxItoMxI} gives that $\Theta$ yields a bijection $\premots\cong\premotshom$ hence we only need to check that $\Theta(\Mtcmag(N,N'))\subseteq \Mtcmaghom(N,N')$ and $\Theta^{-1}(\Mtcmaghom(N,N'))\subseteq \Mtcmag(N,N')$.
	
	If $\mot{f}{}{N}{N'}$ is a motion, then $\Theta (f)(N\times \{1\})=f_1(N)\times \{1\}=N'\times \{1\}$.
	Now suppose $f'\colon M\times \II\to M\times \II$ is a homeomorphism with  $f'(N\times \{1\})=N'\times \{1\}$, then 
	$\Theta^{-1} (f')_1(N)=p_0\circ f'(N\times \{1\})=p_0(N'\times \{1\})=N'$, so $\Theta^{-1}(f)$ is in $\Mtcmag(N,N')$.
\end{proof}

Theorem~\ref{le:motion_MxI}
says that we can add subsets to flare schematics to obtain representations of motions.
In Figure~\ref{fig:my_labelNNN}
we have two more flare schematics corresponding to 
different motions in $\II$. 
Here we have omitted
the image showing 
$F\subset\II\times \II$ since it is the same as for Figures~\ref{fig:flare-id} and \ref{fig:flare-noid}.
These schematics represent motions from various 
intervals, so we mark these intervals in addition to the flares. Notice that the motion represented in Figure~\ref{fig:my_labelNNN}(b) is an \stationary[N''] motion.

\medskip 

Recall the from Definition~\ref{def:worldline} that $\W\left(f\colon N \too N'\right)= \bigcup_{t \in \II} f_t(N) \times \{t\}$. The following proposition follows directly from the definition of $\Theta\colon \Mtcmag(N,N')\to \Mtcmaghom(N,N')$.
\begin{proposition}
For a motion $\mot{f}{}{N}{N'}$ in a manifold $M$
\[
\pushQED{\qed}
\W\left(f\colon N \too N'\right)= 
\Theta(f)(N \times \II) \subseteq M\times \II.\qedhere
\popQED
\]
\end{proposition}

Thus the image of the subsets progressing up the page in a flare schematic is the \textit{worldline} of a motion.
The worldline retains the information about the movement of the chosen subset, and forgets the information about the movement of the ambient space.

\medskip

\begin{figure}
    \centering
    \includegraphics[width=1.965in]{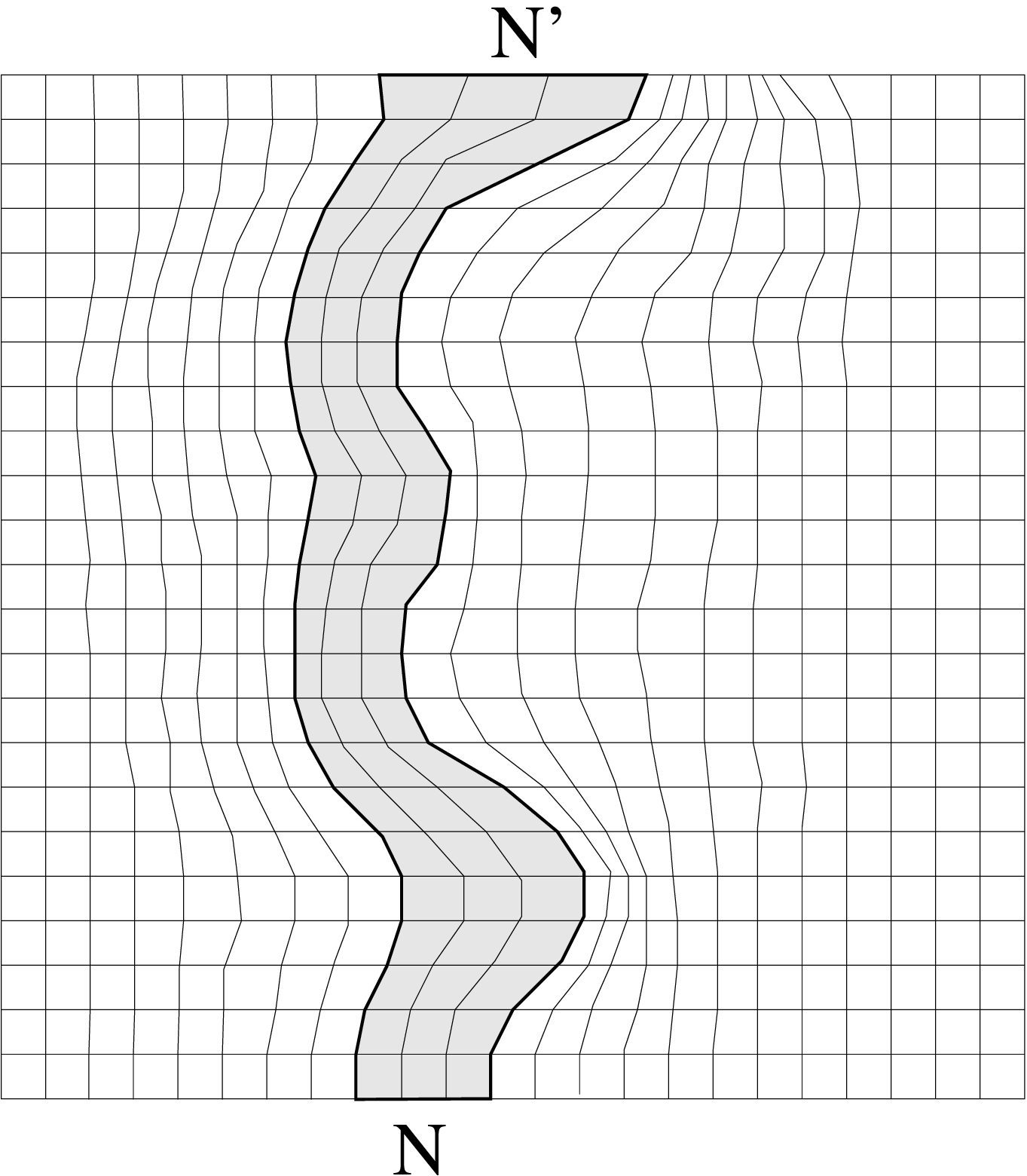}
    \hspace{.6in}
        \includegraphics[width=1.965in]{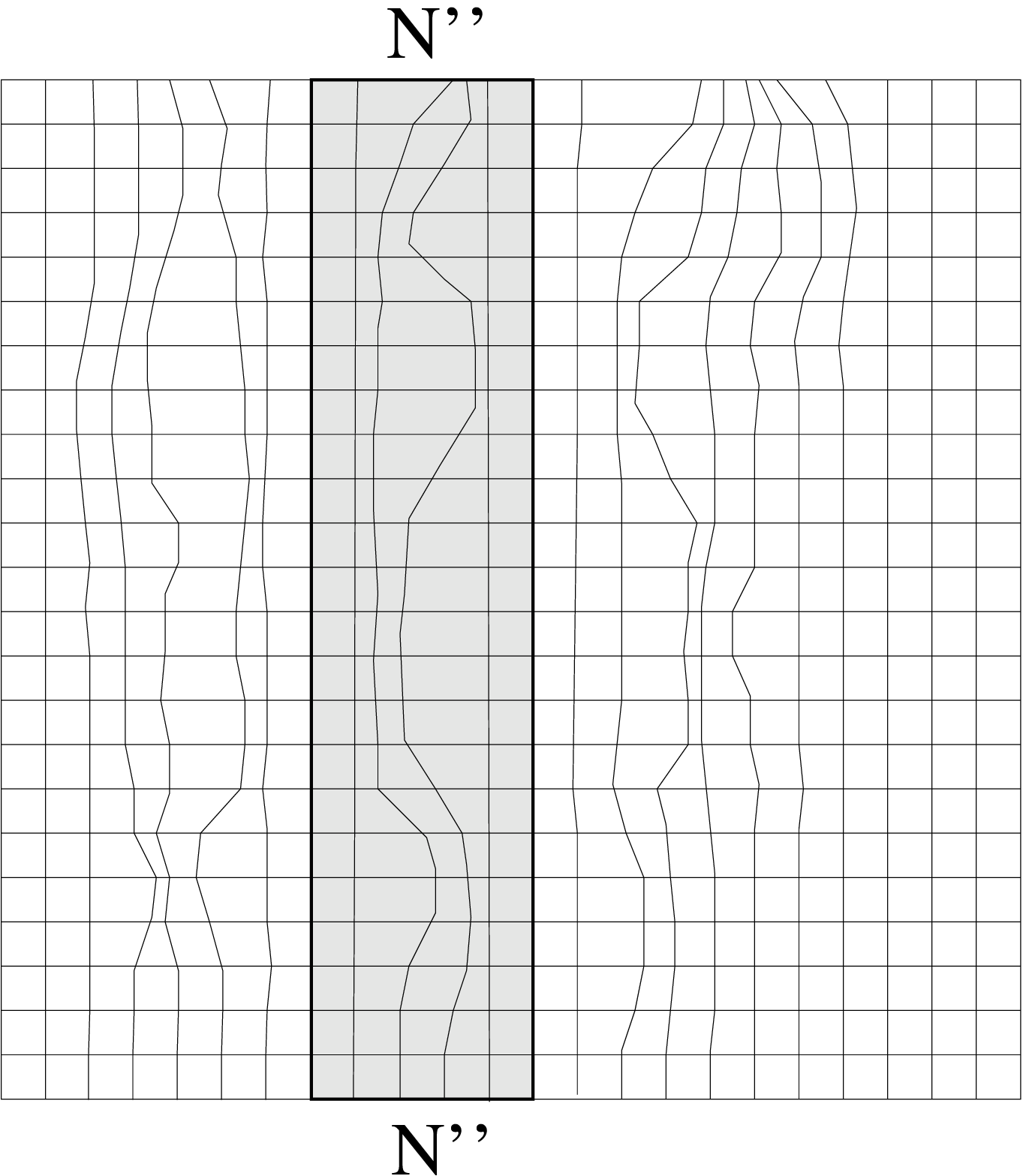}
    \caption{
    Flare schematic for motions (a) from $N$ to $N'$ and (b) from $N''$ to $N''$ in case $M=\II$, where  $N,N',N''$ are intervals in $\II$. The black represents the image of a regular grid subset of $\II\times \II$,
    and the grey shading represents the images of $N$ and $N'$, progressing up the page - the worldlines of the motions.}
    \label{fig:my_labelNNN}
\end{figure}

\begin{figure}[!ht]
    \centering
    \includegraphics[width=2in]{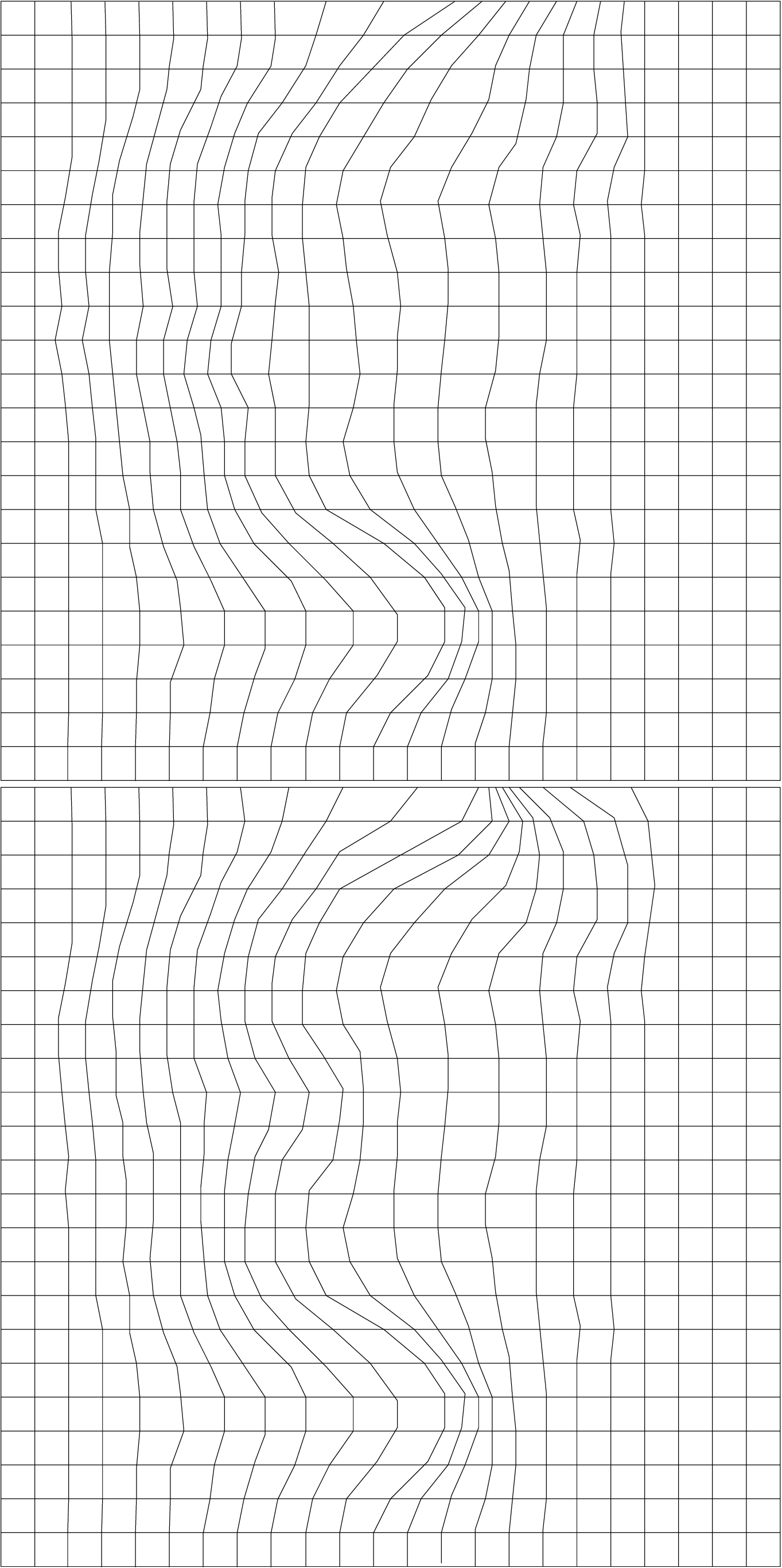}
    \hspace{.421in}
    \includegraphics[width=2in]{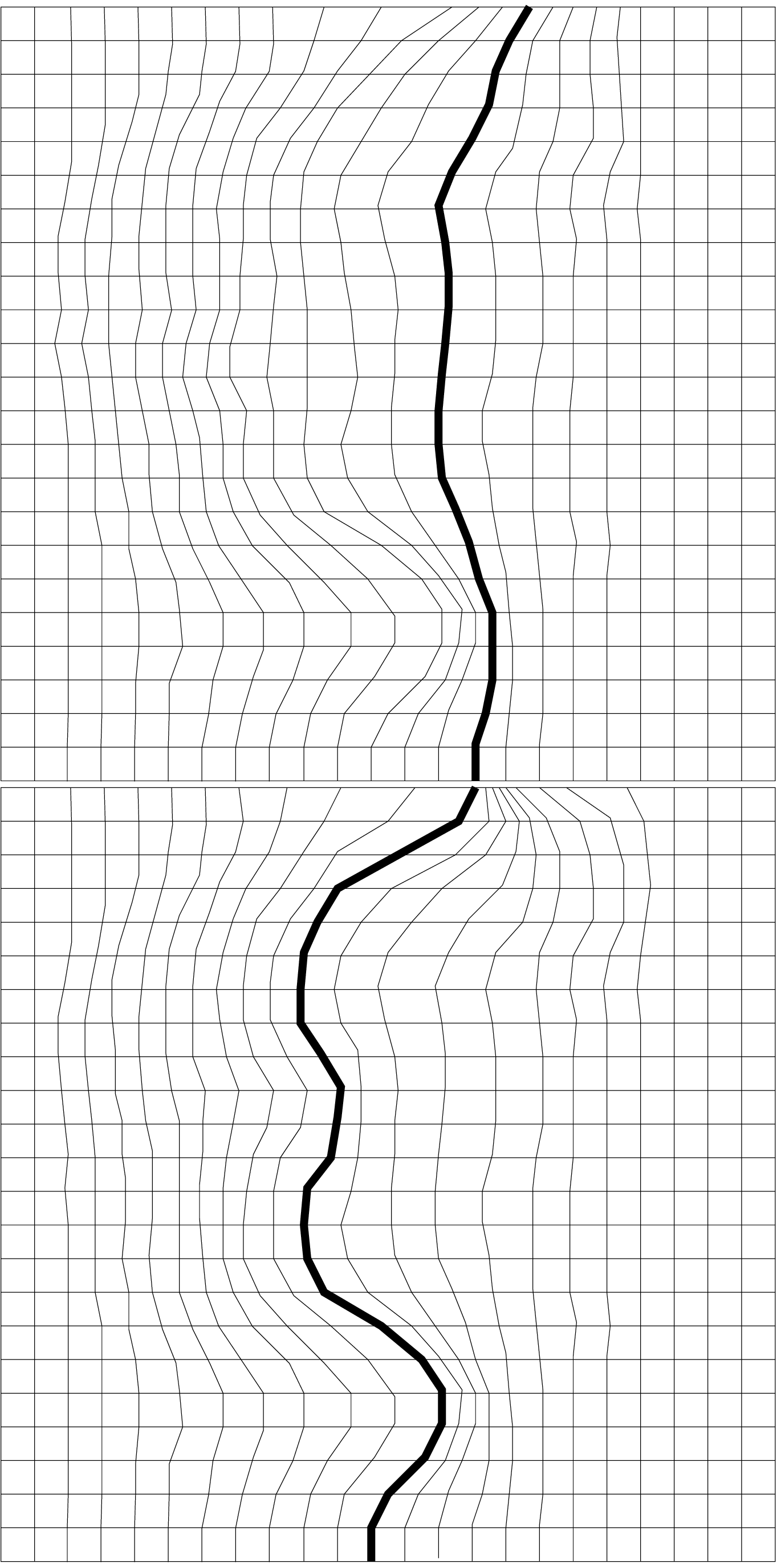}
    \caption { Schematic for  composition of 
motions.
    (a) formal stack of flare schematics of paths;
    (b) formal stack of pictures of paths with a choice of subset $N\subseteq \II$ and the worldline of $N$ marked.
    }
    \label{fig:Fschema32}
\end{figure}

In Fig.\ref{fig:Fschema32} we have an example of the $*$ composition of motions in our flare schematic representation. 
Fig.\ref{fig:Fschema32}(a) simply shows the flare-schematics for two \premot{}s in a formal stack
-- note that this is not itself a flare-schematic for a motion, since the lines
are not matched at the join. 
To turn this picture into a flare schematic, we must trace the images 
 of the marked points along the bottom boundary in the top half of the schematic. This corresponds to composing with $f_1$ in Equation~\eqref{def:comp1}.
In Fig.\ref{fig:Fschema32}(b) we consider what happens when we move to motions.
Lemma~\ref{lem:concatW} gives that simply stacking (and shrinking) the flare schematics corresponding to composable motions will give a schematic with the correct worldline of the composition. Thus if we were to turn the right hand picture into a flare schematic, the bold line representing worldline will remain the same.
We also have that given a formal stack of flare schematics for \premot{}s, choosing a subset
along the bottom boundary
 and tracking it under the first \premot{} 
 determines a choice of subset in the second motion such that paths of self-homeomorphisms become composable motions.

\medskip

A similar situation also applies to the construction of $\bar{f}$.
Suppose we have a flare schematic for a \premot{} $f$ and we turn it upside down (respectively inside-out in the $S^1$ case)
it is not a flare schematic of a motion,
because
$f_{(1-t)}$ is not the identity at $t=0$; but the composition with $f_1^{-1}$ 
in $\bar{f}$
`fixes' this.
However suppose now we add a subset and a worldline to $f$, then the naive `flip' of the flare schematic gives the worldline of $\bar{f}$.

\rem{\label{re:ambientisotopy} If we take the definition of ambient isotopy given by, for example, \cite[Sec.0.3]{kamada} in terms of maps $M\times \II$ to $ M \times \II$, Definition~\ref{de:Mthom} says precisely that $g\in \Mtcmaghom(N,N')$ is an ambient isotopy from 
$N$ to $N'$ in $M$.
There are several other 
definitions of the term `ambient isotopy' of subsets in the literature.
 The equivalences of $\Mtcmag(N,N')$, $\Mtcmagmov(N,N')$ and $\Mtcmaghom(N,N')$ proved in Lemma~\ref{le:movie_mot} and Theorem~\ref{le:motion_MxI} 
prove that the notions of ambient isotopy given by \cite{kamada} and the two notions of isotopy given by \cite[Pg.8]{Crowell_Fox}, in terms of maps from $M\times \II$ to $M$ and in terms of paths in $\Topo^h(M,M)$, are all equivalent.
Thus we have that there exists a motion from $N$ to $N'$ in $M$ if and only if $N$ and $N'$ are ambient isotopic.
In particular knots $K$ and $K'$ in $S^3$ are equivalent if and only if there is a motion in $S^3$ from $K$ to $K'$.
 }

\section{Key alternative ways to understand motion equivalence}\label{sec:equivalence}

In this section we give two alternative ways to understand motion equivalence.
We give more details on each at the start of \S\ref{sec:rel groupoid} and \S\ref{sec:laminated} below. 

\subsection{Relative path homotopy between motions}
\label{sec:rel groupoid}

In this section we introduce the relation, relative path-homotopy on the sets $\Mtcmag(N,N')$.
We prove in Theorem~\ref{th:mg2} that this relation is the same as the relation $\simm$ constructed in the previous section. 

Relative path-homotopy is the same equivalence relation used in the construction of the relative fundamental set $\pi_1(X,Y,*)$ of a pointed pair of spaces.
Thus it will allow us to use the relative homotopy long exact sequence to prove the relationship between motion groupoids and mapping class groupoids in Section~\ref{sec:mg_to_mcg}.
For $X$ a space and $f,g$ paths, recall the set of homotopies
$\phomotopy{X}{f}{g}$.
For $M$ a manifold, $f,g$ flows and $N,N'$ subsets,
let
$$
\Topo(\II^2,\TOPO^h(M,M))\rel{f}{N}{N'}{g} 
 = \hspace{8.3cm}
 $$
\[
 \{ H \in \Topo(\II^2,\TOPO^h(M,M)) \; | \; H(t,0)=f_t,
 H(t,1) = g_t, H(0,s)=\Id_M, H(1,s)(N)=N'=f_1(N) \}
\]

\defn{\label{de:rp equiv}
	Fix a manifold $M$. 
	Define a relation on 
	$\Mtcmag(N,N')$
	as follows.
	Let $\mot{f}{}{N}{N'}\simrp{}\mot{g}{}{N}{N'}$ if {$\Topo(\II^2,\TOPO^h(M,M))\rel{f}{N}{N'}{g}\neq \emptyset$. This means that}  there exists a continuous map:
	\[
	H\colon \II \times \II \to \TOPO^h(M,M),
	\]
	such that:
	\begin{itemize}
		\item for any fixed $s\in \II$, $t\mapsto H(t,s)$ is a motion from $N$ to $N'$,
		i.e. 
		$H(0,s)=\Id_M$ and
		$H(1,s) \in \Hom[M](N,N')$,
		{meaning that} $H(1,s)(N)=(N')$,
		\item for all $t\in \II$, $H(t,0)=f_t$, and
		\item for all $t\in \II$, $H(t,1)=g_t$.
	\end{itemize}
	Notation: We call such a map a {\em relative path-homotopy}, and the motions $\mot{f}{}{N}{N'}$ and $\mot{g}{}{N}{N'}$ are said to be {\em relative path-homotopic}. 
}

\lemm{\label{le:rpe}
Fix a manifold $M$. 
For each pair $N,N'$, the relation
$\simrp$ is an equivalence relation on $\Mtcmag(N,N')$.
\\ 
{Notation:}
		We call $\simrp$ equivalence classes {\em relative path-equivalence} classes and use $\classrp{\mot{f}{}{N}{N'}}$ for the {equivalence} class of $\mot{f}{}{N}{N'}$.
		}
		
		Figure~\ref{fig:rpath_homotopy} gives examples of relative path-homotopic, and non relative path-homotopic motions in our schema introduced in Figure~\ref{fig:my_lalala}. 
\begin{figure}
		\centering
		\def\svgwidth{0.6\columnwidth}
		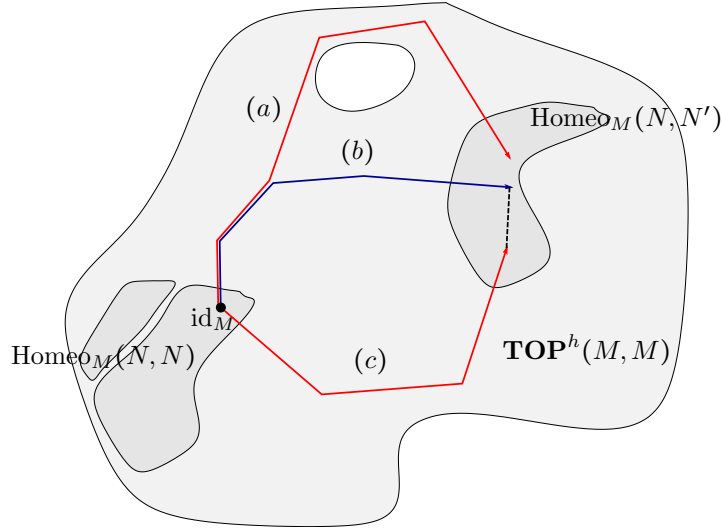
		\caption{Let $M$ be a manifold, and $N,N'\subset M$ subsets.
		Here we use the same schema as in Figure~\ref{fig:my_lalala}.
		For illustration we suppose
		$\TOPO^h(M,M)$ is a connected region of the plane homeomorphic to $S^1\times \II$.
		The paths labelled $(a)$, $(b)$ and $(c)$ represent motions from $N$ to $N'$ in $M$. There is a relative path-homotopy from $(b)$ to $(c)$, but not from $(a)$ to $(b)$ or to $(c)$.} 
		\label{fig:rpath_homotopy}
\end{figure}
	
\begin{proof}
	Let $\mot{f}{}{N}{N'}$, $\mot{g}{}{N}{N'}$ and $\mot{h}{}{N}{N'}$ be motions.
	We can prove reflexivity by observing that the homotopy $H(t,s)=f_t$ for all  $s\in \II$ is a relative path-homotopy from $\mot{f}{}{N}{N'}$ to itself. 
	
	For symmetry let $H_{f,g}$ be the relative path-homotopy from $\mot{f}{}{N}{N'}$ to $\mot{g}{}{N}{N'}$.
 	Then the function $H_{g,f}(t,s)=H_{f,g}(t,1-s)$ is a relative path-homotopy from $\mot{g}{}{N}{N'}$ to $\mot{f}{}{N}{N'}$.
	
	For transitivity let $H_{g,h}$ be the relative path-homotopy from $\mot{g}{}{N}{N'}$ to $\mot{h}{}{N}{N'}$.
	Then 
	\[
	H_{f,h}(t,s)=
	\begin{cases}
	H_{f,g}(t,2s) & 0\leq s \leq \frac{1}{2} \\
	H_{g,h}(t,2(s-\frac{1}{2})) &\frac{1}{2}\leq s\leq 1
	\end{cases}
	\]
	is a relative path-homotopy from $\mot{f}{}{N}{N'}$ to $\mot{h}{}{N}{N'}$.
\end{proof}

Recall the definition of path-equivalence from
Definition~\ref{de:pe} and also path-equivalence of motions, \ppm{$\simp$,} from Lemma~\ref{le:simp_cong}.

\prop{\label{pr:p_implies_rp}
Let $\mot{f}{}{N}{N'} \;\simp\; \mot{g}{}{N}{N'}$ be path-equivalent motions, then $\mot{f}{}{N}{N'}\;\simrp\; \mot{g}{}{N}{N'}$.}
\begin{proof}
A path-homotopy, $H$, from $f$ to $g$ has fixed endpoints, thus for fixed $s$, $t\mapsto H(t,s)$ is a motion from $N$ to $N'$. Hence $H$ is a relative path-homotopy {from} $\mot{f}{}{N}{N'}$ to $ \mot{g}{}{N}{N'}$. 
\end{proof}

Recall the construction of $\Hom$ from \S\ref{ss:selfhomeos}.
Morphisms in $\Hom(N,N')$ are triples $\shmor{f}{}{N}{N'}$, although it will often be useful to think of $\Hom(N,N')$ as the projection to the first element of the triple, and further to topologise this as a subset of $\TOPO^h(M,M)$. The meaning will be clear from context.

		\lemm{\label{le:rp_implies_m}
		Suppose we have relative path-{homotopic} motions $\mot{f}{}{N}{N'}\;\simrp\; \mot{f'}{}{N}{N'}$, then 
		${\mot{f}{}{N}{N'}\;\simm\; \mot{f'}{}{N}{N'}}$.}
	\begin{proof}
	Let $H$ be a relative path-homotopy from $\mot{f}{}{N}{N'}$ to $\mot{f'}{}{N}{N'}$.
	We must show that $ \mot{\bar{f}'*f}{}{N}{N}$ is path-equivalent to a stationary motion from $N$ to $N$.
	
	Notice first that $s \mapsto H(1,1-s)$ is a path from $f'_1$ to $f_1$, which is in $\Homn$ for all $s$. We relabel this path as $\gamma$.
	We define a path
	$\tilde{\gamma}$ by $\tilde{\gamma}_s=\gamma_s\circ f'^{-1}_1$, so 
	$\mot{\tilde{\gamma}}{}{N'}{N'}$
	is a stationary motion with $\tilde{\gamma}_1=f_1\circ f'^{-1}_1$.
	
	We can use $H$ to construct a path-homotopy from $f$ to the path composition $\gamma f'$. 
	For example, a suitable path-homotopy is: 
	\[
	H_1(t,s)=\begin{cases}
	H(\frac{2t}{2-s},s) & t\leq 1-\frac{s}{2} \\
	\gamma_{2t-1} & 1-\frac{s}{2} \leq t.
	\end{cases}
	\]
	For fixed $s\in \II$, the path $t \mapsto H_1(t,s)$ starts at the identity, traces the whole of the path $t \mapsto H(t,s)$ followed by the part of the path $\gamma$ starting from $\gamma_{1-s}=H(1,s)$ and ending at $\gamma_1$. 
	Note that the path composition, $\gamma f'$ is precisely the motion composition $\tilde{\gamma}*f'$, so $f\simp \tilde{\gamma}*f'$.

	{By Lemma~\ref{le:simp_cong}, path-equivalence is a congruence on the magmoid $\Mtmagstar$  (see Proposition~\ref{pr:mot_comp}),
	hence we have that $\bar{f}'*f$ is path-equivalent to $ \bar{f}'*(\tilde{\gamma} * f')$.}

	Now recall that $\mot{\tilde{\gamma}}{}{N'}{N'}$ is stationary. Therefore, using the normalcy of  the subgroupoid  $\setstat(N,N)$ of $\Mtcmag/\simp$, proved in Lemma \ref{le:stat_mots}, 
	it follows that the motion
	$\mot{\bar{f}'*(\tilde{\gamma}* f')}{}{N}{N}$ is path-equivalent to a stationary motion from $N$ to $N$.
	Hence $\mot{\bar{f}'*f}{}{N}{N}$ is path-equivalent to a stationary motion.
\end{proof}
		
		\lemm{\label{le:m_implies_rp} 
		Suppose we have motion-equivalent motions $\mot{f}{}{N}{N'}\;\simm\; \mot{f'}{}{N}{N'}$. 
		Then we have 
		$\mot{f}{}{N}{N'}\;\simrp\; \mot{f'}{}{N}{N'}$.}

		\begin{proof}
		We have from Theorem~\ref{th:mg} that $\Mot$ is a groupoid. Hence it follows from the uniqueness of inverses that $\mot{f}{}{N}{N'}\;\simm\; \mot{f'}{}{N}{N'}$ implies 
	$\mot{\bar{f}}{}{N'}{N}\;\simm \;\mot{\bar{f}'}{}{N'}{N}$.
	Thus there exists a path-homotopy, say $H$, from $f'*\bar{f}$ to a stationary motion $\mot{\gamma}{}{N'}{N'}$.

We also have that  ${\mot{f}{}{N}{N'}\;\simrp \; \mot{\gamma*f}{}{N}{N'}}$. To prove this we can use the following function:
	\[
	H_1(t,s)=\begin{cases}
	f_{\frac{2t}{2-s}} & t\leq 1-\frac{s}{2} \\
	\gamma_{2(t+\frac{s}{2}-1)}\circ f_1 & 1-\frac{s}{2}\leq t.
	\end{cases}
	\]
	Notice that $H_1(t,0)$ is the path $f$ and $H_1(t,1)$ is the path $ \gamma *f$.
	For any fixed $s\in \II$, we have $H_1(0,s)=\id_M$, and also: \[H_1(1,s)(N)=\gamma_{s}\circ f_1(N)=\gamma_s(N')=N'.\]
	Note that $H_1$ is continuous as both functions agree when $t=1-\frac{s}{2}$.
	Hence we have that $H_1$ is a relative path-homotopy, proving that ${\mot{f}{}{N}{N'}\;\simrp\; \mot{\gamma*f}{}{N}{N'}}$.

Using the associativity of $\cdot$ (Proposition~\ref{pr:dot_comp}) together with
Lemma \ref{le:pms_star_equiv_dot}, 
the fact that $*$ and $\cdot$ composition are the same up to path-equivalence, the fact that path-equivalence is a congruence on motions (Lemma~\ref{le:simp_cong}), and that $\bar{f}*f\simp\Id_M$ (Lemma~\ref{le:Mtcmag}), we have 
 \[\gamma*f\simp \gamma\cdot f\simp (f'*\bar{f})\cdot f \simp f'\cdot \bar{f}\cdot f\simp f'\cdot(\bar{f}* f)\simp f'\cdot \Id_M=f',\]
	so $\gamma*f\simp f'$.
	
	Since path-equivalence implies relative path-equivalence (Proposition~\ref{pr:p_implies_rp}), we have $\mot{\gamma*f}{}{N}{N'}\;\simrp\; \mot{f'}{}{N}{N'}$, and hence we have  ${\mot{f}{}{N}{N'}\;\simrp\; \mot{f'}{}{N}{N'}}$.
		\end{proof}

All proofs in this section work in exactly the same way restricting to $A$-fixing motions, and defining $A$-fixing relative path-homotopies between $A$-fixing motions in the obvious way.
Hence 
from  Lemmas~\ref{le:rp_implies_m} and \ref{le:m_implies_rp}
we have: 
\begin{theorem}\label{th:mg2}
For a manifold $M$ and an $A$-fixing motion $\mot{f}{A}{N}{N'}$ in $M$ we have
\[
        \classrp{\mot{f}{A}{N}{N'}}=\classm{\mot{f}{A}{N}{N'}}.
        \]
             In particular, quotienting $\Mtcmag^A$ by relative path-equivalence leads to the same groupoid as quotienting by motion-equivalence. \qed
\end{theorem}

\begin{remark}
Let $M$ be a manifold. Combining Theorems \ref{th:mg} and \ref{th:mg2} we hence can see that the motion groupoid $\Mot[M]$ is such that given $N,N'\subset M$ then $\Mot[M](N,N')$ is the set of relative path homotopy classes of motions $N\too N'$. The composition of $\classrp{f\colon N \too N'}$ with $\classrp{g\colon N' \too N''}$ is $\classrp{g*f\colon N \too N''}$, which by Lemma \ref{le:pms_star_equiv_dot}  is the same as  $\classrp{g\cdot f\colon N \too N''}$. In particular, indeed the compositions $*$ and $\cdot$ in the magmoids   $\Mtcmag^{\;\cdot}$ and \ppm{$\Mtcmag^{\;\ast}$}  (the notation is as in \S \ref{ss:motions}) do descend to the quotient under $\simrp$.

For instance, if $f,f'\colon N \too N'$ are relative path-homotopic via $H\colon \II \times \II \to \TOPO^h(M,M)$ and $g,g'\colon N' \too N''$ are relative path-homotopic via $J\colon \II \times \II \to \TOPO^h(M,M)$, then the functions $K,K'\colon \II\times \II \to \TOPO^h(M,M)$ below are relative path-homotopies connecting $g*f\colon N \too N''$ to $g'*f'\colon N \too N''$, and  $g\cdot f\colon N \too N''$ to $g'\cdot f'\colon N \too N''$, respectively,
\begin{align*}
K(t,s)&=\begin{cases} H(2t,s), \textrm{if } t \in [0,1/2]\\
J(2t-1,s)\circ H(1,s), \textrm{if } t \in [1/2,1],\end{cases}\\
K'(t,s)&= J(t,s)\circ H(t,s).
\end{align*}
The same formulae also work in the $A$-fixing case.
 \end{remark}

\begin{remark}\label{rem:dahm_rel_path}
Suppose that the motions $f,g\colon N \too N'$ have the same worldline (Definition \ref{def:worldline}), then they are motion equivalent, by  Proposition \ref{prop:dahm}.
By Theorem~\ref{th:mg2} it follows that $f\colon N \too N'$ and $g\colon N \too N'$ are relative path-homotopic. An explicit relative path-homotopy from $g\colon N \too N'$ to $f\colon N \too N'$, is, for example, $H\colon \II \times \II \to \TOPO^h(M,M)$, defined as:
\[H(t,s)=g_t\circ g_{ts}^{-1} \circ f_{ts}.
\]   
To see that $H$ is  a relative path-homotopy, note that if $s \in I$ then $g_s(N)=f_s(N)$.
As above, the same formulae  also work in the $A$-fixing case.
\end{remark}

\subsection{Level preserving isotopies between worldlines of motions}\label{sec:laminated}

Recall from Definition~\ref{def:worldline} that
the worldline of a motion $\mot{f}{}{N}{N'}$ in $M$ is
\[\W\left(f\colon N \too N'\right)= \bigcup_{t \in [0,1]} f_t(N) \times \{t\}\subseteq M\times \II.
\]
In Section~\ref{ss:stat_mot} we had Proposition~\ref{prop:dahm}, which said that two motions represent the same morphism in the groupoid $\Mot$ if their worldlines are the same{; cf. also Remark~\ref{rem:dahm_rel_path} just above.} 
Here we prove that we can strengthen this result, by which we mean that we can formulate motion equivalence entirely in terms of a relation on worldlines of motions.

We begin by defining when two subsets of a space of the form $M\times \II $ are \textit{level preserving ambient isotopic} relative to some fixed subsets. The main result in this section is Theorem~\ref{th:me_bflpai} which says that
two motions are motion-equivalent if, and only if, their worldlines are level preserving ambient isotopic relative to $M\times (\{0\}\cup \{1\})$, pointwise.

\medskip

The following definitions of level preserving homeomorphism and level preserving isotopy can be found in \cite{waldhausen}.

 \defn{For a space $M$, a homeomorphism $f\colon M\times \II\to M\times \II$ is called {\em level preserving} if it is of the form $f(m,t)=(p_0\circ f(m,t),t)$, where $p_0$ is the projection to the first component of the product.}

\defn{\label{de:lpi}
For a space $M$, a {\em level preserving isotopy} of $M\times \II$ is an isotopy {of $M\times \II$} through level preserving maps. This means there exists a continuous map
$
H\colon (M\times \II)\times \II \to M\times \II
$,
such that:
\begin{itemize}
    \item for each fixed $s\in\II$, $(m,t)\mapsto H(m,t,s)$ is a level preserving homeomorphism from $M\times \II$ to $M\times \II$, 
    \item for all $m\in M$ and $t\in\II$, $H(m,t,0)=(m,t)$.
\end{itemize}
}

\rem{\label{rem:lpi_premot}Definition~\ref{de:lpi} says that a level preserving isotopy of $M\times \II$ is an element of $\premo{M\times \II}^{mov}$ (Definition~\ref{de:premotsmov}), and thus  {it is in canonical correspondence with a \premot{} of $M\times \II$}. 
We will  be using the existence of {level preserving isotopies} to construct 
{an equivalence relation between subsets of $M\times\II$}, 
rather then keeping track of the isotopy itself.
{Therefore we use the term “isotopy”, rather than “\premot{}”, since we} feel it is unhelpful to conflate the two notions.
}

\defn{Let $M$ be a space, and let $E,P\subseteq M\times \II$ be fixed subsets.
Two subsets $K,L\subseteq M\times \II$ are {\em level preserving ambient isotopic}, relative to $E$ setwise and $P$ pointwise, if
there is a level preserving isotopy of $M\times \II$, such that:
\begin{itemize}
    \item $H(K\times \{1\})=L$,
    \item for all $s\in\II$, $H(E\times \{s\})=E$, and
    \item for all $p\in P$, and $s\in \II$, $H(p,s)=p$.
\end{itemize}
{We will say that $H$ is a level preserving isotopy,  relative to $E$ setwise and $P$ pointwise, from, or connecting, $K$ to $L$.}
}

\begin{lemma}
For a space $M$, and subsets $K,L \subset M\times \II$, $K\sim L$ if there exists a level preserving isotopy  $M\times \II$ from $K$ to $L$, relative to $E$ setwise and $P$ pointwise, is an equivalence relation on subsets of $M\times \II$.
\end{lemma}
\begin{proof}
Let $J,K,L\subseteq M\times \II$ be subsets, such that there are level preserving ambient isotopies relative to $E$ setwise and $P$ pointwise, $H_1\colon (M\times\II)\times \II\to M\times \II$, from $J$ to $K$, and $H_2\colon (M\times\II)\times \II\to M\times \II$, from $K$ to $L$. 
Reflexivity is proved by choosing the map $(m,t,s)\mapsto (m,t)$ which relates $K$ to itself.

{Let $f\colon M\times \II \to M\times \II$ be the inverse of the homeomorphism $(m,t) \mapsto H_1(m,t,1)$.
The map $(m,t,s)\mapsto H_1(f(m,t),1-s)$ then} relates $K$ to $J$, so we have symmetry. 

For transitivity, we may choose the map $(M\times \II)\times \II\to M\times \II$, $(m,t,s)\mapsto H_2(H_1(m,t,s),s)$, which  {relates} $J$ to $L$.
\end{proof}

\rem{By Remark~\ref{rem:lpi_premot} level preserving isotopies are \premot{}s, thus they can be composed as \premot{}s. In the previous proof, we essentially used the $\cdot$ composition to prove transitivity, we could have also used $*$ composition and used Lemma~\ref{le:comp_continuous_eachvariable} to prove continuity, as opposed to Theorem~\ref{le:top_group}. {On the other hand, symmetry was proven by using the reverse of flows in Proposition~\ref{pr:rev}.}}

Below we have two theorems proving that two \textit{a priori} different equivalence relations on worldlines of motions, each defined using level preserving ambient isotopies, both induce the same equivalence relation on motions as motion equivalence.
\begin{lemma}\label{le:hatH}
Let $M$ be a manifold. Suppose we have  motion equivalent motions  $\mot{f}{}{N}{N'}\simm \mot{f'}{}{N}{N'}$ in $M$, and so, by Theorem~\ref{th:mg2}, there exists a relative path homotopy, say $H\colon \II\times \II\to \TOPO^h(M,M)$, from $f$ to $f'$. Then there exists a continuous map: \begin{align*}
 \hat{H}\colon M\times \II\times\II &\to M\times \II \\
(m,t,s)& \mapsto\big((H(t,s)\circ f_t^{-1})(m),t\big).
\end{align*}
Moreover $\hat{H}$ is a level preserving ambient isotopy from the worldline $\W(\mot{f}{}{N}{N'})$ to the worldline $\W(\mot{f'}{}{N}{N'})$, relative to $M\times \{0\}$  pointwise, and $N'\times \{1\}$  setwise.
\end{lemma}
\begin{proof}
We first prove that $\hat{H}$ is continuous.
Define a map $H'\colon \II\times \II\to \TOPO^h(M,M)$ by $(t,s)\mapsto H(t,s)\circ f_t^{-1}$. This is continuous since, by Theorem~\ref{le:top_group}, $\TOPO^h(M,M)$, with composition, is a topological group.
Notice in particular, that for all $t\in\II$, $H'(t,0)=\id_M$.
 Using the bijection $\Phi$, following from the product-hom adjunction (Lemma~\ref{th:tensorhom}), with $X=\II\times \II$, $Y=M$ and $Z=M$, we have a continuous map $\Phi(H')\colon M\times \II \times \II \to M$. The map $\hat{H}\colon M\times \II \times \II \to M\times \II$ is explicitly defined by: \[(m,t,s)\mapsto (\Phi(H')(m,t,s),t)=((H(t,s)\circ f_t^{-1})(m),t),\] which is continuous since each component is continuous.

From its construction, it is immediate that $\hat{H}\colon M\times \II \times \II \to M\times \II$   is a level preserving isotopy, and it follows directly from the definition of a relative path homotopy (Definition \ref{de:rp equiv}) that $\hat{H}$ fixes $M\times \{0\}$ pointwise, and $N'\times \{1\}$ setwise.

Now note that:
\begin{align*}
   \hat{H}\big(\W (\mot{f}{}{N}{N'})\times \{1\}\big)
   &=\hat{H}\bigg(\bigg(\bigcup_{t \in [0,1]} f_t(N) \times \{t\}\bigg)\times \{1\} \bigg)\\
&=\bigcup_{t \in [0,1]} H(t,1)(N) \times \{t\} \\
&=\bigcup_{t \in [0,1]} f'_t(N) \times \{t\} \\
&=\W(f'\colon N \too N'\big). \qedhere
\end{align*}
\end{proof}
Conversely we have the following.
\begin{lemma}\label{le:levelpreserving_implies_motionequiv}
 Let $M$ be a manifold and let $\mot{f,f'}{}{N}{N'}$ be motions in $M$. 
Suppose that their worldlines $\W(f\colon N \too N'\big)$ and $\W(f'\colon N \too N'\big)$ are level preserving ambient isotopic, relative to $M\times \{0\}$ pointwise and $N'\times \{1\}$ setwise. Then $\mot{f}{}{N}{N'}\simm \mot{f'}{}{N}{N'}$.
\end{lemma}
\begin{proof}
There exists a level preserving isotopy, say $H\colon M\times \II\times \II \to M\times \II$, from $\W(f\colon N \too N'\big)$ to $\W(f'\colon N \too N'\big)$,
relative to $M\times \{0\}$ pointwise and $N'\times \{1\}$ setwise. 
We construct a motion $\mot{g}{}{N}{N'}$, and a relative path homotopy $\hat{H}$ from $\mot{f}{}{N}{N'}$ to $\mot{g}{}{N}{N'}$, and then show that
$\W(\mot{g}{}{N}{N'})= \W(\mot{f'}{}{N}{N'})$.
It follows (by Proposition \ref{prop:dahm}) that the motions $f'\colon N\too N'$ and $g\colon N \too N'$ are motion equivalent, and hence $\mot{f}{}{N}{N'}\simm\mot{f'}{}{N}{N'}$.

There is a continuous map from $M\times \II\times \II$ to $M$ given by $(m,t,s)\mapsto p_0\big( H(f_t(m),t,s)\big)$, where $p_0$ is the projection to the first coordinate.
Applying the bijection $\Phi^{-1}$, as in the proof of the product-hom adjunction (Lemma~\ref{th:tensorhom}), 
 we obtain a continuous map $\hat{H}\colon \II\times \II \to \TOPO^h(M,M)$, defined by $(t,s)\mapsto \big (m\mapsto p_0 (H(f_t(m),t,s)) \big)$.
Notice that $\hat{H}(t,0)=f_t$, because $H(m,t,0)=(m,t)$.

Let $g_t=\hat{H}(t,1)$. Then by construction $t \mapsto g_t$ is a path in  $\TOPO^h(M,M)$, starting at $\id_M$. {(It starts at $\id_M$ given that $H$ is relative to $M\times \{0\}$, pointwise.)} Also 
\begin{align*}
g_1(N)&= \{ p_0(H(f_1(n),1,1)) \mid n \in N\} \\
&= \{ p_0(H(n',1,1)) \mid n' \in N'\}\\&=N'.
\end{align*}
{(Where we now used the fact that $H$ is relative to $N'\times \{1\}$, setwise.)}
So we have a motion
$\mot{g}{}{N}{N'}$. Moreover, using {again} that $H(N'\times \{1\} \times \{s\})=N'\times \{1\}$,  for all $s \in \II$, the function $\hat{H}$ is a relative path homotopy from $\mot{f}{}{N}{N'}$ to $\mot{g}{}{N}{N'}$, and hence, using Theorem \ref{th:mg2}, $\mot{f}{}{N}{N'}\simm \mot{g}{}{N}{N'}$.

The worldline of $g$ is the same as the worldline of $f'$, since we have
\begin{align*}
\W(g\colon N \too N')&=\bigcup_{t \in [0,1]} g_t(N) \times \{t\}\\
&=\bigcup_{t \in [0,1]} \hat{H}(t,1)(N) \times \{t\} \\
&=\bigcup_{t \in [0,1]} \bigg(p_0 \big(H( f_t(N)\times \{t\}\times \{1\})\big)\bigg) \times \{t\} \\
&=\bigcup_{t \in [0,1]} \big(p_0 ( f'_t(N)\times \{t\})\big) \times \{t\} \\
&=\bigcup_{t \in [0,1]} f'_t(N) \times \{t\} \\&=
\W(f'\colon N \too N'). \qedhere
\end{align*}
\end{proof}

\begin{theorem}\label{th:lpai_setwise}
 Let $M$ be a manifold. Let $A$ be a subset of $M$. Two $A$-fixing motions $\mot{f,f'}{}{N}{N'}$ in $M$ are motion equivalent if, and only if, their worldlines are level preserving ambient isotopic, relative to $(M\times \{0\})\cup (A\times \II)$ pointwise, and $N'\times \{1\}$ setwise.
\end{theorem}
\begin{proof}
Lemma~\ref{le:hatH} gives that motion equivalence implies a level preserving ambient isotopy between the worldlines 
{relative to}
$M\times \{0\}$ pointwise, and $N'\times \{1\}$ setwise. It is straightforward to check that $\hat {H}$, as in the previous proof, is relative to $A\times\II$ pointwise if $H$ fixes $A$.

Lemma~\ref{le:levelpreserving_implies_motionequiv} gives the reverse implication, and again it is straightforward to check that the constructed relative path homotopy is $A$-fixing if the level preserving isotopy $H$ fixes $A\times \II$.
\end{proof}

The previous result can be refined.
\begin{theorem}\label{th:me_bflpai}
Let $M$ be a manifold. Let $A$ be a subset of $M$. Two $A$-fixing motions $\mot{f,f'}{}{N}{N'}$ are motion equivalent if, and only if, their worldlines are level preserving ambient isotopic, relative to $(M\times (\{0,1\}))\cup (A\times \II)$, pointwise.
\end{theorem}
\begin{proof}
Suppose  that the worldlines $\W(f\colon N \too N'\big)$ and $\W(f'\colon N \too N'\big)$ are level preserving ambient isotopic, relative to $(M\times (\{0,1\}))\cup (A\times \II)$, pointwise.
Then, clearly, using the same isotopy, they are level preserving ambient isotopic, relative to $(M\times \{0\})\cup (A\times\II)$, pointwise, and to $N'\times \{1\}$, setwise. Using Theorem~\ref{th:lpai_setwise}, we see that $\mot{f}{}{N}{N'}\simm\mot{f'}{}{N}{N'}$.

Now suppose $\mot{f}{}{N}{N'}\simm \mot{f'}{}{N}{N'}$. By Lemma~\ref{le:grpd_motstar}, there is a path homotopy $F\colon \II\times \II \to \TOPO^h(M,M)$ from $f$ to $\Id_M*f$. 
Since $F$ is a path homotopy, it has fixed endpoints and so, for all $s\in\II$, we have $F(0,s)=\id_M$ and $F(1,s)=f_1$. Therefore, $\hat{F}$,
constructed as in Lemma~\ref{le:hatH}:
\begin{align*}
 \hat{F}\colon M\times \II\times\II &\to M\times \II \\
(m,t,s)& \mapsto\big((F(t,s)\circ f_t^{-1})(m),t\big),
\end{align*}
is a level preserving ambient isotopy from $\W(\mot{f}{}{N}{N'})$ to $\W(\mot{\Id_M*f}{}{N}{N'})$, which 
{is relative to}
$(M\times (\{0,1\}))\cup (A\times \II)$ pointwise.
Similarly there is a path homotopy $F'$ from $\Id_M*f' $ to $f'$. The same construction gives a level preserving
ambient isotopy $\hat{F'}$ from $\W(\mot{\Id_M*f'}{}{N}{N'})$ to $\W(\mot{f'}{}{N}{N'})$ which 
{is relative to}
$(M\times (\{0,1\}))\cup (A\times \II)$ pointwise.

It remains to construct a level preserving
ambient isotopy from $\W(\mot{\Id_M*f}{}{N}{N'})$ to $\W(\mot{\Id_M*f'}{}{N}{N'})$, which 
{is relative to}
$(M\times (\{0,1\}))\cup (A\times \II)$ pointwise.
Since $\mot{f}{}{N}{N'}\simm \mot{f'}{}{N}{N'}$, there exists a relative path homotopy $H\colon \II\times \II \to \TOPO^h(M,M)$ from $f$ to $f'$.
Thus there exists a level preserving isotopy $\hat{H}\colon M\times \II\times \II \to M\times \II$ from $\W(\mot{f}{}{N}{N'})$ to $\W(\mot{f'}{}{N}{N'})$, 
{relative to}
$(M\times \{0\})\cup(A\times \II)$ pointwise and $N'\times \{1\}$ setwise, constructed as in Lemma~\ref{le:hatH}.

We now consider the function $\hat{J}\colon M\times \II\times \II\to M\times \II$, defined as:
\[
\hat{J}(m,t,s)=
\begin{cases}
\big(p_0( \hat{H}(m,2t,s)),t\big), & 0\leq t \leq \frac{1}{2} \\
\bigg(p_0\Big( \hat{H}\big(m,1,s(1-2(t-\frac{1}{2}))\big)\Big),t\bigg), &\frac{1}{2}\leq t \leq 1.
\end{cases}
\]
Let us see that $\hat{J}$ is a level preserving isotopy of $M\times \II$, 
{relative to}
$M\times \{0,1\}$, pointwise.
By construction, it follows that, for each $s\in\II$, we have a level preserving homeomorphism  $M\times \II \to M\times \II$, sending $(m,t)$ to $\hat{J}(m,t,s)$, and also that $J(m,t,0)=(m,t).$
That $\hat{J}$ 
{is relative to}
$M\times \{0\}$ pointwise, follows directly from the fact that $\hat{H}$ is so. Moreover
for all $m\in M$, $s\in\II$, $\hat{J}(m,1,s)=(m,1)$ since 
$
\hat{H}(m,1,0)=(m,1)$.
That $\hat{J}$ 
{is relative to}
$A\times \II$ pointwise follows directly from the fact that $H$ fixes $A$.

We now show that $\hat{J}$ is a level preserving ambient isotopy from $\W(\mot{\Id_M*f}{}{N}{N'})$ to $\W(\mot{\Id_M*f'}{}{N}{N'})$.
First notice that for all $t\in [1/2,1]$ and $s\in \II$, $\hat{J}(N'\times\{t\}\times \{s\})=N'\times \{t\}$, since $\hat{H}(N'\times \{1\}\times \{s\})=N'\times \{1\}$.
We have \phantom{\qedhere}
\begin{align*}
   \hat{J}\big(\W (\mot{\Id_M*f}{}{N}{N'})\times \{1\}\big)
   &=\hat{J}\bigg(\Big( \Big (\bigcup_{t \in [0,1/2]} f_{2t}(N) \times \{t\}\Big) \cup \Big(\bigcup_{t \in [1/2,1]} N' \times \{t\}\Big)\Big)\times \{1\}\bigg) \\
   &=\hat{H}\bigg(\Big(\bigcup_{t \in [0,1/2]} f_{2t}(N) \times \{t\}\Big)\times \{1\}\bigg)\cup \hat{J}\bigg(\Big(\bigcup_{t \in [1/2,1]} N' \times \{t\}\Big)\times \{1\}\bigg) \\
   &=\Big(\bigcup_{t \in [0,1/2]} f'_{2t}(N) \times \{t\}\Big)\cup \Big(\bigcup_{t \in [1/2,1]} N' \times \{t\}\Big)\\
&=\W\big(\Id_M*f'\colon N \too N'\big). &&  \hspace{-1em}{\qedsymbol}
\end{align*}
\end{proof}
\begin{remark} Suppose that $f,g\colon N \too N'$ are motion equivalent, fixing $A$. In the proof of Theorem~\ref{th:me_bflpai}, we constructed a level preserving isotopy from $\W(g\colon N \too N'\big)$ to $\W(f\colon N \too N'\big)$, relative to $M\times \{0,1\} \cup A\times \II$, pointwise, by concatenating three level preserving isotopies. 

There are other ways to construct such a  level preserving isotopy from $\W(g\colon N \too N'\big)$ to $\W(f\colon N \too N'\big)$. Let us  show another construction.

Since $f,g\colon N \too N'$ are motion equivalent, fixing $A$, then $f^{-1}\cdot g\colon N \too N$ is path-homotopic, in $\TOPO_A^h(M,M)$, see \S \ref{ss:Afixmot}, to an $N$-stationary motion $\eta\colon N \too N$, fixing $A$; we are using the notation in Lemma~ \ref{le:dot_premot_comp}. (Recall also that $(f^{-1})_t=f_t^{-1}.$) Applying Lemma~\ref{le:grpd_motstar}, it follows that there exists a path-homotopy from $\Id\colon N \too N$ to $f\cdot \eta \cdot g^{-1}\colon N \to N$, fixing $A$. Call it $H\colon \II\times \II \to \TOPO^h_A(M,M).$ Then the function below is a level preserving isotopy from $\W(g\colon N \too N'\big)$ to $\W(f\colon N \too N'\big) $, relative to $M\times \{0,1\} \cup A\times \II$, pointwise:
\[
J(m,t,s)=\big(H(t,s)(m),t).
\]
\end{remark}

 The previous theorem allows us to give a proof of the (well-known) result that the automorphism group of a point in the motion groupoid $\R^n$ is trivial.

\begin{example}
Consider $M=V$, where $V$ is a finite dimensional space, with the unique Hausdorff topology that makes it a topological vector space. (In the case when $V=\R^n$, then this is just $\R^n$,  with the usual topology.) Note that $V$ is a manifold.
Let $v \in V$. Let us determine $\Mot[V](\{v\},\{v\})$. 

Let $f\colon \{v\} \too \{v\}$ be any motion. 
Let us prove that $f$ is motion equivalent to the identity motion $\Id_V\colon \{v\} \too \{v\}.$
For each $t \in  [0,1]$, let $v_t=v-f_t(v)$. Hence $v_0=v_1=\vec{0}$, the zero vector in $V$. Also the mapping $t \in \II  \mapsto v_t \in V$ is continuous.
Consider the mapping 
\begin{align*}
H\colon V \times \II \times \II &\longrightarrow V\times \II\\
(w,t,s) &\longmapsto (w+sv_t,t).
\end{align*}
Then clearly $H$ is a level preserving isotopy of $V\times I$, which 
{is relative to}
$V\times \{0,1\}$, pointwise. Also clearly:
\[H\big (\W(f\colon \{v\} \too \{v\}) \times \{1\}\big)=\W(\Id\colon \{v\} \too \{v\}).\]
So, by Theorem~\ref{th:me_bflpai}, $f\colon \{v\} \too \{v\}$  is motion equivalent to the identity motion $\Id\colon \{v\} \too \{v\}.$

The same argument, with the obvious modifications, proves that that the full subcategory of $\Mot[V]$, with objects the singletons has a unique arrow between each ordered pair of objects.
\end{example}

\begin{defin}\label{de:wequivalence}
Let $M$ be a manifold, $A$ a subset of $M$, and $N,N'\subseteq M$ be subsets.
Define a relation on 
$\Mtcmag(N,N')$
by $\mot{f}{}{N}{N'}\sim \mot{f'}{}{N}{N'}$
if $\W(\mot{f}{}{N}{N'})=\W(\mot{f'}{}{N}{N'})$, or if 
$f\simp f'$ (here $f\simp f'$ means that $f$ and $f'$ are path-homotopic, as maps $\II\to \TOPO^h_A(M,M)$).
Let $\simww$ be the equivalence (i.e. transitive, symmetric and reflexive) closure
of this relation. Denote the class of a motion $\mot{f}{}{N}{N'}$ under this relation by $\classww{\mot{f}{}{N}{N'}}$.
\end{defin}

\begin{theorem}\label{th:motequiv_worldlineequiv}
    Let $M$ be a manifold and $N,N'\subseteq M$ be subsets, and $\mot{f}{}{N}{N'}$ be an $A$-fixing motion for some $A\subset M$. Then 
    \[
    \classm{\mot{f}{}{N}{N'}}=\classww{\mot{f}{}{N}{N'}}.
    \]
\end{theorem}

\begin{proof}
By Proposition~\ref{prop:dahm}, motions which have the same worldline are motion equivalent, and path equivalent motions are motion equivalent (Proposition~\ref{pr:pe_implies_me}). Each of these results works in the same way in the $A$-fixing case.
Hence $\simww$ implies $\simm$.

Suppose two motions $\mot{f,f'}{}{N}{N'}$ are motion equivalent, fixing $A$. 
Then, by Theorem~\ref{th:me_bflpai}, there exists a level preserving ambient isotopy $H$ 
relative to
$(M\times (\{0,1\}))\cup(A\times \II)$ from $\W(\mot{f}{}{N}{N'})$ to $\W(\mot{f'}{}{N}{N'})$.
From such an $H$, the map $\hat{H}\colon \II\times \II \to \TOPO^h(M,M)$ constructed as in the proof of Lemma~\ref{le:levelpreserving_implies_motionequiv} is a path homotopy from the motion $f$ to a motion $g$, with $\W(\mot{g}{}{N}{N'})=\W(\mot{f'}{}{N}{N'})$.
Thus $f\simww g\simww f'$.
\end{proof}
\begin{remark} Note that the main idea of the second step of the proof of the previous theorem is that if $f,f'\colon N \too N'$ are motion equivalent, fixing $A$, then $f$ is path-homotopic to $g$, where $g$ has the same 
worldline as $f'$. Again, there are multiple ways to construct such 
a path-homotopy. 
Alternatively to what we did in the proof, we can for instance use the following construction. We know that $f^{-1} \cdot f'\colon N \too N$ is path-homotopic to $\eta\colon N \too N$, $N$-stationary. So we have a path-homotopy $H\colon \II\times\II \to \TOPO^h_A(M,M)$ from $f\colon N \too N'$ to $f'\cdot \eta^{-1} \colon N \too N'$. Now note that the later motion has the same worldline as $f'$.
\end{remark}

\section{The groupoid \texorpdfstring{$\FMot{M}^A$}{}: Connections to Artin braids}
\label{sec:braids}

Recall from \S\ref{sec:Intro} that we define an engine as a construction that takes a manifold (possibly with a subset) as input and produces a power-set magmoid,
and ultimately a groupoid, as output.
More specifically we have in mind that the group obtained 
when the engine is
restricted all the way down to fixed finite subsets of $\R^2$, 
is some version of the braid group.
In this section we begin by giving 
an engine that, 
when restricted to fixed finite subsets of $\R^2$,
coincides with the definition of Artin braids \cite{artin}.
Our definition also generalises that of \cite[Sec.II]{dahm}, who defined braids similarly to Artin,  albeit in general topological manifolds.

We call our morphisms here, prior to applying any equivalence, 
`\fakemotion{}s'.
The quotient 
of the magmoid of \fakemotion{}s
by {\it strong isotopy} is then a groupoid, $\FMot{M}^A$ (where $A\subseteq M$). There is a `forgetful' functor $\mathbf{T}\colon \Mot^A \to \FMot{M}^A$. 

In Section~\ref{sec:artinbraids}  we will restrict to the full subgroupoids of  $\Mot^A$  and $\FMot{M}^A$  with objects the set of finite subsets of a manifold. (Essentially treated by Dahm in the group case.) As we will see, the functor $\T$ will then restrict to an isomorphism between the full subgroupoids. 
Since for a finite subset $K\subset \R^2$, $\FMot{\R^2}(K,K)$ 
is isomorphic to
the braid group on $|K|$ strands, $\Braid_{|K|}$,
 we can then use Artin's presentation of braid groups \cite{artin} to prove that certain subgroupoids of our construction of motion groupoids have finite presentation.

In addition to relating them to classical constructions, in \S\ref{sec:magmoidFMt} we can also consider another overarching question:
how should we think of these algebraic structures that we construct --- how useful are they?
This raises the interesting question of how to characterise magmoids generally --- a hard problem because they are often very wild, and their representation theory is presently in its infancy. One of the simplest non-trivial questions we can ask, is how to parameterise the connected components of a given magmoid.

\subsection{The magmoid \texorpdfstring{$\FMotmag{M}$}{FMotmag} of fake motions and its connected components}\label{sec:magmoidFMt}

\medskip
Consider a manifold $M$ and a  
subset $N\subseteq M$ and define 
$\Topo_{\Braid{}}(N\times\II,M)$ to be the subset of elements 
$f\in \Topo(N\times\II,M)$ such that: 
\\ (I) for each $t\in\II$, the map from $N$ to $M$ given by $n \mapsto f(n,t)$ is an embedding,
and
\\ 
 (II) for all $n \in N$,
$f(n,0)=n$.

We note that, when $N$ is a submanifold of $M$, a common name for an element of $\Topo_{\Braid{}}(N\times\II,M)$ is an {\em isotopy of $N$ inside $M$}.

\begin{defin}\label{de:fakemotion}
Let $N,N'\subseteq M$. We say that  $f \in\Topo_{\Braid{}}(N\times\II,M)$ is a {\em \fakemotion}, from $N$ to $N'$, and we write $f\colon N \not \too N'$, if $f(N,1)=N'$.
Denote by $ \dahm{N}{M}{N'}\subseteq \Topo_{\Braid{}}(N\times\II,M)$, the subset of \fakemotion{}s from $N$ to $N'$.
(Note here that $N,N'$ are necessarily homeomorphic via $f(-,1)$.)
\end{defin}

The reason for the nomenclature `\fakemotion{}' of $N$ in $M$ is that it looks like $f$ describes a reversible evolution of $N$ over time, but there is no intrinsic guarantee that there is a deformation of $M$ that will induce $f$.
We return to this point shortly.

\newcommand{\ox}{\smallsquare} 

\begin{lemma}\label{le:mag_fakemotion}
{Let $M$ be a manifold.}
There is a partial composition 
of \fakemotion{}s 
$\ox\colon \dahm{N}{M}{N'}\times \dahm{N'}{M}{N''} \to \dahm{N}{M}{N''}$ 
that 
maps $(f,g) $ to 
$g\ox f$, formally given by
\begin{align}\label{eq:composition_fakemotion}
g\smallsquare f(n,t):= 
\begin{cases} f(n,2t), & 0\leq t \leq 1/2 \\
g(f(n,1),2t-1), &  1/2\leq t\leq 1\, .
\end{cases}
\end{align}
Hence 
$(\Power M,\dahm{N}{M}{N'},\smallsquare)$
is a magmoid of \fakemotion{}s. 
We denote it $\FMotmag{M}$.
\end{lemma}
\begin{proof}
{Observe first that the second row of \eqref{eq:composition_fakemotion} gives a well-defined map in $ \Set(N\times[1/2,1],M)$ by the matching middle $N'$ of the domain of $\ox$. 
Furthermore, since
$g(f(n,1),0)=f(n,1)$,
the segments agree at $t=1/2$, hence $g\ox f \in \Set(N\times\II,M)$ is well-defined. Since  $g\ox f$ is continuous when restricted to $N\times  [0,1/2]$ and $N\times [1/2,1]$, closed subsets of $N\times \II$, it follows that $g\ox f\colon N \times \II \to M$ is continuous. }

It is {immediate} that the map {$N \to (g \ox f)(N,t)$, such that $n \mapsto (g \ox f)(n,t)$ is an embedding, whenever $t \in [0,1/2]$.}
 {If $t \in [1/2,1]$, the map $n \mapsto (g \ox f)(n,t)$ is an embedding as it is the composition of the homeomorphisms $N \ni n \mapsto f(n,1) \in N'$ with the embedding $N'\ni n' \mapsto g(n',2t-1) \in M$.} 
Finally,
$g(f(N,1),1)=g(N',1)=N''$.
\end{proof}

We return now to the question of which objects are connected in $\FMotmag{M}$.
The following examples will be sufficient to intrigue the reader.

\exa{\label{ex:fakeflow_pttobdy}Let $x\in (0,1)$ be any point. There exists $f\in \Topo_{\Braid{}}(\{x\}\times \II,[0,1])$, given by $f(x,t)=x(1-t)$. 
So we have a \fakemotion\ $f\colon \{x\} \not\too \{0\}$, sending an interior point in $\II$ to a boundary point.
}

{The following example  gives a concrete way to write down the process, see e.g. \cite[page 2]{Burde}, of deforming any knot, here a \textit{string knot},
continuously into an unknotted string knot, in a way that cannot in general be induced by a movement of space.}

\exa{\label{ex:fakeflow_knot}Let $M= [-1,1]\times \R^2$, with the topology induced by $\R^3$. We let $g\colon [-1,1] \to M$ be an embedding (so equivalently, since $[-1,1]$ is compact, $g$ is continuous and injective), such that $g(-1)=(-1,0,0)$ and $g(1)=(1,0,0)$. 
Also suppose that $g\big( (-1,1)\big) \subset \mathrm{int}(M)$. Let $K=g([-1,1])\subset M$. Observe  that $K$ is
a (possibly wild and also non-trivial) string knot.
And then consider $H\colon [-1,1] \times \II \to M$, defined as:
\[H(t,s):=\begin{cases}
(t,0,0), \textrm{ if } \big ( t \in [-1,-1+s] 
\textrm{ or } t \in [1-s,1] 
\big) 
\textrm{ and } s<1, \\
(1-s){g}\left(\frac{t}{1-s}\right), \textrm{if } t \in [-1+s,1-s],  \textrm{ and } s<1,\\
(t,0,0), \textrm { if } s=1.
\end{cases}
 \]
Define $f\colon K \times \II \to M$ as $f(k,t)=H(g^{-1}(k),t)$.
We then have a \fakemotion\ from $K$ to the trivial string knot $[-1,1] \times \{(0,0)\}$.
 }

\exa{\label{ex:fakeflow_circletoint}
Let $M=\R^2$,
$S^1=\{x\in \R^2 \,|\, |x|=1\}$
and $y\in S^1$ be any point.
Then it is straightforward to construct a \fakemotion{}  $\fakemot{f}{S^1\setminus\{y\}}{(0,1)\times \{0\}}$. {Similarly, if $H'$ is the Hopf link, with one point taken out of one of the components, then there exists a fake motion $H'\not \too L$, where $L'$ is the unlink, where we have taken a point out of one of the components.}
}

{This leads us naturally back to the question of connection to the other structures introduced in this paper, and indeed to classical constructions. To address both of these, we proceed as follows.}

\medskip

Recall from Lemma~\ref{le:movie_mot} that there is a bijection $\Phi\colon \Mtcmag(N,N')\to \Mtcmagmov(N,N')\subset \Topo(M\times \II,M)$ such that $\Phi(f)(m,t)=f_t(m)$.
Since motions are made from  homeomorphisms of $M$,
the restriction of $\Phi(f)$ to $N\times \II$ lies in $\Topo_{\Braid{}}(N\times\II,M)$.
Thus there is a set map 
\[
T' : \Motc{M}{N}{N} \rightarrow 
 \dahm{N}{M}{N'}
\]
given by $T'(f)(n,t) = f_t (n)$.

 An interesting question is to ask, for some manifold $M$ and subsets $N$ and $N'$, which elements of $\dahm{N}{M}{N'}$ have a preimage under $T'$. Notice that no such preimage exists in Examples~\ref{ex:fakeflow_pttobdy} and \ref{ex:fakeflow_circletoint}. Also, the fake motion $f$ in Example \ref{ex:fakeflow_knot} does not have a preimage under $T'$  whenever $K$ is a non-trivial string knot.

For fake motions of submanifolds, establishing a preimage amounts to establishing an isotopy extension property (to an ambient isotopy) in the context of topological manifolds.  The Edwards-Kirby isotopy extension theorem  for topological manifolds \cite{Edwards_Kirby} is a powerful general result in this direction. A recent treatment of this result  can be found in chapter 19, due to Arunima Ray, of \cite{TopManifolds}.

\medskip

In the following result we can see that if a preimage of a fake motion exists then it will be unique under motion equivalence.
\begin{lemma} \label{le:T'equal>motionequivalent}
Let $\mot{f,g}{}{N}{N'}$ be motions in $M$.
If $T'(\mot{f}{}{N}{N'})=T'(\mot{g}{}{N}{N'})$, then $f\simm g$. 
\end{lemma}
\begin{proof}
It follows directly from the definition of the worldline that $T'(\mot{f}{}{N}{N'})=T'(\mot{g}{}{N}{N'})$
implies  $\W(\mot{f}{}{N}{N'})=\W(\mot{g}{}{N}{N'})$. By Proposition~\ref{prop:dahm} motions with the same worldline are  motion equivalent.
\end{proof}

\subsection{The groupoid $\FMot{M}$ and the T-functor}

\begin{defin} \label{de:strongiso}
For manifold $M$ and subsets $N,N'\subseteq M$, two \fakemotion{}s $f,f'\colon N \times \II \to M$ in $\dahm{N}{M}{N'}$ are said to be {\em strongly isotopic}\footnote{This nomenclature is borrowed from \cite{artin}, where strong isotopies between braids are introduced.} if there exists a homotopy $H\colon (N \times \II) \times \II \to M$ from $f$ to $f'$ (i.e.  
$H(n,t,0)=f(n,t)$ and $H(n,t,1)=f'(n,t)$)
such that, for each $s \in \II$, the map 
from $N\times \II $ to $M$ given by $(n,t) \mapsto H(n,t,s)$ is in $\dahm{N}{M}{N'}$.\\
Notation: We write $f\simfk f'$ to mean the \fakemotion{}s $f$ and $f'$ are strongly isotopic.
\end{defin}

\rem{Let $M$ be a manifold, if $N\subseteq M$ is locally compact, then, by the product-hom adjunction (Lemma \ref{th:tensorhom}), there is a one-to-one correspondence between \fakemotion{}s of $N$ in $M$ and maps $f\colon \II \to \TOPO(N,M)$ satisfying  that $f_0$ is the inclusion of $N$ in $M$, and  $f_t\colon N \to M$ is always an embedding. 
However, note that, in general, a subset $N\subseteq M$ need not be locally compact, so we do not know if this correspondence holds in general. In particular, we do not know whether there is a way to write {the} strong isotopy relation mirroring relative path-homotopy, as a map from $\II^2$, when $N$ is not locally compact.
}

\begin{lemma}\label{lem:Fake-braids}
Fix a manifold $M$.
(I) The relations $\simfk$ on each $\dahm{N}{M}{N'}$
is a congruence on $\FMotmag{M}$.
(II) The quotient is a groupoid.  We denote it \[\FMot{M}=(\Power M,\dahm{N}{M}{N'}/\simfk, \;\smallsquare \;, \;\classfk{f}\mapsto\classfk{\bar{f}} \;  ,\; \classfk{\id_{\Braid,N}}).\]
\end{lemma}
\noindent We have used the notation $\classfk{f}$ or $\classfk{f\colon N \not \too N'}$ to denote the strong isotopy equivalence class of a fake motion $f\in \dahm{N}{M}{N'}$. Note that in contrast with the motion case, the source and target of a fake motion are uniquely specified by $f\colon N\times \II\to M$.
\begin{proof}
(I). Suitable homotopies to prove that indeed isotopy between fake motions defines an equivalence relation follow similar constructions to those in the proofs of Proposition~\ref{pr:pe}.
In order to prove that the composition  $\smallsquare$ of fake motions descends to the quotient, let $f,f'\in \dahm{N}{M}{N'}$, $g,g'\in \dahm{N'}{M}{N''}$ with strong isotopies $H_f$ from $f$ to $f'$ and $H_g$ from $g$ to $g'$, then 
\[H(n,t,s)= \begin{cases} 
H_f(n,2t,s), & 0\leq t \leq 1/2, \\
H_g\big(H_f(n,1,s), 2t-1,s\big), & 1/2\leq t\leq 1
\end{cases} \] 
is a strong isotopy from $g\smallsquare f$ to $g'\smallsquare f'$. (Note that it may be that $N$ is not closed in $M$. Since $N\times [0,1/2]$ and $N\times [1/2,1]$ are nevertheless closed in $N\times [0,1]$, we can {still} apply the glue lemma to see that $H$ is continuous.)

 (II) For $N\subseteq M$, the identity on $N$ is given by the strong  isotopy class of $\id_{\Braid,N}\colon N\times \II\to M$, $(n,t)\mapsto n$ (the \emph{identity fake motion}).
For a \fakemotion\ $f\in \dahm{N}{M}{N'}$, and fixed $t\in [0,1]$, consider the map $\sigma_{f}\colon N \to N'$ given by $n\mapsto f(n,1)$. Note that $\sigma_{f}$ is a homeomorphism. Define $\bar{f}\colon N' \not \too N$ as  $(n',t)\mapsto 
f({\sigma_{f}^{-1}}(n'),1-t)$.
Explicit homotopies proving the identity, associativity and inverse axioms follow a similar construction to those in the proof of Proposition~\ref{pr:fundamentalgroupoid}. To illustrate the idea we give the homotopies which prove the inverse axiom. 
Thus the following function is a strong isotopy $\bar{f}\ox f$ to $\id_{\Braid,{N}}$:
	\[
	H_{in}(n,t,s)=\begin{cases}
		f(n,2t), & 0 \leq t\leq \frac{1}{2}-\frac{s}{2}, \\
		f(n,1-s),& \frac{1}{2}-\frac{s}{2}\leq t \leq \frac{1}{2}+\frac{s}{2}, \\
		f(n , 2-2t),& \frac{1}{2}+\frac{s}{2}\leq t \leq 1.
	\end{cases}
	\]
	Similarly 
		\[
	H'_{in}(n,t,s)=\begin{cases}
		f(\sigma_f^{-1}(n),1-2t), & 0 \leq t\leq \frac{1}{2}-\frac{s}{2}, \\
		f(\sigma_f^{-1}(n),s), & \frac{1}{2}-\frac{s}{2}\leq t \leq \frac{1}{2}+\frac{s}{2} ,\\
		f(\sigma_f^{-1}(n) , 2t-1),& \frac{1}{2}+\frac{s}{2}\leq t \leq 1.
	\end{cases}
	\]
	is a suitable strong isotopy from $f\ox \bar{f}$ to $\id_{\Braid,{N'}}$.
\end{proof}

As with the motion groupoid, we can construct a version of $\FMot{M}$ which depends on a fixed choice of subset $A\subset M$.
Let $M$ be a manifold $M$, and $A\subset M$ a subset. Given $N,N'\subseteq M$, a fake motion $f$, from $N$ to $N'$, relative to $A$, is by defined to be a fake motion $f\in \dahm{N}{M}{N'}$, such that furthermore:
\begin{enumerate}
    \item[(III)] If $a \in N\cap A$, then $f(a,t)=a$, for all $t$, and
    \item[(IV)] If $a \in N \setminus A$, then $f(t,a)\notin A$ for all $t$.
\end{enumerate}
We let $\dahmA{N}{M}{N'}\subset \dahm{N}{M}{N'}$ be the set of fake motions from $N$ to $N'$, relative to $A$. Note that the composition $\ox$ restricts to fake-motions relative to $A$, thus there is a magmoid $\FMotmag{M}^A$. If we restrict to strong isotopies $H$ connecting $f,f'\in \dahmA{N}{M}{N'} $, such that $H( -,-,s)$ is always in $\dahmA{N}{M}{N'}$ this a magmoid congruence on $\FMotmag{M}^A$, called strong isotopy relative to $A$ (also denoted $\simfk$). We hence have a quotient groupoid, denoted $\FMot{M}^A$.

\begin{example}
{In} Example~\ref{ex:fakeflow_pttobdy}, $f$ is not a fake motion relative to $\{0,1\}$. On the other hand, looking at Example~\ref{ex:fakeflow_knot}, $f$ is a fake motion relative to $ \{-1,1\} \times \R^2$.
\end{example}
\medskip

Note that $T'$ restricts to a set map
$\MotAc{M}{N}{N'} \rightarrow 
 \dahmA{N}{M}{N'}$, which we denote also by $T'$.
For a manifold $M$, and subset $A\subset M$, the map $T'$ extends to a functor from the motion groupoid to the groupoid of fake motions.
\begin{theorem}\label{thm:MottoFmot}
Let $M$ be a manifold, and $A\subset M$ a subset. Given $N,N'\subseteq M$,
there is a well defined functor of groupoids, \[\mathbf{T}\colon \Mot^A \to \FMot{M}^A,\]
which is the identity on objects, and  sends $\classm{\mot{f}{}{N}{N'}}$ to
$\classfk{\fakemot{T'(f)}{N}{N'}}$ where $T'(f)(n,t)= f_t(n)$.
 (Here $\simm$ denotes motion equivalence, fixing $A$ and  $\simfk$ denotes strong isotopy of fake motions, relative to $A$.)
\end{theorem}
\begin{proof}
Suppose motions $\mot{f,f'}{}{N}{N'},$ fixing $A$, are motion equivalent, fixing $A$. Then, by Theorem~\ref{th:mg2}, there exists a 
relative path-homotopy
$H\colon \II \times \II \to 	\TOPO^h_A(M,M) $, from $f$ to $f'$. And then, by construction,  the map $H'\colon ( N\times \II)\times \II \to M$ such that 
$H'\big((n,t),s\big)= 
H(t,s)(n) 
$
defines a strong isotopy, from $T'(f)$ to $T'(f')$, relative to $A$. 
Moreover, given motions $f,f'\colon K \too K$, fixing $A$, we have that $T'(f)\ox T'(f')=T'{(f*f')}$. The latter follows immediately from the conventions for the composition $\ox$ and $*$ given in \eqref{eq:composition_fakemotion} and \eqref{def:comp1}. From this it follows that $T_A$ indeed preserves the compositions in $\Mot^A$ and  $\FMot{M}^A$.
\end{proof}

\subsection{Motion groupoids and Artin braid group(oid)s}\label{sec:artinbraids}

Armed with Theorem~\ref{th:mg2}, 
we can now unpack our 
observation from Section~\ref{ss:stat_mot} that at least certain
non-trivial motion groups are finitely presented.
(We will see in Proposition~\ref{rem:mcgI}
that some are not.) 

Our strategy here is simply to make contact with Artin's proof that Artin's braid group 
$\braid$ 
(defined in (\ref{def:Bn}) below)
is presented by Artin's presentation
\cite{artin}.
(The strategy is simple; the contact is not.)

Let $M$ be a manifold and $K$ be a finite subset of $M\setminus \partial M$. We will define $\mathbf{B}(M,K)$ as the group of morphisms from $K$ to $K$ in $\FMot{M}^{\partial M}$. Then we prove that $\T$ restricts to a group isomorphism from $\Mot^{\partial M}(K,K)$ to $\mathbf{B}(M,K)$ (Theorem~\ref{th:geometric_braids}).
In particular we have an isomorphism when we further specialise to $M=\R^2$.

Of course a major point of our construction of motion groupoids, and fake motion groupoids, is that
it coincides with Dahm's construction when restricted to the
group case, so in this final step 
we will be able 
essentially to follow Dahm's version.

\begin{defin}Suppose that $M$ is a manifold and $K\subset M\setminus \partial M$ a finite subset. An {\em \Artinb} 
of $K$ in $M$ is an element 
$f\in \dahmfix{K}{M}{K}{\partial M}$.
\end{defin}

 It follows from Lemma~\ref{lem:Fake-braids} that \Artinb{}s
 compose under $\ox$ (defined in \eqref{eq:composition_fakemotion}), the composition descends to the quotient under  
strong isotopy, and moreover defines a group structure on the set of equivalence classes of Artin braids. 

\begin{defin} \label{de:braidgroup}
Let $M$ be a manifold and $K\subset M\setminus \partial M$ a finite subset. Denote by $\mathbf{B}(M,K)$ the group $\FMot{M}^{\partial M}(K,K)$ ($\FMot{M}^{A}$ is defined in Lemma~\ref{lem:Fake-braids}). 
\end{defin}

\rem{\label{rem:braidD^2partialD^2} It follows directly from the definition that $\FMot{M\setminus \partial M}(K,K)=\FMot{M}^{\partial M}(K,K)$. 
}

Note that, by a similar argument as used in Example~\ref{ex:finiteset}, if $M$ is connected then the isomorphism type of  $\mathbf{B}(M,K)$ depends only on the cardinality of $K$. (Albeit up to non-canonical isomorphism.)

In case when $M=\R^2$, and 
$K=K_n = \{ (0,1),(0,2),...,(0,n) \}$, 
we define 
\begin{equation} \label{def:Bn}
\braid=\mathbf{B}(\R^2,K_n)
=\FMot{\R^2}(K_n , K_n)
.
\end{equation}
Then $\braid$ coincides with the 
definition of the braid group on $n$-strands, arising from Artin's construction in \cite[Thm.5]{artin} (and in \cite{artin25}).

\begin{lemma}
\label{le:geobraid_grouphom} 
Let $M$ be a manifold. 
Let
$n$ be a positive integer and $K\subset M{\setminus \partial M}$ a finite subset with $n$ elements.
The formal map
\[
T\colon \Mot[M]^{\partial M}(K,K) \to \mathbf{B}(M,K)
\]
that
sends the motion equivalence class of a motion $f\colon K \too K$
 to the 
strong isotopy
class of the \Artinb{}  
$X_f\colon K\times \II\to M$ of $K$ in $M$, such that
$(k,t)  \mapsto f_t(k) $, is well-defined, and a group homomorphism.

\end{lemma}
\begin{proof}
This is a restriction of the functor given in  Theorem~\ref{thm:MottoFmot}.
\end{proof}

\begin{remark}\label{rem:partial_moves}
Note that the previous result holds when the boundary of the manifold $M$ is not necessarily fixed, meaning that we have a homomorphism $\T\colon \Mot[M](K,K)\to \mathbf{B}(M,K),$ defined by using the same formulae as in the boundary-fixing setting. This follows from the fact that, given $x \in \mathrm{int}(M)$, then if $f\colon M \to M $ is any homeomorphism we have that $f(x) \in  \mathrm{int}(M)$. 
\end{remark}

The proof of the following theorem closely follows 
ideas present in \cite[Thm.II.1.2]{dahm}.
We note that our setup is not precisely the same as Dahm's, who restricts to homeomorphisms with compact support and does not treat manifolds with boundary. 

\begin{theorem}
\label{th:geometric_braids} 
Let $M$ be a manifold. 
Let
$n$ be a positive integer and $K\subset M{\setminus \partial M}$ a finite subset with $n$ elements.
The group homomorphism
\[
T\colon \Mot[M]^{\partial M}(K,K) \to \mathbf{B}(M,K)
\]
given in Lemma~\ref{le:geobraid_grouphom} is a group isomorphism.

In particular we have $\Mot[\R^2](K,K)\cong \braid.$ 
\end{theorem}

To prove 
Theorem~\ref{th:geometric_braids}
we will require some preliminary lemmas.

\smallskip

\newcommand{\OB}{OB} 
\newcommand{\B}{B} 

 Recall that 
for $M$ a topological manifold, possibly with non-empty boundary, the group of homeomorphisms $M \to M$ fixing the boundary, with the compact-open topology, is denoted $\TOPO^h_{\partial M}(M,M)$.
 For the remainder of this section, we also put 
 $B^n=\{x \in \R^n\mid |x| \leq 1\}$,  
 $\OB^n=\{x \in \R^n\mid |x| < 1\}=B^n\setminus \partial B^n$ 
 and  $S^n=\partial B^n$.

\begin{lemma}Let $n$ be a non-negative integer. 
There exists a continuous function:
\[F^n\colon OB^n \times OB^n \to \TOPO^h_{\partial B^n}(B^n,B^n),
\]
such that, 
given any $x,y \in OB^n$, we have $F^n(x,y)(x)=y$, and moreover $F^n(x,x)\colon B^n \to B^n$ is the identity function.
\end{lemma}
There are several different ways to construct  $F^n$ with the given properties.  
 We give such a construction in the following proof. See also \cite[\S 4]{Milnor_top} for a construction that can be done in the smooth setting.

\begin{proof}
 Given $x,w \in \R^n$ we put $(x,w]=\{x+t(w-x) \mid t \in (0,1]\}.$ 
Observe that for each $x\in \OB^n$, the union 
$ \{x\} \cup \bigcup_{w\in S^n} (x,w]$
gives a `radial' partition of $\B^n$. 
Thus, for each $(x,y) \in \OB^n \times \OB^n $,
there is an element 
$F^n(x,y)$ 
in $\Set(\B^n, \B^n)$, in fact a bijection,
that sends the segment 
$[x,w]=\{tx +(1-t)w \mid t \in [0,1]\}\subset B^n$, linearly, to the segment $[y,w]$, so that $tx +(1-t)w  \mapsto ty +(1-t)w.$
By construction, $F^n({x,y})\colon B^n \to B^n$ is  the identity when restricted to the boundary of $B^n.$ Also $F^n({x,y})(x)=y$ and $F^n(y,x)$ is the inverse of $F^n(x,y)$. This finishes our construction of $F^n(x,y)$.

Let us give some detail as to why indeed we defined a continuous function $F^n\colon OB^n \times OB^n \to \TOPO^h_{\partial B^n}(B^n,B^n)$. We use Lemma~\ref{th:tensorhom}.

Note that $F^n(0,y)(z)=y+z-y|z|.$ So $F^n(0,y)\colon B^n \to B^n$ is continuous, and since $F^n(0,y)$ is a bijection, and $B^n$ is compact, $F^n(0,y)$ is a homeomorphism. Furthermore, the fact that $(y,z)\mapsto F^n(0,y)(z)$ is continuous, jointly in $y \in OB^n$ and $z \in B^n$, gives that we have a continuous function
$OB^n \ni y \mapsto F^n(0,y) \in \TOPO^h_{\partial B^n}(B^n,B^n)$, if the latter is given the compact-open topology.  (By Lemma \ref{th:tensorhom})

Since  $\TOPO^h(B^n,B^n)$ is a topological group, Theorem \ref{le:top_group}, it follows that the map 
$OB^n \ni x \mapsto F^n(x,0) \in \TOPO^h_{\partial B^n}(B^n,B^n)$ is continuous. (Where we note that $F^n(x,0)$ is the inverse of $F^n(0,x)$.) And then, again using that $\TOPO^h(B^n,B^n)$ is a topological group, it follows that the function
\[OB^n\times OB^n \ni (x,y) \mapsto F^n(0,y) \circ F^n(x,0) \in\TOPO^h_{\partial B^n}(B^n,B^n)\] is continuous.
Now note $F^n(x,y)=F^n(0,y) \circ F^n(x,0)$.
\end{proof}

From here  until the end of the proof of Theorem~\ref{th:geometric_braids}, we fix some choice of $F^n\colon OB^n \times OB^n \to \TOPO^h_{\partial B^n }(B^n,B^n),$ satisfying the conditions of the previous lemma.
  
 \medskip 
  
  The following lemma is a
special case of the isotopy extension theorem for topological manifolds. (As already mentioned, for that theorem see \cite{Edwards_Kirby} and  \cite[Chapter 19]{TopManifolds}.)

 \begin{lemma}\label{lem:xhat}
Let $X\colon K \times \II \to M \setminus \partial M$ be an Artin braid of $K$ in $M$, or more generally any element of $\Topo_{\Braid{}}^{\partial M}(K\times\II,M)$.
 There exists a continuous function $\hat{X} \colon \II \to \TOPO^h_{\partial M}(M,M)$, with $\hat{X}(0)=\id_M$, and such that if $k\in K$ and $t\in \II$ we have $\hat{X}(t)(k)=X(k,t)$.
 \end{lemma}
 Before we start the proof, note that the proof is quite simple in the case when $M=B^n$ and $|K|=1$. In this case, putting $K=\{x_0\}$ and $X(x_0,t)=x_t$, then $\hat{X}(t):=F^n(x_0,x_t)$ works. The proof of the result is strongly rooted in this idea.
  
\begin{proof}
Let $n=\dim(M)$.
For each $t\in \II$ and $k\in K$, we can find a subset $U_k^t\subset M\setminus \partial M$, together with a homeomorphism $f_k^t\colon U_k^t\to  B^n$, such that $X(k,t) \in \mathrm{int}(U_k^t)$, and such that, moreover, for fixed $t$, all of the $U_k^t$, $k\in K$, are pairwise disjoint. The continuity of $X$ implies that there exists $\epsilon_t >0$ such that $X(k,s) \in \mathrm{int}(U_k^t)$, if $|s-t|\leq \epsilon_t$ for all $k \in K$.
(We have used the fact that the manifold is, by assumption, Hausdorff, and also that Artin braids take values only in the interior of $M$, by definition the set of points of $M$ that have a neighbourhood homeomorphic to an open ball.)

The intervals $(t-\epsilon_t, t+\epsilon_t)$ form an open cover of $\II$. We then use the fact that $\II$ is compact, and the Lebesgue number lemma (see e.g. \cite{munkres}), to find a positive integer $N$ and $0=v_0 < v_1 < \dots < v_{N+1}=1,$ such that each interval $[v_i,v_{i+1}]$ is contained in an interval of the form $(t_i-\epsilon_{t_i}, t_i+\epsilon_{t_i})$, for some $t_i\in [0,1]$. 

For each $i \in \{0,\dots,N\}$, define 
$Y_i\colon [v_{i},v_{i+1}] \to  \TOPO^h_{\partial M}(M,M)$ as\footnote{{Here note that, by the invariance of domain theorem, $\mathrm{int}(U_k^t)$ can be seen both as the interior of the set $U_k^t$, in $M$, and as the interior of the manifold $U_k^t$. This can be used to show that the function below is continuous, applying the pasting lemma to the closed subsets $M\setminus \mathrm{int}(U_k^t)$ and $U_k^t$.}}:
\[Y_i(t)(x):=
\begin{cases}
 (f_k^{t_i})^{-1} \Big (F^n\big( \; f_k^{t_i}(X(k,t_i))\; ,\;
 f_k^{t_i}( X(k,t)) \; \big)\Big), & \textrm{ if } x \in U_k^{t_i}, \textrm{ for some } k \in K\\
 x, & \textrm{otherwise}.
\end{cases}
\]
And finally $\hat{X}\colon \II \to \TOPO^h_{\partial M}(M,M)$, defined as:
\begin{equation}\label{eq:xhat}
 \hat{X}(t)=
\begin{cases}
Y_0(t), & \textrm{ if } t \in [v_0,v_1],\\
Y_1(t)\circ Y_0(v_1), & \textrm{ if } t \in [v_1,v_2],\\
Y_2(t)\circ Y_1(v_2) \circ Y_0(v_1), & \textrm{ if } t \in [v_2,v_3],\\
  \dots\\
Y_{N}(t)\circ Y_{N-1}(v_N) \circ \dots \circ Y_2(v_3)\circ Y_1(v_2) \circ Y_0(v_1), &\textrm{ if } t \in [v_{N},v_{N+1}],
\end{cases}
\end{equation}
has the required properties. 
\end{proof}

\begin{proof} (Of Theorem~\ref{th:geometric_braids})
The surjectivity of $ 
T\colon \Mot[M]^{\partial M}(K,K) \to \mathbf{B}(M,K)$
follows directly from the above lemma.

Let us now prove the injectivity of $T$\footnote{A quicker proof of  injectivity, relying on less elementary results, follows from Lees- \cite{Lees}  parametrised   isotopy extension theorem, {in the form} stated in \cite[Theorem 3.9]{Kupers2015}.}. 
For this we follow a similar technique to the proof of Lemma~\ref{lem:xhat}, with an extra dimension.
Consider relative path-equivalent motions $f,f'\colon K \too K$, fixing $\partial M$.  Suppose that the geometric braids $X_f,X_{f'}\colon K \times I \to M\setminus \partial M$ are strongly isotopic. Let $H\colon (K\times I) \times I \to M\setminus \partial M $ be a strong isotopy connecting $X_f$ and $X_{f'}$. For each $(t,s)\in \II \times \II $, there exists pairwise disjoint subsets $U^{(t,s)}_k\subset M \setminus \partial M$, $k \in K$, together with homeomorphisms:
\[
f^{(t,s)}_k\colon U^{(t,s)}_k \to B^n
\]
and such that $H(k,t,s) \in \mathrm{int}(U^{(t,s)}_k).$ And again, using the continuity of $H$, there exists an open rectangle $I_{t,s}\subset \II\times \II$, with $(t,s) \in I_{t,s}$, and such that
 $H(k,t',s') \in \mathrm{int}(U^{(t,s)}_k)$, for all $k \in K$ and $(t',s')\in I_{t,s}$.
 
 The $I_{t,s}$ form an open cover of $\II \times \II$. By using the Lebesgue number theorem, there exists a positive integer $N$ and $0=m_0 < m_1 <\dots< m_{N}< m_{N+1}=1 $, such that  each $[m_i,m_{i+1}]\times [m_j,m_{j+1}] $ is contained in some $I_{t_i,s_j}$.
 
 Given $s \in [0,1]$, let $X_s\colon K\times \II \to M\setminus \partial M$ be defined by $X_s(k,t)=H(k,t,s)$. 
 
For $j\in \{0,1,\dots, N\}$, following the notation at the end of the proof of Lemma \ref{lem:xhat},  if $s \in [m_j,m_{j+1}]$, define $\hat{X}_s^j\colon \II \to  \TOPO^h_{\partial M}(M,M) $ using Equation \eqref{eq:xhat}, with $X_s$ instead of $X$, and using the homeomorphisms below: 
 \begin{equation}\label{list_homeos}
f^{(t_i,s_j)}_k\colon U^{(t_i,s_j)}_k \to B^n
, \textrm{ where } i\in \{0,1,\dots, N\}, \textrm{ and } k \in K.
\end{equation}
Then, by construction, the map $\II \times [m_j,m_{j+1}] \to  \TOPO^h_{\partial M}(M,M)$ given by $(t,s) \mapsto \hat{X}_s^{j}(t)$ defines a relative path-homotopy from $\hat{X}_{m_j}^{j}\colon K \too K$ to $\hat{X}_{m_{j+1}}^{j}\colon K \too K$, fixing $\partial M$. (Albeit using the interval $[m_j,m_{j+1}]$ rather than $[0,1]$.) So $(\hat{X}_{m_j}^j\colon K \too K) \simrp (\hat{X}_{m_{j+1}}^j\colon K \too K)$, for each $j\in {0,1,\dots, N}$. (Here and below, $\simrp$ means relative path-homotopy fixing $\partial M$.)

Note that 
$\hat{X}_{m_{j+1}}^j\colon K \too K$  and $\hat{X}_{m_{j+1}}^{j+1}\colon K \too K$ are not necessarily the same motion: 
the lists of homeomorphisms,
\begin{equation*}
f^{(t_i,s_j)}_k\colon U^{(t_i,s_j)}_k \to B^n
, \textrm{ where } i\in \{0,1,\dots, N\}, \textrm{ and } k \in K,
\end{equation*}
and 
\begin{equation*}
f^{(t_i,s_{j+1})}_k\colon U^{(t_i,s_{j+1})}_k \to B^n
, \textrm{ where } i\in \{0,1,\dots, N\}, \textrm{ and } k \in K,
\end{equation*}
may be different.
However, by construction, the two motions $\hat{X}_{m_{j+1}}^j\colon K \too K$  and $\hat{X}_{m_{j+1}}^{j+1}\colon K \too K$ have the same worldline, which explicitly is
\[\bigcup_{t\in [0,1]} X_{m_{j+1}}(K \times \{t\}) \times \{t\} .\]
And therefore, using Theorem  \ref{th:mg2} combined with Proposition \ref{prop:dahm} (or equivalently Remark \ref{rem:dahm_rel_path}) we have:
$(\hat{X}_{m_{j+1}}^j\colon K \too K) \simrp (\hat{X}_{m_{j+1}}^{j+1}\colon K \too K)$. 

Since the discussion holds for all $j \in \{0,1,\dots, N\}$, by transitivity, we have that
$(\hat{X}^0_{m_0}\colon K \too K)\simrp (\hat{X}^{N+1}_{m_{N+1}}\colon K \too K)$.

To finish proving that $f,f'\colon K \too K$ are relative path homotopic, fixing $\partial M$, we now observe that the worldline of $f\colon K \too K$ is the same as that of  $\hat{X}^0_{m_0}\colon K \too K,$ explicitly 
\[\bigcup_{t\in [0,1]} X(K \times \{t\}) \times \{t\} .\]
Analogously the worldline of $f'\colon K \too K$ is the same as that of  $\hat{X}^{N+1}_{m_{N+1}}\colon K \too K$. So we can apply Remark  \ref{rem:dahm_rel_path} again to prove that  $f,f'\colon K \too K$ are relative path homotopic, fixing $\partial M$.
\end{proof}

\rem{ It follows from Theorem~\ref{th:geometric_braids} together with the finite presentation of $\braid$ in \cite{artin}, that $\Mot[\R^2](K,K)$ is finitely generated.}

Note that exactly the same argument as used in the proof of Theorem~\ref{th:geometric_braids} will give the following result.
\begin{theorem}\label{th:groupoid_artinbraid}
Let $M$ be a manifold. The groupoid map $\mathbf{T} \colon \Mot^{\partial M} \to \FMot{M}^{\partial M}$ restricts to an isomorphism between the  respective full subgroupoids with objects the finite subsets of $M\setminus \partial M$. \qed.
\end{theorem}
Also, applying the same argument to the map from Remark~\ref{rem:partial_moves} we also have: 
\begin{theorem}\label{th:geometricbraid_nofix}Let $M$ be a manifold.  We have an isomorphism between  the full subgroupoids of $\Mot$ and $\FMot{M}^{\partial M}$ with objects the finite subsets of $M\setminus \partial M$. \qed.
\end{theorem}

\rem{ Let $f\colon \Mot[\R^2]\to \Mot[D^2 \setminus \partial D^2]$ be a homeomorphism, and $K=\{(0,1),(0,2),\dots,(0,n)\}$. By Example~\ref{ex:R2D2}, $\Mot[\R^2](K,K)\cong  \Mot[D^2 \setminus \partial D^2](f(K),f(K))$.  Thus, by Theorem~\ref{th:geometric_braids} there is an isomorphism $\braid{}\cong \Braid(D^2\setminus \partial D^2,f(K))$. }

\rem{\label{rem:noboundaryfix}
Let $K\subset M\setminus \partial M$ be a finite subset.
From Remark~\ref{rem:braidD^2partialD^2}
we have that $\Braid(M\setminus \partial M,K)\cong \Braid (M,K)$, together with Theorem~\ref{th:geometric_braids} this implies $\Mot[M\setminus \partial M](K,K)\cong \Mot[M]^{\partial M}(K,K)$, and with Theorem~\ref{th:geometricbraid_nofix} we have $\Mot[M]^{\partial M}(K,K)\cong\Mot[M](K,K)$.
}

\begin{remark}\label{rem:braidDahm}
Notice that the proof of Theorem~\ref{th:geometric_braids} works in exactly the same way if we had restricted to the subset of $\TOPO^h(M,M)$ of homeomorphisms with compact support in our construction of the motion groupoid, as is the set up in \cite{dahm}. Thus this implies in the Artin braid setting, the two constructions are isomorphic.
\end{remark}

Finally in our discussion of the full subgroupoids of $\Mot$ and $\FMot{M}$, with objects the finite subsets of $M\setminus \partial M$, we have the following result.

\begin{lemma}\label{le:motion_exists_between_finitesets}
{Let $M$ be a connected manifold. Given a pair of finite subsets $K$ and $K'$ contained in the interior of $M$, of the same cardinality, then there exists a motion $\mot{f}{}{K}{K'}$ of $M$. Moreover, we can suppose that the motion fixes the boundary $M$.}
\end{lemma}
\begin{proof} {(Sketch.)}
{For the case $\dim(M)=1$, note that $M$ is homeomorphic to $S^1, [0,1),[0,1]$ or $\R$. The proof in this case is left to the reader. Cf. Example \ref{ex:finiteset}, and also the maps constructed in Theorem~\ref{th:mcgI} below.} 

{So suppose that $\dim(M)\ge 2$. Then $M$ remains path-connected if a finite subset is taken out of $M$. Let $K=\{x_1,\dots, x_n\}$ and $K'=\{y_1,\dots, y_n\}$. 
 We consider a path $\gamma_1$ connecting $x_1$ to $y_1$, from which we can construct a fake motion  $X_1\colon \{x_1\} \times \II\to M$, sending $(x_1,t)$ to $\gamma_1(t)$. So we have  $X_1\colon \{x_1\} \not \too \{y_1\}$. By Lemma \ref{lem:xhat}, we hence have a motion $f^1\colon \{x_1\} \too \{y_1\}$ extending $X_1$, meaning that $T'(f_1)=X_1$. {We then} consider a path $\gamma_2$ from $f_1^1(x_2)$ to $y_2$ that avoids $y_1$. {This} leads to a fake motion $X_2\colon \{y_1,f_1^1(x_2)) \not \too \{y_1,y_2\}$, where $y_1$ does not move. This can again be extended to a motion $f^2\colon \{y_1,f_1^1(x_2)) \too \{y_1,y_2\}$, meaning that $T'(f^2)=X_2$. {This process can be continued. We} consider a path $\gamma_3$ from $f_1^2\circ f_1^1(x_3)$ to $y_3$ that avoids $y_1$ and $y_2$, {from which we} construct another fake motion $X_3\colon \{y_1,y_2,f_1^2\circ f_1^1(x_3)\} \not \too \{y_1,y_2,y_3\} $  that does not move either $y_1$ or $y_2$, and restricts to $\gamma_3$ over $\{f_1^2\circ f_1^1(x_3)\}\times \II$.
 We can lift $X_3$ to another motion $f^3\colon \{y_1,y_2,f_1^2\circ f_1^1(x_3)\}  \too \{y_1,y_2,y_3\} $. If we only have three points, in $K$ and $K'$, then a motion $\{x_1,x_2,x_3\} \too \{y_1,y_2,y_3\}$ can be $f^3*f^2*f^1.$ It is clear how  the induction can be completed.}
\end{proof}

{It therefore follows that the full subgroupoid of the motion groupoid of a connected manifold $M$ with objects the finite subsets of the interior has one connected component for each possible cardinality. {The same holds for the fake motion groupoid.}}

\section{Mapping class groupoid \texorpdfstring{$\mcg^A$}{MCG(M)}}\label{sec:mcg}

\subsection{The mapping class groupoid \texorpdfstring{$\mcg$}{MCG(M)}}\label{ss:mcg}
In this section, we construct the mapping class groupoid $\mcg$ associated to a manifold $M$.
We do this by constructing a congruence on $\Hom$ (see Definition~ \ref{Def:homeoMA}), so the morphisms in $\mcg$ are certain equivalence classes of self-homeomorphisms of $M$ (together with a pair of subsets of $M$). 
Compare this with motions, which keep track of an entire path in $\TOPO^h(M,M)$.

As we will see, mapping class groupoids contain 
the classical {\it mapping class groups}.

\medskip
{ Recall from Section~\ref{ss:selfhomeos} that for a \axiomM{} $M$ and for subsets $N,N'\subseteq M$, morphisms in $\Hom[M](N,N')$ are triples denoted $\shmor{f}{}{N}{N'}$ where $\sh{f}\in \Topo^h(M,M)$ and $\sh{f}(N)=N'$.}
Where convenient, we also think of the elements of $\Hom[M](N,N')$ as the projection to the first coordinate of each triple i.e. $\sh{f}\in \Topo^h(M,M)$ such that $\sh{f}(N)=N'$.
From here we will also use $\Hom(N,N')$ to denote the same set together with the subspace topology induced from $\TOPO^h(M,M)$.
		   
\defn{\label{de:isotopy}
Let $M$ be a \axiomM{} and $N,N'\subseteq M$.
For any $\shmor{f}{}{N}{N'}$ and $\shmor{g}{}{N}{N'}$ in $\Hom(N,N')$, $\shmor{f}{}{N}{N'}$ is said to be {\em isotopic} to $\shmor{g}{}{N}{N'}$, denoted by $\simi$, if there exists a continuous map: 
		\[
		H\colon M\times \II \to M
		\]
		such that 
		\begin{itemize}
		\item for all fixed $s\in \II$, the map $m\mapsto H(m,s)$ is in $\Homn$,
		\item for all $m\in M$, $H(m,0) = \mc{f}(m)$, and 
		\item for all $m\in M$, $H(m,1) = \mc{g}(m)$.
		\end{itemize}
We call such a map an {\em isotopy} from $\shmor{f}{}{N}{N'}$ to $\shmor{g}{}{N}{N'}$.}

More generally, let $A\subset M$ be a subset, and 
$\shmor{f}{}{N}{N'}$ and $\shmor{g}{}{N}{N'}$ homeomorphisms in $\Hom^A(N,N')$, 
then $\sh{f}$ and $\sh{g}$ are said to be {\em $A$-fixing isotopic} if an  isotopy $H$ from $\shmor{f}{}{N}{N'}$ to $\shmor{g}{}{N}{N'}$ exists, satisfying moreover that for all $a\in A$ and $t \in \II$, $H(a,t)=a$. We write:
$\shmor{f}{A}{N}{N'} \simiA\shmor{g}{A}{N}{N'}$. Such a map $H\colon M\times \II \to M$ is called an {\em $A$-fixing isotopy}.

\lemm{\label{le:isotopy}
Let $M$ be a \axiomM{}.
For all pairs $N,N'\subseteq M$, the relation $\simi$ is an equivalence relation on $\Hom(N,N')$.\\
{Notation:} We call this equivalence relation {\em isotopy equivalence}.
We denote the equivalence class of $\shmor{f}{}{N}{N'}$, up to isotopy equivalence, as 
$\classi{\shmor{f}{}{N}{N'}}$.
}

\begin{proof}
Let $\shmor{f}{}{N}{N'}$, $\shmor{g}{}{N}{N'}$ and $\shmor{h}{}{N}{N'}$ be in $\Hom(N,N')$ with $(\shmor{f}{}{N}{N'})\simi (\shmor{g}{}{N}{N'})$ and $(\shmor{g}{}{N}{N'})\simi (\shmor{h}{}{N}{N'})$.
	Then
	there exists some isotopy, say $H_{\mc{f},\mc{g}}$, from $\shmor{f}{}{N}{N'}$ to $\shmor{g}{}{N}{N'}$ and 
	an isotopy, say $H_{\mc{g},\mc{h}}$, from $\shmor{g}{}{N}{N'}$ to $ \shmor{h}{}{N}{N'}$.
	
	We first check reflexivity.
	The map $H(m,s)=\mc{f}(m)$ for all $s\in \II$ is an isotopy from $\shmor{f}{}{N}{N'}$ to itself.
	For symmetry, $H(m,s)=H_{\mc{f},\mc{g}}(m,1-s)$ is an isotopy from $\shmor{g}{}{N}{N'}$ to $\shmor{f}{}{N}{N'}$.
	For transitivity, 
	\[
	H(m,s)=\begin{cases}
	H_{\mc{f},\mc{g}}(m,2s), & 0\leq s\leq \frac{1}{2} \\
	H_{\mc{g},\mc{h}}(m,2(s-\frac{1}{2})), & \frac{1}{2} \leq s \leq 1
	\end{cases}
	\]
	is an isotopy from $\shmor{f}{}{N}{N'}$ to $\shmor{h}{}{N}{N'}$.
\end{proof}

\lemm{\label{le:isotopy_comp}
Let $M$ be a manifold.
The family of relations $(\Hom(N,N'),\simi)$
for all pairs $N,N'\subseteq M$
 are a congruence on $\Hom$.
}

\begin{proof}
We have that $\simi$ is an equivalence relation on each $\Hom(N,N')$ from Lemma~\ref{le:isotopy}.

We check that the composition descends to a well defined composition on equivalence classes.
Suppose there exists an isotopy, say $H_{\mc{f},\mc{f'}}$, from $\shmor{f}{}{N}{N'}$ to  $\shmor{f'}{}{N}{N'}$ and another isotopy, say $H_{\mc{g},\mc{g'}}$ from
$\shmor{g}{}{N}{N'}$ to $\shmor{g'}{}{N}{N'}$.
Then: 
\[
H(m,s)= H_{\mc{g},\mc{g'}}\left( H_{\mc{f},\mc{f'}}(m,s),s \right)
\]
is an isotopy from $\shmor{g\circ f}{}{N}{N''}$ to $\shmor{g'\circ f'}{}{N}{N''}$. 
Note that $H$ is continuous since it is the composition of the continuous functions $H_{\mc{g},\mc{g'}}$ and $(m,s) \in M\times \II \mapsto \left ( H_{\mc{g},\mc{g'}}(m,s),s\right)\in M\times \II.$
\end{proof}

In the following theorem, as with Theorem~\ref{th:mg}, we drop the information about subsets in the tuple defining a groupoid, to keep the notation readable; so $\classi{f\colon X \to Y}$ is denoted $\classi{f}$. Again the subsets are essential to construct the relation $\simi$.

\begin{theorem}\label{th:mcg}
Let $M$ be a \axiomM{}.
There is a groupoid 
\[\mcg=(\Power M, \Hom(N,N')/\simi,\circ,\classi{\id_M},\classi{\sh{f}}\mapsto\classi{\sh{f}^{-1}}). \]
We call this the {\em mapping class groupoid of $M$}.
\end{theorem}
\begin{proof}
This is the quotient $\Hom/\simi$. Lemma~\ref{le:isotopy_comp} gives that $\simi$ is a congruence and Proposition~\ref{pr:gpd_cong} gives that the quotient of a groupoid by a congruence is still a groupoid, with the given identity and inverse.
\end{proof}

Recall that, if $X$ is a topological space, then the relation $x\sim y$ if, and only if, $x$ and $y$ can be connected by a path in $X$, is an equivalence relation on (the underlying set of) $X$. We let $\pi_0(X)$, the \emph{set of path-components of $X$}, 
denote the quotient of the set $X$
by this equivalence relation.

\lemm{\label{le:MCG_pi0}
Let $M$ be a \axiomM{}. We have
the following equality of sets
			\begin{align*}
			\mcg(N,N')&=\pi_0(\Hom(N,N')).
			\end{align*}
 			}
\begin{proof}
Using Theorem \ref{th:tensorhom}, a continuous map $M\times \II \to M$ from $\shmor{f}{}{N}{N'}$ to $\shmor{g}{}{N}{N'}$ satisfying the conditions in Definition~\ref{de:isotopy} corresponds to a path $\II\to \Homn$ 
from $\sh{f}$ to $\sh{g}$, 
hence the equivalence relations on each side are the same.
\end{proof}

For an oriented compact manifold $M$, let $\mcg^{+}\subset \mcg$ be the subgroupoid whose morphisms are only those equivalence classes containing an orientation preserving homeomorphism.
Note that, if $\sh{f}$ is an orientation preserving, then so is every $\sh{g}\in \classi{\sh{f}}$ \cite[Thm.1.6]{Hirsch}, and composition of orientation preserving maps is orientation preserving, thus $\mcg^{+}$ is indeed a subgroupoid.

\prop{\label{pr:mcgroup}
For an oriented compact manifold $M$ without boundary, it follows directly from the definitions that the mapping class group of \cite{Kassel_Turaev} is $\mcg[M]^{+}(N,N)$ (cf. \cite{damiani,birman}).
In the case $M$ is a surface, $\mcg[M]^{+}(\emptyset,\emptyset)$ is the classical mapping class group, see e.g.\cite{farb}. \qed
}

Thus, to 
help
understand specific mapping class groupoids, we can make use of  
results on mapping class (sub)groups 
e.g. in \cite{birman,farb,hamstrom,hatcherthurston}. 
It is an interesting question as to the relationship between a 
connected groupoid and the groups it contains.
The complement consists of morphism sets that are not groups,
but they are 
(respectively left and right)
$G$-sets for the the two different objects.
Each morphism set 
$ \mcg[M](N,N')$
is generated by the action of either
group on any single element.
But of course as soon as $N \neq N'$ this set is not even pointed,
so there is no analogue of the identity morphism.

\medskip

The following proposition shows that, as with motion groupoids, automorphism groups in mapping class groupoids are only finitely generated for sufficiently tame choices of objects.

\begin{theorem}
\label{th:mcgI}
	Let $M=\II$ and $N=\II\cap \mathbb{Q}$.
	Then  $\mcg[\II](N,N)$ is uncountably infinite.
	\end{theorem}
	\begin{proof}
	We begin by constructing certain elements of $\Topo^h(\II,\II)$.
	Choose points $x,x'\in N\setminus \{0,1\}$, then there is a unique piecewise linear, orientation preserving map with precisely two linear segments sending $x$ to $x'$ and moreover this map sends $N$ to itself. Denote this by $\phi_{xx'}$.
	Let us fix the point $x$, then varying $x'$ gives a countably infinite choice of maps $\phi_{xx'}$.
	We prove by contradiction that all such $\phi_{xx'}$ represent distinct equivalence classes in $\mcg[\II](N,N)$.
	
	Let $x,x',x''\in N\setminus\{0,1\}$ and suppose $\phi_{xx'}\colon N\to N$ is isotopic to $\phi_{xx''}\colon N\to N$ in $\mcg[\II](N,N)$.
	Then for all $n\in N$ we have a path $\phi_{xx'}(n)$ to $\phi_{xx''}(n)$ in $\Q$, and hence a path $\phi_{xx'}(x')$ to $\phi_{xx''}(x'')$.
	But all paths $\II \to \Q$ are constant, which follows from the intermediate value theorem.
	Hence $\phi_{xx'}=\phi_{xx''}$.
	Therefore any pair of distinct maps of the described form are not isotopic.
	
	More generally a piecewise linear map can be defined as follows. Starting from $t=0$, each segment is be defined by choosing an upper bound $t\in\{0,1]$ and a gradient (which is bounded by condition that the map is well defined). Repeating with the condition that the upper bound must be distinct from the upper bound of the previous section until $t=1$ is chosen, defines a map. Choosing rational gradients, and rational bounds is sufficient to ensure such a map sends $N$ to itself.
	By the same argument as above distinct such maps are non isotopic.
	Allowing for an infinite number of segments, then the cardinality of maps that can be constructed this way is a countable product of countable sets, thus uncountable. (More precisely it has the cardinality of the continuum.)
\end{proof}

We now give some examples of automorphism groups in the mapping class groupoid which follow easily from classical results on homeomorphism spaces.
We include proofs to set the scene for the next section.

\prop{\label{calc:MCGS1}
We have 
$
\mcg[S^1](\emptyset,\emptyset) \cong  \Z / 2\Z.
$
}
\begin{proof}
It is proven in \cite[Theorem 1.1.2]{hamstrom} that the space of orientation preserving homeomorphisms $S^1 \to S^1$ is homotopic to $S^1$, and in particular that it is path-connected. Since the identity map $S^1 \to S^1$ is orientation preserving, the path-component of the identity in $\TOPO^h(S^1,S^1)$ is the set of orientation preserving homeomorphisms from $S^1$ to itself.

Let $\mathrm{deg}\colon \Topo^h(S^1,S^1)\to \{\pm 1\}$ be the restriction of the degree map (see for example \cite[Sec.2.2]{hatcher}), thus $\mathrm{deg}(\sh{f})=+1$ if $\sh{f}$ is orientation preserving and $-1$ if $\sh{f}$ is orientation reversing. We show that $\mathrm{deg}$ descends to an isomorphism $\pi_0(\Hom[S^1](\emptyset,\emptyset))\to \{\pm 1\}\cong \Z/2\Z$.
	If $\sh{f}$ and $\sh{g}$ are orientation reversing, then $\sh{f}\circ \sh{g}^{-1}$ is orientation preserving, and hence can be connected by a path to $\mathrm{id}_{S^1}$. It follows that $\sh{f}$ and $\sh{g}$ can be connected by a path in  $\TOPO^h(S^1,S^1)$. Moreover maps in the same path component must have the same degree, thus are either both orientation reversing, or both orientation preserving.
\end{proof}

It is also known that the group of orientation preserving homeomorphisms of $S^2$ and of $S^3$ is path-connected, this follows from 
\cite[Thm.15]{Fisher}. In dimensions 2 and 3 a more recent discussion is in \cite[Thm.3.1]{Aceto_Bregman_Davis_Park_Ray}. (The dimension 2 case is also present in \cite[Thm.1.2.2.]{hamstrom}.) Therefore the same argument as above gives:

\prop{\label{prop:mcgs2s3}
We have 
$
\mcg[S^2](\emptyset,\emptyset)
\cong 
\Z / 2\Z$ and $\mcg[S^3](\emptyset,\emptyset)
\cong
\Z / 2\Z$. \qed
}

\subsection{Pointwise \texorpdfstring{$A$}{A}-fixing mapping class groupoid, \texorpdfstring{$\mcg^A$}{mcgMA}}
\label{sec:A-fixing}

Here we discuss a version of the mapping class groupoid which fixes a distinguished subset.

\begin{theorem}\label{th:mcga}
	Let $M$ be a manifold and $A\subset M$ a subset.
	There is a groupoid
	\[
	\mcg^A=(\Power M,\Hom^A(N,N')/\simiA,\circ,\classiA{\id_M},
	\classiA{\sh{f}}\mapsto \classiA{{\sh{f}}^{-1}}).
	\]
\end{theorem}
\begin{proof}
	This is the quotient $\Hom^A/\simiA$.
	The proofs of Lemmas \ref{le:isotopy} and \ref{le:isotopy_comp} proceed in 
	the same way for $A$-fixing self-homeomorphisms.
	All constructed homotopies will be $A$-fixing for all $s\in \II$.
	Proposition~\ref{pr:gpd_cong} gives that the quotient of a groupoid by a congruence is still a groupoid.
\end{proof}

We can now state an analogous version of Proposition~\ref{pr:mcgroup} for manifolds with boundary.
\begin{proposition}\label{pr:mcgroupA}
For an oriented manifold $M$ with non empty boundary $\partial M$, it follows from the definitions that the mapping class group of \cite{Kassel_Turaev} is the automorphism group $\mcg[M]^{\partial M}(N,N)$. (Note, when $\partial M$ is non-empty, boundary fixing homeomorphisms are necessarily orientation preserving (cf. \cite{damiani,birman}).
In the case $M$ is a surface, and $N=\emptyset$, $\mcg[M]^{\partial M}(\emptyset,\emptyset)$ is the \ppm{classical} mapping class group of e.g. \cite{farb}. \qed
\end{proposition}

By Proposition~\ref{pr:mcgroupA} we can follow the proof classical result \cite[Lem.2.1]{farb} in the proof of the following proposition.

\prop{\label{ex:mcg_disk}
The 
group $\mcgfix{D^2}{\partial D^2}(\emptyset,\emptyset)$
is trivial.}
\begin{proof} This follows from the  Alexander trick \cite{alexander}, as we now recall. 
Suppose we have a homeomorphism $\sh{f}\colon D^2 \to D^2,$ that restricts to the identity on $\partial D^2$.
Define 
\[
			f_t(x) = \begin{cases}
				t\, \sh{f}(x/t) & 0\leq |x| \leq t, \\
				x &  t\leq |x| \leq 1.
			\end{cases}
			\]
Notice that $f_0=\id_{D^2}$ and $f_1=\sh{f}$ and each $f_t$ is continuous.
Moreover:
\ali{
H\colon D^2\times \II &\to D^2, \\
(x,t)&\mapsto f_t(x)
}
is a continuous map.
So we have constructed an isotopy from any boundary preserving self-homeomorphism of $D^2$ to $\id_{D^2}$. Also by construction $H$ is $\partial(D^2)$-fixing.
		\end{proof}

The same argument can be expanded to give a proof that the space of maps $D^n \to D^n$ fixing the boundary is contractible; see \cite[Thm.1.1.3.2.]{hamstrom}.

\rem{The automorphism group $\mcgfix{D^2}{\partial D^2}(K,K)$ is finitely presented when $K$ is a finite subset in the interior of $D^2$. This follows from Theorem~\ref{th:braids_and_mapping_class} below, along with the finite presentation of the braid group $\braid$ given in \cite{artin}.}

Noting Proposition~\ref{pr:mcgroupA},
the following definition is consistent with \cite[Def.2.5]{damiani}.
\defin{\label{de:loopbraid} 
Let $D^3$ be the $3$-disk, and $C$ a fixed choice of subset consisting of $n$ disjoint unknotted, unlinked circles. 
Then the {\em extended loop braid group} 
may be 
defined as
\[
LB^{ext}_n= \mcgfix{D^3}{\partial D^3}(C,C).
\]}
\rem{
In particular $LB^{ext}_n=\mcgfix{D^3}{\partial D^3}(C,C)$ has a finite presentation, given in \cite{damiani}.}

\medskip

 Adjusting the proof of Lemma~\ref{le:MCG_pi0} to consider paths in $\Hom^A$ we have:

\lemm{\label{le:MCGA_pi0}
Let $M$ be a \axiomM{}. We have
the following equality of sets
			\[
			\pushQED{\qed}
			\mcg^A(N,N')=\pi_0(\Hom^A(N,N')). \qedhere
			\popQED
			\]
 			}

\mathversion{normal2} 
        
\section{Functor \texorpdfstring{$\FA$}{FA} from \texorpdfstring{$\Mot[M]^A$}{MotM}
 to 
 \texorpdfstring{$\mcg[M]^A$}{mcg}}\label{sec:mg_to_mcg}

 Here we 
 construct
 a functor $\F\colon \Mot^A \to \mcg^A$,
and
prove this is an isomorphism of categories if $\pi_1(\TOPO^h(M,M),\id_M)$ and $\pi_0(\TOPO^h(M,M))$ are trivial. Indeed we prove $\F$ is full if and only 
if $\pi_0(\TOPO^h(M,M))$ is trivial, and that $\F$ is faithful if $\pi_1(\TOPO^h(M,M),\id_M)$ is trivial. 
(This will thus generalise, for example, the known equivalence 
between realisations of braid groups 
as mapping class groups and as motion groups \cite{dahm, goldsmith, birman};
and similarly for loop braid groups
\cite{damiani}.)

\begin{theorem}\label{le:functor_Mot_to_MCG}
	Let $M$ be a \axiomM{}.
	There is a functor
	\[
	\F\colon \Mot^A\to \mcg^A
	\]
	which is the identity on objects, and on morphisms we have
	\[
	\F\left(\classm{\mot{f}{}{N}{N'}}\right)= \classi{\shmot{f_1}{}{N}{N'}}.
	\]
\end{theorem}
\begin{proof}
	We first check that $\F$ is well defined.
	By Theorem~\ref{th:mg2} two motions $\mot{f}{}{N}{N'}$ and $\mot{f'}{}{N}{N'}$ are motion-equivalent if, and only if, they are relative path-equivalent, i.e. we have a relative path-homotopy:
	\[
	H\colon \II \times \II \to \TOPO^h(M,M).
	\]
	Then $H(1,s)$ is a path $f_1$ to $f'_1$ such that for all $s\in \II$, $H(1,s)\in \Homn$.
	Hence $\shmot{f_1}{}{N}{N'}$ and $\shmot{{f'_1}}{}{N}{N'}$ are isotopic.
	
	We check $\F$ preserves composition.
	For $\classm{\mot{f}{}{N}{N'}}$ and $\classm{\mot{g}{}{N'}{N''}}$ in $\Mot$ we have:
	\begin{align*}
	\F\left(\classm{\mot{g}{}{N'}{N''}}
	*\classm{\mot{f}{}{N}{N'}}\right)&=\F\left(\classm{\mot{g*f}{}{N'}{N''}}\right)
	=\classi{\shmot{(g*f)_1}{}{N}{N''}}\\
	&=\classi{\shmot{g_1\circ f_1}{}{N}{N''}}
	=\classi{\shmot{g_1}{}{N'}{N''}}\circ\classi{\shmot{f_1}{}{N}{N'}}\\
	&=	\F\left(\classm{\mot{g}{}{N'}{N''}}\right)\circ \F\left(\classm{\mot{f}{}{N}{N'}}\right). \qedhere
	\end{align*}
\end{proof}
\lemm{\label{le:full_Mot_MCG}
	Let $M$ be a \axiomM{}.
	The functor 
	\[
	\F\colon \Mot^A \to \mcg^A
	\]
	defined in Theorem \ref{le:functor_Mot_to_MCG}			is full if and only if
	$\pi_0(\TOPO^h_A(M,M))$
	is trivial.}
\begin{proof}
	Suppose $\pi_0(\TOPO_A^h(M,M))$ is trivial and let $\classi{\shmor{f}{}{N}{N'}}\in \mcg (N,N')$.
	Then, since $\TOPO_A^h(M,M)$ is path connected, there exists a path $f\in\TOPO^h_A(M,M)$ with $f_0=\id_M$ and $f_1=\sh{f}$, and, since $\sh{f}(N)=N'$, $\mot{f}{}{N}{N'}$ is a motion. Hence 
	$\F(\classm{\mot{f}{}{N}{N'}})=\classi{\shmor{f}{}{N}{N'}}$.
	
	Now suppose $\pi_0(\TOPO^h_A(M,M))$ is non-trivial. Let $\sh{f}$
	be a self-homeomorphism in a path-component of $\TOPO^h_A(M,M)$ which does not contain $\id_M$. 
	Then $\classi{\shmor{f}{}{\emptyset}{\emptyset}}\in \mcg(\emptyset,\emptyset)$,
	and all representatives
	are in the same path-component.
	Hence there is no path $f$ in  $\TOPO^h_A(M,M)$, with $f_0=\id_M$, 
	and with $\shmot{f_1}{}{\emptyset}{\emptyset} \in \classi{\shmor{f}{}{\emptyset}{\emptyset}}$. 
\end{proof}

The following example shows that the functor $\F$ may 
restrict to 
a surjection on some automorphism groups and not on others.

\exa{
	Let $M=S^3$. 
	Recall from Lemma~\ref{le:MCG_pi0} that $\mcg(N,N')=\pi_0(\Hom(N,N'))$, where $\Hom(N,N')\subseteq \TOPO^h(M,M)$ is given the subspace topology, and $\Homempty=\TOPO^h(M,M)$.
	By Proposition \ref{prop:mcgs2s3} $\pi_0(\TOPO^h(M,M))=\mcg[S^3](\emptyset,\emptyset)
	\cong 
	\Z/2\Z$, corresponding to an orientation preserving and orientation reversing component,  
	and so by the previous lemma $\F$ is not full.
	
	Consider $K\subset S^3$ a knot which is not isotopic to its mirror image.
	Then $\mcg[S^3](K,K)$ contains only orientation preserving self-homeomorphisms, which are therefore in the same connected component as the identity. Hence, by the first part of the proof of the previous lemma, the restriction $\F\colon \Mot[S^3](K,K)\to \mcg[S^3](K,K)$ is full.
}	

\subsection{Applying 
the long exact sequence of relative homotopy groups}

Here
we prove that, for any fixed $N\subseteq M$, the functor $\F\colon \Mot \to \mcg$ restricts to a group isomorphism $\mathrm{F}\colon\Mot(N,N)\xrightarrow{\sim} \mcg(N,N)$  if $\pi_1(\Homempty,\id_M)$ and $\pi_0(\Homempty)$ are both trivial (Lemma~\ref{le:grp_iso_Mot_to_MCG}).
For this we will need the homotopy long exact sequence of a pair. We briefly introduce this here, see \cite[Sec.4.1]{hatcher} or \cite[Ch.9]{May} for 
further
exposition.

\medskip

In this section we work in the non $A$-fixing case to avoid overloading the notation. Everything proceeds in exactly the same way for the $A$-fixing case.

\defn{Let $n$ be a positive integer. Let $\II^{n-1}\times \{1\}$ be the face of $\II^n$ with last coordinate $1$,
	and let $J^{n-1}$ be the closure of $\partial \II^n\setminus ( \II^{n-1}\times \{1\})$, i.e. the union of all remaining faces of $\II^n$.}
	So for instance, for $n=1$, the triple $(\II,\partial\II, J^0)$ is $(\II, \{0,1\}, \{0\})$.
   \defn{\label{de:relative_homotopy}
   Let $X$ be a topological space, $Y\subseteq X$ a subset and $x_0\in Y$ a point. 
	For fixed integer $n\geq 1$ we define a relation on the set of continuous maps
	\[
	\gamma\colon (\II^n,\partial \II^n,J^{n-1})\to (X,Y,\{x_0\}) 
	\]
	as follows. We say $\gamma\sim \gamma'$ if there exists $H\colon \II^n\times \II \to X$ such that 
	\begin{itemize}
		\item for all $s \in \II$, $x\mapsto H(x,s)$ defines a map $(\II^n,\partial \II^n,J^{n-1})\to (X,Y,\{x_0\})$,
		\item for all $x\in \II^n$, $H(x,0)=\gamma(x)$, and
		\item for all $x\in \II^n$,  $H(x,1)=\gamma'(x)$.
	\end{itemize}}
	
	\prop{
	For each $n$ 
	the relation described in Definition~\ref{de:relative_homotopy} is an equivalence relation.\\
	Notation: We will call the set of equivalence classes of maps from $(\II^n,\partial \II^n,J^{n-1})$ to $(X,Y,\{x_0\})$  the {\em $n^{th}$ relative homotopy set} and denote it $\pi_n(X,Y,x_0)$.
}
\begin{proof}
	We omit this proof. It is similar to Lemma \ref{le:rpe}. See also \cite[Sec.4.1]{hatcher}.
\end{proof}

Recall that we also use $\Hom(N,N')$ to denote the set together with the subspace topology from $\TOPO^h(M,M)$.

\lemm{\label{le:Mot_pi1}
	Let $M$ be a \axiomM{} and $N\subseteq M$ a subset.
	Then as sets
	\[\Mot(N,N)=\pi_1(\Homempty,\Hom(N,N),\id_M).\]
}
\begin{proof} This follows from Theorem \ref{th:mg2}, combined with the previous definition.
\end{proof}
\noindent Notation: The equivalence relation used to construct $\pi_1(\Homempty,\Hom(N,N),\id_M)$ is precisely relative path-equivalence of motions $N\too N$. We abuse notation, and use $\classrp{\gamma}$ for the equivalence class of a continuous map $\gamma$ in any relative homotopy set.

\medskip

\lemm{ {\rm{(Cf. e.g. \cite[Sec.4.1]{hatcher}.)}} $\;$ 
	Let $X$ be a topological space, $Y\subseteq X$ a subset and $x_0\in Y$ a point.
	For $n\geq 2$, given continuous maps $\beta,\gamma\colon (\II^n,\partial \II^n,J^{n-1})\to (X,Y,\{x_0\}) $,
	define
\begin{equation}\label{def+}
	( \gamma+\beta)(t_1,\ldots,t_n)=
	\begin{cases}
	\beta(2t_1,\ldots,t_n) & 0\leq t_1\leq \frac{1}{2},\\
	\gamma(2(t_1-\frac{1}{2}),\ldots,t_n) & \frac{1}{2} \leq t_1\leq 1.
	\end{cases}
\end{equation}
	Then there is a well-defined binary composition 
			\ali{
		+\colon \pi_n(X,Y,x_0)\times \pi_n(X,Y,x_0) &\to \pi_n(X,Y,x_0)  \\  
		(\classrp{\beta},\classrp{\gamma})&\mapsto \classrp{\gamma +\beta}.} 
		
		 Moreover the set $\pi_n(X,A,x_0)$ becomes a group with $+$. The identity is the equivalence class of the constant path $(t_1,\dots,t_n)\mapsto x_0$.
	The inverse of $\classrp{\gamma}\in \pi_n(X,Y,x_0)$ is the equivalence class of $(t_1,\dots,t_n)\mapsto \gamma(1-t_1,\dots ,t_n)$. \qed
}

For $n=1$ there is not in general a sensible way to make $\pi_1(X,Y,x_0)$ into a group. This is however possible when  $X$ is topological group, and $Y$ is a subgroup.
This is the case we will be interested in.

\medskip

In the theorem below, given a continuous map $\gamma \colon (\II^n,\partial\II^n)\to (X,\{x_0\})$, of pairs, the set of maps 
$\gamma' \colon (\II^n,\partial\II^n)\to (X,\{x_0\})$ homotopic to $\gamma$, relative to $\partial \II^n$ is denoted $\classp{\gamma}$. For $n=1$, this coincides with the conventions in Definition \ref{de:pe}. Then for $n\ge 1$, the group $\pi_n(X,x_0)$ is the set of such homotopy classes, with the operation induced by $+$ in \eqref{def+}. (We will continue to denote the operation in $\pi_1(X,x_0)$ by the null symbol, since the group $\pi_1(X,x_0)$ is not necessarily abelian).

\medskip

For a space $X$, we use $\pi_0(X,y)$, where $y\in X$, to denote
the set $\pi_0(X)$, together with 
the path-component of $y$, so $\pi_0(X,y)$ is a pointed set (a set together with an element of it). 
The set $ \pi_1(X,A,x_0)$
also naturally becomes a pointed set with distinguished element the homotopy class of the appropriate constant map to $x_0$.
In the last three stages of the exact sequence in 
Theorem~\ref{th:long_exact_sequence} below,
where group structures are not defined in general, exactness 
means that the image of a map
is equal 
to
the set of elements sent to the homotopy class of the constant map to $\{x_0\}$ by the following map.

\begin{theorem}\label{th:long_exact_sequence}
(See for example \cite[Sec.4.1]{hatcher}.)
Let $X$ be a space, $Y\subseteq X$ a subset, $x_0\in Y$ a point, and let the maps $i\colon (Y,\{x_0\})\to (X,\{x_0\})$ and $j\colon (X,\{x_0\},\{x_0\})\to (X,Y,\{x_0\})$ be the inclusions. Then we define
	\ali{
		i_*^n\colon \pi_n(Y,x_0) &\to \pi_n(X,x_0) \\
		\classp{\gamma} &\mapsto \classp{i \circ \gamma}
	}
	and 	
	\ali{
		j_*^n\colon \pi_n(X,x_0)&\to \pi_n(X,Y,x_0)\\
		\classp{\gamma} &\mapsto \classrp{j\circ \gamma}.
	}
	We also define a map which is the following restriction:
	\ali{
		\partial^n\colon \pi_n(X,Y,x_0)&\to \pi_{n-1}(Y,x_0)\\
		\classrp{\gamma}  &\mapsto \classp{\gamma\vert_{\II^{n-1} \times \{1\}}}.
	}
	Note that for $n=1$, we have
	\ali{
		\partial^1\colon \pi_1(X,Y,x_0)&\to \pi_0(Y,x_0), \\
		\classrp{\gamma} &\mapsto \classp{\gamma(1)}.
	}
	There is a long exact sequence: 
\begin{multline}\label{lespair}
	\ldots \to
	\pi_n(Y,\{x_0\})
	\xrightarrow{i^n_{*}} \pi_n(X,x_0)
	\xrightarrow{j^n_{*}} \pi_n(X,Y,x_0)
\xrightarrow{\partial^n} \pi_{n-1}(Y,x_0)
	\xrightarrow{i^{n-1}_{*}} \dots
\\	\dots \to \pi_{1}(Y,x_0) \xrightarrow{i^1_*} \pi_{1}(X,x_0)
	\xrightarrow{j^1_{*}} \pi_1(X,Y,x_0)
	\xrightarrow{\partial^1} \pi_{0}(Y,x_0)
\xrightarrow{i^0_{*}}\pi_0(X,x_0).
\end{multline}
	\qed
\end{theorem}

A version of
the following 
long exact sequence 
appears in \cite[Prop.3.2]{goldsmith};
there all homeomorphisms have compact support and in particular, the motion automorphism group is replaced by the motion group of \cite{goldsmith} (see Remark~\ref{rem:DahmMG}), keeping in mind also Propositions~\ref{pr:mcgroup} and \ref{pr:mcgroupA}.

\lemm{\label{le:long_exact_sequence2}
	Let $M$ be a \axiomM{} and fix a subset $N\subseteq M$. Then we have a long exact sequence 
	\begin{multline} \label{eq:les}
	\ldots \to
	\pi_n(\Hom(N,N),\id_M)
	\xrightarrow{i^n_{*}} \pi_n(\Homempty,\id_M)
	\xrightarrow{j^n_{*}}\\ \pi_n(\Homempty,\Hom(N,N),\id_M)
	\xrightarrow{\partial^n} \pi_{n-1}(\Hom(N,N),\id_M)
	\xrightarrow{i^{n-1}_{*}}\\
	\dots \xrightarrow{\partial^2}
	\pi_1(\Hom(N,N),\id_M)
	\xrightarrow{i^1_{*}} \pi_1(\Homempty,\id_M)\\
	\xrightarrow{j^1_{*}}			 \Mot(N,N) 
	\xrightarrow{\F} \mcg(N,N)
	\xrightarrow{i_*^0} 
	\mcg(\emptyset,\emptyset)
	\end{multline}
	where all maps are group maps
	and $\F$ is the appropriate restriction of the functor defined in Theorem \ref{le:functor_Mot_to_MCG}.
}
\begin{proof}
We get the first part of the sequence by substituting into \eqref{lespair} $X= \Homempty$, $Y=\Hom(N,N)$ and $x_0=\id_M$.
	We have that as sets $\mcg(N,N)=\pi_0(\Hom(N,N))$ from Lemma \ref{le:MCG_pi0}, and in particular $\mcg(\emptyset,\emptyset)=\pi_0(\Homempty)$. 
	From Lemma \ref{le:Mot_pi1} we have $\Mot(N,N)=\pi_1(\Homempty,\Hom(N,N),\id_M)$
	as sets.
	Notice also that, as a set map, $\F\colon \Mot(N,N)\to \mcg(N,N)$ is precisely $\partial ^1$.
	In the last few steps we have replaced pointed sets with groups that have the same underlying set, we show that the sequence is indeed an exact sequence of groups.	

	Note that each group identity is the distinguished object in the corresponding pointed set, hence exactness follows from the pointed set case, and we only need to check that each map is indeed a group homomorphism.
	We have that $\F$ is a group map, as it is the restriction of a functor of groupoids. It remains to show that $j_*^1$ and $i_*^0$ become group maps.
	We check that $j_*^1$ preserves composition. 
	Let $g$ and $f$ be paths from $\id_M$ to $\id_M$ in $\Homempty$.
	Then $gf$, the concatenation of paths as in Proposition \ref{pr:path_comp}, is a well defined \premot{} and it  is precisely the \premot{} $g*f$ as $f_1=\id_M$. 
	Hence we have 
	\[
	j_*^1(\classp{gf})=\classrp{\mot{gf}{}{N}{N}}=
	\classrp{\mot{g* f}{}{N}{N}}
	=\classrp{\mot{g}{}{N}{N}}*\classrp{\mot{f}{}{N}{N}}=j_*^1(\classp{g})*j_*^1(\classp{f}).
	\]
	The composition in $\mcg(N,N)$ is composition of homeomorphisms, 
	hence the composition is the same in the source and target of $i^0_*$, and $i^0_*$ is an inclusion. Thus composition is preserved.
\end{proof}

\lemm{\label{le:grp_iso_Mot_to_MCG}
	Suppose $M$ is a \axiomM{} and fix a subset $N\subseteq M$. Suppose
	\begin{itemize}
		\item $\pi_1(\Homempty,\id_M)$ is trivial, and
		\item  $\mcg(\emptyset,\emptyset)=\pi_0(\Homempty)=\pi_0(\TOPO^h(M,M))$ is trivial.
	\end{itemize}
	Then there is a group isomorphism
	\[
	\F\colon\Mot(N,N)\xrightarrow{\sim} \mcg(N,N).
	\]
}
\begin{proof}
	Using the conditions of the lemma, the long exact sequence in Lemma \ref{le:long_exact_sequence2} gives short exact sequence
	\[
	1
	\to \Mot(N,N)
	\to \mcg(N,N)
	\to 1.\qedhere
	\]
\end{proof}

\rem{
More generally \eqref{lespair} becomes an exact sequence of groups when $X$ is a topological group, $Y$ a subgroup and $x_0=1_X$, which generalises the case treated in Lemma~\ref{le:long_exact_sequence2}.
This can be seen as follows.
Using $.$ to denote the composition in $X$, a pointwise composition in $\pi_n(X,1_X)$ is defined as follows.
The composition of a pair $(\classp{\gamma},\classp{\beta})$ is the class of $(t_1,\dots, t_n)\mapsto \gamma(t_1,\ldots, t_n).\beta(t_1,\ldots, t_n)$.
That this is equivalent to the composition given by \eqref{def+} follows from the same idea as in the proof of Lemma \ref{le:pms_star_equiv_dot}.
The same composition gives a composition of paths in $\pi_n(X,Y,1_X)$ for $n\geq 1$, which again is the same as \eqref{def+} for $n\geq 2$.
Composition in $\pi_0(X)$ and $\pi_0(Y)$ is given by the group composition.
It is straightforward to check that all maps become group maps.}

\subsection{Isomorphisms \texorpdfstring{$\Mot^A$}{MotM} to \texorpdfstring{$\mcg^A$}{MCGM}}

Here we give conditions under which $\Mot^A$ and $\mcg^A$ are isomorphic categories.

\begin{theorem}\label{th:mot_to_mcg}
	Let $M$ be a \axiomM{}.
	If
	\begin{itemize}
		\item $\pi_1(\Homempty,\id_M)=\pi_1(\TOPO^h(M,M),\id_M)$ is trivial, and
		\item $\mcg(\emptyset,\emptyset)=\pi_0(\Homempty)=\pi_0(\TOPO^h(M,M))$ is trivial,
	\end{itemize}
	the functor 
	\[ 
	\F \colon \Mot^A
	\to \mcg^A
	\]
	defined in Theorem \ref{le:functor_Mot_to_MCG} is an isomorphism of categories.
\end{theorem}
\begin{proof}
	\ppm{First note that by construction $\F$ is the identity map on objects.}
	Suppose $\pi_1(\Homempty,\id_M)$
	and
	$\pi_0(\Homempty)$ are trivial.
	We have from Lemma \ref{le:full_Mot_MCG} that $\F$ is full.
	We check $\F$ is faithful.
	Let
	$\classm{\mot{f}{}{N}{N'}}$
	and $\classm{\mot{f'}{}{N}{N'}}$ be in $\Mot(N,N')$.
	If
	\[
	\F\left(
	\classm{\mot{f}{}{N}{N'}}
	\right)
	=
	\F\left(
	\classm{\mot{f'}{}{N}{N'}}
	\right),
	\]
	then
	\ali{
		\classi{\shmot{\id_M}{}{N}{N}}&=
		\F(\classm{\mot{f'}{}{N}{N'}})^{-1}\circ\F(\classm{\mot{f}{}{N}{N'}})\\
		&=
		\F\big(\classm{\mot{f'}{}{N}{N'}}^{-1}*\classm{\mot{f}{}{N}{N'}}\big)\\
		&=\F(\classm{\mot{\bar{f'}*f}{}{N}{N}}).}
	By Lemma \ref{le:grp_iso_Mot_to_MCG} this is true if and only if
	\[
	\classm{\mot{\bar{f'}*f}{}{N}{N}}=\classm{\mot{\Id_M}{}{N}{N}}
	\]
	which is equivalent to saying $\Id_M*(\bar{f'}*f)$ is path-equivalent to a stationary motion, and hence that $\bar{f'}*f$ is path-equivalent to the stationary motion (since $\Id_M*(\bar{f'}*f)\simp \bar{f'}*f$).
	So we have $\classm{\mot{f}{}{N}{N'}}=
	\classm{\mot{f'}{}{N}{N'}}$.

We have seen that $\F \colon \Mot
\to \mcg$ is such that,
restricted to objects
it is a bijection, and moreover proved that, given objects $N,N' \in \Power{M}$, then $\F\colon \Mot(N,N')
\to \mcg(N,N')$ is a bijection. Therefore $\F$
is, by definition, an isomorphism of categories (see for example \cite[Page 14]{MacLane}).

	The proofs of the 
	preceding 
	lemmas and the theorem proceed in the same way for the $A$-fixing case.
\end{proof}
		
 \rem{We note that for some manifolds $M$, we may have that $\F$ is not an isomophism, but does restrict to an isomorphims on some full subgroupoid given by a particular choice of subsets of $M$.
 See Example~\ref{ex:2sphere} below.}		

\subsection{Examples of \texorpdfstring{$\F\colon \Mot^A\to \mcg^A$}{F:Mot to mcg}}\label{ss:longexactseq}
 Here we give examples of $M$ for which $\F$ is an isomorphism, and examples for which it is not. Even when we do not have a category isomorphism, the long exact sequence in \eqref{eq:les} will often be useful to obtain results about motion groupoids from results about mapping class groupoids.

\medskip

\subsubsection{Example 1: the disk \texorpdfstring{$D^m$}{Dm}, fixing the boundary}\label{sec:D^m}

\prop{\label{pr:D^m}
Let $D^m$ be the $m$-dimensional disk. Then there is an isomorphism
\[
\F\colon \Mot[D^m]^{\partial D^m}\xrightarrow{\sim} \mcgfix{D^m}{\partial D^m}
\]
	with $\F$ as in Theorem~\ref{le:functor_Mot_to_MCG}. }
\begin{proof}
We proved in Proposition \ref{ex:mcg_disk} that $\mcg[D^2]^{\partial D^2}(\emptyset,\emptyset)=\pi_0\big(\Hom[D^2]^{\partial D^2}(\emptyset,\emptyset)\big)$ is trivial.
Also $\Hom[D^2]^{\partial D^2}(\emptyset,\emptyset)$ is contractible, see e.g. Theorem 1.1.3.2 of \cite{hamstrom}.
In fact, for all integer $m$, $Homeo_{D^m}^{\partial D^m}(\emptyset,\emptyset)$ contractible for all $m$. This follows from the Alexander Trick \cite{alexander}.
Hence by Theorem \ref{th:mot_to_mcg} we have the result.
\end{proof}

One consequence of Proposition~\ref{pr:D^m}
is the following  
result,
well-known via  Proposition~\ref{pr:mcgroupA}. 
A similar result can be found in \cite[Thm.1]{BB}, for instance, framed  using instead groups of diffeomorphisms of $D^2$ to formulate  mapping class groups of $D^2$.

\begin{theorem}\label{th:braids_and_mapping_class}
Let $n$ be a positive integer. Let $K\subset \mathrm{int}(D^2)$ have cardinality $n$. 
Choose a homeomorphism $f\colon \R^2\to \intt{D^2}$ and let $K'\subset \R^2$ be such that $f(K')=K$.

We have an isomorphism:
\[ \mcgfix{D^2}{\partial D^2}(K,K) \cong \mathbf{B}(\R^2,K'). \] (Recall from \S\ref{sec:artinbraids} that $\mathbf{B}(\R^2,K)$ is isomorphic, by construction, to the 
Artin braid group $\braid$ on $n$ strands, as formulated in  \cite{artin}.)
\end{theorem}
\begin{proof} 

From Theorem~\ref{th:geometric_braids} we have an isomorphism $\Mot[\R^2](K',K')\cong \Braid(\R^2,K')$. By Lemma~\ref{Relating_Mot} it follows that $\Mot[\intt{D^2}](K,K)\cong \Braid(\R^2,K')$, and with Remark~\ref{rem:noboundaryfix}
it follows that $\Mot[D^2]^{\partial D^2}(K,K) \cong \Braid(\R^2,K')$.
Combining with Proposition~\ref{pr:D^m} we have the result.
\end{proof}

The following corollary is stated in \cite{damiani}, and is essentially present in \cite{goldsmith} (see Proposition~3.2, Corollary~3.6 and Section~5), although the setting of the latter is $\R^3$ restricted to maps with compact support.
\begin{corollary}
    \label{ex:3disk_iso}
Let $D^3$ be the $3$-disk, and $C$ a fixed choice of subset consisting of $n$ disjoint unknotted, unlinked circles.
Then we have an isomorphism
	\[
			\F\colon \Mot[D^3]^{\partial D^3}(C,C)
			\xrightarrow{\sim} LB^{ext}_n=\mcg[D^3]^{\partial D^3}(C,C).
			\] 
			(The group $LB^{ext}_n$ was defined in Definition~\ref{de:loopbraid}.)\qed
\end{corollary}

In particular notice that the previous corollary implies $\Mot[D^3]^{\partial D^3}(C,C)$ is finitely generated, since the loop braid group has a finite presentation \cite[Prop.3.16]{damiani}.

\prop{\label{rem:mcgI}
Let $N=\II\cap \Q$, then
$\Mot[\II](N,N)$ is uncountably infinite.
}
\proof{
It follows from Proposition~\ref{pr:D^m}
 that we have an isomorphism 
\[
\F\colon \Mot[\II]^{\partial \II}\to \mcg[\II]^{\partial \II},
\]
and, since all motions in $\II$ are boundary fixing, an isomorphism
$ \Mot[\II]\to \mcg[\II]^{\partial \II}
$.
All mapping classes considered in Theorem~\ref{th:mcgI} are boundary fixing, thus this isomorphism implies $\Mot[\II](N,N)$ uncountably infinite.
\qed 
}
		
\subsubsection{Example 2: the disk \texorpdfstring{$D^2$}{D2}, without fixing the boundary}
\label{sec:D2_movebdy}

In this section we consider the (non-isomorphism) functor 
\[
\F\colon \Mot[D^2]\to \mcg[D^2],
\]
giving an explicit example of an motion contained in the kernel.

Let $P_2\subset D^2\setminus \partial D^2 $ be a subset consisting of two points equidistant from and
on a line 
through
the centre of the disk.
Let $\tau_{\pi}$ be the path in $\TOPO^h(D^2,D^2)$ such that $\ppm{(\tau_{\pi})_ t}$ is a $\pi t$ rotation of the disk.
There is a motion $\mot{\tau_\pi}{}{P_2}{P_2}$, and this motion represents a non-trivial equivalence class in $\Mot[D^2]$. This can be seen by noting that there is no stationary motion which exchanges the points in $P_2$. 
By a similar argument $\F(\classm{\mot{\tau_\pi}{}{P_2}{P_2}})$ also represents a non trivial element of $\mcg[D^2]$.

Now consider the motion $\mot{\tau_\pi *\tau_\pi}{}{P_2}{P_2}$.
It is intuitively clear this motion is non-trivial in $\Mot[D^2]$ by considering its image as a homeomorphism $D^2\times \II\to D^2\times \II$, see Figure~\ref{fig:2pi_rotation}. A proof follows from the fact that the worldlines of the trajectory of the points in $P_2$ transcribe a non-trivial braid.
However its endpoint is a $2\pi$ rotation, which is precisely $\id_M$, hence represents $\shmor{\id_{\mathit{D}^\mathrm{2}}}{}{P_2}{P_2}$ in $\mcg[D^2]$.
\begin{figure}[h]
	\centering
	\includegraphics[scale =0.2]{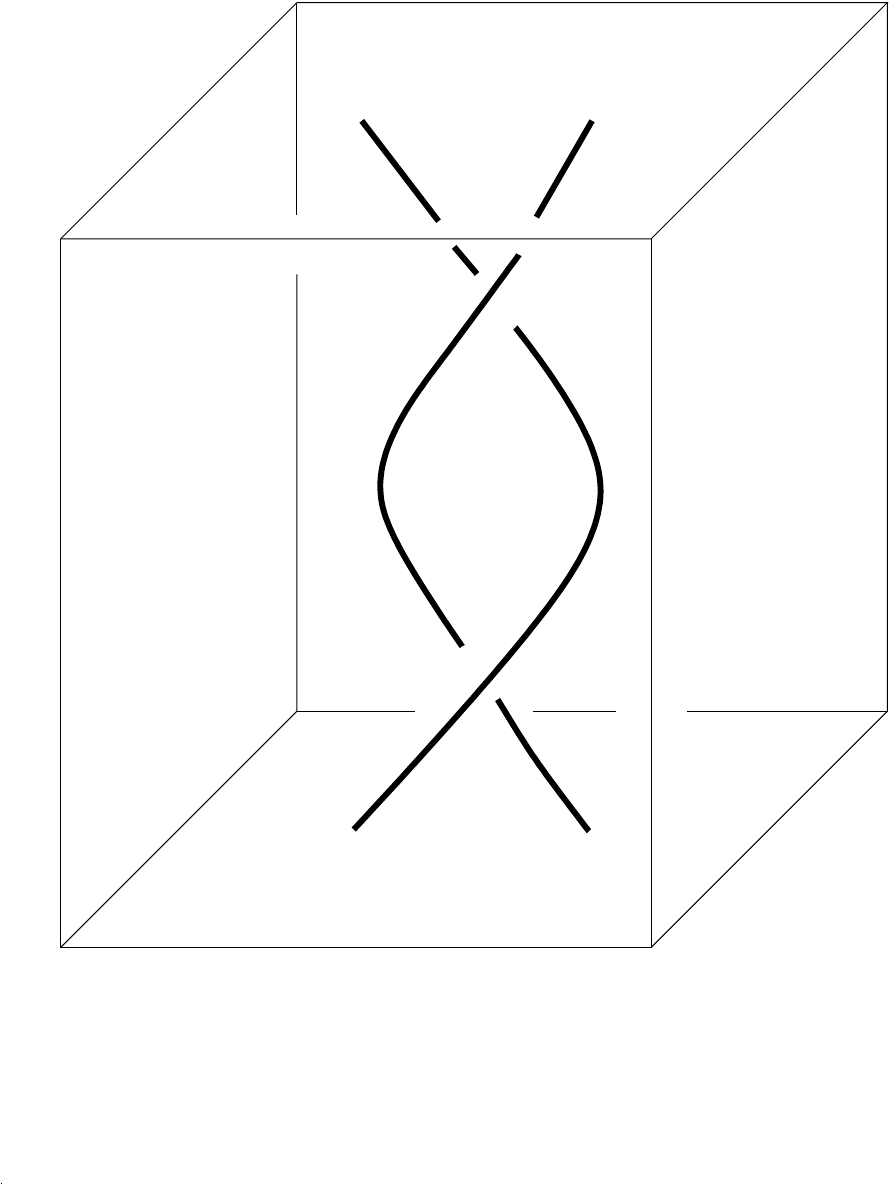}
	\vspace{-1em}
	\caption{Movement of two points during motion $\mot{\tau_\pi *\tau_\pi}{}{P_2}{P_2}$ (see text), mapped into $\Mot[\II^2]$, and represented as the image of a homeomorphism $\II^3\to \II^3$.}
	\label{fig:2pi_rotation}
\end{figure}

In fact, the map $\F\colon\Mot[D^2]\to \mcg[D^2]$ is neither full nor faithful.
The space $\TOPO^h(D^2,D^2)$ is homotopy equivalent to $S^1\sqcup S^1$, where the first connected component corresponds to orientation preserving homeomorphisms and the second orientation reversing.
That the space of orientation preserving homeomorphisms is homotopy equivalent to $S^1$ follows from \cite[Sec.1.1]{hamstrom}. We arrive at $S^1\sqcup S^1$ following the same argument used in the proof of Proposition~\ref{calc:MCGS1} (with the definition of degree map generalised to the disk), together with the fact that all path components of topological groups are homeomorphic.
In particular we have from \cite{hamstrom} that $\pi_1(\Hom[D^2](\emptyset,\emptyset),\id_{D^2})=\mathbb{Z}$ where the single generating element corresponds to the $2\pi$ rotation.
Also $\pi_0(\Hom[D^2](\emptyset,\emptyset))=\mathbb{Z}/2\mathbb{Z}$.
The last steps of the exact sequence from Lemma~\ref{le:long_exact_sequence2} \ppm{are}:
		\[		\ldots \to
	\pi_1(\Hom[D^2](N,N),{\id}_{D^2})
				\xrightarrow{i^1_{*}}
		\mathbb{Z}
		\xrightarrow{j^1_{*}} \Mot[D^2](N,N)
		\xrightarrow{\F} \mcg[D^2](N,N)
	\xrightarrow{i^0_{*}}\mathbb{Z}/2\mathbb{Z}.
		\]

 \subsubsection{Example 3: the circle \texorpdfstring{$S^1$}{S1}} \label{ex:S1}

The unit circle $S^1$ provides another example for which we can explicitly construct elements of the non-trivial kernel of $\F\colon \Mot[S^1] \to \mcg[S^1]$. Let $P\subset S^1$ be a subset containing a single point in $S^1$. Consider the \premot{} $f$ of $S^1$, where $f_t$ is a $2\pi t$ rotation of the circle. There is a motion $f\colon P \too P $, whose worldline is depicted figure \ref{fig:Point_circle}.
\begin{figure}[h]
		\centering
		\def\svgwidth{0.2\columnwidth}
		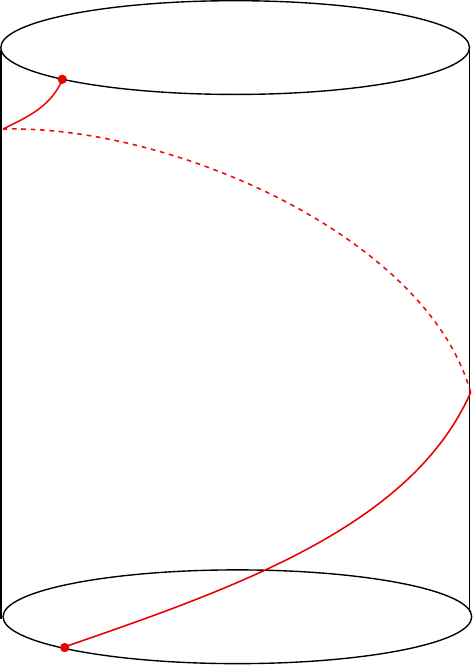
		\caption{Example of motion of circle which is a $2\pi$ rotation carrying a point to itself.} 
		\label{fig:Point_circle}
\end{figure}
Similarly to the construction for the  disk, we then have:
\begin{lemma}
The morphism $\classm{f\colon P \too P} $ in the motion groupoid $\Mot[S^1]$ is non trivial, and moreover we have an exact sequence:
	\begin{align}\label{eq:S1exact_seq}
			\ldots \to \{1\}\to 
			\mathbb{Z}
			\xrightarrow{\cong} \Mot[{S^1}](P,P)
			\xrightarrow{0} \mcg[{S^1}](P,P)
			\xrightarrow{\cong}
		 \mathbb{Z}/2\mathbb{Z}.
			\end{align}
\end{lemma}
\newcommand{\Rot}{\mathrm{Rot}}

In the proof of this lemma, we consider the homomorphism $\Rot\colon S^1 \to \Hom[S^1](\emptyset,\emptyset)$, the topological space of homeomorphisms of $S^1$, sending a point with coordinate $\theta$ of $S^1$ to a rotation of $S^1$ along angle $\theta$.
Explicitly, seeing $S^1$ as a subset of the complex plane, $\Rot$ is defined as:
\begin{align*}
\Rot\colon S^1 &\longrightarrow \Hom[S^1](\emptyset,\emptyset)\\
w &\longmapsto \left (\Rot_w\colon S^1 \ni z \mapsto wz \in S^1\right).
\end{align*}
This is an injective homomorphism, hence an embedding, since $S^1$ is compact and $\Hom[S^1](\emptyset,\emptyset)$ is Hausdorff.
\begin{proof}
The last steps \ppm{in} the exact sequence in Lemma \ref{le:long_exact_sequence2} are:
\begin{multline*}
	\dots \xrightarrow{\partial^2}
	\pi_1(\Hom[S^1](P,P),\id_{S^1})
	\xrightarrow{i^1_{*}} \pi_1(\Hom[S^1](\emptyset,\emptyset),\id_{S^1})\\
	\xrightarrow{j^1_{*}}			 \Mot[S^1](P,P) 
	\xrightarrow{\F} \mcg[S^1](P,P)
	\xrightarrow{i_*^0} \mcg[S^1](\emptyset,\emptyset).
\end{multline*}

Now, the subspace of $\Hom[S^1](P,P)$ containing orientation preserving maps is contractible
(see \S1.1.2 combined with Thm~1.1.1 \cite{hamstrom}).
Also every point on a path in $\Hom[S^1](P,P)$ starting at $\id_M$ is an orientation preserving map (this is argued in the proof of Proposition~\ref{calc:MCGS1}).
Hence $\pi_1(\Hom[S^1](P,P),{\id}_{S^1})$ is trivial, and also $\F$ is the zero map, as the endpoint of a motion of $S^1$ must be orientation preserving.
Also from \cite[\S 1.1.2]{hamstrom} we have that $\Rot(S^1)$ is a strong deformation retract of the path-component of $\Hom[S^1](\emptyset,\emptyset)$, containing the identity (i.e. the subgroup of orientation preserving homeomorphisms $S^1 \to S^1$, as argued in Proposition~\ref{calc:MCGS1}). By Proposition~\ref{calc:MCGS1}, $
\mcg[S^1](\emptyset,\emptyset)=\Z / 2\Z
$, hence the exact sequence becomes 
:
		\begin{equation}\label{eq:exact_sequenceS1}
			\ldots \xrightarrow{\partial^2} \{1\}	\xrightarrow{i^1_{*}} 
			\mathbb{Z}
				\xrightarrow{j^1_{*}}
			 \Mot[{S^1}](P,P)
			 	\xrightarrow{0} 
			 \mcg[{S^1}](P,P)
			\xrightarrow{i_*^0}  \mathbb{Z}/2\mathbb{Z}.
		\end{equation}
In particular, $j^1_*\colon \Z \to \Mot[{S^1}](P,P)$ is injective. That $\classm{f\colon P \too P} $ is non trivial follows from the fact that it is, by construction, the image of the generator of $\pi_1(\Rot(S^1),\id_{S^1}) \cong \Z$ arising (via $\Rot$) from a positively oriented loop in $S^1.$

Finally notice $i_*^0$ is not the zero map as there is clearly an orientation reversing homeomorphism of $S^1$ sending $P$ to $P$.
Hence the exact sequence in \eqref{eq:exact_sequenceS1} becomes \eqref{eq:S1exact_seq}.
 	In particular the motion-equivalence class of $f$ is non-trivial in $\Mot[S^1](P,P)$; and its image in $\mcg[S^1](P,P)$ is trivial.	
\end{proof}

Note that we can directly see that $\F\left ( \classm{f\colon P \too P} \right)=\classi{\shmot{f_1}{}{P}{P}}= \classi{\shmot{\id_{S^1}}{}{P}{P}}$, where the latter equality follows from the fact that a $2 \pi$ rotation of $S^1$ is the identity function. However we need the previous theorem to assert that $f$ is non-trivial in $\Mot[S^1](P,P)$.

\subsubsection{Example \ppm{4}: the \texorpdfstring{$2$}{}-sphere \texorpdfstring{$S^2$}{}} \label{ex:2sphere}

The following example shows that $\F$ can restrict to an isomorphism when considering certain subsets (and orientation preserving maps) but not others.

Let $M=S^2$ and $P_2$ be the subset containing $2$ points in the sphere.

\noindent From Section 1.2 of \cite{hamstrom} we have the following,
\ali{
\pi_1(\Hom[S^2](P_2,P_2),\id_{S^2}) &= \Z  \\
\pi_1(\Hom[S^2](\emptyset,\emptyset),\id_{S^2}) &= \Z/2\Z \\
\pi_0(\Hom[S^2](\emptyset,\emptyset),\id_{S^2}) &= \Z/2\Z.}

So the exact sequence becomes
\[\dots \to \Z\to \Z/2\Z \to \Mot[S^2](P_2,P_2) \to \mcg[S^2](P_2,P_2) \to \Z/2\Z. \]
Also from \cite{hamstrom}, the map $\pi_1(\Hom[S^2](P_2,P_2),\id_{S^2}) \to 
\pi_1(\Hom[S^2](\emptyset,\emptyset),\id_{S^2})$ is surjective, with the non trivial element in $\pi_1(\Hom[S^2](\emptyset,\emptyset),\id_{S^2})$ represented by a path which maps $t\in\II$ to a $2\pi t$ rotation about some chosen axis.
Hence the map $\Z/2\Z \to \Mot[S^2](P_2,P_2)$ is the zero map, and the same rotation is trivial in $\Mot[S^2](P_2,P_2)$.

This can be seen directly by choosing the points to be antipodal, say the north and south pole. Now consider a $2\pi$ rotation with axis through north and south pole. 
This is a path fixing both points, hence a stationary path which is equivalent to the identity.

Looking back at the exact sequence, we have that the map $\Mot[S^2](P_2,P_2)\to \mcg[S^2](P_2,P_2)$ is injective.
Combining the results of Hamstrom with the result on pg.52 of \cite{farb}, we have that $\mcg[S^2](P_2,P_2)\cong \Z/2\Z \times \Z/2\Z$ where the non trivial element in the first copy of $\Z/2\Z $ is represented by a self-homeomorphism which swaps the points by an orientation preserving self-homeomorphism, and the non trivial element in the second component is represented by a self-homeomorphism swapping the two points with is orientation reversing.
Hence a motion which swaps the two points represents a non trivial morphism in $\Mot[S^2](P_2,P_2)$.

Recall that $\mcg[S^2]^{+}$ is the mapping class groupoid constructed using only orientation preserving homeomorphisms. Then we have a group isomorphism
\[
\Mot[S^2](P_2,P_2) \simeq \mcg[S^2]^{+}(P_2,P_2).
\]
Note this does not extend to a category isomorphism. Considering instead the subset consisting of three points the groups are non isomorphic. Intuitively we can see this by arguing that we cannot place three points on the sphere such that any $2\pi$ rotation is a stationary motion. Arguing in a similar way to the preceding examples, a $2\pi$ rotation of the sphere represents the identity morphism in the mapping class groupoid.

\setlength\emergencystretch{2em}
\normalem
\bibliographystyle{alphaurl}

\bibliography{bib.bib}

\appendix
\section{Appendix}
\subsection{Proof of Theorem~\ref{le:top_group}}
\label{sec:TopGroupProof}

This section revisits the proof of Theorem 4 in \cite{arens} using our notation.

Recall that a space $X$ is said to be {\em locally compact} if each $x\in X$ has an open neighbourhood which is contained in a compact set.
If $X$ is Hausdorff, 
$X$ is locally compact if and only if for each $x \in X$ and open set $U\subset X$ containing $x$, there exists an open set $V$ containing $x$ with $\bar{V}$ compact and $\bar{V}\subset U$ (where $\bar{V}$ is the closure of $V$)  \cite[Thm~29.2]{munkres}.

\lemm{\label{le:locally_compact}
Let $X$ be a locally compact Hausdorff space.
Let $K\subset X$ be compact and $U\subset X$ be open with $K\subset U$.
Then there exists an open set $V$ with $K\subset V\subset\overline{V}\subset U$, where $\overline{V}$ is compact.}
\begin{proof}
Since $X$ is locally compact Hausdorff, for every $x\in K$ there is an open set $V(x)\subset U$ with $\overline{V(x)}\subset U$ compact.
The set of all $V(x)$ is a cover for $K$, and $K$ is compact so there exists a finite subcover.
Hence we have
\[
K\subset \bigcup_{i\in \{1,\dots ,n\}}V(x_i)
\subset \bigcup_{i\in \{1,\dots ,n\}}\overline{V(x_i)}\subset U\]
for some finite set $\{x_1,\dots x_n\}\subset K$.
We can choose $V=\bigcup_{i\in \{1,\dots ,n\}}{V(x_i)}$,
noting that (since the union is finite)
$\overline{V}=\bigcup_{i\in \{1,\dots ,n\}}\overline{V(x_i)}$,
and hence $\overline{V}$ is compact, since it is a finite union of compact subsets. 
\end{proof}

\lemm{\label{le:comp_cont}
Let $X$ be a locally compact Hausdorff space. Then the composition of homeomorphisms
\ali{
\circ \colon \TOPO^h(X,X) \times \TOPO^h(X,X) &\to \TOPO^h(X,X) \\
(\sh{f},\sh{g})&\mapsto \sh{g}\circ\sh{f} 
}
is continuous.
}
\begin{proof}
Let $\coball{X}{X}{K}{U}$ be an element of the subbasis of $\tauco{X}{X}$.
Now suppose $\sh{h}\in \coball{X}{X}{K}{U}$ is in the image of $\circ$, so $\sh{h}=\sh{g}\circ \sh{f}$ for some $\sh{g},\sh{f}\in \Topo^h(X,X)$.
 We show that for all such $\sh{h}$, we can construct an open set in $V\in\Topo^h(X,X) \times \Topo^h(X,X)$ with $(\sh{f},\sh{g})\in V$ and for all $(\sh{f}',\sh{g}')\in V$,
$\sh{g}'\circ \sh{f}'\in \coball{X}{X}{K}{U}$. 

We have $\sh{g}(K)\subset \sh{f}^{-1}(U)$, and so by Lemma \ref{le:locally_compact} there exists an open set $W$ with $\sh{g}(K)\subset W\subset \overline{W}\subset f^{-1}(U)$, and $\overline{W}$ compact.
Now $\coball{X}{X}{K}{W}\times \coball{X}{X}{\overline{W}}{U}$ is an open set containing $(\sh{f},\sh{g})$ and for any $(\sh{f}',\sh{g}')\in \coball{X}{X}{K}{W}\times \coball{X}{X}{\overline{W}}{U}$, $\sh{g}\circ\sh{f}\in \coball{X}{X}{K}{U}$.
\end{proof}

There is a more general version of the previous Lemma where the maps are not necessarily homeomorphisms, see Theorem 2.2 of \cite{dugundji}.

\lemm{\label{le:co_subbasis}
Let $X$ be a locally connected, locally compact Hausdorff space.
Then the sets $\coball{X}{X}{L}{U}$ where $L$ is compact, connected and has non empty interior, and $U$ is open, form a subbasis for the compact open topology.
}
\begin{proof}Again we follow the argument in \cite{arens}.
Let $\sh{h}\in \Topo^h(X,X)$.
We show that for any $\coball{X}{X}{K}{U}$ containing $\sh{h}$ where $K$ is compact and $U$ is open, there exists a 
subset of $\coball{X}{X}{K}{U}$  containing $\sh{h}$ of the form $\coball{X}{X}{L_1}{U}\cap\dots \cap \coball{X}{X}{L_n}{U}$ where each $L_i$ is compact, connected and has non empty interior.

Since $\sh{h}$ is continuous, for
each $x\in K$ we can find an open set $V(x)$ containing {$x$} such that { $\sh{h}(V(x))\subset U$.} 
{Since $X$ is locally compact } {and Hausdorff}{, we can then find another $V'(x)$, open in $X$, such that} \[{x \in V'(x)\subset \overline{V'(x)}\subset V(x)},\] {with $\overline{V'(x)}$ compact.}
Now since $X$ is locally connected, there exists a connected open set $V''(x)$ such that {$x \in V''(x)\subset {V'(x)}$.} 
Also $\overline{V''(x)}$ is compact, since $\overline{V''(x)}\subset {\overline{V'(x)}}$ and closed subsets of compact spaces are compact. {Furthermore $\overline{V''(x)}\subset V(x)$, so $\sh{h}(\overline{V''(x)}) \subset U.$}

The $V''(x)$ cover $K$ and so there exists a finite subcover by $V(x_i)$ for some finite set of $x_i\in K$ with $i\in \{1,\dots,n\}$.
Clearly:
\[
\sh{h}\in\bigcap_{i\in \{1,\dots,n\}}\coball{X}{X}{\overline{V''(x_i)}}{U}\subset \coball{X}{X}{K}{U}.
\]
\end{proof}

\begin{lemma}\label{le:inv_cont}
Let $X$ be a locally connected, locally compact Hausdorff space.
Then the inverse map
\ali{
(-)^{-1}\colon\TOPO^h(X,X)&\to \TOPO^h(X,X) \\
\sh{f}&\mapsto \sh{f}^{-1}
}
is continuous.
\end{lemma}
\begin{proof}
Throughout the  proof, only,   we will put $(-)^{-1}=T$. So $T\colon \TOPO^h(X,X)\to \TOPO^h(X,X)$ is the function such that $T(\sh{h})=\sh{h}^{-1}$.

By Lemma \ref{le:co_subbasis}, in order to prove that $T$ is continuous,  we only need to prove that the inverse images under $T$ of sets of the  form $\coball{X}{X}{L}{U}$, with $L$ compact, connected and with non-empty interior, and $U$ open, are open in $\TOPO^h(X,X)$.

Let $L\subset X$ be compact, connected, and with a non-empty interior. Let $U$ be open in $X$.
We show that for any $\sh{f}^{-1}\in \coball{X}{X}{L}{U}$, we can construct an open subset of $ \TOPO^h(X,X)$, containing $\sh{f}$, which is a subset of  $T^{-1}(\coball{X}{X}{L}{U})$.

Since $\sh{f}^{-1}$ is a homeomorphism, it sends compact subsets to compact subsets. So $\sh{f}^{-1}(L)$ is compact. Also $\sh{f}^{-1}(L) \subset U$, since $\sh{f}^{-1}\in \coball{X}{X}{L}{U}$. 

Using Lemma \ref{le:locally_compact}, we can choose an open set $V\subset X$ such that $f^{-1}(L)\subset V\subset\overline{V}\subset U$, with $\overline{V}$ compact, and then an open set $W\subset X$ with $\overline{V}\subset W\subset \overline{W}\subset U$, with $\overline{W}$ compact. In full:
\[ 
\sh{f}^{-1}(L)\subset V\subset\overline{V}\subset  W\subset \overline{W}\subset U.
\]
Therefore \[
\sh{f}\big ( (X\setminus V) \cap \overline{W}\big)=\big  (X\setminus \sh{f}(V)\big) \cap \sh{f}(\overline{W}) \subset (X\setminus L) \cap f(U).\] 

We can also choose an $x\in X$ such that $\sh {f}(x)\in int(L)$ (where $int(L)$ is the interior of $L$).
So there exists an open set (in $\TOPO^h(X,X)$):
\[
B_{XX}\left(\{x\},int(L)\right)\cap B_{XX}\left(
(X\setminus V) \cap \overline{W}, (X\setminus L) \cap \sh{f}(U)\right)
\]
containing $\sh{f}$, which we denote $U_0$. 
We claim that $U_0 \in  T^{-1}(\coball{X}{X}{L}{U}).$

Let $\sh{h}\in U_0$. We have:  $\sh{h}\left( (X\setminus V) \cap \overline{W}\right)\subset (X\setminus L) \cap \sh{f}(U)$.
Taking complements and reversing the inclusion we have 
\[L\cup (X\setminus \sh{f}(U))\subset \sh{h}(V\cup (X\setminus \overline{W}))
=\sh{h}(V)\cup \sh{h}(X\setminus \overline{W})
.\]
Now $\sh{h}(V)$ and $\sh{h}(X\setminus \overline{W})$ are disjoint open sets, and $L$ is connected\footnote{This is where the crucial fact that $L$ can be chosen to be connected is used.}, so either $L$ is   contained in $\sh{h}(V)$ or $L$ is contained in $\sh{h}(X\setminus \overline{W})$, but not both.  We claim that $L\subset \sh{h}(V)$. 

Note that since $\sh{h} \in B_{XX}\left(\{x\},int(L)\right)$, we have $\sh{h}(x) \in int(L)$. Since $\sh{f}(x) \in int(L)$, by construction, we have  $x\in \sh{f}^{-1}(int(L))\subset V.$ So $\sh{h}(x) \in \sh{h}(V).$ So $L\cap \sh{h}(V) \neq \emptyset$. So $L\subset \sh{h}(V)$.

Since $L\subset \sh{h}(V)$, we have  $\sh{h}^{-1}(L)\subset V\subset U$. Hence
$\sh{h}^{-1}\in \coball{X}{X}{L}{U}$.
\end{proof}

\begin{proof} (Of Theorem~\ref{le:top_group})
In Lemma \ref{le:comp_cont} we prove that the composition is continuous if $X$ is locally compact Hausdorff.
In Lemma \ref{le:inv_cont} we prove that the inverse map is continuous. 
\end{proof}

\end{document}